\newif\iffinal
  \finaltrue

\documentclass[12pt]{book}
\usepackage[a4paper,vmargin=3cm,centering]{geometry}

\input{macros}  % some useful commands

\iffinal
\else
  \IfFileExists{crop_pages}{
  \paperwidth 15.6cm
  \paperheight 25.5cm
  \oddsidemargin  -0.8in
  \evensidemargin -0.8in
  \topmargin      -1.3in
  }{}
\fi

\begin{document}

\frontmatter 
\pagenumbering{gobble}

\begin{titlepage}
\newgeometry{hmargin=1cm,vmargin=2cm,centering}
\begin{center}
\vspace*{2cm}
{\LARGE Dissertation}\\

\vspace{8mm}
{\Huge \bf Reactive Synthesis:}

\vspace{3mm}
%{\bf \huge {$\frak b$ranching~logic~{\!\LARGE\&\!}~parameter\emph{i}sed~systems}}
%{\bf \huge {$\frak{branching}$~logic~{\!\LARGE\&\!}~parameter\emph{i}sed~systems}}
%{\bf \huge {branc$\frak h$ing~logic~{\LARGE\&}~parameter\emph{i}sed~systems}}
{\bf \huge {branching~logic~{\LARGE\&}~parameter\emph{i}zed~systems}}
\vspace{8mm}

{\LARGE Ayrat Khalimov}

\vspace{1cm}

{\large
\begin{tabular}{rl}
Advisor:    & Roderick Bloem \\&Graz University of Technology, Austria\\[0.5ex]
Reviewer: & Sven Schewe \\&University of Liverpool, UK\\
Dean of Studies: & Denis Helic\\&Graz University of Technology, Austria\\
\end{tabular}}

\vfill

\begin{figure}[!ht]
\centerline{\includegraphics[width=3.3cm,keepaspectratio=true]{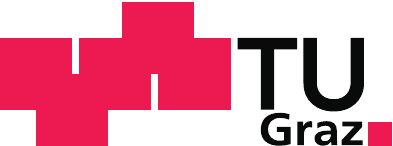}}
\end{figure}

{\large
Institute for Applied Information Processing and Communications\\
Graz University of Technology\\
A-8010 Graz, Austria\\
}

{\large Januar 2018}
\end{center}

%\noindent
%\underline{\hspace*{3cm}}\\
%{\footnotesize \copyright ~ Copyright 2017 by the author}

\end{titlepage}

\restoregeometry

\cleardoublepage

\pagestyle{plain}
\pagenumbering{roman}

%\vspace*{0.2cm}

\begin{center}
{\Large\bfseries Acknowledgements}
\end{center}
This work would not be possible without RiSE network (established by Roderick Bloem and \boxed{\text{Helmut Veith}}).
I stumbled upon the poster with the PhD position by chance,
during a relaxed walk at EPFL where I was doing an internship.\\
I am grateful to my advisor Roderick Bloem,
who honestly answered my questions, patiently directed me by asking questions and pitching ideas,
and who always listened.
Sasha Rubin showed how to be rigid and develop theories,
with Swen Jacobs we wondered a lot around parameterised synthesis bouncing the token from each other,
and Sven Schewe convinced me that tree automata are easy-peasy.\\
My colleagues Robert K\"onighofer, Georg Hofferek, and Bettina K\"onighofer helped
in the initial integration and changed my attitude towards people outside of Russia.
Our secretaries Martina Piewald, Melanie Blauensteiner, Ursula Urwanisch, and Angelika Wagner
enabled me to focus on my work without administrative distractions.\\
Dedicated to my grandfather Rauf and our large family.

\cleardoublepage

%\vspace{5cm}
\begin{center}
{\Large\bfseries Statutory Declaration}
\end{center}
\vspace{5mm}
\noindent
I declare that I have authored this thesis independently, that I have
not used other than the declared sources/resources, and that I have
explicitly marked all material which has been quoted either literally
or by content from the used sources.

\vspace{2cm}

\noindent
\begin{tabular}{ccc}
\hspace*{5cm}     & \hspace*{1.5cm}   & \hspace*{5.7cm}\\
\dotfill          &                 & \dotfill\\
place, date       &                 & signature\\
\end{tabular}

% --- English Abstract --------------------------------------------------

\cleardoublepage

%\begin{changemargin}{1.4cm}{1.4cm}

\begin{center}
{\Large\bfseries Abstract}
\end{center}
%\vspace*{1mm}
Reactive synthesis is an automatic way
to translate a human intention expressed in some logic
into a system of some kind.
This thesis has two parts, devoted to logic and to systems.

\parbf{Part I}
In 1963 Alonzo Church introduced the synthesis problem~\cite{Church63}
for specifications in monadic second-order logic.
%Later Amir Pnueli introduced linear temporal logic (LTL)~\cite{pnueli1977temporal}.
Nowadays most model checkers and synthesizers use linear temporal logic (LTL)~\cite{pnueli1977temporal}.
LTL reasons about system runs in a linear fashion.
With LTL we can ask ``does a run reach a particular state?'' or
``does a run visits a particular state infinitely often?''.
LTL is linear in its nature, leaving the designer without structural properties,
which are expressible in computation tree logic (CTL) and its generalization \CTLstar~\cite{ctl-origin,ctlstar-origin}.
%\CTLstar and CTL were introduced by Emerson and Halpern and Clarke~\cite{ctl-origin,ctlstar-origin}.
With \CTLstar we can ask
``does a run never visit a particular state but it has a \emph{possibility}
  to reach it?''.
Such properties are important---they allow for fine-tuning the system structure.

In Part I, we develop two new approaches to \CTLstar synthesis.
The first approach is an extension (actually, two) of
the SMT-based bounded synthesis~\cite{BS}.
We describe two extensions:
one follows bottom-up \CTLstar model checking,
another one follows the automata framework~\cite{ATA}.
Then we develop the approach that reduces \CTLstar synthesis to LTL synthesis.
The reduction turns any LTL synthesiser into a \CTLstar synthesiser.
The approaches were implemented and are available online.

\parbf{Part II}
Modern systems become more and more distributed.
Such distributed systems are typically parameterized by the number of processes:
they should work for any number of processes.
The parameterized synthesis problem~\cite{JB14} asks,
given a parameterized specification, to find a process template,
that can be cloned to form a correctly behaving system of any size.
At the core of the method is the cutoff reduction technique:
reduce reasoning about systems with an arbitrary number of processes
to reasoning about systems of a fixed cutoff size.
The intrinsic parameter, hidden in the parameterized synthesis problem,
is how the processes are connected and how they communicate,
i.e., the system architecture.

In Part II, we study parameterized synthesis for two system architectures.
The first architecture is guarded systems~\cite{EmersonK03}
and is inspired by cache coherence protocols.
In guarded systems,
processes transitions are enabled or disabled depending on the existence of other processes in certain local states.
The existing cutoff results~\cite{EmersonK03} for guarded protocols
are restricted to closed systems, and are of limited use for liveness properties.
We close these gaps and prove tight cutoffs for open systems
with liveness properties, and also cutoffs for detecting deadlocks.

The second architecture is token-ring systems~\cite{Emerso03},
where the single token circulates processes arranged in a ring.
The experiments with the existing parameterized synthesis method~\cite{JB14}
showed that it does not scale to large specifications.
First, we optimize the method by refining the cutoff reduction,
using modularity and abstraction.
The evaluation show several orders of magnitude speed-ups.
Second, we perform parameterized synthesis case study on the industrial arbiter protocol AMBA~\cite{AMBAspec}.
We describe new tricks%
---a new cutoff extension and decompositional synthesis---%
that, together with the previously described optimizations, allowed us to synthesize AMBA
in a parameterized setting, for the first time.

       % Title, Abstract, Pledge, Ack

%%%%
%\phantomsection
%\addcontentsline{toc}{chapter}{Contents}
\tableofcontents

%\listoffigures
%\addcontentsline{toc}{chapter}{List of Figures}

%\listoftables
%\addcontentsline{toc}{chapter}{List of Tables}

%%%%

\mainmatter

\chapter{Introduction}

%\parbf{What is reactive synthesis?}
{\small
\noindent
\li
\-[Dave:~] Hey, Elli, how can I calculate the week number from the date?
\vspace{-3mm}
\-[Elli:~] ...prints the C-function.
\vspace{-3mm}
\-[Dave:~] Great. Can this function also output three last requested dates?
\vspace{-3mm}
\-[Elli:~] ...prints another C-function.
\vspace{-3mm}
\-[Dave:~] Thanks. Can you also make it output the name of the requesting person?
\vspace{-3mm}
\-[Elli:~] I am sorry, Dave, I am afraid I can't do that.
\il
}

In synthesis,
we describe the required behaviour and ask the computer to find the solution
with such a behaviour.
(In the dialog above,
 Dave asks Elli to find a function that,
 given a date, outputs the week number to which the date belongs.)
In \emph{reactive} synthesis, we are interested not in simple ``do-and-forget'' functions,
but rather in functions that interact with the user
akin to functions with an internal state.
(In the dialog above, the second C-function is reactive.)
It is not always possible to find a solution,
in which case the synthesizer (Elli) outputs ``specification is unrealizable''.
(In the last dialogue request, the specification became unrealizable,
 because the person name is not available to the function to be synthesized.)
%To summarise, reactive synthesis is an automatic way
%to translate a human intention expressed in some language
%into a system of some kind.

%\parbf{What is reactive synthesis problem?}

In 1963, Alonzo Church introduced the reactive synthesis problem~\cite{Church63}:
given a formula in Monadic Second Order Logic of One Successor,
and the inputs and the outputs of a circuit,
find such a circuit such that all behaviors of the circuit satisfy the formula.
(The circuit behaviour is an infinite string of inputs combined with the outputs.)
Church's problem was solved by Rabin~\cite{Rabin69} and
by B\"uchi and Landweber~\cite{BL69} in 1969.\ak{how?}

%\parbf{Recent progress of reactive synthesis}

Recent research in reactive synthesis focused on specifications
given in Linear Temporal Logic (LTL),
introduced by Pnueli~\cite{pnueli1977temporal} in 1977.
LTL has temporal operators, like $\G$ (always) and $\F$ (eventually),
and allows one to state properties like ``every request is eventually granted'':
$\G(r \impl \F g)$.
A system satisfies a given LTL property if all its computations satisfy it.
Pnueli and Rosner proved~\cite{DBLP:conf/popl/PnueliR89}
that the LTL synthesis problem is 2EXPTIME-complete.
Their approach translates a given LTL formula into a nondeterministic B\"uchi automaton,
then determinises it into a deterministic parity automaton
with the aid of involved Safra construction~\cite{Safra},
turns the automaton into a game, and solves the game.
Recent research focused on how to overcome the high complexity and Safra construction:
the work~\cite{Bloem12} considered the synthesis for a subset of LTL called GR(1),
the work~\cite{BS,KupfermanV05} considered bounding the system size and gave a name to Bounded Synthesis,
by combining the previous bounding with efficient data structures---Anti-chains Synthesis~\cite{Filiot11}.
The SYNTCOMP competition~\cite{syntcomp} is another recent initiative
with the goal to advance efficient synthesisers and popularise reactive synthesis.

%\parbf{The issues with reactive synthesis}

Despite substantial progress,
reactive synthesis is not as widespread as model checking.
The major reason, I believe,
is that writing the specifications---especially \emph{complete} specifications---is hard.
The issue is less pronounced in model checking,
because we do not need all the properties,
only those to model check.

%\parbf{In light of this issue, two directions to proceed}

In light of this issue, there are two directions to proceed.
First, we can develop synthesis approaches for richer logics,
which can ease writing the specifications.
Second, we can find application contexts where high specification costs are acceptable.
This thesis targets both directions:
we develop new synthesis approaches for the logic called \CTLstar,
and we delve into synthesis of distributed algorithms.

\subsection*{Part I: Excursion into Branching Logic}

%\parbf{when was CTL introduced?}

Computation Tree Logic (CTL)~\cite{ctl-origin}
was introduced by Emerson and Clarke in 1981
to circumvent the high complexity (PSPACE-complete) of the LTL model checking problem
and to be able to specify structural properties.
In 1986 Emerson and Halpern introduced a generalization,
Computation Tree Star Logic (\CTLstar)~\cite{ctlstar-origin},
that subsumes both CTL and LTL.

%\parbf{what is \CTLstar?}

In contrast to LTL, which reasons about (linear) computation runs,
\CTLstar reasons about (branching) computation trees.
We can get such a tree by unfolding the system transition structure.
\CTLstar has---in addition to temporal operators---path quantifiers:
$\A$ (on all paths) and $\E$ (there exists a path).
Such path quantifiers allow us to reason about branching structure of trees,
not just about their ``linear'' paths.
For example, \CTLstar formula ``$\AGEF reset$'' says:
``on all tree paths, from every tree node,
  there should be a path into a node where `reset' holds''.
We cannot express such a property using LTL alone.

%\parbf{what advantages does CTL* have over LTL?}

Despite \CTLstar being more expressible than LTL,
the complexity of \CTLstar synthesis (2EXPTIME-complete)
stays the same.
This prompted us to look into approaches to \CTLstar synthesis.

%\parbf{what is the standard solution to CTL*?}

The standard solution~\cite{informatio} to \CTLstar synthesis
turns the \CTLstar formula into an alternating hesitant tree automaton,
removes nondeterminism and derives a universal co-B\"uchi tree automaton,
determinises it using Safra construction~\cite{Safra} into a parity tree automaton,
and, finally, checks its non-emptiness.
If it is empty, then the specification is unrealisable,
otherwise we can extract the system from the proof of the non-emptiness.
This approach is hard to implement correctly and efficiently,
due to the involved Safra construction%
\footnote{It was a common belief that the Safra construction
  is difficult to implement and results in impractical algorithms.
  However, the belief might be wrong,
  as SYNTCOMP~\cite{syntcomp} in 2017 showed:
  the LTL synthesiser {\tt ltl-synt} that used Safra construction performed very well.}.

%\parbf{what is our contribution to CTL* synthesis?}

Part I contribution is two practical approaches to \CTLstar synthesis.

\subsubsection*{Contribution I.1: \CTLstar Bounded Synthesis}

We developed two bounded synthesis approaches for the \CTLstar specifications.
Let us recall how the SMT-based bounded synthesis by Schewe and Finkbeiner~\cite{BS} works:
we bound the system size, and
encode the resulting synthesis problem into an SMT query\footnote{%
  Satisfiability Modulo Theory (SMT)~\cite{SMT} query is a
  set of constraints over in a given theory.
  For example, in Linear Integer Arithmetic theory,
  the constraints talk about integer variables,
  use operations plus, minus, and the comparison relations.
  Such a query asks whether there are values for integer variables
  that make the constraint true.}.
The query encodes the model checking question:
whether a system---which is yet unknown---is accepted by the automaton.
Bounding the system size makes it possible to encode such a model checking query
into an SMT query.
To solve such a query,
an SMT solver efficiently enumerates every possible system of a given size,
and checks if it is correct.
Thus, if the SMT query is satisfiable, then we extract the system (of the given size),
otherwise increase the system size and repeat.
The loop stops when the bound on the system size%
---provided by the user or from the theory---is reached.

Our first bounded synthesiser for \CTLstar
resembles bottom-up \CTLstar model checking~\cite{PrinciplesMC}:
it introduces an atom for each subformula of the \CTLstar formula,
and encodes into an SMT query whether the atom holds in a system state,
for every state.
We also require the top-level atom, representing the whole \CTLstar formula,
to hold in the initial system state,
Hence, if the SMT query is satisfiable,
then there is a system of the given size,
which satisfies the \CTLstar formula.
Otherwise, increase the system size and repeat.

Our second bounded synthesiser for \CTLstar uses the automata framework~\cite{ATA}:
translate the \CTLstar formula into an alternating hesitant automaton,
then encode into an SMT query
whether there is a system of a given size that is accepted by the automaton.
Conceptually, the approach is the same as the previous one,
except that we do not introduce atoms for subformulas explicitly
and instead use their automata representation.

The results constitute Chapter~\ref{chap:bosy:ctlstar} and were published in:
\li
\-[\cite{CTLstarCAV}]
   \emph{Bounded Synthesis for Streett, Rabin, and \CTLstar},
   by Ayrat Khalimov and Roderick Bloem,
   at CAV conference, 2017
\il

\subsubsection*{Contribution I.2: \CTLstar-via-LTL Synthesis}

We reduce synthesis for \CTLstar properties to synthesis for LTL.
In the context of model checking this is impossible%
---\CTLstar is more expressive than LTL.
Yet, in synthesis we have knowledge of the system structure
\emph{and} we can add new outputs.
These outputs can be used to encode witnesses of
the satisfaction of \CTLstar subformulas directly into the system.
This way, we construct an LTL formula, over old and new outputs and original inputs,
which is realisable if, and only if, the original \CTLstar formula is realisable.
The \CTLstar-via-LTL synthesis approach preserves the problem complexity,
although it might produce systems that are larger than necessary.
Furthermore,
the approach directly benefits from the performance advances of LTL synthesisers.
The results constitute Chapter~\ref{chap:ctl-via-ltl} and were published in:
\li
\-[\cite{CTLsynt-via-LTLsynt}]
  \emph{\CTLstar Synthesis via LTL Synthesis},
  by Roderick Bloem and Sven Schewe and Ayrat Khalimov,
  at SYNT workshop, 2017
\il

\subsection*{Part II: Excursion into Parameterized Systems}

%\parbf{why do we consider parameterized systems?}

Modern systems become more and more distributed.
Distributed systems are hard to implement and even harder to debug.
Yet, the failure of such systems may be unacceptable.
Thus, substantial efforts are devoted to ensure the correctness of distributed systems.
In Part II,
we look into the hard task of automatic synthesis of distributed parameterized systems.

%\parbf{what is the parametrized synthesis problem?}

Most distributed systems, algorithms, and data structures are \emph{parameterized}:
they should work for a varied, not a priori fixed, number of the components.
%We call such systems parameterized (by the number of components).
The parameterized synthesis problem~\cite{JB14} asks, given
a parameterized specification, to find a process template,
that can be cloned to form a correctly behaving system of any size.
An example parameterized specification is:
\[ \begin{array}{ll}
  \forall i \neq j.~ & \G \neg ( g_i \land g_j ) \land \\
  \forall i.~ & \G (r_i \impl \F g_i).
  \end{array}
\]
The synthesizer should find a process template, having input $r$ and output $g$,
such that a system composed of any number of such processes,
satisfies the above specification.
The related question is that of parametrized model checking
where the process template is given.
The intrinsic parameter,
hidden in the parameterized synthesis problem,
is how the processes are connected and how they communicate,
i.e., the system architecture.
The survey of existing cutoff and decidability results
for many different system architectures can be found in~\cite{BloemETAL15}.
We focus on two system architectures: guarded systems and token-ring systems.

%\parbf{what is the standard solution to parameterized synthesis problem?}

A common approach to solve the parameterized synthesis and model checking problems
is to use the cutoff reduction~\cite{Emerso03}:
reduce reasoning about systems with an arbitrary number of processes
to reasoning about systems of a fixed cutoff size.
For example,
if we consider the parameterized specification mentioned above and token-ring systems,
then it is enough to consider a system with 4 processes:
if it is correct, then any larger system is correct.

%\parbf{what is our contribution?}

\subsubsection*{Contribution II.1: Cutoffs for Parameterized Guarded Systems}

Guarded systems~\cite{EmersonK03} are inspired by cache coherence protocols found
in most modern processors.
A cache coherence protocol is usually described by states,
where transitions between states happen depending on whether or not
there is a processor in a particular state.
I.e., the transitions are guarded.
Inspired by this, in guarded systems,
processes transitions are enabled or disabled depending
on the existence of other processes in certain local states.
Our contribution concerns both parameterized synthesis and parameterized verification.
Our work stems from the observation that
existing cutoff results for guarded systems
(i) are restricted to closed systems, and
(ii) are of limited use for liveness properties
because reductions do not preserve fairness.
We close these gaps and obtain new cutoff results for open systems with
liveness properties under fairness assumptions.
Furthermore, we obtain cutoffs for the detecting deadlocks,
which are of paramount importance in synthesis.
Finally, we prove tightness or asymptotic tightness for the new cutoffs.
The results constitute Chapter~\ref{chap:guarded-systems}
and were published in:
\li
\-[\cite{AJK16}]
   \emph{Tight Cutoffs for Guarded Protocols with Fairness},
   by Simon Au{\ss}erlechner and Swen Jacobs and Ayrat Khalimov,
   at VMCAI conference, 2016
\il

%\parbf{what is the problem with the existing parameterized synthesis?}

\subsubsection*{Contribution II.2: Case Study of Parameterized Token-ring AMBA}

In token-ring systems, a single token circulates in the system.
A process possessing the token knows that no other process has the token.
Based on this information, the process can, for example, raise the grant signal.
If all processes raise the grant only when they posses the token,
then the grants will be mutually exclusive.
Thus, the token serves as the resource token.

The experiments with the existing parameterized synthesis method~\cite{JB14}
showed that it does not scale to large specifications.
First, we optimize the method by refining the cutoff reduction.
The experiments show speed-ups of several orders of magnitude.
Second, we perform parameterized synthesis case study
on the industrial arbiter protocol AMBA~\cite{AMBAspec}.
We describe new cutoff extension and decompositional synthesis tailored to AMBA
that, together with the previously mentioned optimizations,
allowed us to synthesize AMBA in parameterized setting, for the first time.
The results constitute Chapter~\ref{chap:token-systems}
and were published in:
\li
\-[\cite{Khalimov13}]
   \emph{Towards Efficient Parameterized Synthesis},
   by Ayrat Khalimov and Swen Jacobs and Roderick Bloem,
   at VMCAI conference, 2013
\-[\cite{party}]
   \emph{PARTY: Parameterized Synthesis of Token Rings},
   by Ayrat Khalimov and Swen Jacobs and Roderick Bloem,
   at CAV conference, 2013
\-[\cite{BJK14}]
   \emph{Parameterized Synthesis Case Study: AMBA AHB},
   by Ayrat Khalimov and Swen Jacobs and Roderick Bloem,
   at SYNT workshop, 2014
\il

\subsection*{Other Results}

Here are the results that did not make their way into the thesis:
\li
\-[\cite{BloemETAL15}]
   \emph{Decidability of Parameterized Verification},
   book of 170 pages,
   by Roderick Bloem and
               Swen Jacobs and
               Ayrat Khalimov and
               Igor Konnov and
               Sasha Rubin and
               Helmut Veith and
               Josef Widder.\\
In this book we consider the important case of systems parameterized by the number of processes in the system
and where each process is independent of that number.
The literature in this area produced a wealth of computational models for systems based on token passing,
broadcast communication, guarded transitions, and other communication primitives.
We introduce a computational model that unites the central synchronization and
communication primitives of many models.
We survey existing decidability and undecidability results,
and provide a systematic overview of the basic problems in this research area.

\-[\cite{DBLP:journals/sigact/BloemJKKRVW16}]
   \emph{Decidability in Parameterized Verification},
   the journal version of the above book; appeared in SIGACT News in 2016.

\-[\cite{DBLP:journals/corr/Khalimov16}]
   \emph{Specification Format for Reactive Synthesis Problems},
   by Ayrat Khalimov,
   at SYNT workshop, 2015.\\
To do synthesis, we need a specification.
Writing specifications is hard.
In this paper, we propose a user-friendly format to ease
the specification work, in particularly, that of specifying partial implementations.
Also, we provide scripts to convert specifications in the new format into the SYNTCOMP format,
thus benefiting from state of the art synthesizers.

\-[\cite{AJKR14}]
   \emph{Parameterized Model Checking of Token-Passing Systems},
   by Benjamin Aminof and
               Swen Jacobs and
               Ayrat Khalimov and
               Sasha Rubin,
   at VMCAI conference, 2014.\\
In this paper, we revisit the parameterized model checking problem for token-passing
systems and specifications in indexed $\CTLstarmX$.
%In foundational work, Emerson and Namjoshi (1995, 2003) showed that
%parameterized model checking of indexed \CTLstarmX in uni-directional token rings
%can be reduced to checking rings up to some cutoff size.
%Clarke et al. (2004) showed a similar result for general topologies and indexed \LTLmX,
%provided processes cannot choose the directions for sending or receiving the token.
We unify and substantially extend the results of Emerson and Namjoshi~\cite{Emerso95b,Emerso03}
and Clarke et al.~\cite{Clarke04c}
by systematically exploring fragments of indexed \CTLstarmX with respect to general network topologies.
For each fragment we establish whether a cutoff exists,
and for some concrete topologies, such as rings, cliques and stars, we infer small cutoffs.
Finally, we show that the problem becomes undecidable, and thus no cutoffs exist,
if processes are allowed to choose the directions in which they send or from which they receive the token.

\-[\cite{2017arXiv171204291K}]
  \emph{OpenSEA: Semi-Formal Methods for Soft Error Analysis},
  by Patrick Klampfl and Robert K\"onighofer and Roderick Bloem and Ayrat Khalimov and
  Aiman Abu-Yonis  and Shiri Moran,
  on arxiv, 2017.\\
Due to alpha-particles and cosmic rays, modern circuits are prone to bit flips.
To alleviate the problem, designers develop protection circuits,
but they are hard to implement right.
This leads to bugs: an undetected fault can bring miscalculations,
the protection that alarms about harmless faults incurs performance penalty.
In this paper, we use formal methods on designer’s input tests, while keeping time-location open.
This idea is at the core of the tool OpenSEA.
OpenSEA can
(i) find latches vulnerable to and protected against faults,
(ii) find tests that exhibit checker false alarms,
(iii) use fixed and open inputs, and
(iv) use environment assumptions.
Evaluation on a number of industrial designs shows that OpenSEA produces valuable results.
\il

\part[Excursion into Branching Logic]{Excursion Into \\ Branching Logic \\ \ \\ \LARGE Approaches to \CTLstar Synthesis}
%\part{Two Approaches to Synthesis from \CTLstar}

\section*{Overview of Part I}\label{chap:ctlstar:overview}

The reactive synthesis problem was introduced by Alonzo Church~\cite{Church63}.
Given a specification as a formula in Monadic Second Order Logic of One Successor (MSO),
the question is to produce a circuit such that \emph{all} its behaviors satisfy the formula.
Later Pnueli introduced Linear Temporal Logic (LTL)~\cite{pnueli1977temporal}
and together with Rosner solved the synthesis problem for LTL~\cite{DBLP:conf/popl/PnueliR89}.
Now LTL is the main basic logic for specifications.
Both these logics, MSO and LTL, are \emph{linear}:
they describe the set of behaviours,
but do not allow for specifying \emph{structural} properties of the systems.

To be able to specify structural properties
(and to circumvent a relatively high complexity of the verification wrt.\ LTL),
Emerson and Clarke introduced Computation Tree Logic (CTL)~\cite{ctl-origin}.
Later Emerson and Halpern introduced Computation Tree Star Logic (\CTLstar)~\cite{ctlstar-origin}
that subsumed both CTL and LTL.

Let us briefly compare LTL and \CTLstar.

LTL reasons about \emph{computations}.
The logic has \emph{temporal} operators, e.g., $\G$ (always) and $\F$ (eventually),
and can describe properties like ``every request is eventually granted'':
$\G(r \impl \F g)$.
A system satisfies such an LTL property iff \emph{all} its computations satisfy it.
Thus a system is characterized by its computations.

In contrast, \CTLstar reasons about computation \emph{trees}.
Thus, a system is viewed as a tree (cf.\ set of linear paths for LTL),
and we can get such a tree by unfolding the system.
\CTLstar has---in addition to temporal operators---\emph{path quantifiers}:
$\A$ (on all paths) and $\E$ (there exists a path).
Such path quantifiers allow us to reason about branching structure of trees,
not just about the set of its ``linear'' paths.
For example, the \CTLstar formula ``$\AGEF reset$'' says:
``on all tree paths, from every tree node,
  there should be a path into a node where `reset' holds''.
We cannot express such a property using LTL alone.

This part of the thesis explores synthesis approaches from properties in \CTLstar.
It consists of two chapters.

In Chapter~\ref{chap:bosy:ctlstar}
we introduce two approaches to synthesis from \CTLstar.
Both approaches follow the Bounded Synthesis approach
introduced by Finkbeiner and Schewe~\cite{BS}.
In Bounded Synthesis, we repeatedly search for a system of increasing sizes,
until we find a solution.
Bounded Synthesis is very flexible and can be easily adapted
to do e.g. distributed synthesis.
We extend Bounded Synthesis to specifications in \CTLstar and beyond.

The disadvantage of Bounded Synthesis is that it is susceptible to system size:
it works well when the specification admits a small implementation,
but less well when no small implementation exists.
The same holds for our \CTLstar Bounded Synthesis.

In Chapter~\ref{chap:ctl-via-ltl}, partly to overcome this disadvantage,
we introduce a reduction of the \CTLstar synthesis problem to the LTL synthesis problem.
After applying the reduction, any \emph{LTL} synthesiser can do $CTL^*$ synthesis.
Notice that for model checking such a reduction is impossible---%
\CTLstar is more expressive than LTL.
Yet, in synthesis we control the system structure,
which enables the reduction.
The \CTLstar-via-LTL synthesis approach preserves the problem complexity,
although it might increase the size of a system.

The approaches differ in how they ensure the satisfaction of existential \CTLstar subformulas
(recall that universal \CTLstar subformulas, just like LTL, talk about system paths as a whole,
 while existential \CTLstar subformulas specify the existence of a system path).
Recall from Section~\ref{defs:bounded_synthesis} that
bounded synthesis encodes the LTL synthesis problem into the SMT satisfaction problem.
The SMT constraints annotate the states of a \emph{product}
(of a yet unknown system with an automaton expressing a given \CTLstar formula)
with information that ensures that all lassos in the product are not ``bad'' (for universal subformulas)
and that there are ``good'' lassos (for existential subformulas).
In contrast, \CTLstar-via-LTL synthesis produces an LTL formula that
talks about \emph{system} paths and has no direct access to the product.
Hence we move annotations into a system which may increase its size.

This thesis part is organized as follows.
In the next Chapter~\ref{chap:defs} we introduce the definitions
which are used in both chapters.
Chapter~\ref{chap:bosy:ctlstar} focuses on extensions of Bounded Synthesis to \CTLstar,
while Chapter~\ref{chap:ctl-via-ltl} describes the \CTLstar-to-LTL synthesis reduction.
Both chapters depend on the definitions section,
but are independent of each other.

\newcommand{\nocontentsline}[3]{}
\newcommand{\toclesslab}[3]{\bgroup\let\addcontentsline=\nocontentsline#1{#2\label{#3}}\egroup}

%\toclesslab\section{Preliminaries}{chap:defs}
\chapter{Common Definitions for Part I}\label{chap:defs}
Notation:
$\bbB = \{\true,\false\}$ is the set of Boolean values,
$\bbN$ is the set of natural numbers (excluding $0$),
$\bbN_0 = \bbN\cup\{0\}$,
$[k]$ is the set $\{i \in \bbN \| i \leq k\}$
and $[0,k]$ is the set $[k] \cup \{0\}$ for $k \in \bbN$.

The powerset of $A$ is denoted by $2^A$.
We often write $(a,x)$ instead of $a \cup x$ (that is from $2^{A \cup X}$),
and $a \cup x$ instead of $(a,x)$ (that is from $2^A \times 2^X$),
when $a \in 2^A$, $x \in 2^X$ and $A \cap X = \emptyset$.

We denote substitution by the symbol $\mapsto$.
E.g., $(a \land b) [a \mapsto x]$ is $x \land b$.

All systems and automata are finite,
paths are infinite,
and trees have only infinitely long paths but are finitely-branching---%
unless explicitly stated.

\toclesslab\section{Moore Systems}{defs:moore-systems}
%\section{Moore Systems}\label{defs:moore-systems}

A \emph{(Moore) system} $M$ is a tuple
$(I, O, T, t_0, \tau, out)$
where
$I$ and $O$ are disjoint sets of input and output variables,
$T$ is the set of states, $t_0 \in T$ is the initial state,
$\tau: T \times 2^I \to T$ is a transition function,
$out: T \to 2^O$ is the output function that
labels each state with a set of output variables.
Note that systems have no dead ends and have a transition for every input.
We write $t \trans{io} t'$ when $t' = \tau(t,i)$ and $out(t) = o$.
We abuse the notation and define $\tau(t,w)$ for $w_1 w_2 ... w_n \in (2^I)^+$
to be the system state $t_n$ such that $t_0 \trans{io_0} t_1 \trans{io_1} ... \trans{io_{n-1}} t_n$,
i.e., $\tau(t,w)$ is the state where the system ends after reading the word $w$,
when starting from the initial state.

A \emph{system path} is a sequence $t_1 t_2 ... \in T^\omega$
such that for every $i\in \bbN$ there is $e \in 2^I$ with $\tau(t_i,e) = t_{i+1}$.
An \emph{input-labeled system path} is a sequence $(t_1,e_1) (t_2,e_2) ... \in (T\times 2^I)^\omega$
where $\tau(t_i,e_i) = t_{i+1}$ for every $i\in \bbN$.
We sometimes use notation $t_1 \trans{e_1} t_2 \trans{e_2} t_3 ...$
to describe the input-labeled system path $(t_1,e_1) (t_2,e_2) ...$.
A \emph{system computation starting from $t_1 \in T$} is a sequence $(o_1\cup e_1) (o_2\cup e_2) ... \in (2^{I\cup O})^\omega$
for which there exists an input-labeled system path $(t_1,e_1) (t_2,e_2) ...$ 
and $o_i=out(t_i)$ for every $i \in \bbN$.
We write \emph{system computation} to mean system computation starting from the initial state.
Note that since systems are Moore,
the output $o_i$ cannot ``react'' to input $e_i$---%
the outputs are ``delayed'' with respect to inputs.

\begin{remark}\label{rem:inputs-shift}
There are two ways to group inputs and outputs into computations.
The first way is to introduce an initial transition $\tau_I: 2^I \to T$
instead of using the initial state $t_0$.
Then the input-labeled system path $\trans{e_1} t_1 \trans{e_2} t_2 \trans{e_3} t_3 ...$
corresponds to the computation $(e_1, out(t_1)) (e_2, out(t_2)) (e_3, out(t_3)) ...$.
Another way is to avoid using the initial transition---use the initial state $t_0$ instead---and ``shift'' inputs and outputs.
Then an input-labeled system path $t_0 \trans{e_1} t_1 \trans{e_2} t_2 ...$
corresponds to the computation $(out(t_0), e_1) (out(t_1), e_2) ...$.
We use the second approach.
\end{remark}

\ak{add example}

\toclesslab\section{Trees}{defs:trees}
%\section{Trees}\label{defs:trees}

A \emph{(infinite) tree} is a tuple $(D, L, V \subseteq D^*, l:V \to L)$,
where
\li
\- $D$ is the set of directions (in our case, finite),
\- $L$ is the set of node labels (in our case, finite),
\- $V$ is the (infinite) set of nodes satisfying:
   (i) $\epsilon \in V$ is called the root (the empty sequence),
  (ii) $V$ is closed under prefix operation (i.e., every node is connected to the root),
 (iii) for every $n \in V$ there exists a $d \in D$ such that $n\cdot d \in V$
       (i.e., there are no leafs),
\- $l$ is the node labeling function.
\il
A tree $(D,L,V,l)$ is \emph{exhaustive} iff $V=D^*$.
A tree is \emph{non-labeled} iff $|L|=1$ and then we omit $L$ and $l$.

A \emph{tree path} is a sequence $n_1 n_2 ... \in V^\omega$,
such that, for every $i$, there is $d \in D$ such that $n_{i+1} = n_i \cdot d$.

%An \emph{$L$-labeled $D$-directed tree} is a tuple $(V,l)$, where
%$V=D^*$ is the (infinite) set of tree nodes
%and $\epsilon \in V$ (the empty sequence) is called the root node,
%$l: V \to L$ is a labeling function.
%An \emph{(infinite) tree path} is a sequence $n_1 n_2 ... \in V^\omega$,
%such that, for every $i$, there is $d \in D$ and $n_{i+1} = n_i \cdot d$.

In contexts where $I$ and $O$ are inputs and outputs,
we call an exhaustive tree $(D,L,V,l)$
a \emph{computation tree},
where $D=2^I$, $L=2^O$, $V=D^*$, and $l:V \to \O$.
We omit $D$ and $L$ when they are clear from the context.

With every system $M=(I, O, T, t_0, \tau, out)$ we associate
the computation tree $(D, L, V, l)$ such that, for every $n\in V$:
$l(n)=out(\tau(t_0,n))$.
We call such a tree a \emph{system computation tree}.

A computation tree is \emph{regular}
iff it is a system computation tree for some (finite) system.

\ak{example}

\section*{Two Views on the System}\label{defs:two-views-on-system}

Later we introduce logics \CTLstar and LTL to distinguish correct from buggy systems.
The two logics look at systems from two sides.

On one side,
we can associate with a system $M$
a set of its computations $b(M) \subseteq (2^{I\cup O})^\omega$.
A formula $\varphi$ in Linear Temporal Logic (LTL) (introduced later)
describes a set of infinite words $L(\varphi)$.
Thus, we can use an LTL formula to specify all correct computations.
Then a system $M$ is correct wrt.\ LTL formula $\varphi$
iff $b(M) \subseteq L(\varphi)$,
i.e., all system computations satisfy $\varphi$.

On the other side, we might want to specify structural properties of systems.
E.g.,
whether from every system state
we can branch into a state satisfying $p$ and
we can branch into a state satisfying $\neg p$.
In this case, characterizing a system by its set of computations---e.g. using LTL---is not possible.
Instead, we associate with a system its computation tree.
A formula $\Phi$ in Computation Tree Logic (defined later)
describes a set of computation trees $L(\Phi)$.
Thus, we can use such a formula to describe a set of all correct computation trees.
Then a system $M$ is correct wrt.\ $\Phi$ iff $(V,l) \in L(\Phi)$,
i.e., the system computation tree satisfies $\Phi$.

\toclesslab\section{Logics: \CTLstar with Inputs and LTL}{defs:ctlstar}
%\section{Logics: \CTLstar with Inputs and LTL}\label{defs:ctlstar}

\subsection*{\CTLstar with inputs (release PNF)}

Fix two disjoint sets: inputs $I$ and outputs $O$.
Below we define \CTLstar with inputs, in release positive normal form%
\footnote{This form is sometimes called negation normal form.
  For the name, we follow~\cite{PrinciplesMC}.
  Note that without the release operator $\R$%
  ---the dual of the until operator $\U$---%
  the logic is less expressive due to the restriction on negations.
  That explains the name ``release PNF''.}.
The definition differentiates inputs and outputs (see Remark~\ref{rem:ctlstar-subtle}).

\parbf{Syntax}
\emph{State formulas} have the grammar:
$$
\Phi = \true \| \false \|
       o \| \neg o \| \Phi \land \Phi \| \Phi \lor \Phi \|
       \A \varphi \| \E \varphi
$$
where $o \in O$ and $\varphi$ is a path formula. \emph{Path formulas} are defined by the grammar:
$$
\varphi = \Phi \|
      i \| \neg i \|
      \varphi \land \varphi \| \varphi \lor \varphi \|
      \X \varphi \|
      \varphi \U \varphi \|
      \varphi \R \varphi,
$$
where $i \in I$.
The temporal operators $\G$ and $\F$ are defined as usual.

The above grammar describes the \CTLstar formulas in positive normal form.
The general \CTLstar formula
(in which negations can appear anywhere)
can be converted into the formula of this form with no size blowup,
using the equivalence $\neg (a \U b) \equiv \neg a \R \neg b$ and some others.

\parbf{Semantics}
We define the semantics of \CTLstar with respect to a computation tree $(V,l)$
(where $D=2^I$ and $L=2^O$).
The definition is very similar to the standard one~\cite{PrinciplesMC},
except for a few cases involving inputs
(marked with ``+'').

Let $n \in V$ and $o \in O$.
Then:
\li
\- $n \not\models \Phi$ iff $n \models \Phi$ does not hold,
\- $n \models \true$ and $n \not\models \false$,
\- $n \models o$ iff $o \in l(n)$, $n \models \neg o$ iff $o \not\in l(n)$,
\- $n \models \Phi_1 \land \Phi_2$ iff $n \models \Phi_1$ and $n \models \Phi_2$.
   Similarly for $\Phi_1\lor\Phi_2$.
\- $n \models \A \varphi$ iff for all tree paths $\pi$ starting from $n$:
   $\pi \models \varphi$.
   For $\E\varphi$, replace ``for all'' with ``there exists''.
\il

Let $\pi = n_1 n_2 ... \in V^\omega$ be a tree path,
$i \in I$, and $n_2 = n_1 \cdot e$ where $e \in \I$.
For $k \in \bbN$, define $\pi_{[k:]} = n_k n_{k+1} ...$,
i.e., the suffix of $\pi$ starting in $n_k$.
Then:
\li
\- $\pi \models \Phi$ iff $n_1 \models \Phi$,
\-[+] $\pi \models i$ iff $i \in e$, $\pi \models \neg i$ iff $i \not\in e$,
\- $\pi \models \varphi_1 \land \varphi_2$ iff $\pi \models \varphi_1$ and $\pi \models \varphi_2$.
   Similarly for $\varphi_1 \lor \varphi_2$.
\- $\pi \models \X \varphi$ iff $\pi_{[2:]} \models \varphi$,
\- $\pi \models \varphi_1 \U \varphi_2$ iff
   $\exists l\in\bbN: (\pi_{[l:]} \models \varphi_2 \land \forall m \in [1,l-1]: \pi_{[m:]} \models \varphi_1)$,
\- $\pi \models \varphi_1 \R \varphi_2$ iff
   $(\forall l\in\bbN: \pi_{[l:]} \models \varphi_2) \lor 
    (\exists l\in\bbN: \pi_{[l:]} \models \varphi_1 \land \forall m\in [1,l]: \pi_{[m:]} \models \varphi_2)$.
\il

A \emph{computation tree $(V,l)$ satisfies a \CTLstar state formula $\Phi$},
written $(V,l) \models \Phi$,
iff the root node satisfies it.
A \emph{system $M$ satisfies a \CTLstar state formula $\Phi$},
written $M \models \Phi$,
iff its computation tree satisfies it.

\begin{remark}[Subtleties]\label{rem:ctlstar-subtle}
Note that $(V,l) \models i\land o$ is not defined,
since $i \land o$ is not a state formula.
Let $r \in I$ and $g \in O$.
By the semantics, $\E r \equiv \true$ and $\E \neg r \equiv \true$,
while $\E g \equiv g$ and $\E \neg g \equiv \neg g$.
These facts are the consequences of the way we group inputs with outputs
(see also Remark~\ref{rem:inputs-shift}).
\end{remark}

\subsection*{LTL}
The syntax of LTL formulas (in general form) is:
$$
\phi = \true \| p \| \neg p \| \phi \land \phi \| \neg \phi \| \phi \U \phi \| \X \phi,
$$
where $p \in I \cup O$.
The temporal operators $\G$ and $\F$ are defined as usual,
and $\false = \neg \true$.
The semantics is standard (see, e.g., \cite{PrinciplesMC}).
A computation tree $(V,l)$ satisfies an LTL formula $\phi$,
written $(V,l) \models \phi$,
iff all tree paths starting in the root satisfy it.
A system satisfies an LTL formula iff its computation tree satisfies it
(equivalently, every system computation starting from the initial state
 satisfies the LTL formula).

\toclesslab\section{Tree Automata}{defs:tree-automata}
%\section{Tree Automata}\label{defs:tree-automata}

Tree automata consume infinite trees and output ``accept'' or ``reject''.
Since every Moore system has a corresponding computation tree,
tree automata can be used to differentiate buggy Moore machines from correct ones.
Also, \CTLstar can be translated into a special type of alternating tree automata.
Thus, tree automata are the excellent tool for model checking and synthesis.

We start with a general definition of alternating tree automata,
then introduce different acceptance conditions,
then introduce alternating hesitant tree automata.

\parit{Notation ${\cal B}^+(S)$}
For a finite non-empty set $S$,
let ${\cal B}^+(S)$ be the set of all positive Boolean formulas over elements of $S$,
i.e., every such a formula $\phi$ has the syntax:
$\phi = e \| \phi \land \phi \| \phi \lor \phi$, where $e \in S$.
Note that $\false \not\in {\cal B}^+(S)$ and $\true \not\in {\cal B}^+(S)$.
As we will see later, these are not limitations in our context.
Also, since the set $S$ is finite,
any Boolean formula over atoms in $S$ and which is semantically different from $\true$ and $\false$
is equivalent to some formula in ${\cal B}^+(S)$.
Furthermore,
every formula $\phi \in {\cal B}^+(S)$ can be rewritten into formula $\phi' \in {\cal B}^+(S)$ in disjunctive normal form (DNF)
or into formula $\phi'' \in {\cal B}^+(S)$ in conjunctive normal form (CNF).
We assume that formulas in ${\cal B}^+(S)$ (and thus CNF and DNF formulas)
have neither redundant atoms, conjuncts, nor disjuncts.
%Finally,
%every formula $\phi \in {\cal B}^+(S)$ can be seen as a non-empty set of non-empty sets of atoms from $S$

%For a finite set $S$,
%let ${\cal B}^+(S)$ denote the set of all positive Boolean formulas over elements of $S$ or $\true$,
%i.e., every such a formula $\phi$ has the syntax:
%$\phi = \true \| e \| \phi \land \phi \| \phi \lor \phi$, where $e \in S$.
%Note that $\false \not\in {\cal B}^+(S)$.

\subsection*{Alternating tree automata}

An \emph{alternating tree automaton} is a tuple
$(\Sigma, D, Q, q_0, \delta, acc)$,
where $\Sigma$ is the set of node propositions,
$D$ is the set of directions, $q_0 \subseteq Q$ is the initial state,
$\delta: Q \times \Sigma \to \mathcal{B}^+(D \times Q)$
is the transition relation,
and $acc: Q^\omega \to \bbB$ is an acceptance condition.
For simplicity we assume that $\delta$ is total wrt.\ directions.
Thus, it is worth noting about $\delta(q,\sigma)$, for every $(q,\sigma) \in Q\times\Sigma$:
\li
\- if we rewrite $\delta(q,\sigma)$ into DNF,
   then each conjunct mentions each direction at least once (``totalness'' wrt.\ directions).
\- $\delta(q,\sigma) \neq \false$ and $\delta(q,\sigma) \neq \true$.
   These are not limitations,
   because we can emulate $\true$ and $\false$ by introducing additional states and modifying $acc$.
\il
The above means that $\delta$ has a transition for every possible argument and direction.

Fix two disjoint sets, inputs $I$ and outputs $O$.

Tree automata consume exhaustive trees like $(D, L=\Sigma, V=D^*, l:V \to \Sigma)$
and produce run-trees.

A \emph{run-tree} of an alternating tree automaton
$(\Sigma=2^O, D=2^I, Q, q_0, \delta, acc)$
on a computation tree $(V=(2^I)^*, l:V \to \O)$
is a tree
with directions $2^I \times Q$,
labels $V \times Q$,
nodes $V' \subseteq (2^I \times Q)^*$,
labeling function $l'$
such that
\li
\- $l'(\epsilon) = (\epsilon, q_0)$,
\- if $v \in V'$ with $l'(v) = (n,q)$, then:\\
   there exists $\{(d_1,q_1),...,(d_k,q_k)\}$ that satisfies $\delta(q, l(n))$
   and $n \cdot (d_i, q_i) \in V'$ for every $i \in [1,k]$.
\il
Intuitively,
we run the alternating tree automaton on the computation tree:
\li
\-[(1)]
   We mark the root node of the computation tree with the automaton initial state $q_0$.
   We say that initially, in the node $\epsilon$,
   there is only one copy of the automaton and it has state $q_0$.
\-[(2)]
   We read the label $l(n)$ of the current node $n$ of the computation tree
   and consult the transition function $\delta(q,l(n))$.
   The latter gives a set of conjuncts of atoms of the form $(d',q') \in D\times Q$.
   We nondeterministically choose one such conjunction $\{(d_1,q_1), ..., (d_k,q_k)\}$
   and send a copy of the alternating automaton into each direction $d_i$ in the state $q_i$.
   Note that we can send up to $|Q|$ copies of the automaton into one direction
   (but into different automaton states).
   That is why a run-tree defined above has directions $2^I\times Q$
   rather than $2^I$.
\-[(3)]
   We repeat step (2) for every copy of the automaton.
   As a result we get a run-tree:
   a tree labeled with nodes of the computation tree and
   states of the automaton.
\il

A \emph{run-tree is accepting}
iff every run-tree path starting from the root is accepting.
A run-tree path $v_1 v_2 ...$ is accepting
iff $acc(q_1q_2...)$ holds ($acc$ is defined later),
where $q_i$ for every $i\in \bbN$ is the automaton state part of $l'(v_i)$.
Note that every run-tree path is infinite.
(In particularly, we do not have \emph{finite} paths that end with $\true$ nor $\false$,
 by definition of $\delta:Q\times\Sigma\to{\cal B}^+(D\times Q)$.)

An \emph{alternating tree automaton $A=(\Sigma=2^O, D=2^I, Q, q_0, \delta, acc)$
accepts a computation tree $(V=(2^I)^*, l:V \to \O)$},
written $(V,l) \models A$,
iff
the automaton has an accepting run-tree on that computation tree.
An alternating tree automaton is \emph{non-empty} iff there exists a computation tree accepted by it.

Similarly,
\emph{a Moore system $M=(I, O, T, t_0, \tau, out)$
is accepted
by the alternating tree automaton $A=(\Sigma=2^O, D=2^I, Q, q_0, \delta, acc)$},
written $M \models A$,
iff $(V,l) \models A$,
where $(V=(2^I)^*,l:V\to 2^O)$ is the system computation tree.

Let us define different variations of an acceptance condition $acc: Q^\omega \to \bbB$.
For a given infinite sequence $\pi \in Q^\omega$,
let $\Inf(\pi)$ be the elements of $Q$ appearing in $\pi$ infinitely often.
Then:
\li
\- \emph{B\"uchi acceptance} is defined by a set $F \subseteq Q$:
   $acc(\pi)$ holds iff $\Inf(\pi) \cap F \neq \emptyset$.
   We often call the states of $F$ \emph{accepting}.
\- \emph{Co-B\"uchi acceptance} is defined by a set $F \subseteq Q$:
   $acc(\pi)$ holds iff $\Inf(\pi) \cap F = \emptyset$.
   We often call the states of $F$ \emph{rejecting}.
\- \emph{Streett acceptance} is defined by pairs $\{(A_i\subseteq Q,G_i \subseteq Q)\}_{i\in[k]}$:
   $acc(\pi)$ holds
   iff
   $\forall i \in [k]: \Inf(\pi)\cap A_i \neq \emptyset \impl \Inf(\pi) \cap G_i\neq\emptyset$.
\- \emph{Rabin acceptance} is defined by pairs $\{(F_i,I_i)\}_{i \in [k]}$:
   $acc(\pi)$ holds
   iff
   $\exists i \in [k]: \Inf(\pi) \cap F_i = \emptyset \land \Inf(\pi) \cap I_i \neq \emptyset$.
\- \emph{Parity acceptance} is defined by a priority function
   $p: Q \to [0,k]$:
   $acc(\pi)$ holds iff the minimal priority appearing infinitely often in $p(\pi)$ is even.
\il
In addition to the above acceptance conditions,
we define generalized versions.
Generalized B\"uchi acceptance condition is defined by a set $\{F_i\}_{i \in [k]}$:
$acc(\pi)$ holds iff the B\"uchi condition holds wrt.\ every $F_i$ where $i \in [k]$.
Similarly define Generalized co-B\"uchi%
\footnote{We stress that, in our work, Generalized co-B\"uchi
  for a set $\{F_i\}_{i \in [k]}$ means:
  $acc(\pi)$ holds iff the co-B\"uchi condition holds wrt.\ \emph{every} $F_i$ where $i \in [k]$.
  But often Generalized co-B\"uchi acceptance means that
  \emph{there exists} $F_i$ that is visited finitely often where $i \in [k]$.%
},
Streett, Rabin, and Parity conditions.

\subsection*{Nondeterministic and universal tree automata}

Depending on the form of $\delta(q,\sigma)$ (for every $(q,\sigma) \in Q\times\Sigma$),
we distinguish the following special cases of alternating tree automata.
\li
\- \emph{Universal} tree automata:
   $\delta(q,\sigma)$ is a conjunction of variables of $Q\times D$,
   where each direction is mentioned at least once.
\- \emph{Deterministic} tree automata:
   $\delta(q,\sigma)$ is a conjunction of variables of $Q \times D$ and
   each direction is mentioned exactly once.
\- \emph{Nondeterministic} tree automata:
   let $\delta(q,\sigma)$ be rewritten in DNF.
   Then each conjunct mentions each direction exactly once.
\il

\ak{state known facts: (1) non-empty iff regularly non-empty, (2) non-empt complexity, (3) conversion of \CTLstar into AHTs, (4) determinisation?}

\subsection*{Alternating hesitant tree automata (AHT)} \label{page:defs:aht}

An \emph{alternating hesitant tree automaton (AHT)} is an alternating tree automaton
$(\Sigma, D, Q, q_0, \delta, acc)$
with the following acceptance condition and structural restrictions.
The restrictions reflect the fact that AHTs are tailored for \CTLstar formulas.
\li 
\- $Q$ can be partitioned into $Q^N_1,\dots ,Q^N_{k_N}$, $Q^U_1,\dots
,Q^U_{k_U}$, where superscript $N$ means
nondeterministic and $U$ means universal.
Let $Q^N = \bigcup Q^N_i$ and $Q^U = \bigcup Q^U_i$.
(Intuitively,
 nondeterministic state sets describe \E-quantified subformulas of the \CTLstar formula,
 while universal state sets describe \A-quantified subformulas.)

\- There is a partial order on $\{Q^N_1,\dots ,Q^N_{k_N},Q^U_1,\dots , Q^U_{k_U}\}$.
   (Intuitively, this is because state subformulas can
    be ordered according to their relative nesting.)

\- The transition function $\delta$ satisfies: for every $q \in Q$, $a \in \Sigma$
   \li
   \- if $q \in Q^N_i$, then:
      $\delta(q,a)$ contains only disjunctively related\footnotemark[1] elements of $Q^N_i$;
      every element of $\delta(q,a)$ outside of $Q^N_i$ belongs to a lower set;
   \- if $q \in Q^U_i$, then:
      $\delta(q,a)$ contains only conjunctively related\footnotemark[1] elements of $Q^U_i$;
      every element of $\delta(q,a)$ outside of $Q^U_i$ belongs to a lower set.
   \il
   \ak{figure of disj-conj related sets: and why?!}
\il
\footnotetext[1]{In a Boolean formula, atoms $E$ are disjunctively [conjunctively] related
  iff the formula can be written into DNF [CNF] in such a way that each cube [clause] has at most one element from $E$.}

Finally, $acc: Q^\omega \to \bbB$ of AHTs is defined by a set $Acc \subseteq Q$:
$acc(\pi)$ holds for $\pi=q_1q_2...\in Q^\omega$ iff one of the following holds.
\li
\- The sequence $\pi$ eventually stays in some $Q^U_i$ and
   $\Inf(\pi) \cap (Acc\cap Q^U) = \emptyset$
   (co-B\"uchi acceptance).
   Let us denote $F = Acc \cap Q^U$.
\- The sequence $\pi$ eventually stays in some $Q^N_i$ and
   $\Inf(\pi) \cap (Acc \cap Q^N) \neq \emptyset$
   (B\"uchi acceptance).
   Let us denote $I = (Acc \cap Q^N) \cup (Q^U \!\setminus\! Acc)$.
\il
Due to the restrictions on the structure of hesitant automata,
this acceptance is equivalent to the Rabin acceptance with one pair $(F,I)$.

\ak{examples}

\toclesslab\section{Word Automata}{defs:word-automata}
%\section{Word Automata}\label{defs:word-automata}

In contrast to tree automata that consume infinite trees,
word automata consume infinite words.
Every LTL formula can be translated into a word automaton,
but a \CTLstar formula, in general, cannot.
This is because a \CTLstar formula describes a set of trees,
while an LTL formula describes a set of words.

We start with alternating word automata.
Such automata can concisely represent LTL formulas
(without incurring an exponential blow-up in its size).
Then we define two specializations: nondeterministic and universal word automata.
Such automata are often used as input to synthesis algorithms,
because they are simpler to work with
(although translation of an LTL formula into such an automaton can incur
 and exponential blow-up).
Finally, we define alternating hesitant word automata.
They are useful for model checking and synthesis from AHTs
(and thus from \CTLstar formulas).

\subsection*{Alternating word automata}
%How a tree automaton $(\Sigma, D, Q, q_0, \delta, acc)$ looks like when $|D|=1$
%(thus trees have one branch only)?
%The transition function becomes
%$\delta: Q\times\Sigma \to {\cal B}^+(Q)$.

%An \emph{alternating word automaton $(\Sigma, Q, q_0, \delta, acc)$}
%is an alternating tree automaton $(\Sigma, D, Q, q_0, \delta, acc)$ with $|D|=1$.
%Run tree becomes run path.
%The difficulty arises when defining $M \models A$, where $A$ is a word automaton.
%We cannot define $M \models A$ iff $(V,l) \models A$,
%since we requested $A$ to have only one direction ($|2^I|=1$),
%while system can have $|I|>0$.
%We can overcome this by defining $\Sigma=2^I\cup 2^O$,
%but then we need to consider non-deterministic systems.
%
%Usually we use $D$ to describe inputs: $D=2^I$.
%Although now we require $|D|=1$, we can still use word automata to
%distinguish correct from buggy Moore systems.
%The idea is to consider $\Sigma=2^I \cup 2^O$
%(cf. $\Sigma=2^O$ in the previous section).

\begin{remark}[Re-using definitions from tree automata]
An infinite word can be viewed as a tree with a single branch.
Thus it is tempting to derive definitions for word automata from those of tree automata.
Without additional tricks this will not work for the following reason.
In our work, branching degree of every computation tree is, by definition, $|2^I|$.
Thus, considering only single-branch trees is equivalent to
having systems with no inputs: $|2^I|=1 \Iff |I|=0$.
But we are interested in the general case: $|I| \in \bbN_0$.
(To reuse definitions, we could move the inputs into the outputs,
 consider nondeterministic systems without edge labels,
 and require that each state has a successor containing a label $e$, for every $e \in 2^I$.)
\end{remark}

An \emph{alternating word automaton} is a tuple $(\Sigma, Q, q_0, \delta, acc)$
where $\Sigma$ is an alphabet, $Q$ is a set of states, $q_0 \in Q$ is initial,
$\delta: Q \times \Sigma \to {\cal B}^+(Q)$ is a transition function,
and $acc:Q^\omega \to \bbB$ is a path acceptance condition.
Note that $\delta(q) \neq \false$ and $\delta(q) \neq \true$ for every $q\in Q$,
i.e., there is a successor state for every letter.
These are not real limitations, they are in place to simplify definitions.
% ak: not sure if we use this restriction

Given a word $w = a_1 a_2 ... \in \Sigma^\omega$,
let $\Pref(w) = \{\epsilon, a_1, a_1 a_2, a_1 a_2 a_3, ...\}$
denote the set of all its prefixes (including the empty prefix).
Also,
for a finite non-empty word $w \in \Sigma^+$,
let $last(w)$ denote its last letter.

A \emph{run-tree} of an alternating word automaton
$(\Sigma, Q, q_0, \delta, acc)$
on an infinite word $w = a_1 a_2 ... \in \Sigma^\omega$
is a tree
with directions $Q$,
labels $\Pref(w) \times Q$,
nodes $V' \subseteq (\Sigma \times Q)^*$,
labeling function $l': V' \to \Pref(w) \times Q$
such that
\li
\- $l'(\epsilon) = (\epsilon, q_0)$,
\- if $v \in V'$ with $l'(v) = (p,q)$, then:\\
   there exists $\{q_1,...,q_k\}$ that satisfies $\delta(q, last(p))$
   and $p \cdot q_i \in V'$ for every $i \in [1,k]$.
\il
This definition coincides with the definition of run-trees of tree automata
when restricted to trees with a singe path.

A \emph{run-tree is accepting}
iff every run-tree path starting from the root is accepting.
A run-tree path $v_1 v_2 ...$ is accepting
iff $acc(q_1 q_2 ...)$ holds,
where $q_i$, for every $i$, is an automaton state of the label $l'(v_i)$.

An \emph{alternating word automaton $A=(\Sigma, Q, q_0, \delta, acc)$
accepts an infinite word $w=a_1 a_2... \in \Sigma^\omega$},
written $w \models A$,
iff
the automaton has an accepting run-tree on the word.
An alternating word automaton is \emph{non-empty}
iff there exists an infinite word accepted by it.

So far we re-used definitions from Section~\ref{defs:tree-automata},
but now we depart.

In Section~\ref{defs:ctlstar} about \CTLstar, we introduced two path quantifiers:
$\A$ (on all paths) and $\E$ (there exists a path).
Accordingly, we define $M\models \A(A)$ and $M \models \E(A)$:
$M \models \E(A)$ iff there is a system computation accepted by the automaton;
$M \models \A(A)$ iff every system computation is accepted by the automaton.

\begin{remark}[Tree automata vs. Word automata]
Consider \CTLstar formula $\EX g \land \EX \neg g$,
outputs $O = \{g\}$, inputs $I = \{r\}$.
No alternating \emph{word} automaton, when prefixed with $\E$ or $\A$,
can describe this language,
while there is an alternating \emph{tree} automaton for it\ak{see Fig.XXX}.
\end{remark}

\subsection*{Nondeterministic and universal word automata}

We now look closer at nondeterministic and universal word automata.
Their definitions coincide
with those for tree automata when assuming the single direction,
but for clarity we recall them here.

Depending on the form of $\delta(q,a)$ (for every $(q,a) \in Q\times\Sigma$),
we distinguish:
\li
\- \emph{Deterministic} word automata: $\delta(q,a)$ is a single state.
   There is no choice: we know exactly into which state to proceed
   from $q$ when reading $a$.
   Thus, a run-tree degenerates into a line.
\- \emph{Nondeterministic} word automata: $\delta(q,a)$ is a (non-empty) disjunction of states.
   We reading $a$ in state $q$, we choose one of the states and proceed.
   A run-tree degenerates into a line.
\- \emph{Universal} word automata: $\delta(q,a)$ is a conjunction of states
   $q_1\land...\land q_k$ where $k \geq 1$.
   Thus, when reading $a$ we send one copy of the automaton into every $q_1 ... q_k$.
   A run-tree is indeed a tree.
\il
In all the cases above,
the transition function can be expressed as $\delta: Q\times\Sigma \to 2^Q \setminus \{\emptyset\}$.
We often use this notation instead of $\delta: Q\times\Sigma \to {\cal B}^+(Q)$.

\subsection*{Alternating hesitant word automata (AHW)} \label{page:defs:ahw}

An \emph{alternating hesitant word automaton} (AHW)
$A=(\Sigma, Q, q_0, \delta: Q\times\Sigma \to {\cal B}^+(Q), Acc \subseteq Q)$
is an alternating word automaton with
the similar to AHTs structural restrictions and the same as for AHTs acceptance condition.

\subsection*{Tree variants of word automata} \label{page:defs:tree_variants}

We define the following tree variants, $A_\A$ and $A_\E$, of a word automaton $A$.
Given a nondeterministic word automaton
$A = (2^{I\cup O}, Q, q_0, \delta: Q\times 2^{I\cup O} \to 2^Q, acc)$,
let $A_\E = (2^O, D=2^I, Q,q_0, \delta': Q\times 2^O \to 2^{Q\times D}, acc)$
be the nondeterministic tree automaton defined in the most natural way:
for every $(q,o) \in Q \times 2^O$,
$$
\delta'(q,o) = \bigvee_{d \in D} \delta\big(q,(o,d)\big) \left[q' \mapsto (d,q')\right].
$$
We define $A_\A$ in the same way.

We will use $A_\E$ to talk about
``the product between system $M$ and nondeterministic \emph{word} automaton $A$''.
Since we do not define such a product---it is defined for \emph{tree} automata only---%
we use $A_\E$ instead.
Similarly, when we have universal word automata, we use $A_\A$.

$A_\A$ and $A_\E$ satisfy the following property:
for every system $M$,
\li
\- for universal $A$: $M \models \A(A) ~\Iff~ M \models A_\A$,
\- for nondeterministic $A$: $M \models \E(A) ~\Iff~ M \models A_\E$.
\il

%We mention the important property of AHWs.
%Recall that the states of an AHW can be partitioned into
%``existential'' sets $Q^N_1, ..., Q^N_{k_N}$ and
%``universal'' sets $Q^U_1, ..., Q^U_{k_U}$.
%Such sets are ordered and the transition function ensure the following:
%Every path in every run-tree of the AHW
%gets trapped in some $Q^N_i$ or in $Q^U_j$.
%We use this observation when encoding
%the \CTLstar model checking problem into an SMT query.

%In Section~\ref{defs:model-checking-approach} we use 1-letter AHWs

%\emph{A path in automaton $A$}
%is a sequence $q_0 q_1 ... \in Q^\omega$ starting in the initial state
%such that there exists $a_i \in \Sigma$ for every $i \geq 0$
%such that $(q_i,a_i,q_{i+1}) \in \delta(q_i)$.
%\emph{A sequence $a_0 a_1 \dots  \in \Sigma^\omega$ generates a path}
%$\pi = q_0 q_1\dots $ iff for every $i \geq 0$: $(q_i,a_i,q_{i+1}) \in \delta$.
%A \emph{path $\pi$ is accepted} iff $acc(\pi)$ holds.
%
%We distinguish two types of automata: universal and non-deterministic.
%The type defines when the automaton accepts a given infinite sequence.
%
%A \emph{non-deterministic automaton $A$ accepts an infinite sequence}
%from $\Sigma^\omega$
%iff there exists an accepted path generated by the sequence.
%Universal automata require \emph{all} paths generated by the sequence to be accepted.
%For an automaton $A$, write $L(A)$ for the set of all infinite sequences accepted by $A$.
%Equvalently,
%in both cases:
%an automaton accepts an infinite sequence iff there is an accepting run.

\section*{Automata Abbreviations}\label{defs:automata-abbrv}
We use the standard three letter abbreviation for automata
$$(\{A,U,N,D\}\times\{B,C,S,P,R\}\times\{W,T\}) \cup \{AHT, AHW\}.$$
For example,
NBW means Nondeterministic B\"uchi Word automaton,
UCT means Universal co-B\"uchi Tree automaton,
AHT means Alternating Hesitant Tree automaton.

\toclesslab\section{Approaches to Model Checking}{defs:model-checking-approach}
%\section{Approaches to Model Checking}\label{defs:model-checking-approach}

This section treats the automata-theoretic approach to \CTLstar and LTL model checking~\cite{ATA},
as well as the classical bottom-up approach to \CTLstar model checking~\cite{DBLP:journals/toplas/ClarkeES86}.
Let us start with the definition.

\smallskip
The \emph{model checking problem} is:

\smallskip
\noindent
\emph{%
Given: a Moore system $M$, formula $\Phi$ in some logic\\
Return: does $M \models \Phi$?
}

\smallskip
\noindent
Depending on the logic of $\Phi$,
we have \CTLstar and LTL model checking problems.

\medskip
We now briefly describe the automata-based solution to \CTLstar (and thus LTL)
model checking problem
introduced by Kupferman, Vardi, and Wolper~\cite{ATA}.
The idea is to translate a given formula $\Phi$ into a tree automaton:
\li
\- into AHT, when $\Phi$ is in \CTLstar,
\- into UCT, when $\Phi$ is in LTL.
\il
Then we build the ``product'' of the system and the automaton,
which will be a word automaton (an alternating one for \CTLstar, a universal one for LTL).
Such a word automaton captures the joint behaviours
of the system and the automaton for $\Phi$.
The nice property of the product is that it is empty iff
the system is accepted by the tree automaton for $\Phi$
(equivalently: iff the system satisfies $\Phi$).
Thus, we reduce the model checking problem to checking the non-emptiness
of a word automaton.
Below we define the product.

\subsection*{Product of system and tree automaton}

Let us provide the intuition.
%The product of a system and a tree automaton will allows us
%to reduce the question ``whether the system is accepted by the automaton''
%into the non-emptiness question of the product automaton.

A system $M$ can be seen as a deterministic tree automaton $A_M$
that accepts only its own computation tree.
Then the question of whether the system is accepted
by a given tree automaton $A$
is equivalent to checking whether the intersection $A_M \land A$ is non-empty.
The product of a system $M$ and a tree automaton $A$,
written $M\otimes A$,
can be seen as the intersection tree automaton $A_M \land A$,
from which we remove labels and directions%
\footnote{%
 We can remove labels because,
 for each state of $A_s$,
 the label is uniquely defined---it is $out(t)$ for the corresponding system state $t$.
 We can remove directions,
 because in the non-emptiness check of a 1-letter tree automaton
 we do not distinguish directions.%
}
and which we treat as a \emph{word} automaton.
The important property is that the product $M\otimes A$---a 1-letter alternating word automaton---
is empty iff $A_M \land A$ is empty.

%Then the product of the system and a tree automaton $A$
%is the intersection $A_s \land A$.
%Thus, every transition of the product describes
%both a transition of the system and a transition of the automaton $A$.
%Now let us ask ``Is the automaton $A_s \land A$ non-empty?''
%Since the only computation tree it can accept belongs to the system,
%this question is equivalent to ``$M \in L(A)$?'' (equiv.: ``$M \models A$?'').
%
%If we are interested only in the non-emptiness of $A_s \land A$,
%we can treat it as a \emph{1-letter word} automaton for the following reasons.
%
%Here is why 1-letter alphabet is enough.
%Every system state $t$ has exactly one output label $out(t)$.
%Hence the tree automaton $A_s$, in every its state,
%rejects all output labels except one ($out(t)$, for the corresponding $t$).
%The same holds for the intersection $A_s \land A$.
%Thus, when checking the non-emptiness of $A_s \land A$,
%we have no choice of letters---they are already fixed by the system.

%Here is why we can treat $A_s \land A$ as a \emph{word} automaton.
%With 1-letter alphabet $\Sigma=\{\sigma\}$,
%whe transition function of $A_s \land A$,
%$\delta_{A_s\land A}: Q\times\Sigma \to {\cal B}^+(D \times Q)$
%can be seen as
%$\delta_{A_s\land A}: Q \to {\cal B}^+(D \times Q)$.

A \emph{1-letter alternating hesitant word automaton (1-AHW)} is
an AHW $(Q, q_0, \delta: Q \to \mathcal{B}^+(Q), Acc \subseteq Q)$,
whose alphabet has only one letter (not shown in the tuple).
Informally,
an 1-AHW is an and-or graph of a restricted form
plus a Rabin acceptance condition.

%A \emph{run of a 1-AHW} 
%$(Q, q_0, \delta: Q \to \mathcal{B}^+(Q), Acc \subseteq Q)$
%is a labeled tree defined in a standard way.
%Its nodes are from $Q^*$,
%the root is $q_0$,
%the labeling $l$ maps a node (in $Q^*$) to the last element (in $Q$),
%and for any reachable node, $l(succ(n)) \models \delta(l(n))$ where $l(succ(n))$
%is the set of labels of the successors of node $n$.
%A \emph{run is accepting} if all paths of the tree
%satisfy the acceptance condition.
%A run tree \emph{path satisfies the acceptance condition $Acc$}
%iff one of the following holds:
%\li
%\- the corresponding path in the 1-AHW gets trapped in some $Q^U$
%   and visits $Acc\cap Q^U$ only finitely often, or
%\- the corresponding path in the 1-AHW gets trapped in some $Q^N$ 
%   and visits some state of $Acc \cap Q^N$ infinitely often.
%\il
%Intuitively, the 1-AHW acceptance condition is a mix of B\"uchi
%and co-B\"uchi acceptance conditions.
%It can also be seen as a Rabin acceptance with one pair
%$(F,I)$ where
%$F = Acc \cap Q^U$ and $I = (Acc \cap Q^N) \cup (Q^U\!\setminus\! Acc$).
%
%Note that any path of a run tree of a 1-AHW
%is trapped in some $Q^N_i$ or $Q^U_i$.
%
%The \emph{non-emptiness question of 1-AHW} is
%``does the automaton has an accepting run?''.

A \emph{product} of an AHT
$A\!\!=\!\!(2^O, 2^I, Q, q_0, \delta, Acc)$ and a system $M\!\!=\!\!(I,O,T,t_0,\tau,out)$,
written $M \otimes A$,
is a 1-AHW
${(Q \times T, (q_0,t_0), \Delta, Acc')}$
such that
$${Acc'=\{(q,t) \mid q \in Acc\}}$$
and for every $(q,t) \in Q \times T$:
$$
{\Delta(q,t)=\delta(q,out(t))[(d,q') \mapsto (\tau(t,d),q')]}.
$$
As before, $M \otimes A$ is non-empty iff there exists an infinite word accepted by it.
Since $M \otimes A$ has a one-letter alphabet---let it be $\Sigma = \{l\}$---%
checking the non-emptiness of $M \otimes A$ means checking
whether the infinite word $l^\omega$ is accepted.

Recall that the states of an AHW can be partitioned into
``existential'' sets $Q^N_1, ..., Q^N_{k_N}$ and
``universal'' sets $Q^U_1, ..., Q^U_{k_U}$.
These sets are ordered, and the transition function of the AHW
satisfies the restriction that ensures the following:
every infinite path of the AHW gets trapped in some $Q^N_i$ or $Q^U_j$.

\subsubsection{\CTLstar model checking using the product}

For the general case of \CTLstar formulas and AHTs we know the following.
\begin{proposition}[\cite{ATA}]
A system $M$ is accepted by an AHT $A$ iff their product $M \otimes A$ is non-empty.
\end{proposition}
\begin{corollary}
A system $M$ satisfies \CTLstar formula $\Phi$ iff the product $M \otimes A_\Phi$
is non-empty,
where $A_\Phi$ is an AHT for $\Phi$.
\end{corollary}
The 1-AHW non-emptiness problem can be decided in linear time wrt.\ the size of the automaton~\cite{ATA},
therefore the approach is EXPTIME wrt.\ the size of a given \CTLstar formula.
It is known~\cite{PrinciplesMC} that \CTLstar model checking problem is PSPACE-complete.

\subsection*{LTL model checking using the product}

Consider the special case of LTL properties.
We can check whether a given system $M$ satisfies an LTL formula $\varphi$
as follows.
First, construct a UCT $A_\varphi$ for $\varphi$
(we do not need alternating automata for LTL).
Second, build the product $M \otimes A_\varphi$.
Such a product is a 1-letter \emph{universal} word automaton (1-UCW).
Then the model checking is equivalent to checking non-emptiness of the 1-UCW.
\begin{proposition}\label{defs:prop:ltl_mc_via_product}\label{page:defs:prop:ltl_mc_via_product}
A system $M$ satisfies an LTL formula $\varphi$
iff
the product $M \otimes A_\varphi$ is non-empty,
where $A_\varphi$ is the UCT for $\varphi$.
\end{proposition}
For the same arguments as for \CTLstar,
the approach is EXPTIME wrt.\ the size of a given LTL formula.

\subsection*{Bottom-up \CTLstar model checking}\label{page:defs:bottom-up-mc}

Here is the classical construction~\cite{DBLP:journals/toplas/ClarkeES86}
(see also \cite[p.427]{PrinciplesMC})
for \CTLstar model checking.

We will need the following notions of $F$, $P$, and $\widetilde\Phi$.
Let $\Phi$ be a \CTLstar formula with inputs $I$ and outputs $O$.
Let $F'=\{f'_1,...,f'_k\}$ be the set of \CTLstar subformulas of the form $\A\varphi$ or $\E\varphi$,
where $\varphi$ is a path formula.
The subformulas $\{f'_1,...,f'_k\}$ can be ordered wrt.\ path quantifier nesting depth $d$.
Let us assume that $f'_1,...,f'_k$ are ordered wrt.\ $d$ in increasing order,
i.e., $d(f'_i) \leq d(f'_{i+1})$ for every $i$.
With every $f'_i$ we associate a proposition $p_i$,
which makes up the set $P=\{p_1,...,p_k\}$.
Let $F=\{f_1,...,f_k\}$ be the set of formulas,
where each $f_i$ is $f'_i$ in which all subformulas were replaced by the corresponding propositions:
For example,
for the \CTLstar formula $\Phi=\EG(g \impl g \U \AX\neg g)$ we have
$F'=\{ f'_1=\AX\neg g, f'_2=\EG(g \impl g \U \AX\neg g) \}$,
$P=\{p_1, p_2\}$, and
$F=\{ f_1=\AX\neg g, f_2=\EG(g \impl g \U p_1) \}$.
Notice about $F$:
(i) $F$ are formulas over atoms $I\cup O\cup P$,
(ii) every $f_i$ is over terms $I \cup O \cup \{p_1,...,p_{j\leq i-1}\}$, and
(iii) they are of the form $\A\varphi$ or $\E\varphi$,
where $\varphi$ is a \CTLstar path formula that has \emph{no path quantifiers}.
Let $\widetilde\Phi$ be $\Phi$ where all subformulas were replaced by the corresponding propositions.
Note that $\widetilde\Phi$ is a \emph{Boolean} formula over $O\cup P$.

Given a \CTLstar formula $\Phi$ and a system $M$.
The bottom up model checker creates the formulas $F$, propositions $P$, and $\widetilde\Phi$.
Then it annotates the system states with propositions from $P$ such that
a proposition $p_i$ holds in a state $t$ iff $t \models f_i$.
It does so in a bottom up manner (inductively):
\li
\- It starts with the proposition $p_1$:
   the formula $f_1$ is a path formula over propositions $I\cup O$.
   We can use LTL model checker to check if $t \models f_1$,
   for every system state $t$.
\- Similarly for $p_i$:
   use LTL model checker to check if $t \models f_i$,
   where $f_i$ can talk about propositions $I \cup O \cup \{p_1,...,p_{i-1}\}$
   whose truth for every system state we already established.
\- Finally, $M \models \Phi$ iff $\widetilde\Phi$ holds in the initial system state.
\il
The complexity of the procedure is EXPTIME wrt.\ $|\Phi|$.

\toclesslab\section{Approaches to Synthesis}{defs:synthesis-approaches} \label{defs:bounded_synthesis} \label{defs:synthesis-problem}
%\section{Approaches to Synthesis}\label{defs:synthesis-approaches} \label{defs:bounded_synthesis} \label{defs:synthesis-problem}

This section describes
(i) the classical game-based approach~\cite{DBLP:conf/popl/PnueliR89} to LTL synthesis,
(ii) a more recent approach~\cite{BS} that avoids automata determinization and uses constraint solvers, and
(iii) an approach to \CTLstar synthesis.
Let us start with the definition.

\smallskip
The \emph{synthesis problem} is:

\smallskip
\noindent
\emph{%
$~$Given: the set of inputs $I$, the set of outputs $O$, formula $\Phi$ in some logic\\
Return: a Moore system with inputs $I$ and outputs $O$ satisfying $\Phi$,}\\
$~~~~~~~~~~~~$\emph{or ``unrealisable'' if no such system exists}

\smallskip
\noindent
The input $\tpl{I,O,\Phi}$ to the problem is called a \emph{specification}.
A specification is \emph{realisable} if the answer is a system,
otherwise the specification is \emph{unrealisable}.
Depending on the logic of $\Phi$,
we have LTL and \CTLstar synthesis problems.
Instead of a \emph{formula} $\Phi$,
we can use tree automata or word automata prefixed with the $\E$ or $\A$ path quantifier.
%Also,
%instead of asking for a system, we can ask for a computation tree.
\ak{note: we can ask for computation tree instead of regular object---but they are equivalent}

\medskip
It is known~\cite{informatio,DBLP:conf/popl/PnueliR89} that
the \CTLstar and LTL synthesis problems are 2EXPTIME-complete.
Below we discuss two approaches to LTL synthesis problem,
game-based approach and bounded synthesis.
Then we discuss an approach to \CTLstar synthesis.

\subsection*{LTL synthesis via reduction to games}

The standard game-based approach~\cite{DBLP:conf/popl/PnueliR89} to synthesis from LTL specifications is as follows.
\li
\- Translate a given LTL formula into a nondeterministic B\"uchi word automaton~\cite{DBLP:journals/iandc/VardiW94}.
   The automaton can be exponentially larger than the LTL formula.

\- Determinise the automaton into a deterministic parity word automaton,
   e.g. using Safra construction~\cite{Safra}.
   The resulting automaton can be exponentially larger than the original one,
   leading to the doubly exponential blow up.

\- Translate the word automaton into a tree automaton,
   by splitting each transition into two transitions according to the input and output labels.

\- Check the non-emptiness of the deterministic parity tree automaton.
   The check can be done by treating the tree automaton as a game, and then solving the game.
   If the game is winning for the system player (it controls the choice of output labels),
   then the specification is realisable,
   otherwise it is unrealisable.
   The particular class of parity games that we get can be solved in polynomial time e.g. using~\cite{DBLP:conf/stoc/CaludeJKL017}.
\il
The method gives 2EXPTIME solution to the synthesis problem.

The 2EXPTIME-hardness comes from the fact that we can encode
into the LTL realisability problem
the acceptance of a given word by an alternating exponential-space bounded Turing machine~\cite{Pnueli1989,Vardi:1985:IUL:22145.22173}.%
\footnote{%
  An alternating Turing machine is, like an alternating automaton,
  has universal and existential transitions.
  A given word is accepted if there is an accepting run-tree of the machine on this word
  (and thus all its branches are accepting).
  A Turing machine is exponential-space bounded iff:
  (i) it terminates on all inputs,
 (ii) it uses $2^{cn}$ number of cells where $n$ is the length of the input word
      and $c$ is a constant.
 The problem of deciding whether a word is accepted by such a machine is 2EXPTIME-complete~\cite{Chandra:1981:ALT:322234.322243}.
}
I.e., given such a Turing machine and an input word,
we can build the LTL specification,
which is realisable
iff
the word is accepted by the machine.
The length of the specification is $2^{cn}$ where $n$ is the length of the input word and $c$ is a constant.
The specification requires a system to output,
in each step,
the set of all successor configurations of the TM until it accepts on all of them.
In each step, with the aid of additional inputs,
we choose one configuration from which to proceed.
Non-deterministic transitions of the TM are emulated using ORs in the LTL formula.

\subsection*{Bounded LTL synthesis via SMT}\label{page:defs:bounded_synthesis}

The idea of bounded synthesis via SMT~\cite{BS} is to reduce the synthesis problem to SMT solving.
The resulting SMT query encodes the \emph{model checking} question%
---the query is satisfiable iff the system satisfies a given specification.
To turn model checking into synthesis, 
we replace the given system by uninterpreted functions.
Therefore, if the query is satisfiable, then the SMT solver produces%
---in addition to YES/NO answer---%
models of the uninterpreted functions that encode the system.
From those models we extract the system, and such a system is correct.

The SMT query encodes the non-emptiness of the product of a system and a UCT,
where UCT represents a given LTL formula
(see also Proposition~\ref{defs:prop:ltl_mc_via_product} on page~\pageref{page:defs:prop:ltl_mc_via_product}).
Recall that such a product is a 1-UCW.
Thus,
the emptiness question reduces to finding a lasso
with a final state of the 1-UCW in the loop.

If a system was given (as in the model checking),
then using SMT solvers in \emph{this} way to solve such a simple graph question does not seem%
\footnote{SMT solvers \emph{are} used in verification, see, for example, papers on solving Horn clauses~\cite{bjornerhorn,RybHornSynth,BeDi}.
 Here we refer only to the way of using them as it is done in bounded synthesis.%
} to be wise
(if we fix a system,
 then the complexity of solving such%
 \footnote{Here ``such queries'' means that they have the same theory as those used by bounded synthesis
 (for example, UFLIA).
 We did not analyse whether the special structure of SMT queries from bounded synthesis
 gives way to a simpler complexity than that of solving general UFLIA queries.%
 }
 SMT queries is NPTIME-hard wrt.\ the size of the formula automaton,
 while the straight graph-based approach is in PTIME wrt.\ the size of the formula automaton;
 also, the SMT-based approach is not symbolic).
But in the case of synthesis the system is not given:
here, an SMT solver plays the role of an efficient guess-verify searcher.

The pseudo-code of the bounded synthesis is:
\begin{verbatim}
convert a given LTL formula into UCT
for system size in {1...bound}:
    encode non-emptiness of system*UCT into SMT query
    solve the query
    if the query is satisfiable:
        return REALIZABLE
return UNREALIZABLE
\end{verbatim}
Let us go through the steps.

\parbf{Automata translation}
A given LTL formula $\varphi$
is translated into a UCT $U$ which accepts a Moore machine $M$ iff $M$ satisfies $\varphi$.
This can be done, for example, using SPOT~\cite{spot} or LTL3BA~\cite{LTL3BA}:
negate the formula,
translate it into a NBW,
treat it as a UCW $A$, and turn the UCW $A$ into a UCT $A_\A$ as described
on page~\pageref{page:defs:tree_variants}.
This $A_\A$ is the sought UCT $U$.

\parbf{Iteration for increasing bounds}
Fix the number of states in a system $M$.
This allows us to encode the non-emptiness problem of $M \otimes U$ into a decidable fragment of SMT.
The bound $bound$ can be either user-chosen or it is the upper bound on the system size ($O(2^{2^{|\varphi|}})$).
A better way, from the practical point of view, is described in Remark~\ref{rem:bosy-unreal}.

\parbf{SMT encoding}
Let inputs and outputs be $I$ and $O$.
Fix the states $T$ of a system $M=(I,O,T,t_0,\tau,out)$.
Let UCT $U = (2^O, 2^I, Q,q_0, \delta': Q\times 2^O \to 2^{Q\times 2^I}, F \subseteq Q)$.
In the SMT query, we use uninterpreted functions to express system functions $\delta$ and $out$.
We also use two uninterpreted functions:
$\reach: Q\times T \to \bbB$ denotes
whether a pair $(q,t) \in Q \times T$ is reachable in the product $M\otimes U$,
and $\rank: Q \times T \to \bbN$ which is called \emph{ranking function} and is used
to ensure the absence of bad lassos
(they visit $q \in F$ in the loop of the lasso).
The constraints are:
\begin{equation}\label{defs:eq:bs}
\begin{aligned}
&\reach(q_0,t_0) \land \\
\bigwedge_{(q,t) \in Q\times T} ~\Big[ &\reach(q,t) \impl
\!\!\!\!\!\!\!\!\!\!\bigwedge_{(d,q') \in \delta(q,out(t))} \!\!\!\!\!\!\!\!\!\!
\reach\!\left(q',\tau(t,d)\right) ~\land~ \rank(q,t) ~\triangleright~ \rank\!\left(q',\tau(t,d)\right)\Big]
\end{aligned}
\end{equation}
where $\triangleright$ is $>$ if $q \in F$, otherwise $\geq$.
The intuition is as follows.
We mark the initial state $(q_0,t_0)$ of the product $M\otimes U$ as reachable.
For every reachable state $(q,t)$,
we mark every successor state $(q',\tau(t,d))$ of the product as reachable,
and we require the rank to strictly decrease if $q \in F$, and non-strictly decrease otherwise.
Thus, all reachable states of the product are marked with $\reach$.
Additionally,
if there is a bad lasso (that has $(q_b,t)$ with $q_b \in F$ inside its loop),
then the query will have an unsatisfiable cycle of constraints $(q_b,t) > ... \geq (q_b,t)$.
Note that this query is satisfiable iff there exists functions $\tau$ and $out$
(and $\rank$ and $\reach$) such that the product does no have a bad lasso~\cite{BS}\ak{which theorem there?}.

\parbf{Solving the SMT query}
To solve the query one can use e.g.\ Z3 solver~\cite{Moura08}.

\begin{remark}[Checking unrealisability]\label{rem:bosy-unreal}
When a given LTL specification is unrealisable,
the above procedure iterates through all system sizes up to a bound $O(2^{2^{|\varphi|}})$.
The bound is computationally difficult to reach on non-toy unrealisable specifications,
making the approach impractical.
To overcome this, we can use the determinacy of the LTL synthesis problem, which states:
an LTL specification is unrealisable iff the dual LTL specification is realisable.
For a specification $\tpl{I,O,\varphi,\textit{Moore}}$
the specification $\tpl{O,I,\neg\varphi,\textit{Mealy}}$
is called dual,
i.e., we swap inputs and outputs, negate the formula,
and search for a Mealy machine instead of a Moore machine.
(Mealy machines are just like Moore machines
 except that the output function $out: T\times 2^I \to 2^O$ also depends on inputs.)
\ak{cite this trick}\ak{provide proof}
Thus, instead of iterating for increasing system bound,
we can run two processes in parallel:
one checks for realisability of the original specification,
another checks for realisability of the dual specification.
The process that finishes first,
returns the answer, while the other process is terminated.
This approach is used in most bounded synthesis implementations~\cite{syntcomp}.
\end{remark}

\subsection*{\CTLstar synthesis}

The standard approach to \CTLstar synthesis~\cite{informatio} is:
translate a given \CTLstar formula $\Phi$ into an alternating Rabin tree automaton~\cite{ATA}
with $\approx 2^{|\Phi|}$ many states and $\approx|\Phi|$ many acceptance pairs,
turn it into a nondeterministic Rabin tree automaton~\cite{MS95}
with $\approx2^{2^{|\Phi|}}$ many states and $\approx2^{|\Phi|}$ many acceptance pairs,
and check its non-emptiness.
The latter check is polynomial in the size of the automaton~\cite{EJ99,DBLP:conf/popl/PnueliR89},
i.e., requires $\approx2^{2^{|\Phi|}}$ time.
Thus, the approach gives a 2EXPTIME algorithm.
The lower bound comes from the 2EXPTIME completeness of the LTL synthesis problem~\cite{Pnueli1989,Vardi:1985:IUL:22145.22173}
and the fact that \CTLstar subsumes LTL.

\chapter{Bounded Synthesis for Streett, Rabin, and \CTLstar}\label{chap:bosy:ctlstar}

\hfill {\footnotesize\textit{This chapter is based on joint work with Roderick Bloem~\cite{CTLstarCAV}}.~~~~~~~~}

\begin{quotation}
\noindent\textbf{Abstract.}
SMT-based Bounded Synthesis uses an SMT solver to synthesize systems from 
LTL properties by going through co-\buchi automata. In this chapter,
we show how to extend
the ranking functions used in Bounded Synthesis, and thus the bounded
synthesis approach, to \buchi, Parity, Rabin, and Streett
conditions. We show that we can handle both existential and universal
properties this way, and therefore, that we can extend Bounded
Synthesis to \CTLstar. Thus, we obtain the first Safraless synthesis
approach and the first synthesis tool for (conjunctions of) the acceptance conditions
mentioned above, and for \CTLstar.
\end{quotation}

\section{Introduction}
For Linear Temporal Logic \cite{pnueli1977temporal},
the standard approach to reactive synthesis involves Safra's relatively complex
construction \cite{Safra}
to determinize B\"uchi automata \cite{DBLP:conf/popl/PnueliR89}.
The difficulty to implement the construction has led to the development
of \emph{Safraless} approaches~\cite{KupfermanV05,BS}.
In this chapter, we focus on one such approach, called Bounded Synthesis,
introduced by Finkbeiner and Schewe~\cite{BS}.

The idea behind Bounded Synthesis is the following.
LTL properties can be translated to \buchi automata \cite{DBLP:journals/iandc/VardiW94}
and verification of LTL properties can be reduced to deciding
emptiness of the product of this automaton and the Kripke structure
representing a system
\cite{DBLP:conf/lop/MannaW81,DBLP:conf/focs/WolperVS83}
(see also Section~\ref{defs:synthesis-approaches}).
This product is a \buchi automaton in its own right.
Finkbeiner and Schewe made two important observations:
(1) Using a ranking function,
    the emptiness problem of \buchi automata
    can be encoded as a Satisfiability modulo Theories (SMT) query, and
(2) by fixing its size, the Kripke structure can be left uninterpreted,
    resulting in an SMT query for a system that fulfills the property.
Because the size of the system is bounded by Safra's construction,
this yields an approach to LTL synthesis that is complete in principle.
(Although proofs of unrealizability are usually computed differently.)  

The reduction to SMT used by Bounded Synthesis provides two benefits:
the performance progress of SMT solvers and the flexibility.
The flexibility allows one to easily adapt the SMT constraints,
produced by Bounded Synthesis,
to build semi-complete synthesizers for
distributed~\cite{BS},
self-stabilising~\cite{Nico},
parameterized~\cite{JB12},
assume-guarantee~\cite{Bloem2015},
probabilistic~\cite{bounded-pctl},
and partially implemented systems.

In this chapter, we extend Bounded Synthesis in two directions.

First,
we show how to directly encode into SMT that some path of a system
is accepted by an $X$ automaton, for $X \in \{$\buchi, co-\buchi,
Parity, Streett, Rabin$\}$.
We do this by introducing new ranking functions.
Therefore we avoid the explicit translation of these automata into \buchi automata,
which would be needed if we were to use the original Bounded Synthesis.

Second, we extend Bounded Synthesis to the branching logic \CTLstar.
\CTLstar formulas allow the user to specify \emph{structural} properties of the system.
For example,
if $g$ is system output and $r$ is system input,
then the \CTLstar formula $\AG\EF g$ says that
a state satisfying $g$ is always reachable;
and the \CTLstar formula $\EFG(g \land r)$ says that
a state satisfying $g$ is reachable and it has a loop when reading $r$ that satisfies $g$.
In both cases, the existential path quantifier $\E$ allows us to refrain
from specifying the exact path that leads to such states.
%Another advantage of branching logics is that they allow us
%to synthesize systems with non-deterministic-like behaviours
%controllable by the user (who is the part of the environment).
%For example,
%the \CTLstar formula
%$
%\E(\neg g \land r \land \X (g \land \neg r \land \X\neg g)) \land
%\E(\neg g \land r \land \X (g \land \neg r \land \X g))
%$
%requires the system to be able to react to the request in two ways,
%by granting it once or twice.
%The formula is unrealizable if the system has only one input $r$,
%but it is realizable if we introduce a fresh input which is not used in the formula.
%The system designer might not be interested in exact way how it is implemented
%(so she only introduces a fresh input),
%but the system user wants to know how to trigger both behaviours.
%Thus, the synthesizer, in addition to a system, should
%output the witnesses for each sub-formula quantified with $\E$.\ak{?remove?: tool does not output this}

In this chapter we show \emph{two} Bounded Synthesis approaches for \CTLstar.
First, we show how to use the ranking functions for
$X$ automata to either decide that some path of a system fulfills such
a condition, or that \emph{all} paths of the system do. Once we have
established this fact, we can extend Bounded Synthesis to logics like
\CTLstar by replacing all state subformulas by fresh atomic
propositions and encoding them each by a \buchi automaton. This
approach follows the classical construction
\cite{DBLP:journals/toplas/ClarkeES86} of model checking \CTLstar,
extending it to synthesis setting.
Alternatively, we show that we can
use a translation of \CTLstar to Alternating Hesitant Tree Automata~\cite{ATA}
to obtain a relatively simple encoding to SMT.

Thus, we obtain a relatively simple, Safraless synthesis procedure to
(conjunctions of) various acceptance conditions and \CTLstar. This
gives us a decision procedure that is efficient when the
specification is satisfied by a small system, but is admittedly
impractical at showing unrealizability.
Just like Bounded Synthesis does for LTL synthesis,
it also gives us a semi-decision procedure for
undecidable problems such as distributed
\cite{DBLP:conf/focs/PnueliR90} or parameterized synthesis \cite{JB12,party}.
We have implemented the \CTLstar synthesis
approach in a tool\footnote{%
Available at \url{https://github.com/5nizza/party-elli}, branch ``cav17''.
}
that to our knowledge is the only tool that
supports \CTLstar synthesis.

The chapter is structured as follows.
In the next section we list the definitions that this chapter uses.
Then in Section~\ref{bosy:synt-from-word-automata} we introduce ranking functions that can be used to verify
and synthesize properties expressed as word automata.
Section~\ref{bosy:ctlstar-synt} contains two approaches to Bounded Synthesis for \CTLstar:
Section~\ref{bosy:ctlstar-synt:direct}
describes the direct encoding into SMT,
in the spirit of bottom-up \CTLstar model checking,
while Section~\ref{bosy:ctlstar-synt:aht} describes the approach via hesitant tree automata.
Section~\ref{bosy:experiments} describes the prototype \CTLstar synthesizer and
the experiments that show applicability of the approach for the synthesis of
small monolithic and distributed systems.
%We conclude in Section~\ref{bosy:conclusion}.

%\section{Definitions}
%
%In this chapter we will use all definitions introduced in Section~\ref{chap:defs}.
%For Section~\ref{bosy:synt-from-word-automata},
%the reader may skip alternating tree automata
%(we will use only their universal and nondeterministic variants) and \CTLstar.

\section{Synthesis from B\"uchi, Streett, Rabin, and Parity Automata}\label{bosy:synt-from-word-automata}

In this section we describe how to verify and synthesize
properties described by \buchi, co-\buchi, Parity, Streett, and Rabin
conditions. 
For each acceptance condition $X \in \{$\buchi, co-\buchi,
Parity, Streett, Rabin$\}$, we can handle the question whether (the word defined
by) some path of a system is in the language of a nondeterministic $X$
automata, as well as the question of whether all paths of the system
are in the language defined by a universal $X$ automaton. There does
not appear to be an easy way to mix these queries (``do all paths of
the system fulfill the property defined by a given nondeterministic
automaton?'').

\subsection{Preliminaries on Ranking}\label{sec:ranking}

In the following, given a system $M = (I,O,T,t_0,\tau,\out)$ and
a nondeterministic (universal) word automaton
$A = (2^{I\cup O}, Q, q_0, \delta, acc)$,
we describe how to build an SMT query $\Phi(M,A)$
that is satisfiable iff some path (all paths, resp.) of $M$ are in $L(A)$.
That is, we focus on the verification problem.
When the verification problem is solved, we obtain the solution to the
synthesis problem easily, following the Bounded Synthesis approach:
given an automaton $A$, we ask the SMT solver whether there is a
system $M$ such that $\Phi(M,A)$ is satisfiable.
More precisely, for increasing $k$,
we fix a set of $k$ states and ask the SMT solver for
a transition relation $\tau$ and a labeling $\out$ (and a few more objects)
for which $\Phi(M,A)$ is satisfiable.

%For the following, fix a system $M = (I,O,T,t_0,\tau,out)$ and
%a word automaton $A = (2^{I\cup O}, Q, q_0, \delta, acc)$.

Our constructions use ranking functions.
A \emph{ranking function} is a function $\rank: Q\times T \to D$
for some totally ordered set $D$ with order $\geq$.
%\footnote{A domain $D$ is \emph{well-founded} wrt.\ relation $>$ iff
%  every non-empty subset has a minimal element:
%  $\forall X \subseteq D (X \neq \emptyset): \exists m \in X: \forall m' \in X: m\not>m'$.
%  (Note that $>$ can be partial.)
%  Equivalently,
%  there is no infinite descending chain: $m_1 > m_2 > m_3 > \dots$.%
%}%
A \emph{rank comparison relation} is a (possibly partial) relation
$\GR \subseteq Q \times D \times D$.
In the following,
we write
$d \GR_q d'$
to mean
$(q,d,d') \in \GR$.
We will usually define $\GR$ using $\geq$ and $\rank$.

We will first establish how to use the ranking functions
to check existential and universal properties,
expressed as $\E(A)$ and $\A(A)$.
Then we define the ranking functions and comparisions for the different acceptance conditions,
i.e., for different types of the word automaton $A$.

Given a rank comparison $\GR$,
we define the following formula to check an existential property $\E(A)$:
\begin{equation*}\label{eq:nondet}
\begin{aligned}
\PhiE(M,A) = & \reach(q_0,t_0) \land \\
\bigwedge_{q,t \in Q\times T} \mkern-12mu &\reach(q,t) \impl \mkern-18mu 
\bigvee_{(q,i \cup o, q') \in \delta} \mkern-18mu 
out(t)\!=\!o \land \reach(q'\!,\tau(t,i)) \land \rank(q,t)\triangleright_q\! \rank(q'\!,\tau(t,i)).
\end{aligned}
\end{equation*}
Similarly, to check a universal property $\A(A)$, we define
\begin{equation*}\label{eq:universal}
\begin{aligned}
\PhiA(M,A) = & \reach(q_0,t_0) \land \\
\bigwedge_{q,t \in Q\times T}\mkern-12mu & \reach(q,t) \impl\mkern-18mu
\bigwedge_{(q,i \cup o, q') \in \delta}\mkern-18mu 
out(t)\!=\!o \impl \reach(q',\tau(t,i)) \land \rank(q,t)\triangleright_q\! \rank(q'\!,\tau(t,i)).
\end{aligned}
\end{equation*}
In these formulas,
\li 
\- the free variable $\reach: Q\times T \to \bbB$
   is an uninterpreted function that marks reachable states
   in the product of $M$ and $A_\E$ or $A_\A$,
   where $A_\E$ and $A_\A$ are the tree automata for $\E(A)$ and $\A(A)$
   (defined on page~\pageref{page:defs:tree_variants}),
   and

\- the free variable $\rank: Q\times T \to D$ is an uninterpreted ranking function. 
\il

Intuitively,
$\PhiE$ will be used to encode that there is an accepting loop in the product automaton,
while $\PhiA$ will be used to ensure that all loops are accepting.

Given a path $\pi=(q_1,t_1) (q_2,t_2)\dots \in (Q \times T)^\omega$,
a rank comparison relation $\GR$,
totally ordered set $D$, and
a ranking function $\rho$,
\emph{$\pi$ satisfies $\GR$ using $\rank$ and $D$},
written $(\pi,D,\rank) \models \GR$,
iff
$\rank(q_i,t_i) \GR_q \rank(q_{i+1},t_{i+1})$ holds for every $i$.

Let us look at the properties of these equations.

\begin{lemma}\label{le:nondet-eq}
For every totally ordered set $D$,
rank comparison relation $\GR$,
ranking function $\rank$,
nondeterministic word automaton $A$, and machine $M$:
  $\PhiE(M,A)$ is satisfiable using $\rank$ and $D$
  iff
  the product $M \otimes A_\E$
  has an infinite path that satisfies $\GR$ using $\rank$ and $D$.
\end{lemma}
\begin{proof}[Proof idea]
Direction $\Leftarrow$.
Let us assume that the product contains a path $\pi=(q_1,t_1) (q_1,t_1)\dots$ such that $(\pi,D,\rank) \models \GR$.
By definition, $\rank(q_i,t_i)\GR\rank(q_{i+1},t_{i+1})$ holds for every $i$.
If we set $\reach(q,t)$ to true for $(q,t) \in \pi$ and to false for all the other states,
then the formula $\PhiE(M,A)$ holds.

Direction $\Rightarrow$.
Let us assume that $\PhiE$ is satisfiable,
then there is a model for $\reach$.
We can use $\reach$ to construct a lasso-shaped infinite path $\pi$
such that $(\pi,D,\rank) \models \GR$ and that belongs to the product.
\end{proof}
A similar result holds for universal word automata.
\begin{lemma}\label{le:universal-eq}
For every well-founded domain $D$,
rank comparison relation $\triangleright$,
ranking function $\rank$,
universal word automaton $A$, and machine $M$:
  $\PhiA(M,A)$ is satisfiable using $\rank$ and $D$
  iff
  in the product $M \otimes A_\A$
  every infinite path satisfies $\triangleright$ using $\rank$ and $D$.
\end{lemma}
\begin{proof}[Proof idea]
Direction $\Leftarrow$.
If we set $\reach$ to true for every $(q,t)$ reachable in the product $M \otimes A_\A$,
then $\PhiA(M,A)$ holds.

Direction $\Rightarrow$.
(Note that $\reach$ may mark some $(q,t)$ with true,
 although it is not reachable in the product $M \otimes A_\A$.
 But for any reachable $(q,t)$, $\reach(q,t)$ holds.)
We prove this direction by contradiction.
Assume that there is an infinite path $\pi = (q_1,t_1)(q_2,t_2)\dots$
such that $(\pi,D,\rank) \not\models \GR$.
Hence there is $i$ such that $\neg\big(\rank(q_i,t_i)\GR \rank(q_{i+1},t_{i+1})\big)$.
Since $(q_i,t_i)$ is reachable (thus $\reach(q_i,t_i)=\true$) and $M\otimes A_\A$ has a transition into $(q_{i+1},t_{i+1})$,
this falsifies $\PhiA(M,A)$ when using $\rank$.
Contradiction.
\end{proof}

These two lemmas will help us to establish the main results:
$M \models A_\E$ whenever $\PhiE(M,A)$ is satisfiable, and
$M \models A_\A$ whenever $\PhiA(M,A)$ is satisfiable,
where the word automata $A$ are nondeterministic and universal respectively,
with different acceptance conditions,
and the form of $\GR$ in $\PhiE$ and $\PhiA$ depends on the acceptance condition.
In the next sections, we describe rank comparison relations
$\triangleright$ for the acceptance conditions B\"uchi, co-B\"uchi,
Streett, Rabin, and Parity.
%We prove correctness only for Streett acceptance%
%---the cases of B\"uchi, co-B\"uchi, and Parity follow as special cases.
%The correctness of $\triangleright$ constraints for
%Rabin acceptance follows from the work of Piterman et al.~\cite{Nir06}.
For didactic purposes,
let us start with the relatively simple B\"uchi and co-\buchi conditions.

\subsection{Ranking for B\"uchi Automata}

B\"uchi conditions were also presented in \cite{VacSy}
and implicitly in \cite{bounded-pctl}.
Given a B\"uchi automaton $A = (2^{I\cup O}, Q, q_0, \delta, F)$, we
define the rank comparison relation $\triangleright^A_B$ as
\begin{equation}\label{eq:rank:buchi}
\rank(q,t) \triangleright^A_B \rank(q',t') =
\begin{cases}
  \true                    & \text{if } q \in F,\\
  \rank(q,t)>\rank(q',t')  & \text{if } q \not\in F.
\end{cases}
\end{equation}

\begin{theorem}[\cite{VacSy,bounded-pctl}]
Let $D$ be $\bbN_0$.
For every universal B\"uchi word automaton $U$,
nondeterministic B\"uchi word automaton $N$, and
machine $M$:
\li
\- $M \models \E(N)$
   iff~
   $\PhiE(M,N)$ is satisfiable,
   where $\GR=\triangleright^N_B$.

\- $M \models \A(U)$
   iff~
   $\PhiA(M,U)$ is satisfiable,
   where $\GR=\triangleright^U_B$.
\il
\end{theorem}
\begin{proof}[Proof idea]
Consider the first item, direction $\Leftarrow$.
If $\PhiE(M,N)$ is satisfiable,
then, using the model of $\reach$, we can extract a lasso-shaped path $\pi = (q_1,t_1)(q_2,t_2)\dots$ of $M \otimes N_\E$,
which satisfies $\GR$ for every $i$.
Such a path visits at least one accepting state of $N$ in its loop part and therefore is B\"uchi accepting.

Consider the direction $\Rightarrow$.
There is an infinite path of $M \otimes N$, in the shape of a lasso,
that has an accepting state in its loop.
We set $\reach(q,t)=\true$ for every state $(q,t)$ visited on the lasso-path,
and set $\rank(q,t)$ to the shortest distance to an accepting state.
Such $\reach$ and $\rank$ make $\PhiE(M,N)$ hold.
%By Lemma~\ref{le:nondet-eq},
%there is an infinite path of $M \otimes N_\E$ that satisfies $\GR$.
%Such a path satisfies $\GR$ and therefore it visits an accepting state infinitely often
%(because $\GR$ cannot have an infinite descending chain).

Consider now the case $M \models \A(U)$.
The direction $\Leftarrow$ is simple, consider the direction $\Rightarrow$.
We describe $\rank$ and $\reach$ that make $\PhiA(M,U)$ hold.
For every $(q,t)$ reachable in $M \otimes U_\A$,
let $\reach(q,t)=\true$.
For every reachable $(q,t)$,
let $\rank(q,t)$ be a longest distance to an accepting state.
These $\rank$ and $\reach$ make $\PhiA(M,U)$ hold.
\end{proof}
Note that in the theorem a machine $M$ is either fixed
(then we solve the model checking problem),
or we fix the number of states in $M$ and express it using uninterpreted functions
(then we solve the bounded synthesis problem).
Also note that we used the set $\bbN$ of natural numbers for $D$,
but we could prove the results for some other large-enough well-founded sets.

\subsection{Ranking for co-B\"uchi Automata}

This case was presented in the original paper~\cite{BS} on Bounded Synthesis.
Given a co-\buchi automaton $A = (2^{I\cup O}, Q, q_0, \delta, F)$,
the ranking constraint relation $\triangleright^A_C$
for co-\buchi is defined as
\begin{equation}\label{eq:rank:cobuchi}
\rank(q,t) \triangleright^A_C \rank(q',t') =
\begin{cases}
  \rank(q,t)>\rank(q',t) \text{~~if } q \in F,\\
  \rank(q,t)\geq\rank(q',t') \text{ if } q \not\in F.
\end{cases}
\end{equation}

\begin{theorem}[\cite{BS}]
Let $D$ be $\bbN_0$.
For every universal co-B\"uchi word automaton $U$,
nondeterministic co-B\"uchi word automaton $N$, and
machine $M$:
\li
\- $M \models \E(N)$
   iff~
   $\PhiE(M,N)$ is satisfiable,
   where $\GR=\GR^N_C$.

\- $M \models \A(U)$
   iff~
   $\PhiA(M,U)$ is satisfiable,
   where $\GR=\GR^U_C$.
\il
\end{theorem}
\begin{proof}[Proof idea]
Consider the first item.
Direction $\Rightarrow$:
$M \otimes N_\E$ has an infinite path, in the shape of a lasso,
that never visits a rejecting state in the loop.
We set $\reach(q,t)=\true$ for all reachable $(q,t)$ in the path,
and set $\rank(q,t)$ to be the number of rejecting states visited before entering the loop.
Direction $\Leftarrow$:
From the model of $\reach$ we can construct an infinite path that is accepted by $N$.
Any such path must be accepting, 
because having a rejecting state $(q,t)$ visited infinitely often
implies having an unsatisfiable cycle of constraints $(q,t) > ... \geq (q,t)$.

Consider the case $M \models \A(U)$, direction $\Rightarrow$.
We set $\reach(q,t)=\true$ for every reachable $(q,t)$ in $M \otimes U_\A$,
and set $\rank(q,t)$ to the maximal number of visits to rejecting states among the paths starting from $(q,t)$.
Such a number is finite,
because all paths visit a rejecting state only finitely often.
The direction $\Leftarrow$ holds,
because every rejecting path visits a rejecting state $(q,t)$ infinitely often,
which implies having an unsatisfiable cycle of constraints $(q,t) > ... \geq (q,t)$.
\end{proof}

\subsection{Ranking for Streett Automata}

The ranking below is our contribution.

Fix a Streett automaton
$A = (2^{I\cup O}, Q, q_0, \delta, \{(A_i,G_i)\}_{i \in [k]})$.
We slightly modify the definitions to have
$\rank: Q\times T \to D^k$ and 
$\GR \subseteq Q \times D^k \times D^k$,
i.e.,
the ranking function consists of $k$ components.

The ranking function $\rank: Q \times T \to D^k$
is defined using $k$ components
$\rank_i: Q \times T \to D$,
$\rank(q,t) = \big(\rank_1(q,t),...,\rank_k(q,t)\big)$.
The rank comparison relation $\GR_S: Q \times D^k \times D^k$ is
%$\GR_S = \bigcap_{i \in [k]} \triangleright_S^{A,i}$,
$\rank(q,t) \GR_S^A \rank(q',t') = \bigwedge_{i \in [k]} \left(\rank_i(q,t) \triangleright_S^{A,i} \rank_i(q',t')\right)$,
where
\begin{equation}\label{eq:streett-ranking}
\rank_i(q,t) \triangleright_S^{A,i} \rank_i(q',t') =
\begin{cases}
  \true                            &\text{if } q \in G_i,\\
  \rank_i(q,t)>\rank_i(q',t')      &\text{if } q \in A_i \land q \not\in G_i,\\
  \rank_i(q,t)\geq\rank_i(q',t')   &\text{if } q \not\in A_i \cup G_i.
\end{cases}
\end{equation}

\begin{theorem}
Let $D$ be $\bbN_0$.
For every universal Streett word automaton $U$,
nondeterministic Streett word automaton $N$, and
machine $M$:
\li
\- $M \models \E(N)$
   iff~
   $\PhiE(M,N)$ is satisfiable,
   where $\GR=\GR^N_S$.

\- $M \models \A(U)$
   iff~
   $\PhiA(M,U)$ is satisfiable,
   where $\GR=\GR^U_S$.
\il
\end{theorem}
\begin{proof}[Proof idea]
We prove only the second item, the first item can be proven similarly.
\ak{can be shortened: 1) remove out-edges from acc 2) assign to others the number of visits to a bad state from A}

Direction $\Rightarrow$.
We construct $\rank = (\rank_1,...,\rank_k)$ and $\reach$ that satisfy $\PhiA(M,U)$.
Set $\reach(q,t) =\true$ for all reachable $(q,t)$ in the product $\Gamma = M \otimes U_\A$,
and for unreachable states set $\reach$ to $\false$ and set $\rank = (0,...,0)$.
Now let us remove all unreachable states from $\Gamma$.
Then for each $i \in [k]$, $\rank_i$ is defined as follows.
\li
\- For every $(q,t) \in \cup_{i\in[k]} G_i \times T$,
   let $\rank_i(q,t) = 0$.

\- Define an \emph{SCC $S$ of a graph} to be any maximal subset of the graph states
   such that for any $s \in S$, $s' \in S$,
   the graph has a path $\pi = s,...,s'$ of length $\geq 2$,
   where the length is the number of states appearing on the path.
   Thus, a single-state SCC can appear only if the state has a self-loop.

\- Remove all outgoing edges from every state $(q,t)$ of $\Gamma$ with $q \in G_i$.
   The resulting graph $\Gamma'$ has no SCCs that have a state $(q,t)$ with $q \in \cup_{i \in [k]} A_i$.

\- Let us define the graph $\Gamma''$.
   Let $\mathcal{S}$ be the set of all SCCs of $\Gamma'$.
   Then $\Gamma''$ has the states
   $V_{\Gamma''} =
   \mathcal{S}
   \cup
   \{ \{ s \} \| s \not\in \cup_{S \in \mathcal{S}} S \}$,
   i.e., each state is either an SCC or a singleton-set
   containing a state outside of any SCC
   (but in both cases, a state of $\Gamma'$ is a set of states of $\Gamma$).
   The edges $E_{\Gamma''}$ of $\Gamma''$ are:
   $(S_1,S_2) \in E_{\Gamma''}$
   iff
   $\exists s_1 \in S_1, s_2 \in S_2: S_1 \neq S_2 \land (s_1,s_2) \in E_{\Gamma'}$.
   Intuitively, $\Gamma''$ is a graph derived from $\Gamma$
   by turning all accepting states into leafs,
   and by making SCCs the new states.
   Note that the graph $\Gamma''$ is a DAG.

\- Given a path $\pi = S_1,...,S_m$ in $\Gamma''$,
   let $nb(\pi)$ be the number of ``bad'' states visited on the path,
   i.e., $nb = |\pi \cap \{ \{ (q,t) \} : q \in \cup_{i \in [k]} A_i \}|$.
   Such a number exists since all paths of $\Gamma''$ are finite.

\- For all $(q,t) \in S \in V_{\Gamma''}$ with $q \not\in G_i$,
   let $\rank_i(q,t)$ be the max number of ``bad'' states
   visited on any path from $S$:
   $\rank_i(q,t) = max(\{ nb(\pi) \| \pi \textit{ is a path from $S$} \})$.
   Such a number exists since the number of paths in $\Gamma''$ is finite.
\il
This concludes the direction $\Rightarrow$.

The direction $\Leftarrow$ is proven by contradiction.
Suppose $\PhiA(M,U)$ is satisfiable with some $\reach$ and $\rank = (\rank_1,...,\rank_k)$,
but $M \otimes U_\A$ is empty.
The latter means that there is a lasso-shaped path that is not accepted by some pair $(A_i,G_i)$:
it visits $A_i$ infinitely often but visits $G_i$ only finitely often.
Thus, the loop part of the path contains state $(q,t)$ with $q \in A_i$ but has no states visiting $G_i$.
Recall that such a path is labeled $\true$ by $\reach$,
because $\reach$ over-approximates the set of reachable states.
Altogether this makes $\PhiA(M,U)$ unsatisfiable,
because of the unsatisfiable cycle of constraints $\rank_i(q,t) > ... \geq \rank(q,t)$.
\end{proof}

\begin{remark}[Comparison with ranking from \cite{Nir06}]
Piterman and Pnueli~\cite{Nir06} introduced ranking functions
to solve Streett \emph{games}.
Our ranking functions can be adapted to solve games, too.
(Recall that our SMT encoding describes \emph{model checking}
 with an uninterpreted system.)
It may seem that in the case of games, our construction uses fewer counters
than \cite{Nir06}, but that is not the case.
Given a DSW with $k$ Streett pairs and $n$ states, a winning strategy in the corresponding Streett game
may require a memory of size $k!$.
In this case, the size of the product system$\times$automaton is $k!n$.
Our construction introduces $2k$ counters with the domain $[k!n] \to
[k!n]$ to associate a rank with each state.
In contrast, \cite{Nir06} introduces $k!k$ counters with the domain
$[n] \to [0,n]$.
Encoding these counters into SAT would require
$2k\cdot k!n \cdot log_2(k!n)$ bits for our construction, and
$k!k\cdot n\cdot log_2(n)$ bits for the construction of \cite{Nir06}.
Thus, our construction introduces $2(1\!+\!\frac{log_2(k!)}{log_2(n)}) \approx 2(1\!+\!log_2(log_2(n)))$ times more bits
(the approximation assumes that $k = log_2(n)$ and $n$ is large).
On the positive side, our construction is much simpler.
\end{remark}

\subsection*{Ranking for Parity Automata}

Given a Parity automaton
$A = (2^{I\cup O}, Q, q_0, \delta, p)$
with priorities $0,\dots,k-1$,
it is known that we can translate it into an equivalent Streett automaton with pairs
$(A_1,G_1), \dots , (A_{m/2},G_{m/2})$,
where
%$A_1 = \{q \| p(q)=1 \}$, $G_1=\{q \| p(q)=0\}$,
$A_i=\{q \| p(q)=2i-1 \}$, $G_i=\{q \| p(q) \in \{0,2,\dots ,2i-2\} \}$.
We can then apply the encoding for Streett automata.
The resulting ranking resembles Jurdzi\'nski's progress measure~\cite{jurdzinski2000small}.

\subsection{Ranking for Rabin Automata}

Given a Rabin automaton
$A = (2^{I\cup O}, Q, q_0, \delta, \{(F_i,I_i)\}_{i\in [k]})$
and a system $M = (I,O,T,t_0,\tau,\out)$,
we use ranking constraints described by Piterman and Pnueli~\cite{Nir06} to
construct a rank comparison relation.
%The expression for $\rank(q,t)\triangleright\rank(q',t')$ is more involved.
The ranking function $\rank: Q \times T \to \bbN_0^{2k+1}$
maps a state of the product to a tuple of numbers
$(b, ~j_1,\!d_1, \dots , j_k,\!d_k)$,
where the numbers have the following meaning.
For each $l \in [k]$,
\li
\- $j_l \in [k]$ is the index of a Rabin pair,
\- $b   \in [0,|Q \times T|]$ is an upper bound on the number of times the set $F_{j_1}$ can be visited from $(q,t)$, 
\- $d_l \in [0, |Q\times T|]$ is the maximal distance from $(q,t)$ to the set $I_{j_l}$,
\il

We define the rank comparison relation
$\GR \subseteq Q \times \bbN_0^{2k+1} \times \bbN_0^{2k+1}$
as follows:
$(b, ~j_1,\!d_1, \dots , j_k,\!d_k) \GR_q (b', ~j'_1,\!d'_1, \dots , j'_k,\!d'_k)$
iff
%$\rank(q,t) \GR_q \rank(q',t')$ iff
%%$(b,~ j_1,d_1, \dots , j_k,d_k) \geq_q (b',~ j_1',d_1', \dots , j_k',d_k')$.
%%The latter holds iff 
there exists $l \in [k]$ such that one of the following holds:
\begin{equation}\label{eq:rabin-rank}
\arraycolsep=4pt\def\arraystretch{1.2}
\begin{array}{ccccc}
b > b', &  & & &  \\
(b,\dots,j_{l-1},d_{l-1}) = (b',\dots,j'_{l-1},d'_{l-1})  &\land & j_l > j_l'&\land&  q \not\in \underset{m\in [l-1]}{\cup} F_{j_m},   \\
(b,\dots ,j_l) = (b',\dots, j'_l)  &\land & d_l > d_l' & \land & q \not\in \underset{m\in[l]}{\cup} F_{j_m},    \\
(b,\dots, j_l) = (b',\dots, j'_l)  &\land & q \in I_{j_l} &\land &  q \not\in \underset{m\in[l]}{\cup} F_{j_m}.
\end{array}
\end{equation}
Here is the intutition.
The first line bounds the number of visits to $F_{j_1}$
($b$ decreases each time $F_{j_1}$ is visited).
% Without the first line it would be impossible to finitely visit $F_1$
% for Rabin aut with single pair $(F_1,I_1)$.
The second line limits the changes of order $j_1,\dots ,j_k$
in the rank $(b,~j_1,d_1,\dots ,j_k,d_k)$ to a finite number.
% Without the second line it would be impossible to 
% infinitely often visit $F_1$ (or $F_2$) in the aut with
% two Rabin pairs $(F_1,I_1), (F_2,I_2)$
% that has two paths that satisfy different pairs
Together, these two lines ensure that on any path some $F_m$
is not visited infinitely often.
The third and fourth lines require $I_{j_l}$ to be visited within $d_l$ steps;
once it is visited, the distance $d_l$ can be reset to any number
$\leq |Q\times T|$.

We can encode the rank comparison constraints in Eq.~\ref{eq:rabin-rank} into SMT as follows.
For each of $j_1,\dots,j_k$ introduce an uninterpreted function:
$Q\times T \to [k]$.
For each of $b,d_1,\dots,d_k$ introduce an uninterpreted function:
$Q\times T \to [0,|Q\times T|]$.
Finally, replace in Eq.~\ref{eq:rabin-rank} counters $b,j,d,b',j',d'$ with
expressions $b(q,t)$, $j(q,t)$, $d(q,t)$, $b(q',t')$, $j(q',t')$, $d(q',t')$ resp.

\subsection*{Ranking for Generalized Automata}

The extension to generalized automata is simple:
replace
$\rank(q,t) \GR \rank(q',t')$ with
$\bigwedge_i \rank^i(q,t) \GR^i \rank^i(q',t')$
where $\rank^i$ and $\GR^i$ are for $i$th automaton acceptance component.

\ak{\hl{for all the cases: complexity of the non-emptiness check using SMT!}}

\subsection{Discussion of Ranking}

\ak{check rankings in \url{https://link.springer.com/content/pdf/10.1007\%2F978-3-540-31980-1_14.pdf}. I bet they are the same! ``Complementation Constructions for Nondeterministic Automata on Infinite Words''}

A close work on rankings is the work by Beyene et al.~\cite{RybHornSynth}
on solving \emph{infinite}-state games using SMT solvers.
Conceptually,
they use co-B\"uchi and B\"uchi ranking functions to encode game winning into SMT,
which was also partially done by Schewe and Finkbeiner \cite{BS} a few years earlier
(for finite-state systems).
The authors focused on co-B\"uchi and B\"uchi automata,
while we also considered Rabin and Streett automata
(for finite-state systems).
Although they claimed their approach can be extended to $\mu$-calculus
(and thus to \CTLstar),
they did not elaborate beyond noting that \CTLstar verification can be reduced to games.
In the next section we introduce two approaches to bounded synthesis from \CTLstar.
Both approaches inherit the ideas on rankings presented in this section.

\section{Bounded Synthesis from \CTLstar} \label{bosy:ctlstar-synt}

We describe two ways to encode model checking for \CTLstar into SMT.
The first one, direct encoding (Section~\ref{bosy:ctlstar-synt:direct}),
resembles bottom-up \CTLstar model checking \cite{DBLP:journals/toplas/ClarkeES86} (see also page~\ref{page:defs:bottom-up-mc}).
The second encoding (Section~\ref{bosy:ctlstar-synt:aht})
follows the automata-theoretic approach~\cite{ATA}
(see also Section~\ref{defs:model-checking-approach})
and goes via hesitant tree automata.
As usual,
replacing a concrete system function with an uninterpreted one of a fixed size gives a
bounded synthesis procedure.

Let us compare the approaches.
In the direct encoding,
the main difficulty is the procedure that generates the constraints:
we need to walk through the formula and
generate constraints for nondeterministic B\"uchi or
universal co-B\"uchi sub-automata.
In the approach via hesitant tree automata,
we first translate a given \CTLstar formula into a hesitant tree automaton $A$,
and then encode the non-emptiness problem of the product of $A$ 
and the system into an SMT query.
In contrast to the direct encoding,
the difficult---from the implementation point of view---part is to construct the automaton $A$,
while the definition of the rank comparison relation is very easy.

In the next section we define \CTLstar with inputs
and then describe two approaches.
The approaches are conceptually the same,
thus automata fans are invited to read Section~\ref{bosy:ctlstar-synt:aht}
about the approach using hesitant automata,
while the readers preferring bottom-up \CTLstar model checking are welcomed to Section~\ref{bosy:ctlstar-synt:direct}.

\subsection{Direct Encoding}\label{bosy:ctlstar-synt:direct}

We reduce the \CTLstar model checking problem into SMT
following the classical bottom-up model checking approach (see page~\pageref{page:defs:bottom-up-mc}).
%except that a system is expressed using uninterpreted functions.

Let $M=(I,O,T,t_0,\tau,out)$ be a machine
%whose transition and output functions
%$\tau: T \times 2^I \to T$ and
%$out: T \to 2^O$
%we want to synthesize,
%and the bound $|T|$ is given.
and $\Phi$ be a \CTLstar state formula (in positive normal form).
We use the notions of $F$, $P$, and $\widetilde\Phi$ defined on page~\pageref{page:defs:bottom-up-mc}:
recall that with every state subformula $\E\varphi$ or $\A\varphi$
we associate a Boolean proposition,
whose truth in a system state $t$ implies that the corresponding subformula holds.
The set $P=\{p_1,...,p_k\}$ is the set of such propositions,
the set $F=\{f_1,...,f_k\}$ is the set of subformulas corresponding to $\{p_1,...,p_k\}$
(note that each $f_i$ is of the form $\A\varphi$ or $\E\varphi$ and $\varphi$ has no path quantifiers),
and $\widetilde\Phi$ is the top-level Boolean formula.
We define the SMT query as follows.
\li

\-[(1)] The query talks about uninterpreted functions
$\reach: Q_\textit{all} \times T \to \bbB$,
$\rank: Q_\textit{all} \times T \to \bbN$,
$\tau: T \times 2^I \to T$,
$out: T \to 2^O$, and
$p: T \to 2^P$.
What is $Q_\textit{all}$ will become clear later.

%\-[1)] Replace each subformula $f_i$ of $\Phi$ that starts
%with \A or \E with a new proposition $p_i$, and let $P = \{p_1,\dots,p_k\}$.
%Thus,
%each $p_i$ corresponds to some formula $\E\varphi_i$ or $\A\varphi_i$,
%where $\varphi_i$ is a path formula over $I \cup O \cup P$.

\-[(2)] For each $f \in \{f_1, \dots, f_k\}$,
we do the following.
If $f$ is of the form $\A \varphi$, we translate $\varphi$ into a UCW%
\footnote{To translate $\varphi$ into a UCW,
          translate $\neg\varphi$ into an NBW and treat it as a UCW.},
otherwise into an NBW;
let the resulting automaton be
$A_\varphi = (2^{I \cup O \cup P},Q,q_0,\delta,F)$.
Note that $\delta \subseteq Q \times 2^I \times 2^O \times 2^P \times Q$,
and it depends on $P$.
For every $(q,t) \in Q \times T$, the query contains the constraints:
\li
\-[(2a)] If $A_\varphi$ is an NBW, then:
\begin{equation*}
\reach(q,t) ~\impl
\bigvee_{(e, q') \in \delta(q,out(t),p(t))}
\reach(q',t') \land \rank(q,t)\triangleright_{\textit{\tiny B}} \rank(q',t')
\end{equation*}

\-[(2b)] If $A_\varphi$ is a UCW, then:
\begin{equation*}
\reach(q,t) ~\impl
\bigwedge_{(e, q') \in \delta(q,out(t),p(t))} %\mkern-10mu
\reach(q',t') \land \rank(q,t)\triangleright_{\textit{\tiny C}} \rank(q',t')
\end{equation*}
\il
In both cases, we have:
  $p(t)=\{p_i \in P \| \reach(q_0^{p_i},t)=\true \}$,
  $q_0^{p_i}$ is the initial state of  $A_{\varphi_i}$,
  $\triangleright_{\textit{\tiny B}}$ and $\triangleright_{\textit{\tiny C}}$
  are the B\"uchi and co-B\"uchi rank comparison relations
  wrt.\ $A_\varphi$ (see Eq.~\ref{eq:rank:buchi}--\ref{eq:rank:cobuchi}),
  and $t' = \tau(t,i)$.
  Intuitively
  $p(t)$ under-approximates the subformulas that hold in $t$:
  if $p_i \in p(t)$, then $t \models f_i$.

\-[(3)] The query contains the constraint
       $\widetilde\Phi[p_i \mapsto \reach(q_0^{p_i},t_0)]$,
       where $q_0^{p_i}$ is the initial state of $A_{\varphi_i}$.
       For example,
       for $\Phi = g \land \AGEF \neg g$ where $g \in O$,
       the constraint is $g(t_0) \land \reach(q_0^{p_2}, t_0)$,
       where
       $p_2$ corresponds to $\AG p_1$,
       $p_1$ corresponds to $\EF \neg g$.
\il

%% \begin{remark}[$\reach(q',t')$ vs. $\reach(q_l,t)$]
%% On the right side of Eq. \ref{eq:direct-ctl-nbw} and \ref{eq:direct-ctl-ucw}
%% for the main automaton there is ``$\reach(q',t')$'' while for subautomata
%% --- ``$\reach(q_l,t)$''.
%% This is the consequence of the semantics of \CTLstar with inputs
%% (also see Remark~\ref{rem:inputs-vs-outputs}):
%% the ``current'' direction does not affect ``adjacent'' state formulas.
%% \end{remark}

\begin{example}\label{ex:direct-encoding}
Let $I=\{r\}$, $O=\{g\}$, $\Phi = g \land \AGEF\neg g$.
We associate $p_1$ with $\EF\neg g$ and $p_2$ with $\AG p_1$.
Automata for $p_1$ and $p_2$ are in Fig.~\ref{fig:direct-encoding:ex:automata}, the SMT constraints are in Fig.~\ref{fig:direct-encoding:ex:constraits}.
\end{example}
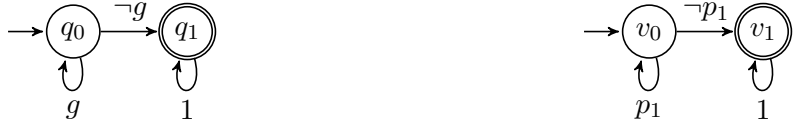
\begin{figure}[tb]\center
\begin{subfigure}[b]{0.45\textwidth}\center
\begin{tikzpicture}[->,>=stealth',shorten >=1pt,auto,node distance=1.5cm,semithick]
  \tikzset{every state/.style={minimum size=7mm,inner sep=0mm}, initial text={}}
  \tikzstyle{dot}=[draw,circle,minimum size=1mm,inner sep=0pt,outer sep=0pt,fill=black]
  \tikzstyle{small}=[minimum size=1mm,inner sep=1pt,outer sep=0pt]
  \tikzstyle{every edge} = [align=center,draw=black]

  \node[state,initial] (q0) {$q_0$};
  \node[state,double] (q1) [right of=q0] {$q_1$};

  \path
  (q0) edge [loop below] node {$g$} (q0)
  (q0) edge node {$\neg g$} (q1)
  (q1) edge [loop below] node {$1$} (q1);
\end{tikzpicture}
\caption{NBW for $\F\neg g$, associated with $p_1$ that encodes the truth of $\EF\neg g$.}
\label{ex:direct-encoding:nbw}
\end{subfigure}
~~~~~
\begin{subfigure}[b]{0.45\textwidth}\center
\begin{tikzpicture}[->,>=stealth',shorten >=1pt,auto,node distance=1.5cm,semithick]
  \tikzset{every state/.style={minimum size=7mm,inner sep=0.3mm}, initial text={}}
  \tikzstyle{dot}=[draw,circle,minimum size=1mm,inner sep=0pt,outer sep=0pt,fill=black]
  \tikzstyle{small}=[minimum size=1mm,inner sep=1pt,outer sep=0pt]
  \tikzstyle{every edge} = [align=center,draw=black]

  \node[state,initial] (v0) {$v_0$};
  \node[state,double] (v1) [right of=v0] {$v_1$};

  \path
  (v0) edge [loop below] node {$p_1$} (v0)
  (v0) edge node {$\neg p_1$} (v1)
  (v1) edge [loop below] node {$1$} (v1);
\end{tikzpicture}
\caption{UCW for $G p_1$, associated with $p_2$ that encodes the truth of $\AG p_1$.}
\label{ex:direct-encoding:ucw}
\end{subfigure}
\caption{Automata for Example~\ref{ex:direct-encoding}}
\label{fig:direct-encoding:ex:automata}
\end{figure}
\begin{figure}[t]\center
\small
\begin{align*}
\text{initial}:~~~
& g(t_0) \land \reach(q_0,t_0) \\
\edge{v_0}{p_1}{v_0}:~~~
& \reach(v_0,t) \land \reach(q_0,t)
  \impl
  \bigwedge_{r\in \bbB} 
  \reach(v_0,\tau(t,r)) \land \rank(v_0,t) \geq \rank(v_0,\tau(t,r)) \\
\edge{v_0}{\neg p_1}{v_1}:~~~
& \reach(v_0,t) \land \neg\reach(q_0,t)
  \impl
  \bigwedge_{r\in \bbB} 
  \reach(v_1,\tau(t,r)) \land \rank(v_0,t) \geq \rank(v_1,\tau(t,r)) \\
\edge{v_1}{\true}{v_1}:~~~
& \reach(v_1,t) 
  \impl
  \bigwedge_{r\in \bbB}
  \reach(v_1,\tau(t,r)) \land \rank(v_1,t) > \rank(v_1,\tau(t,r)) \\
\edge{q_0}{g}{q_0}:~~~
& \reach(q_0,t) \land g(t)
  \impl
  \bigvee_{r\in \bbB} 
  \reach(q_0,\tau(t,r)) \land \rank(q_0,t) > \rank(q_0,\tau(t,r)) \\
\edge{q_0}{\neg g}{q_1}:~~~
& \reach(q_0,t) \land \neg g(t)
  \impl
  \bigvee_{r\in \bbB} 
  \reach(q_1,\tau(t,r)) \land \rank(q_0,t) > \rank(q_1,\tau(t,r)) \\
\edge{q_1}{\true}{q_1}:~~~
& \reach(q_1,t) 
  \impl
  \bigvee_{r\in \bbB}
  \reach(q_1,\tau(t,r))
\end{align*}
\caption{SMT constraints for Example~\ref{ex:direct-encoding} for some $t \in T$.
  The final query is the conjunction of the constraints for every $t \in T$.
  The first line is the initialisation, item (3).
  The second line encodes the transition $\edge{v_0}{p_1}{v_0}$ of the automaton in Fig.~\ref{ex:direct-encoding:ucw}
  and corresponds to item (2b):
  since we do not know whether $p_1$ holds in state $t$,
  we add the assumption $\reach(q_0,t)$.}
\label{fig:direct-encoding:ex:constraits}
\end{figure}

\begin{theorem}[Correctness of direct encoding]
Given a \CTLstar formula $\Phi$ over inputs $I$ and outputs $O$
and a system $M=(I,O,T,t_0,\tau,out)$:
$M \models \Phi$
iff the SMT query is satisfiable.
\end{theorem}
Here is the intuition behind the proof.
The standard bottom-up model checker (see page~\ref{page:defs:bottom-up-mc})
marks every system state with state subformulas it satisfies.
The model checker returns ``Yes'' iff the initial state
satisfies the top-level Boolean formula.
The direct encoding conceptually follows that approach.
If for some system state $t$, $\reach(q_0^{p_i},t)$ holds,
then $t$ satisfies the state formula $f_i$ corresponding to $p_i$.
Thus, if the top-level Boolean constraint (3) holds,
then $t_0 \models \Phi$.
And vice versa:
if a model checker returns ``Yes'',
then the marking it produced can be used to satisfy the SMT constraints.
Finally, the positive normal form of $\Phi$
allows us to get away with encoding of positive obligations only
($\reach(q_0^{p_i},t) \Impl t \models f_i$),
eliminating the need to encode
$\neg\reach(q_0^{p_i},t) \Impl t \models \neg f_i$.

\subsection{Encoding via Alternating Hesitant Tree Automata}\label{bosy:ctlstar-synt:aht}

%We recall how model checking can be done in the case of LTL properties
%(see also Section~\ref{defs:model-checking-problem}).
%First, we convert a given LTL formula into a universal Co-B\"uchi tree automaton.
%Then we build the product between a system and the tree automaton---%
%such a product is a \emph{1-letter} universal Co-B\"uchi \emph{word} automaton.
%Then we check the non-emptiness of the product automaton.
%We can encode the latter question into an SMT query.
%To get a bounded synthesiser,
%we replace a given system with an unknown system of a fixed size.

Let us recall how we can model check and synthesize systems from \CTLstar formulas
(see also Section~\ref{defs:model-checking-approach}).
First,
we convert a given \CTLstar formula into an alternating hesitant tree automaton.
Then we build the product between the system and the automaton---%
such a product is a 1-letter alternating hesitant word automaton.
Then we check the non-emptiness of the product automaton.
We show how to encode the latter check into an SMT query.
Such an SMT query is satisfiable iff the product is non-empty
(thus the system satisfies the formula).
As before, if we want to do synthesis, we replace a given system
with an unknown system of a fixed size.
Then an SMT solver returns a model (from which we extract a system), if such exists,
together with a proof of the non-emptiness.

It is worth refreshing the following definitions:
AHT and AHW (Chapter~\ref{chap:defs}, pages \pageref{page:defs:aht} and \pageref{page:defs:ahw}),
1-AHW and model checking wrt.\ \CTLstar (Section~\ref{defs:model-checking-approach}).

\subsubsection{Encoding non-emptiness of the product into SMT}

We start by converting a given \CTLstar formula $\Phi$ into an AHT.
Then we build the product between a given system and the AHT.
Such a product is a 1-AHW.
We are going to encode the non-emptiness of the 1-AHW into an SMT query.

Let us explain the idea of the encoding.
Recall that the states of the 1-AHW can be partitioned into
``existential'' sets $Q^N_1, ..., Q^N_{k_N}$ and
``universal'' sets $Q^U_1, ..., Q^U_{k_U}$.
Such sets are ordered and the 1-AHW transition function ensures the following:
Every path in every run-tree of the 1-AHW
gets trapped in some $Q^N_i$ or in $Q^U_j$.
Such a path $\pi$ is accepting iff
$\Inf(\pi) \cap Acc \neq \emptyset$ for the case of $Q^N_i$ (B\"uchi acceptance)
or $\Inf(\pi) \cap Acc = \emptyset$ for the case of $Q^U_j$ (co-B\"uchi acceptance).
We will build an SMT query where the SMT solver has to:
(a) resolve nondeterminism in the 1-AHW,
(b) ensure that every path in the resulting universal word automaton is accepting.

Consider a system
$M=(I,O,T,t_0,\tau,out)$
and an AHT
$A = (2^O, 2^I, Q, q_0, \delta: Q \times 2^O \to {\cal B}^+(2^I\times Q), Acc \subseteq Q)$
that corresponds to a given \CTLstar formula $\Phi$.
We encode the non-emptiness of the product $M \otimes A$,
which has the states $Q\times T$,
into the following SMT query:
\begin{equation}\label{eq:ctl-haa}
\begin{aligned}
&\reach(q_0,t_0) \land \\
\bigwedge_{(q,t) \in Q\times T}\!\!\! &\reach(q,t) \impl
\delta\big(q,out(t)\big) ~\Big[(d,q') \mapsto \reach\!\left(q',\tau(t,d)\right) \land \\
&~~~~~~~~~~~~~~~~~~~~~~~~~~~~~~~~~~~~~~~~~~~~
\rank(q,t) ~\triangleright_{q,q'}~ \rank\!\left(q',\tau(t,d)\right)\Big]
\end{aligned}
\end{equation}
where $\triangleright_{q,q'}$\footnote{%
  Here $\triangleright_{q,q'}$
  depends on $q$ \emph{and} $q'$,
  but it can also be defined to depend on $q$ only, as it is originally introduced.%
}
is:
\li
\- if $q$ and $q'$ are in the same $Q_i^N$,
   then the B\"uchi rank comparison $\GR^{Q_i^N}_B$;
\- if $q$ and $q'$ are in the same $Q_i^U$,
   then the co-B\"uchi rank comparison $\GR^{Q_i^N}_C$;
\- otherwise, true.
\il

\begin{theorem}
Given a system $(I,O,T,t_0,\tau,out)$ and \CTLstar formula $\Phi$ over inputs $I$ and outputs $O$:
$system \models \Phi$ iff the SMT query in Eq.~\ref{eq:ctl-haa} is satisfiable.
\end{theorem}
\begin{proof}[Proof idea]
Direction $\Rightarrow$.
Let $(Q, q_0, \delta, Acc)$ be the 1-AHW representing the product system$\otimes$AHT.
We will use the following observation.
\begin{quote}
Observation: % (e.g., \cite[Section5]{Visser2000}):
{\em The 1-AHW non-emptiness can be reduced to solving the following 1-Rabin game.}
The game states are $Q$, the game graph corresponds to $\delta$,
there is one Rabin pair $(F,I)$ with
$F = Acc \cap Q^U$, $I = (Acc \cap Q^N) \cup (Q^U\!\setminus\! Acc$).
Let us view $\delta$ to be in the DNF.
Then, in state $q$ of the game,
the ``existential'' player (Automaton) chooses a disjunct in $\delta(q)$,
while the ``universal'' player (Pathfinder) chooses a state in that disjunct.
Automaton's strategy is winning iff for any Pathfinder's strategy the resulting play
satisfies the Rabin acceptance $(F,I)$.
Note that Automaton has a winning strategy iff the 1-AHW is non-empty;
also, memoryless strategies suffice for Automaton.
\end{quote}
Since the 1-AHW is non-empty, Automaton has a memoryless winning strategy.
We will construct $\reach$ and $\rank$ from this strategy.
For $\reach$:
set it to true if there is a strategy for Pathfinder such that the state will reached.
Let us prove that $\rank$ exists.

Since states from different $Q_i$ can never form a cycle
(due to the partial order),
$\rank$ of states from different $Q_i$ are independent.
Hence we consider two cases separately:
$\rank$ for some $Q^N_i$ and for some $Q^U_i$. 
\li
\-
The case of $Q^N_i$ is simple:
by the definition of the 1-AHW,
we can have only simple loops within $Q^N_i$.
Any such reachable loop visits some state from $Acc \cap Q_i^N$.
Consider such a loop:
assign $\rank$ for state $q$ of the loop to be
the minimal distance from any state $Acc \cap Q_i^N$.

\-
The case of $Q^U_i$:
in contrast, we can have simple and non-simple loops within $Q^U_i$.
But none of such loops visits $Acc \cap Q_i^U$.
Then, for each $q \in Q^U_i$ assign $\rank$
to be the maximum bad-distance from any state of $Q^U_i$.
The bad-distance between $q$ and $q'$ is
the maximum number of $Acc \cap Q^U$ states visited
on any path from $q$ to $q'$.
\il

Direction $\Leftarrow$.
The query is satisfiable means there is a model for $\reach$.
Note that the query is Horn-like ($\ldots  \rightarrow \ldots $),
hence there is a minimal marking $\reach$ of states
that still satisfies the query%
\footnote{Minimal in the sense that it is not possible
          to reduce the number of $\reach$ truth values
          by falsifying some of $\reach$.}%
\footnote{Non-minimality appears when $\delta$
          of the alternating automaton has OR
          and the SMT solver marks with $\reach$
          more than one OR argument.
          Another case is when the solver
          marks some state with $\reach$
          but there is no antecedent requiring that.}.
Wlog., assume $\reach$ is minimal.
Consider the subset of the states of the 1-AHW that are marked with $\reach$,
and call it $U$.
Note that $U$ is a 1-AHW and it has only universal transitions
(i.e., we never mark more than one disjunct of $\delta$ on the right side of
 $\ldots \rightarrow \delta(\ldots)$).
Intuitively,
$U$ is a finite-state representation of the (infinite) run-tree of the original 1-AHW.

\begin{quote}
Claim: {\em the run-tree---the unfolding of $U$---is accepting}.
Suppose it is not:
there is a run-tree path that violates the acceptance.
Consider the case when the path is trapped in some $Q^U_i$.
Then the path visits a state in $Q^U_i \cap Acc$
infinitely often.
But this is impossible since we use co-B\"uchi ranking for $Q_i^U$.
Contradiction.
The case when the path is trapped in some $Q_i^N$ is similar%
---the B\"uchi ranking prevents from not visiting $Acc\cap Q_i^N$ infinitely often.
\end{quote}
Thus, the 1-AHW is non-empty since it has an accepting run-tree (the unfolding of $U$).
\end{proof}

\subsection{Prototype Synthesizer for $\CTLstar$}\label{bosy:experiments}

We implemented both approaches to \CTLstar synthesis described
in Sections~\ref{bosy:ctlstar-synt:direct} and \ref{bosy:ctlstar-synt:aht}
inside the tool PARTY~\cite{party}:
{\small\url{https://github.com/5nizza/party-elli}} (branch ``cav17'').
In this section we illustrate the approach via AHTs.

The synthesizer works as follows:
\li
\-[(1)] Parse the specification that describes inputs, outputs,
       and a \CTLstar formula $\Phi$.
\-[(2)] Convert $\Phi$ into a hesitant tree automaton using the procedure
       described in~\cite{ATA},
       using LTL3BA~\cite{LTL3BA} to convert path formulas into NBWs.
\-[(3)] For each system size $k$ in increasing order (up to some bound):
       \li
       \- encode ``$\exists M_k: M_k\otimes AHT \neq \emptyset$?''
          into SMT using Eq. \ref{eq:ctl-haa} where $|M_k|=k$
       \- call Z3 solver \cite{Moura08}:
       if the solver returns ``unsatisfiable'',
          goto next iteration;
       otherwise
          print the system in the dot graph format.
       \il
\il
This procedure is complete,
because there is a $O(2^{2^{|\Phi|}})$ bound on the size of the system,
although reaching it is impractical.

\begin{figure}
\center
\begin{subfigure}[b]{0.45\textwidth}\center
\begin{tikzpicture}[->,>=stealth',shorten >=1pt,auto,node distance=1.9cm,semithick]
  \tikzset{every state/.style={minimum size=7mm,inner sep=0mm}, initial text={}}
  \tikzstyle{dot}=[draw,circle,minimum size=1mm,inner sep=0pt,outer sep=0pt,fill=black]
  \tikzstyle{small}=[minimum size=1mm,inner sep=1pt,outer sep=0pt]
  \tikzstyle{every edge} = [align=center,draw=black]

  \node[state,initial,green] (q0) {$q_0$};
  \node[dot] (q0_nG) [below of=q0] {};
  \node[state,double,green] (p0) [below left of=q0_nG] {$p_0$};
  \node[state,red] (r0) [below of=q0_nG] {$r_0$};
  \node[state,red] (s0) [right of=q0_nG] {$s_0$};

  \node[dot] (p0_nG) [below of=p0] {};
%  \node[small] (p0_G) [left of=p0] {$\bot$};

  \node[dot] (r0_nG) [below of=r0] {};
  \node[dot] (r0_G) [below right of=r0] {};

  \node[state,green] (t0) [below left of=r0_G] {$t_0$};
  \node[dot] (t0_G) [left of=t0] {};
  \node[small] (t0_nG) [right of=t0] {$\top$};
  
  \node[dot] (s0_G) [above right of=s0] {};
  \node[dot] (s0_nG) [below right of=s0] {};

  \node[state,red,double] (s1) [below of=s0_nG] {$s_1$};
  \node[dot] (s1_nG) [below of=s1] {};

  \path
  (q0) edge node {$\neg g$} (q0_nG)
  (q0_nG) edge node {$\A r$} (s0)
  (q0_nG) edge[near end] node {$\A r$} (r0)
  (q0_nG) edge[left] node {$\E r$} (p0)
  (p0) edge node {$\neg g$} (p0_nG)
  (p0_nG) edge[bend left] node {$\E r$} (p0)
%  (p0) edge node {$g$} (p0_G)
  (r0) edge node {$\neg g$} (r0_nG)
  (r0) edge node {$g$} (r0_G)
  (r0_G) edge node {$\E r$} (t0)
  (r0_nG) edge[bend left] node {$\A r$} (r0)
  (t0) edge node {$g$} (t0_G)
  (t0_G) edge[bend left] node {$\E r$} (t0)
  (s0) edge [sloped,below] node {$\neg g$} (s0_nG)
  (s0) edge [below] node {$g$} (s0_G)
  (s0_G) edge [bend right,above] node {$\A r$} (s0)
  (s0_nG) edge node {$r$} (s1)
  (s0_nG) edge [bend right,right] node {$\A r$} (s0)
  (s1) edge node {$\neg g$} (s1_nG)
  (t0) edge [below] node {$\neg g$} (t0_nG)
  (q0_nG) edge node {$r$} (s1)
  (s1_nG) edge [bend left] node {$\A r$} (s1);
\end{tikzpicture}
\end{subfigure}
\begin{subfigure}[b]{0.45\textwidth}\center
\begin{tikzpicture}[->,>=stealth',shorten >=1pt,auto,node distance=1.3cm,semithick]
  \tikzset{every state/.style={minimum size=7mm,inner sep=0.3mm}, initial text={}}
  \tikzstyle{dot}=[draw,circle,minimum size=1mm,inner sep=0pt,outer sep=0pt,fill=black]
  \tikzstyle{small}=[minimum size=1mm,inner sep=1pt,outer sep=0pt]
  \tikzstyle{every edge} = [align=center,draw=black]

  \node[state,initial,green] (q0) {$q_0,m$};
  \node[dot] (q0_nG) [below of=q0] {};
  \node[state,double,green] (p0) [below left of=q0_nG] {$p_0,m$};
  \node[state,red] (r0) [below of=q0_nG] {$r_0,m$};
  \node[state,red] (s0) [right of=q0_nG] {$s_0,m$};

  \node[dot] (p0_nG) [below of=p0] {};
%  \node[small] (p0_G) [left of=p0] {$\bot$};

  \node[dot] (r0_nG) [below of=r0] {};

  \node[dot] (s0_nG) [below right of=s0] {};

  \node[state,red,double] (s1) [below of=s0_nG] {$s_1,m$};
  \node[dot] (s1_nG) [below of=s1] {};

  \path
  (q0) edge node {} (q0_nG)
  (q0_nG) edge node {} (s0)
  (q0_nG) edge[near end] node {} (r0)
  (q0_nG) edge[left] node {} (p0)
  (p0) edge node{} (p0_nG)
  (p0_nG) edge[bend left] node {} (p0)
%  (p0) edge node {$g$} (p0_G)
  (r0) edge node{} (r0_nG)
  (r0_nG) edge[sloped,bend left,below] node {$\leq$} (r0)
  (s0) edge [sloped,below] node {} (s0_nG)
  (s0_nG) edge [sloped,above] node {$>$} (s1)
  (s0_nG) edge [sloped,bend right,above] node {$\leq$} (s0)
  (s1) edge [above] node {} (s1_nG)
  (q0_nG) edge node {} (s1)
  (s1_nG) edge [sloped,bend left,below] node {$<$} (s1);
\end{tikzpicture}
\end{subfigure}
\caption{On the left: AHT for the \CTLstar formula
         $\EG\neg g \land \AG\EF \neg g \land \AG(r \impl \F g)$.
         Green states are from the nondeterministic partion,
         red states are from the universal partition,
         double states are final
         (a red final state is rejecting, a green final state is accepting).
         State $\top$ denotes an accepting state.
         All transitions going out of the black dots are conjuncted.
         For example,
         $\delta(q_0, \neg g) =
         ((r,p_0) \lor (\neg r,p_0)) \land
         ((r,r_0) \land (\neg r, r_0)) \land
         (r, s_1) \land
         ((r,s_0) \land (\neg r, s_0))$.
         States $s_0$ and $s_1$ describe the property $\AG(r \impl \F g)$,
         state $p_0$ --- $\EG \neg g$,
         states $r_0$ and $t_0$ --- $\AG\EF\neg g$,
         state $t_0$ --- $\EF\neg g$.
         (Note that some states, e.g. $q_0$, do not have a transition for some letters.
          We assume that non-existing transitions go into a non-rejecting self-loop state for red (universal) states,
          and into a non-accepting self-loop state for green (nondeterministic) states.)\\
         On the right side is the product (1-AHW) of the AHT with 
         the one state system that never grants
         (thus it has $\edge{m}{\true}{m}$ and $out(m) = \neg g$).
         The edges are labeled with the relation $\triangleright_{q,q'}$
         defined in Eq.~\ref{eq:ctl-haa}.
         The product has no plausible annotation due to
         the cycle $\edge{(s_1,m)}{>}{(s_1,m)}$,
         thus the system does not satisfy the property.}
\label{fig:aht-resettable-arbiter}
\end{figure}
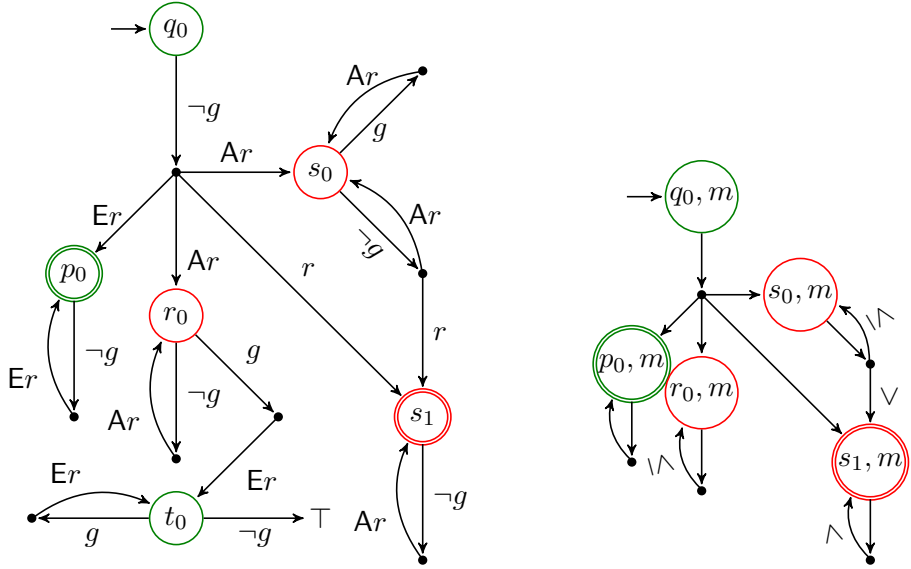
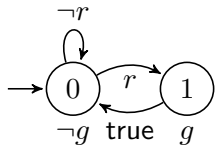
\begin{figure}[hbp]
\centering
\caption{The system that satisfies the property of the resettable 1-arbiter.}
\label{fig:aht-resettable-arbiter:model}
\begin{tikzpicture}[->,>=stealth',shorten >=1pt,auto,node distance=1.5cm,semithick]
  \tikzset{every state/.style={minimum size=7mm,inner sep=0.0mm}, initial text={}}
  \tikzstyle{every edge} = [align=center,draw=black]

  \node[state,initial] (A) [label={below:$\neg g$}] {$0$};
  \node[state] (B) [right of=A, label={below:$g$}] {$1$};

  \path 
  (A) edge [loop above] node {$\neg r$} (A)
  (A) edge [bend left,below] node {$r$} (B)
  (B) edge [bend left,below] node {$\true$} (A);
\end{tikzpicture}
%\vspace{-0.8cm}
\end{figure}

\parbf{Running example: resettable 1-arbiter}
Let $I = \{r\}$, $O=\{g\}$.
Consider a simple \CTLstar property of an arbiter
$$
\EG(\neg g ) \land \AG(r \impl \F g) \land \AG\EF \neg g.
$$
The property says:
there is a path from the initial state where the system never grants
(including the initial state);
every request should be granted;
and finally, a state without the grant should always be reachable.
We now invite the reader to Figure~\ref{fig:aht-resettable-arbiter}.
It contains the AHT produced by our tool,
and on its right side we show the product of the AHT
with the one-state system that does not satisfy the property.
The correct system needs at least two states and is on Figure~\ref{fig:aht-resettable-arbiter:model}.

\parbf{Resettable 2-arbiter}
Let $I = \{r_1, r_2\}$, $O=\{g_1,g_2\}$.
Consider the formula
\begin{align*}
& \EG(\neg g_1 \land \neg g_2) ~\land \AG\EF(\neg g_1 \land \neg g_2) ~\land \\
& \AG(r_1 \impl \F g_1) \land \AG (r_2 \impl \F g_2) \land \AG(\neg (g_1 \land g_2)).
\end{align*}
Note that without the properties with $\E$,
the synthesizer can produce the system in Figure~\ref{fig:vac_resettable_arbiter}
which starts in the state without grants and then always grants one or another client.
Our synthesizer outputs the system in Figure~\ref{fig:resettable_arbiter} (in one second).
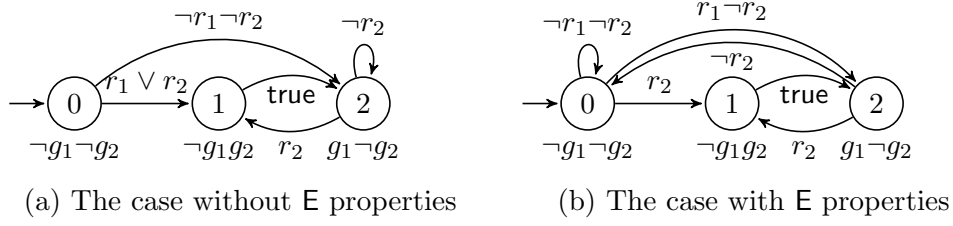
\begin{figure}[tb]
\center
\begin{subfigure}[b]{0.45\textwidth}
\begin{tikzpicture}[->,>=stealth',shorten >=1pt,auto,node distance=1.9cm,semithick]
  \tikzset{every state/.style={minimum size=7mm,inner sep=0.0mm}, initial text={}}
  \tikzstyle{every edge} = [align=center,draw=black]

  \node[state,initial] (0) [label={below:$\neg{g_1}\neg{g_2}$}] {$0$};
  \node[state] (1) [right of=0, label={below:$\neg g_1 g_2$}] {$1$};
  \node[state] (2) [right of=1, label={below:$g_1\neg g_2$}] {$2$};

  \path
  (0) edge node {$r_1 \lor r_2$} (1)
  (1) edge [bend left,below] node {$\true$} (2)
  (2) edge [bend left,below] node {$r_2$} (1)
  (2) edge [loop above] node {$\neg r_2$} (2)
  (0) edge [bend left=40] node {$\neg r_1 \neg r_2$} (2);
\end{tikzpicture}
\caption{The case without $\E$ properties}
\label{fig:vac_resettable_arbiter}
\end{subfigure}
\begin{subfigure}[b]{0.45\textwidth}
\begin{tikzpicture}[->,>=stealth',shorten >=1pt,auto,node distance=1.9cm,semithick]
  \tikzset{every state/.style={minimum size=7mm,inner sep=0.0mm}, initial text={}}
  \tikzstyle{every edge} = [align=center,draw=black]

  \node[state,initial] (0) [label={below:$\neg{g_1}\neg{g_2}$}] {$0$};
  \node[state] (1) [right of=0, label={below:$\neg g_1 g_2$}] {$1$};
  \node[state] (2) [right of=1, label={below:$g_1\neg g_2$}] {$2$};

  \path
  (0) edge node {$r_2$} (1)
  (0) edge [loop above] node {$\neg r_1 \neg r_2$} (0)
  (0) edge [bend left=47] node {$r_1 \neg r_2$} (2)
  (2) edge [bend right=38] node {$\neg r_2$} (0)
  (1) edge [bend left,below] node {$\true$} (2)
  (2) edge [bend left,below] node {$r_2$} (1);
\end{tikzpicture}
\caption{The case with $\E$ properties}
\label{fig:resettable_arbiter}
\end{subfigure}
\caption{Synthesized systems for the resettable arbiter example}
\end{figure}

%% AK: strange example.
%%\paragraph{Non-deterministic arbiter.}
%%Let $I = \{r\}$, $O = \{g\}$.
%%Consider the \CTLstar formula
%%$$
%%\A(r \impl \F g)  ~\land~  \AGF \neg g  ~\land~ 
%%\neg g ~\land~ \E(r \land \X g) \land \E(r \land \X \neg g).
%%$$
%%\begin{wrapfigure}{r}{4.1cm}
%%\vspace{-0.9cm}
%%\begin{tikzpicture}[->,>=stealth',shorten >=1pt,auto,node distance=2cm,semithick]
%%  \tikzset{every state/.style={minimum size=7mm,inner sep=0.0mm}, initial text={}}
%%  \tikzstyle{every edge} = [align=center,draw=black]
%%
%%  \node[state,initial] (0) [label={below:$\neg g$}] {$0$};
%%  \node[state] (1) [above right of=0, label={above:$\neg g$}] {$1$};
%%  \node[state] (2) [below right of=1, label={below:$\neg g$}] {$2$};
%%
%%  \path 
%%  (0) edge [loop above] node {$\neg m \neg r$} (0)
%%  (0) edge [sloped, below] node {$\neg m r$} (1)
%%  (0) edge [bend left=10] node {$m$} (2)
%%  (1) edge [sloped, above] node {$\true$} (2)
%%  (2) edge [bend left=10] node {$\true$} (0);
%%\end{tikzpicture}
%%\vspace{-0.9cm}
%%\end{wrapfigure}
%%The formula is unrealizable due to the last two conflicting clauses.
%%But if we change $I = \{ r, m \}$, then it becomes realizable.
%%Our synthesizer produced the system in the figure on the right (within one second).
%%The new input $m$ makes the system non-deterministic wrt. $r$,
%%thus during the run the user can choose between two transitions,
%%$\edge{0}{\neg m r}{1}$ or $\edge{0}{m r}{2}$.

\parbf{Sender-receiver system}
Consider a sender-receiver system of the following structure.
It has two modules,
the sender (S) with inputs $\{i_1, i_2\}$ and output $wire$
and the receiver (R) with input $wire$ and outputs $\{o_1,o_2\}$.
\begin{wrapfigure}{r}{4cm}
\vspace{-0.2cm}
\centering
\includegraphics[width=4.2cm]{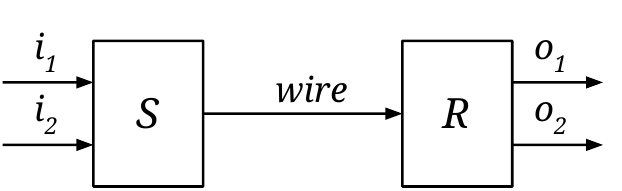}
\vspace{-0.5cm}
\end{wrapfigure}
The sender can send one bit over the wire to the receiver.
We would like to synthesize the sender and receiver modules that satisfy
the following \CTLstar formula over $I = \{i_1,i_2\}$ and $O=\{o_1,o_2\}$:
\begin{align*}
& \AG((i_1 \land i_2)  \impl  \F(o_1  \land  o_2))  \land  \\
& \AG((i_1 \land i_2 \land o_1 \land o_2)  \impl  \X(o_1 \land o_2))  \land  \\
& \AG(~\EF(o_1 \land  \neg o_2)  \land  
       \EF( \neg o_1 \land o_2)  \land  
       \EF( \neg o_1 \land  \neg o_2)  \land  
       \EF(o_1 \land o_2)~).
\end{align*}
Our tool does not support distributed synthesis,
so we manually adapted the SMT query it produced,
by introducing the following uninterpreted functions.
\li
\- For the sender:
   the transition function $\tau_s: T_s \times 2^{\{i_1,i_2\}} \to T_s$
   and the output function $out_s: T_s \times 2^{\{i_1,i_2\}} \to \bbB$.
   We set $T_s$ to have a single state.
\- For the receiver:
   the transition function $\tau_r: T_r \times 2^{\{wire\}} \to T_r$
   and the output functions $o_1: T_r \to \bbB$ and $o_2: T_r \to \bbB$.
   We set $T_r$ to have four states.
\il
It took Z3 solver about 1 minute to find the solution shown in Figure~\ref{fig:sender-receiver}.

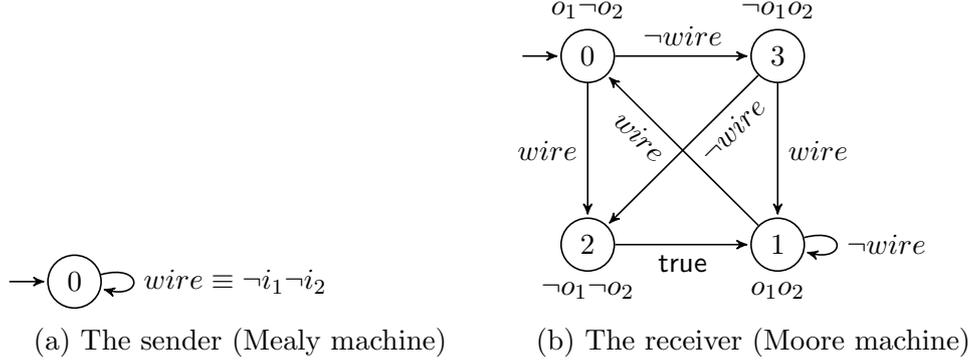
\begin{figure}[tb]
\center
\begin{subfigure}[b]{0.45\textwidth}
\begin{tikzpicture}[->,>=stealth',shorten >=1pt,auto,node distance=2cm,semithick]
  \tikzset{every state/.style={minimum size=7mm,inner sep=0.0mm}, initial text={}}
  \tikzstyle{every edge} = [align=center,draw=black]

  \node[state,initial] (0) {$0$};

  \path 
  (0) edge [loop right] node {$wire \equiv \neg i_1\neg i_2$} (0);
\end{tikzpicture}
\caption{The sender (Mealy machine)}
\end{subfigure}
\begin{subfigure}[b]{0.45\textwidth}
\begin{tikzpicture}[->,>=stealth',shorten >=1pt,auto,node distance=2.5cm,semithick]
  \tikzset{every state/.style={minimum size=7mm,inner sep=0.0mm}, initial text={}}
  \tikzstyle{every edge} = [align=center,draw=black]

  \node[state,initial] (0) [label={above:${o_1}\neg{o_2}$}] {$0$};
  \node[state] (3) [right of=0, label={above:$\neg{o_1}{o_2}$}] {$3$};
  \node[state] (1) [below of=3, label={below:${o_1}{o_2}$}] {$1$};
  \node[state] (2) [left of=1, label={below:$\neg{o}_1\neg{o}_2$}] {$2$};

  \path 
  (0) edge node {$\neg{wire}$} (3)
  (0) edge [left] node {$wire$} (2)

  (1) edge [loop right] node {$\neg{wire}$} (1)
  (1) edge [sloped] node {$wire$} (0)

  (2) edge [below] node {$\true$} (1)

  (3) edge node {${wire}$} (1)
  (3) edge [sloped] node {$\neg{wire}$} (2);
\end{tikzpicture}
\caption{The receiver (Moore machine)}
\end{subfigure}
\caption{The synthesized system for the sender-receiver example}
\label{fig:sender-receiver}
\end{figure}

\subsection{Discussion of Bounded Synthesis from \CTLstar}

%In this chapter,
%we showed how the research on ranking functions~\cite{Nir06,jurdzinski2000small}
%can be used to easily derive synthesis procedures.
We described two approaches to the \CTLstar synthesis and
the only (to our knowledge) synthesizer supporting \CTLstar.
(For CTL synthesis see \cite{klenze2016fast,de2012synthesizing,ctlsat},
 and \cite{bounded-pctl} for PCTL.)
The two approaches are conceptually similar.
The approach via direct encoding is easier to code.
The approach via alternating hesitant automata hints,
for example, at how to reduce \CTLstar synthesis to solving safety games:
via
bounding the number of visits to co-B\"uchi final states and
bounding the distance to B\"uchi final states,
and then determinizing the resulting automaton.
A possible future direction is to extend the approach to the logic ATL* and distributed systems.
In the next chapter, we show how \CTLstar synthesis can be reduced to LTL synthesis,
which avoids developing specialized \CTLstar synthesisers,
presented here.

\chapter{\CTLstar Synthesis via LTL Synthesis}\label{chap:ctl-via-ltl}

\hfill {\footnotesize\textit{This chapter is based on joint work with Roderick Bloem and Sven Schewe~\cite{CTLsynt-via-LTLsynt}}.~~~~~~~~}

\begin{quotation}
\noindent\textbf{Abstract.}
We reduce synthesis for \CTLstar properties to synthesis for \LTL.
In the context of model checking this is impossible --- \CTLstar is more expressive than LTL.
Yet, in synthesis we have knowledge of the system structure \emph{and} we can add new outputs.
These outputs can be used to encode witnesses of the satisfaction of \CTLstar subformulas directly into the system.
This way, we construct an LTL formula, over old and new outputs and original inputs,
which is realisable if, and only if, the original \CTLstar formula is realisable.
The \CTLstar-via-LTL synthesis approach preserves the problem complexity,
although it might increase the minimal system size.
We implemented the reduction,
and evaluated the \CTLstar-via-LTL synthesiser on several examples.
\end{quotation}

\iffinal\else
\textcolor{blue}{\small{
TODOs:
\li
\- related work
\- unrealisability section: elaborate on the witness
\il
}
}
\fi

\section{Introduction} \label{sec:intro}

The problem of reactive synthesis was introduced by Church for Monadic Second Order Logic~\cite{Church63}.
Later Pnueli introduced Linear Temporal Logic (LTL)~\cite{pnueli1977temporal}
and together with Rosner proved 2EXPTIME-completeness
of the reactive synthesis problem for LTL~\cite{DBLP:conf/popl/PnueliR89}.
In parallel, Emerson and Clarke introduced Computation Tree Logic (CTL)~\cite{ctl-origin},
and later Emerson and Halpern introduce Computation Tree Star Logic (\CTLstar)~\cite{ctlstar-origin}
that subsumes both CTL and LTL.
Kupferman and Vardi showed~\cite{informatio} that the synthesis problem for \CTLstar is 2EXPTIME-complete.
%The original approach to synthesis from LTL goes via Safra construction~\cite{DBLP:conf/focs/Safra88,DBLP:conf/lics/Piterman06}
%to determinize B\"uchi automata \cite{DBLP:conf/popl/PnueliR89},
%which proved to be hard to implement efficiently.

Intuitively, LTL allows one to reason about infinite computations.
The logic has \emph{temporal} operators, e.g., $\G$ (always) and $\F$ (eventually),
and allows one to state properties like ``every request is eventually granted''
($\G(r \impl \F g)$).
A system satisfies a given LTL property if \emph{all} its computations satisfy it.

In contrast, CTL and \CTLstar reason about computation trees,
usually derived by unfolding the system.
The logics have---in addition to temporal operators---\emph{path quantifiers}:
$\A$ (on all paths) and $\E$ (there exists a path).
CTL forbids arbitrary nesting of path quantifiers and temporal operators:
they must interleave.
E.g.,
$\AG g$ (``on all paths we always grant'') is a CTL formula,
but $\AGF g$ (``on all paths we infinitely often grant'') is not a CTL formula.
\CTLstar lifts this limitation.

The expressive powers of CTL and LTL are incomparable:
there are systems indistinguishable by CTL but distinguishable by LTL, and vice versa.
One important property inexpressible in LTL is the resettability property:
``there is always a way to reach the `reset' state'' ($\AGEF reset$).

% In synthesis,
% \CTLstar also allows the designer to write ``cooperative'' properties,
% that say that a specific behaviour is possible when the environment cooperates,
% yet it leaves open how the environment should cooperate.
% The synthesiser then needs to find a system and the behaviour of the environment
% that together satisfy the property.
% Existentially quantified properties can also be used to ensure that
% the system does not synthesise the specification vacuously, by falsifying the assumptions.

There was a time when CTL and LTL competed for ``best logic for model checking''~\cite{LTL-vs-CTL}.
Nowadays most model checkers use LTL,
because it is easier for designers to think about paths rather than about trees.
LTL is also prevalent in reactive synthesis.
SYNTCOMP~\cite{syntcomp}---the reactive synthesis competition with the goal to popularise reactive synthesis---%
has two distinct tracks, and both use LTL (or variants) as their specification language.

Yet LTL leaves the designer without \emph{structural} properties.
One solution is to develop general \CTLstar synthesisers like the one we developed in Chapter~\ref{chap:bosy:ctlstar}.
Another solution is to transform the \CTLstar synthesis problem into
the form understandable to LTL synthesisers, i.e., to reduce \CTLstar synthesis to LTL synthesis.
Such a reduction would automatically transfer performance advances in LTL synthesisers
to a \CTLstar synthesiser.
In this chapter we show one such reduction.

Our reduction of \CTLstar synthesis to LTL synthesis works as follows.

First, recall how the standard \CTLstar model checking works (see page~\ref{page:defs:bottom-up-mc}).
The verifier introduces a proposition for every state subformula---formulas starting with an $\A$ or an $\E$ path quantifier---of a given \CTLstar formula.
Then the verifier annotates system states with these propositions,
in the bottom up fashion,
starting with propositions that describe subformulas over original propositions (system inputs and outputs).
Therefore the system satisfies the \CTLstar formula iff the initial system state is annotated
with the proposition describing the whole \CTLstar formula
(assuming that the \CTLstar formula starts with $\A$ or $\E$).

Now let us look into \CTLstar synthesis.
The synthesiser has the flexibility to choose the system structure, 
as long as it satisfies a given specification.
We introduce new propositions---outputs that later can be hidden from the user---%
for state subformulas of the \CTLstar formula,
just like in the model checking case above.
We also introduce additional propositions for existentially quantified subformulas---%
to encode the witnesses of their satisfaction.
Such propositions describe the directions (inputs) the environment should provide
to satisfy existentially quantified path formulas.
The requirement that new propositions indeed denote the truth of the subformulas can be stated in LTL.
For example, for a state subformula $\A\varphi$, we introduce proposition $p_{\A\varphi}$,
and require $\G\left[ p_{\A\varphi} \impl \varphi' \right]$,
where $\varphi'$ is $\varphi$ with state subformulas substituted by the propositions.
For an existential subformula $\E\varphi$,
we introduce proposition $p_{\E\varphi}$ and require,
\emph{roughly}, $\G\left[ p_{\E\varphi} \impl ((\G d_{p_{\E\varphi}}) \impl \varphi')\right]$, which states:
if the proposition $p_{\E\varphi}$ holds, then the path along directions encoded by $d_{p_{\E\varphi}}$
satisfies $\varphi'$ (where $\varphi'$ as before).
We wrote ``roughly'', because
there can be several different witnesses for the same existential subformula
starting at different system states:
they may meet in the same system state,
but depart afterwards---then, to be able to depart from the meeting state,
each witness should have its own direction $d$.
We show that, for each existential subformula, a number $\approx 2^{|\Phi_\CTLstar|}$ of witnesses is sufficient,
where $\Phi_\CTLstar$ is a given \CTLstar formula.
This makes the LTL formula exponential in the size of the \CTLstar formula,
but the special---conjunctive---nature of the LTL formula ensures
that the synthesis complexity is 2EXPTIME wrt.\ $|\Phi_\CTLstar|$.

Our reduction is ``if and only if'' and preserves the synthesis complexity.
However, it may increase the size of the system, and is not very well suited to establish unrealisability.
Of course, to show that a given \CTLstar formula is unrealisable,
one could reduce \CTLstar synthesis to LTL synthesis,
then reduce the LTL synthesis problem to solving parity games,
and derive the unrealisability from there%
\footnote{Reducing LTL synthesis to solving parity games \emph{is} practical, as SYNTCOMP'17~\cite{syntcomp} showed:
  such synthesiser {\tt ltlsynt} was among the fastest.}.
But the standard approach for unrealisability checking---by synthesising the dualised LTL specification---does not seem to be practical.
The reason is that the LTL formula $\Phi_\LTL$ is exponential in the size $|\Phi_\CTLstar|$ of the \CTLstar formula.
The negated LTL formula $\neg \Phi$ (used in the dualised specification)
is a big disjunction (vs.\ big conjunction for $\Phi_\LTL$),
which makes a corresponding universal co-B\"uchi automaton doubly-exponential in $|\Phi_\CTLstar|$
(vs.\ singly-exponential for $\Phi_\LTL$).
The double exponential blow up in the size of the automaton---which is used as input to bounded synthesis---%
makes this unrealisability check impractical~\footnote{%
  This is a conjecture:
  we have \emph{not} proven that the synthesis of dualised LTL formulas, produced by our reduction, takes triply exponential time.}%
.

Finally, we have implemented\footnote{Available at \url{https://github.com/5nizza/party-elli}, branch ``cav17''}
the converter from \CTLstar into LTL,
and evaluated our \CTLstar-via-LTL synthesis approach,
using two LTL synthesisers and \CTLstar synthesiser (Chapter~\ref{chap:bosy:ctlstar}),
on several examples.
The experimental results show that such an approach works very well%
---outperforming the specialised \CTLstar synthesiser
   (Chapter~\ref{chap:bosy:ctlstar})---%
when the number of \CTLstar-specific formulas is small.

The chapter depends on notions defined in Chapter~\ref{chap:defs}
and is structured as follows.
In the next Section~\ref{sec:reductions-to-ltl}
we present the main contribution: the reduction.
Then Section~\ref{sec:ctlstar-unreal} briefly discusses checking unrealisability of \CTLstar specifications.
Section~\ref{sec:experiments} describes the experimental setup, specifications, solvers used, and synthesis timings,
and Section~\ref{sec:conclusion} concludes.
\ak{mention related work section}
%Moore systems (Section~\ref{defs:moore-systems}),
%trees (Section~\ref{defs:trees}),
%\CTLstar and LTL (Section~\ref{defs:ctlstar}),
%(can be skipped) hesitant tree automata (Section~\ref{defs:tree-automata}),
%B\"uchi and co-B\"uchi word automata (Section~\ref{defs:word-automata}),
%synthesis problem (Section~\ref{defs:synthesis-problem}).

\section{Converting \CTLstar to \LTL for Synthesis}
\label{sec:reductions-to-ltl}

\ak{
\li
\- unify: LTL vs. path formula
\- system/tree path: sync the def with its usage
\- define det/universal atm that runs on annotated computation trees?
\il
}

In this section, we describe how and why we can reduce \CTLstar synthesis
to LTL synthesis.
First, we recall the standard approach to \CTLstar synthesis,
then describe, step by step, the reduction and the correctness argument,
and then discuss some properties of the reduction.

\subsection{LTL Encoding}

Let us first look at standard automata based algorithms for \CTLstar synthesis~\cite{informatio}\ak{find the approach in their paper}.
%,ScheweThesis}.  \sven{I am not sure if my thesis really belongs there ... .}
When synthesising a system that realizes a \CTLstar specification, we normally do the following.
\li
\- We turn the \CTLstar formula into an alternating hesitant tree automaton $A$.

\- Move from computation trees to annotated computation trees that move the (memoryless) strategy of the verifier%
   \footnote{Such a strategy maps, in each tree node, an automaton state to a next automaton state and direction.}
   into the label of the computation tree.
   This allows for using the derived universal co-B\"uchi tree automaton $U$,
   making the verifier deterministic: it does not make any decisions, as they are now encoded into the system;

\- Determinise $U$ to a deterministic tree automaton $D$.

\- Play an emptiness game for $D$.

\- If the verifier wins, his winning strategy (after projection of the additional labels) defines a system, if the spoiler wins, the specification is unrealisable.
\il

We draw from this construction and use particular properties of the alternating hesitant tree automaton $A$.
Namely, $A$ is not a general alternating tree automaton,
but is an alternating hesitant tree automaton.
Such an automaton is built from a mix of nondeterministic B\"uchi
and universal co-B\"uchi word automata,
called ``existential word automata'' and ``universal word automata''.
These universal and existential word automata start at any system state [tree node] where a universally or existentially, respectively, quantified subformula is marked as true in the annotated system [annotated computation tree].
We use the term ``existential word automata'' to emphasise that the automaton is not only a non-deterministic word automaton, but it is also used in the alternating tree automaton in a way, where the verifier can pick the system [tree] path, along which it has to accept.

We will use the following notions defined on page~\pageref{page:defs:bottom-up-mc}:
the set $F$ of state subformulas of a given \CTLstar formula $\Phi$,
the set of corresponding propositions $P$,
and the top-level Boolean formula $\widetilde\Phi$.

\begin{example}[Word and tree automata]
Consider the formula $\EG\EX(g \land \X(g \land \F\neg g))$
where inputs $I=\{r\}$ and outputs $O=\{g\}$.
The set $F=\{ f_{\EG}=\EG p_{\EX}, f_{\EX}=\EX(g \land \X(g \land \F\neg g)) \}$,
the set $P=\{ p_{\EX}, p_{\EG} \}$, and
$\widetilde\Phi=p_{\EG}$.
Figure~\ref{fig:automata} shows the nondeterministic word automata for the path formulas of the subformulas,
and the alternating (actually, nondeterministic) tree automaton for the whole formula.
In what follows, we work mostly with word automata.

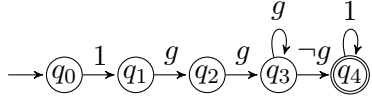
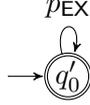
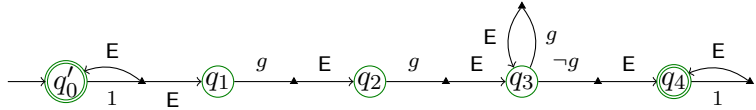
\begin{figure}[tb]
\begin{subfigure}[t]{\linewidth}\center
\begin{tikzpicture}[->,>=stealth',shorten >=1pt,auto,node distance=0.94cm]
  \tikzset{every state/.style={minimum size=3mm,inner sep=0.2mm}, initial text={}}
  \tikzstyle{every edge} = [align=center,draw=black]

  \node[state,initial] (0) {$q_0$};
  \node[state] (1) [right of=0] {$q_1$};
  \node[state] (2) [right of=1] {$q_2$};
  \node[state] (3) [right of=2] {$q_3$};
  \node[state,double] (4) [right of=3] {$q_4$};

  \path
  (0) edge node {$1$} (1)
  (1) edge node {$g$} (2)
  (2) edge node {$g$} (3)
  (3) edge [loop above] node {$g$} (3)
  (3) edge node {$\neg g$} (4)
  (4) edge [loop above] node {$1$} (4);
\end{tikzpicture}
\caption{%
  NBW for $\X (g \land \X (g \land \F\neg g))$,
  the alphabet $\Sigma=2^{\{r,g\}}$.
  Transitions to the non-accepting state $sink$ are not shown.}
\label{fig:nbw-x}
\end{subfigure}
\begin{subfigure}[t]{\linewidth}\center
\begin{tikzpicture}[->,>=stealth',shorten >=1pt,auto,node distance=0.94cm]
  \tikzset{every state/.style={minimum size=3mm,inner sep=0.2mm}, initial text={}}
  \tikzstyle{every edge} = [align=center,draw=black]

  \node[state,initial,double] (0) {$q_0'$};

  \path (0) edge [loop above] node {$p_{\EX}$} (0);
\end{tikzpicture}
\caption{%
  NBW for $\G(p_{\EX})$,
  the alphabet $\Sigma=2^{\{r, g, p_{\EX}\}}$.
  The transition to the non-accepting state $sink$ is omitted.}
\label{fig:nbw-g}
\end{subfigure}
\begin{subfigure}[t]{\linewidth}\center
\begin{tikzpicture}
	\begin{pgfonlayer}{nodelayer}
		\node [style=gn, initial, double] (0) at (0, -0) {$q_0'$};
		\node [style=gn] (1) at (2, -0) {$q_1$};
		\node [style=uptriangle] (2) at (1, -0) {};
		\node [style=gn] (3) at (4, -0) {$q_2$};
		\node [style=gn] (4) at (6, -0) {$q_3$};
		\node [style=gn, double] (5) at (8, -0) {$q_4$};
		\node [style=uptriangle] (6) at (7, -0) {};
		\node [style=uptriangle] (7) at (5, -0) {};
		\node [style=uptriangle] (8) at (3, -0) {};
		\node [style=uptriangle] (9) at (9, -0) {};
		\node [style=uptriangle] (10) at (6, 1) {};
	\end{pgfonlayer}
	\begin{pgfonlayer}{edgelayer}
		\draw [style=simple] (0) to node[below]{$_1$} (2);
		\draw [style=arrow] (2) to node[below=]{$_{\sf E}$} (1);
		\draw [style=arrow, bend right, looseness=1.00] (2) to node[above]{$_{\sf E}$} (0);
		\draw [style=simple] (1) to node[above]{$_g$} (8);
		\draw [style=simple] (3) to node[above]{$_g$} (7);
		\draw [style=simple] (4) to node[above]{$_{\neg g}$} (6);
		\draw [style=simple] (5) to node[below]{$_1$} (9);
		\draw [style=arrow] (8) to node[above]{$_{\sf E}$} (3);
		\draw [style=arrow] (7) to node[above]{$_{\sf E}$} (4);
		\draw [style=arrow] (6) to node[above]{$_{\sf E}$} (5);
		\draw [style=arrow, bend right, looseness=1.00] (9) to node[above]{$_{\sf E}$} (5);
		\draw [style=simple, bend right, looseness=1.00] (4) to node[right]{$_g$} (10);
		\draw [style=arrow, bend right, looseness=1.00] (10) to node[left]{$_{\sf E}$} (4);
	\end{pgfonlayer}
\end{tikzpicture}
\caption{Alternating hesitant tree automaton for $\EG\EX(g \land \X(g \land \F\neg g))$ (actually it is nondeterministic).
  The green color of the states indicate that they are from
  the nondeterministic partition of the states
  (and thus double-circled states are from the B\"uchi acceptance condition).
  The edges starting in the filled triangle are connected with $\land$.
  Edge label $\E$ abbreviates the set of edges, for each tree direction, connected with $\lor$.
  Thus, the transition from $q'_0$ is $((q'_0, r) \lor (q'_0,\neg r)) \land ((q_1, r) \lor (q_1,\neg r))$.
  To get an alternating automaton for $\AG\EX(...)$,
  replace in the self-loop edge of $q_0'$ label $\E$ with $\A$,
  and make the state non-rejecting
  (these also move the state into the universal partition of the states).}
\end{subfigure}
\caption{Word and tree automata.}
\label{fig:automata}
\end{figure}
\end{example}

\smallskip

We are going to show how and why we can reduce \CTLstar-synthesis to LTL synthesis.
The argument is split into steps (a), (b), (c), (d), and (e).
Figure~\ref{fig:discussion-summary} summarises the steps.
\begin{figure}[tbp]
\begin{subfigure}{\linewidth}\center
\includegraphics[width=\textwidth]{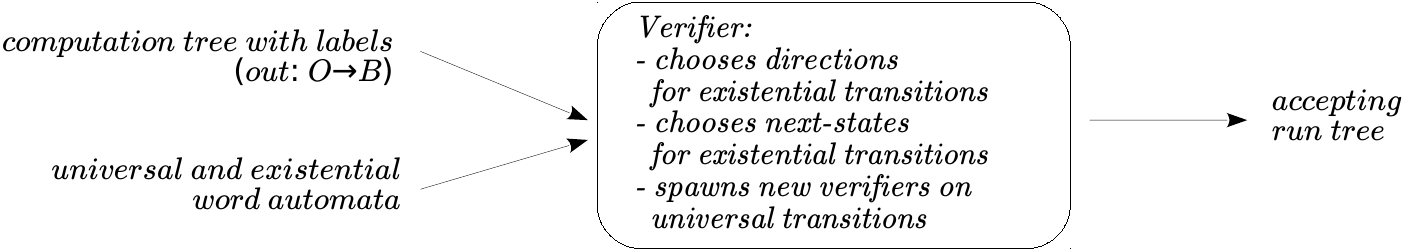}
\caption{The verifier takes a computation tree, universal and existential word automata, and the top-level proposition, that together encode a given \CTLstar formula. It produces an accepting run tree (if the computation tree satisfies the formula).}
\label{fig:stepA}
\end{subfigure}
\vspace{0.1cm}
\hrule
\vspace{0.2cm}
\begin{subfigure}{\linewidth}\center
\includegraphics[width=\textwidth]{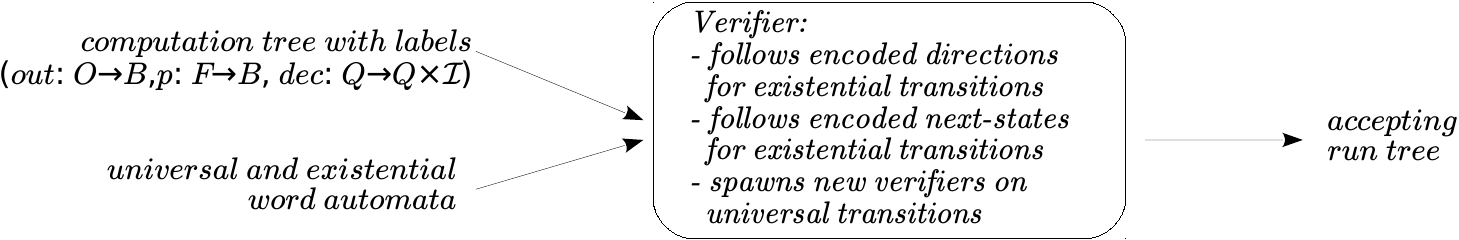}
\caption{We encode the verifier decisions into annotated computation trees,
  making the verifier deterministic.
  Figure~\ref{fig:annotated-computation-tree} shows such an annotated computation tree.}
\label{fig:stepB}
\end{subfigure}
\vspace{0.1cm}
\hrule
\vspace{0.2cm}
\begin{subfigure}{\linewidth}\center
\includegraphics[width=\textwidth]{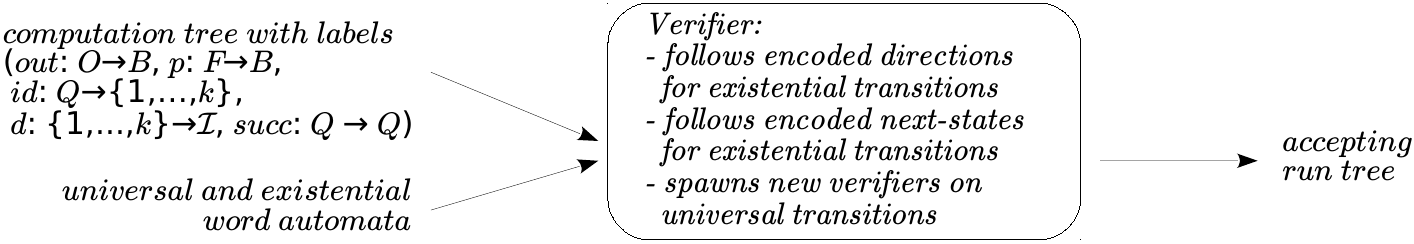}
\caption{The new annotation is a re-phrasing of the previous one.
Figure~\ref{fig:relabeled-tree} gives an example.}
\label{fig:stepC}
\end{subfigure}
\vspace{0.1cm}
\hrule
\vspace{0.2cm}
\begin{subfigure}{\linewidth}\center
\includegraphics[width=\textwidth]{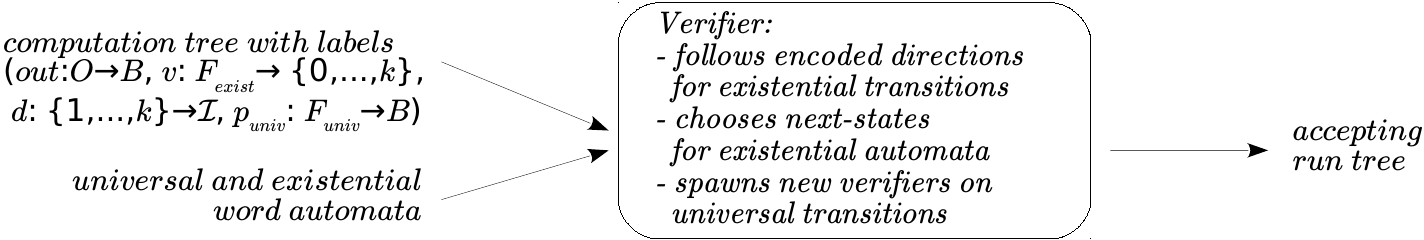}
\caption{We keep directions in the annotation but remove next-states---now the verifier has to choose.
  Figure~\ref{fig:lean-numbered-tree} gives an example.
  The change
  from the label $id: Q \to \{1,...,k\}$
  to the label $v:F_\textit{exist} \to \{0,...,k\}$
  is the reason why the system size can increase.}
\label{fig:stepD}
\end{subfigure}
\vspace{0.1cm}
\hrule
\vspace{0.2cm}
\begin{subfigure}{\linewidth}\center
\includegraphics[width=\textwidth]{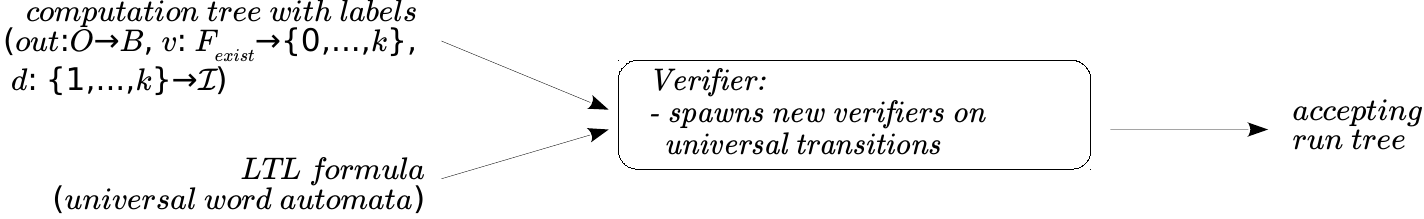}
\caption{Now the obligation of the verifier can be stated in LTL (or using universal co-B\"uchi word automata).}
\label{fig:stepE}
\end{subfigure}
\vspace{0.1cm}
\hrule
\vspace{0.2cm}
\caption{Steps in the correctness argument.
  We transform the input to the verifier and its task, step by step.
  We begin with a computation tree labeled with $2^O$
  and end with a computation tree labeled additionally with $v: F_{exist} \to \{0,...,k\}$ and $d:\{1,...,k\} \to \cal I$.
  (Calligraphic $\cal I$ denotes $2^I$.)
  This is about verifying a given computation tree (labels are fixed),
  in synthesis we would search for such a tree.}
\label{fig:discussion-summary}
\end{figure}

\parbf{Step A (the starting point)}
The verifier takes as input:
a computation tree,
universal and existential word automata for the \CTLstar subformulas, and
the top-level proposition corresponding to the whole \CTLstar formula.
It has to produce an accepting run tree
(if the computation tree satisfies the formula).

\parbf{Step B}
Given a computation tree,
the verifier maps each tree node to a (universal or existential word) automaton state,
and moves from a node according to the quantification of the automaton
(either in all tree directions or in one direction).
The decision
in which tree direction to move and which automaton state to pick for the successor node,
constitutes the strategy of the verifier.
Each time the verifier has to move in several tree directions
(this happens when the node is annotated with a \emph{universal} word automaton state),
we spawn a new version of the verifier,
for each tree direction and transition of the universal word automaton.

The strategy of the verifier is a mapping of states of the existential word automata
to a decision,
which consists of a tree direction
(the continuation of the tree path, along which the automaton shall accept)
and an automaton successor state transition.
For every node $n$, this is a mapping
$dec: Q \rightarrow \I \times Q$ \label{page:decision-mapping}
such that $dec(q)=(e,q')$ implies that $q' \in \delta\big(q,(l(n),e)\big)$,
where $\delta$ corresponds to the existential word automaton to which $q$ belongs,
and $l(n) \in \O$ is a label of the current tree node $n$%
\footnote{%
  The verifier, when in the tree node or system state, moves according to this strategy.}%
.
Note that strategies are defined per-node-basis,
i.e., $dec$ may differ in two different nodes $n_1$ and $n_2$.
(All node labels depend on the current node, but we will omit specifying this explicitly.)
Notice that the strategy is memoryless wrt.\ the history of automata states\ak{note why there exists such a strategy}.

We call a model, in which every state is additionally annotated with a verifier strategy, an \emph{annotated model}.
Similarly, an \emph{annotated computation tree} is a computation tree in which every node is additionally annotated with a verifier strategy.
Thus, in both cases, every system state [node] is labeled with:
(i) original propositional labeling $out: O \to \bbB$,
(ii) propositional labeling for universal and existential subformulas $F=F_\textit{univ}\cupdot F_\textit{exist}$, $p: F \to \bbB$, and
(iii) decision labeling $dec: Q \to \I \times Q$ where $Q$ are the states of all existential automata.
\label{page:def:annotated-tree}

\ak{define when an annotated model/tree is accepted}\

\ak{state the relation btw accepted system and accepted annotated system}

\begin{example}
Figure~\ref{fig:annotated-model-tree} shows an annotated system and computation tree.
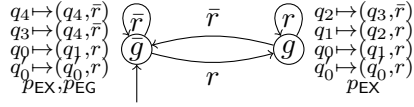
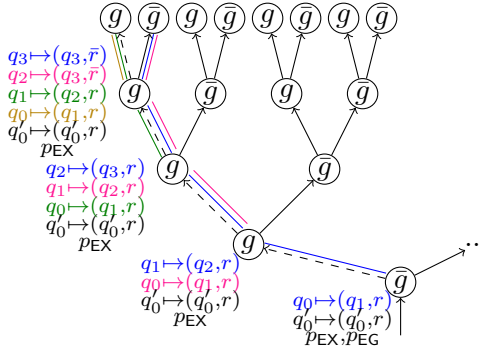
\begin{figure}
\begin{subfigure}[t]{\linewidth}\center
\begin{tikzpicture}
	\begin{pgfonlayer}{nodelayer}
		\node [style=wn, initial below] (0) at (0, -0) {$\bar g$};
		\node [style=wn] (1) at (2, -0) {$g$};
		\node [style=textual] (2) at (-1, -0) {$\color{black}{_{q_0 \mapsto (q_1,r)}}$};
		\node [style=textual] (3) at (-1, -0.25) {$_{q_0' \mapsto (q_0',r)}$};
		\node [style=textual] (4) at (3, -0) {$\color{black}{_{q_0 \mapsto (q_1,r)}}$};
		\node [style=textual] (5) at (3, -0.25) {$_{q_0' \mapsto (q_0',r)}$};
		\node [style=textual] (6) at (-1, 0.25) {$\color{black}{_{q_3 \mapsto (q_4,\bar r)}}$};
		\node [style=textual] (7) at (3, 0.25) {$\color{black}{_{q_1 \mapsto (q_2,r)}}$};
		\node [style=textual] (8) at (3, 0.5) {$\color{black}{_{q_2 \mapsto (q_3,\bar r)}}$};
		\node [style=textual] (9) at (-1, 0.5) {$\color{black}{_{q_4 \mapsto (q_4,\bar r)}}$};
		\node [style=textual] (10) at (3, -0.5) {$_{p_{\sf EX}}$};
		\node [style=textual] (11) at (-1, -0.5) {$_{p_{\sf EX}, p_{\sf EG}}$};
	\end{pgfonlayer}
	\begin{pgfonlayer}{edgelayer}
		\draw [style=arrow, bend right=15, looseness=1.00] (0) to node[below]{$r$} (1);
		\draw [style=arrow, bend right=15, looseness=1.00] (1) to node[above]{$\bar r$} (0);
		\draw [style=arrow, in=120, out=60, loop] (0) to node[below]{$\bar r$} ();
		\draw [style=arrow, in=120, out=60, loop] (1) to node[below]{$r$} ();
	\end{pgfonlayer}
\end{tikzpicture}
\caption{An annotated system satisfying $\EG\EX(g \land \X(g \land \F\neg g))$.
  Near the nodes is the annotation that encodes the winning strategy of the verifier,
  the label $p_{\EX}$ means the subformula $\EX(g \land \X(g \land \F\neg g))$ holds,
  the label $p_{\EG}$ means the subformula $\EG p_{\EX}$ holds.}
\label{fig:annotated-model}
\end{subfigure}

\vspace{0.5cm}
\begin{subfigure}[t]{\linewidth}\center
\center
\begin{tikzpicture}
	\begin{pgfonlayer}{nodelayer}
		\node [style=invisible] (0) at (1, 1) {...};
		\node [initial below, style=wn] (1) at (0, 0.5) {$\bar g$};
		\node [style=wn] (2) at (-2, 1) {$g$};
		\node [style=wn] (3) at (-3, 2) {$g$};
		\node [style=wn] (4) at (-1, 2) {$\bar g$};
		\node [style=wn] (5) at (-0.5, 3) {$\bar g$};
		\node [style=wn] (6) at (-1.5, 3) {$g$};
		\node [style=wn] (7) at (-2.5, 3) {$\bar g$};
		\node [style=wn] (8) at (-3.5, 3) {$g$};
		\node [style=wn] (9) at (-3.75, 4) {$g$};
		\node [style=wn] (10) at (-3.25, 4) {$\bar g$};
		\node [style=wn] (11) at (-2.75, 4) {$g$};
		\node [style=wn] (12) at (-2.25, 4) {$\bar g$};
		\node [style=wn] (13) at (-1.75, 4) {$g$};
		\node [style=wn] (14) at (-1.25, 4) {$\bar g$};
		\node [style=wn] (15) at (-0.75, 4) {$g$};
		\node [style=wn] (16) at (-0.25, 4) {$\bar g$};
		\node [style=textual] (17) at (-0.75, 0.25) {$\color{blue}{_{q_0 \mapsto (q_1,r)}}$};
		\node [style=textual] (18) at (-0.75, -0) {$_{q_0' \mapsto (q_0',r)}$};
		\node [style=textual] (19) at (-2.75, 0.75) {$\color{blue}{_{q_1 \mapsto (q_2,r)}}$};
		\node [style=textual] (20) at (-2.75, 0.5) {$\color{DeepPink}{_{q_0 \mapsto (q_1,r)}}$};
		\node [style=textual] (21) at (-2.75, 0.25) {$_{q_0' \mapsto (q_0',r)}$};
		\node [style=textual] (22) at (-4, 2) {$\color{blue}{_{q_2 \mapsto (q_3,r)}}$};
		\node [style=textual] (23) at (-4, 1.75) {$\color{DeepPink}{_{q_1 \mapsto (q_2,r)}}$};
		\node [style=textual] (24) at (-4, 1.25) {$_{q_0' \mapsto (q_0',r)}$};
		\node [style=textual] (25) at (-4.5, 3.25) {$\color{DeepPink}_{q_2 \mapsto (q_3,\bar r)}$};
		\node [style=textual] (26) at (-4.5, 3.5) {$\color{blue}_{q_3 \mapsto (q_3,\bar r)}$};
		\node [style=textual] (27) at (-4.5, 2.5) {$_{q_0' \mapsto (q_0',r)}$};
		\node [style=textual] (28) at (-4, 1.5) {$\color{Green}{_{q_0 \mapsto (q_1,r)}}$};
		\node [style=textual] (29) at (-4.5, 3) {$\color{Green}_{q_1 \mapsto (q_2,r)}$};
		\node [style=textual] (30) at (-4.5, 2.75) {$\color{DarkGoldenrod}{_{q_0 \mapsto (q_1,r)}}$};
		\node [style=textual] (31) at (-0.75, -0.25) {$_{p_{\sf EX}, p_{\sf EG}}$};
		\node [style=textual] (32) at (-2.75, -0) {$_{p_{\sf EX}}$};
		\node [style=textual] (33) at (-4, 1) {$_{p_{\sf EX}}$};
		\node [style=textual] (34) at (-4.5, 2.25) {$_{p_{\sf EX}}$};
	\end{pgfonlayer}
	\begin{pgfonlayer}{edgelayer}
		\draw [style=arrow] (1) to node{} (0);
		\draw [style=arrow] (4) to (5);
		\draw [style=arrow] (4) to (6);
		\draw [style=arrow] (3) to (7);
		\draw [style=dashed arrow] (3) to (8);
		\draw [style=blue, transform canvas={yshift=0.3mm,xshift=0.8mm}] (3) to (8);
		\draw [style=pink, transform canvas={yshift=0.8mm,xshift=1.4mm}] (3) to (8);
		\draw [style=green, transform canvas={yshift=-0.3mm,xshift=-0.5mm}] (3) to (8);
		\draw [style=dashed arrow, in=-45, out=135, looseness=0.75] (2) to node{} (3);
		\draw [style=blue, transform canvas={yshift=0.7mm,xshift=0.7mm}] (2) to (3);
		\draw [style=pink, transform canvas={yshift=1.3mm,xshift=1.3mm}] (2) to (3);
		\draw [style=arrow] (2) to (4);
		\draw [style=dashed arrow] (1) to node[right]{} (2);
		\draw [style=blue, transform canvas={yshift=0.7mm,xshift=0.1mm}] (1) to (2);
		\draw [style=dashed arrow] (8) to (9);
		\draw [style=green, transform canvas={xshift=-0.5mm,yshift=-0.1mm}] (8) to (9);
		\draw [style=yellow, transform canvas={xshift=-1mm,yshift=-0.2mm}] (8) to (9);
		\draw [style=arrow] (8) to (10);
		\draw [style=blue, transform canvas={xshift=0.5mm,yshift=-0.1mm}] (8) to (10);
		\draw [style=pink, transform canvas={xshift=1mm,yshift=-0.2mm}] (8) to (10);
		\draw [style=arrow] (7) to (11);
		\draw [style=arrow] (7) to (12);
		\draw [style=arrow] (6) to (13);
		\draw [style=arrow] (6) to (14);
		\draw [style=arrow] (5) to (15);
		\draw [style=arrow] (5) to (16);
	\end{pgfonlayer}
\end{tikzpicture}
\caption{%
  An annotated computation tree that satisfies
  $\EG\EX(g \land \X(g \land \F\neg g))$.
  The root node is called $\epsilon$,
  its left child $\mathrm r$, the left child of node $\mathrm r$ is $\mathrm{rr}$, and so on.
  Let $p_{\EG}$ correspond to $\EG(p_{\EX})$
  and let $p_{\EX}$ correspond to $\EX(g \land \X(g \land \F\neg g))$.
  The annotation for the verifier strategy is on the left side of nodes,
  and decisions for non mapped states are irrelevant.
  Paths used by the winning strategy are depicted using dashed and colored lines.
  The black dashed path witnesses $p_{\EG}$,
  the blue path witnesses $p_{\EX}$ starting in the root node $\epsilon$,
  the pink path witnesses $p_{\EX}$ starting in the node $\mathrm r$, and so on.
  The pink and blue paths share the tail.
  Note that this particular annotated computation tree
  is \emph{not} the unfolding of the annotated system above:
  in the annotated system the right state maps $q_2 \mapsto (q_3, \bar r)$,
  while in the tree the node $\mathrm{rr}$ has $q_2 \mapsto (q_3, r)$.
  (This is done to illustrate that mapped out tree paths can share the tail.)
  \ak{put on the right side the atm that eats such trees...and produces acc \emph{run} trees?}}
\label{fig:annotated-computation-tree}
\end{subfigure}
\caption{Annotated system and computation tree.}
\label{fig:annotated-model-tree}
\end{figure}
\end{example}

\parbf{Step C}
\ak{unify: nondet and existential atm}
The verifier strategy (encoded in the annotated computation tree) encodes both
the words on which the nondeterministic automata are interpreted and
witnesses of acceptance
(accepting automata paths on those words).
%that this word is accepted by them.
For the encoding in LTL that we will later use,
it is enough to map out the word,
%(tail of a path through the tree---or its labeling?)
%on which a nondeterministic automaton is interpreted,
and replace the witness by what it actually means:
that the automaton word satisfies the respective path formula.
I.e.,
if a proposition $p$ corresponding to an existential formula $\E \varphi$ holds in a tree node,
then it will be enough to require that $\varphi$ holds
on the path starting in that node and that follows the directions encoded in the tree.

\begin{example}
Let us look at Figure~\ref{fig:annotated-computation-tree}
to understand the notions of mapped out paths and words.
For every proposition marking a tree node there is a mapped out path.
Consider the root node labeled with $p_{\EX}$ and $p_{\EG}$
and look at $p_{\EG}$ first.
The proposition $p_{\EG}$ corresponds to $\EG p_{\EX}$
and is associated with the NBW in Figure~\ref{fig:nbw-g}
that has the initial state $q'_0$.
We consult the strategy $q'_0 \mapsto (q'_0,r)$
and move in direction $r$ into node $\mathrm r$
(note that the root is labeled $p_{\EX}$ and thus we \emph{can} transit $\edge{q'_0}{r}{q'_0}$).
In the node $\mathrm r$ we consult the strategy $q'_0 \mapsto (q'_0,r)$
and again move in direction $r$ into node $\mathrm{rr}$, and so on.
This way we map out the tree path $\epsilon, \mathrm{r}, \mathrm{rr}, ...$
for $p_{\EG}$ from the root,
and the corresponding mapped out word is $(\bar g,p_{\EX},r) (g,p_{\EX},r)^\omega$.
Now consider the root label $p_{\EX}$ that corresponds to $\EX(g \land \X(g \land \F\neg g))$
and is associated with the NBW in Figure~\ref{fig:nbw-x}.
We consult the strategy $q_0 \mapsto (q_1,r)$
that tells us to move in direction $r$ into node $\mathrm{r}$
(again, note that the root label $g$ makes it possible to transit $\edge{q_0}{g}{q_1}$).
In the node $\mathrm r$ we consult the strategy $q_1 \mapsto (q_2,r)$ and move into node $\mathrm{rr}$,
while the automaton state advances to $q_2$.
From the node $\mathrm{rr}$ the strategy $q_2 \mapsto (q_3,r)$ directs us into node $\mathrm{rrr}$,
then the node $\mathrm{rrr}$ has the strategy $q_3 \mapsto (q_3,\bar r)$, and so on.
Thus, from the root for the proposition $p_{\EX}$ the strategy maps out the path
$\epsilon, \mathrm{r}, \mathrm{rr}, \mathrm{rrr}, \mathrm{rr\bar r}, ...$
and the word $(\bar g, r)(g,r)(g,r)(g,\bar r)(\bar g,\bar r)^\omega$.
\end{example}

%Consider the run-tree.
%For every $p$,
%take all infinite paths of the run-tree
%that start in $q_0^p$ and only move in directions $2^I \times Q$.
%Now transform the paths by hiding directions $Q$ (this transforms the nodes, too).
%You will get paths that are mapped out by $p$.
%Union such paths for all $p$.

Let \emph{two tree paths be equivalent} if they share a tail
(equivalently, if one is the tail of the other).
Our interest will be in equivalence of mapped out tree paths.
%To map out the word, we look at the set of tree paths
%--- and here on the tail of the paths, not their labeling ---

There is a simple sufficient condition for two mapped out tree paths to be equivalent:
if they pass through the same node of the annotated computation tree in the same automaton state,
then they have the same future, and are therefore equivalent.
The condition is sufficient but not necessary%
\footnote{%
 Recall that each mapped out tree path corresponds to at least one copy of the verifier that ensures the path is accepting.
 When two verifiers go along the same tree path,
 it can be annotated with different automata states (for example, corresponding to different automata).
 Then such paths do not satisfy the sufficient condition, although they are trivially equivalent.}%
.

\begin{example}
In Figure~\ref{fig:annotated-computation-tree} the blue and pink paths are equivalent,
since they share a tail.
The sufficient condition fires in the node $\mathrm{rrr\bar r}$,
where the tree paths meet in the automaton state $q_3$
%The verifier strategy in that node could be $q_3 \mapsto (q_4,r)$ or $q_3 \mapsto (q_4,\neg r)$.
\end{example}

The sufficient condition implies that we cannot have more non-equivalent tree paths
passing through a tree node than there are states in all existential word automata,
let us call this number $k$:
$k = sum_{\E\varphi \in F_\textit{exist}} |Q_\varphi|$,
where $Q_\varphi$ are the states of an NBW for $\varphi$.
%of existential automata in $A$, call this number $k$.
%Consequently, once we know (or have a lower bound for) this number, say $k$,
For each tree node, we assign unique numbers from $\{1,...,k\}$ to equivalence classes,
and thus any two non-equivalent tree paths that go through the same tree node have different numbers.
As this is an intermediate step in our translation, we are wasteful with the labeling:
\li
\-[(1)] for every node $n$,
        we map existential word automata states to numbers (IDs) using
        $id: Q \to \{1,\ldots,k\}$,
        we also use labels
        $d:\{1,\ldots,k\}\to \I$ (``direction to take'') and
        $succ: Q \to Q$ (``successor to take''),
        such that $succ(q) \in \delta\Big(q, \big(l(n),d(id(q))\big)\Big)$, \emph{and}

\-[(2)] we maintain the same state ID along the chosen direction:\\
        $id_n(q) = id_{n\cdot e}(succ_n(q))$,
        where the subscript denotes a node to which the label belongs
        and $e=d_n(id_n(q))$.
\il

Note that every annotated computation tree can be re-labeled in the above way.
Indeed:
the item (1) alone can be viewed as a re-phrasing of the labeling $dec$ that we had before on page \pageref{page:decision-mapping},
and the requirement (2) is satisfiable because a tree path maintains its equivalence class.
This step is shown in Figure~\ref{fig:stepC}, the labels are:
$(out:O\to\bbB,
p:F\to\bbB,
id: Q \to \{1,\ldots,k\},
d: \{1,\ldots,k\}\to \I,
succ: Q \to Q)$.

Figure~\ref{fig:relabeled-tree} shows a re-labeled computation tree of Figure~\ref{fig:annotated-computation-tree}.
\begin{figure}[tb]
\center
\begin{tikzpicture}
	\begin{pgfonlayer}{nodelayer}
		\node [style=invisible] (0) at (1, 1) {...};
		\node [initial below, style=wn] (1) at (0, 0.5) {$\bar g$};
		\node [style=wn] (2) at (-2, 1) {$g$};
		\node [style=wn] (3) at (-3, 2) {$g$};
		\node [style=wn] (4) at (-1, 2) {$\bar g$};
		\node [style=wn] (5) at (-0.5, 3) {$\bar g$};
		\node [style=wn] (6) at (-1.5, 3) {$g$};
		\node [style=wn] (7) at (-2.5, 3) {$\bar g$};
		\node [style=wn] (8) at (-3.5, 3) {$g$};
		\node [style=wn] (9) at (-3.75, 4) {$g$};
		\node [style=wn] (10) at (-3.25, 4) {$\bar g$};
		\node [style=wn] (11) at (-2.75, 4) {$g$};
		\node [style=wn] (12) at (-2.25, 4) {$\bar g$};
		\node [style=wn] (13) at (-1.75, 4) {$g$};
		\node [style=wn] (14) at (-1.25, 4) {$\bar g$};
		\node [style=wn] (15) at (-0.75, 4) {$g$};
		\node [style=wn] (16) at (-0.25, 4) {$\bar g$};
		\node [style=textual] (17) at (-1, -0) {$\color{blue}{_{q_0 \mapsto (1, q_1), 1 \mapsto r}}$};
		\node [style=textual] (18) at (-3, 0.75) {$\color{blue}{_{q_1 \mapsto (1,q_2), 1 \mapsto r}}$};
		\node [style=textual] (19) at (-3, 0.5) {$\color{DeepPink}{_{q_0 \mapsto (1,q_1)~~~~~}}$};
		\node [style=textual] (20) at (-4.25, 2) {$\color{blue}{_{q_2 \mapsto (1,q_3), 1 \mapsto r}}$};
		\node [style=textual] (21) at (-4.25, 1.75) {$\color{DeepPink}{_{q_1 \mapsto (1,q_2)~~~~~}}$};
		\node [style=textual] (22) at (-4.75, 3.25) {$\color{DeepPink}_{q_2 \mapsto (1,q_3)~~~~~}$};
		\node [style=textual] (23) at (-4.75, 3.5) {$\color{blue}_{q_3 \mapsto (1,q_3), 1 \mapsto \bar r}$};
		\node [style=textual] (24) at (-4.25, 1.5) {$\color{Green}{_{q_0 \mapsto (2,q_1), 2 \mapsto r}}$};
		\node [style=textual] (25) at (-4.75, 3) {$\color{Green}_{q_1 \mapsto (2,q_2), 2 \mapsto r}$};
		\node [style=textual] (26) at (-4.75, 2.75) {$\color{DarkGoldenrod}{_{q_0 \mapsto (3,q_1), 3 \mapsto r}}$};
		\node [style=textual] (27) at (-1, -0.25) {$\color{black}{_{q_0' \mapsto (4, q_0'), 4 \mapsto r}}$};
		\node [style=textual] (28) at (-3, 0.25) {$\color{black}{_{q_0' \mapsto (4, q_0'), 4 \mapsto r}}$};
		\node [style=textual] (29) at (-4.25, 1.25) {$\color{black}{_{q_0' \mapsto (4, q_0'), 4 \mapsto r}}$};
		\node [style=textual] (30) at (-4.75, 2.5) {$\color{black}{_{q_0' \mapsto (4, q_0'), 4 \mapsto r}}$};
		\node [style=textual] (31) at (-1.5, -0.5) {$_{p_{\sf EX}, p_{\sf EG}}$};
		\node [style=textual] (32) at (-3.75, -0) {$_{p_{\sf EX}}$};
		\node [style=textual] (33) at (-5, 1) {$_{p_{\sf EX}}$};
		\node [style=textual] (34) at (-5.5, 2.25) {$_{p_{\sf EX}}$};
	\end{pgfonlayer}
	\begin{pgfonlayer}{edgelayer}
		\draw [style=arrow] (1) to node{} (0);
		\draw [style=arrow] (4) to (5);
		\draw [style=arrow] (4) to (6);
		\draw [style=arrow] (3) to (7);
		\draw [style=dashed arrow] (3) to (8);
		\draw [style=blue, transform canvas={yshift=0.3mm,xshift=0.8mm}] (3) to (8);
		\draw [style=pink, transform canvas={yshift=0.8mm,xshift=1.4mm}] (3) to (8);
		\draw [style=green, transform canvas={yshift=-0.3mm,xshift=-0.5mm}] (3) to (8);
		\draw [style=dashed arrow, in=-45, out=135, looseness=0.75] (2) to node{} (3);
		\draw [style=blue, transform canvas={yshift=0.7mm,xshift=0.7mm}] (2) to (3);
		\draw [style=pink, transform canvas={yshift=1.3mm,xshift=1.3mm}] (2) to (3);
		\draw [style=arrow] (2) to (4);
		\draw [style=dashed arrow] (1) to node[right]{} (2);
		\draw [style=blue, transform canvas={yshift=0.7mm,xshift=0.1mm}] (1) to (2);
		\draw [style=dashed arrow] (8) to (9);
		\draw [style=green, transform canvas={xshift=-0.5mm,yshift=-0.1mm}] (8) to (9);
		\draw [style=yellow, transform canvas={xshift=-1mm,yshift=-0.2mm}] (8) to (9);
		\draw [style=arrow] (8) to (10);
		\draw [style=blue, transform canvas={xshift=0.5mm,yshift=-0.1mm}] (8) to (10);
		\draw [style=pink, transform canvas={xshift=1mm,yshift=-0.2mm}] (8) to (10);
		\draw [style=arrow] (7) to (11);
		\draw [style=arrow] (7) to (12);
		\draw [style=arrow] (6) to (13);
		\draw [style=arrow] (6) to (14);
		\draw [style=arrow] (5) to (15);
		\draw [style=arrow] (5) to (16);
	\end{pgfonlayer}
\end{tikzpicture}
\caption{%
  A re-labeled computation tree.
  Notation ``$q_0 \mapsto (1,q_1)$'' means
  $id(q_0) = 1$ and $succ(q_0) = q_1$, and ``$1 \mapsto r$'' means $d$ maps $1$ to $\{r\}$.
  Since the blue and pink paths are equivalent,
  the label $id$ maps the corresponding automata states in the nodes
  to the same number, $1$.
  The IDs of the green and yellow paths differ implying that they are not equivalent and hence do not share the tail (their tails cannot be seen in the figure).
}
\label{fig:relabeled-tree}
\end{figure}

\parbf{Step D}
In the new annotation with labels $(out, p, id, d, succ)$,
labeling $d$ alone maps out the tree path for each ID.
The remainder of the information is mainly there to establish that the corresponding word\ak{which alphabet?}
is accepted by the respective word automaton (equivalently: satisfies the respective path formula).
%---or, likewise, that the corresponding word satisfies the respective path formula.
If we use only $d$, then the only missing information is where the path starts and which path formula it belongs to---the information originally encoded by $p$.

%We address these two points by providing this information through changing from rich to numbered computation trees: where the rich computation trees have a \emph{propositional} label for each existentially quantified path formula, we replace this propositional label by an \emph{ID}, where $0$ encodes that no claim that this subformula holds is made (similar to a proposition being ``false'' in the rich model), whereas an ID between $1$ and $k$ is interpreted like a ``true'' flag, but also requires that a respective witness is marked out with the given ID.

We address these two points by using \emph{numbered} computation trees.
Recall that the annotated computation trees have a \emph{propositional} labeling
$p: F \to \bbB$ that labels nodes with subformulas.
In the numbered computation trees,
we replace $p$ for \emph{existential} subformulas $F_\textit{exist} \subseteq F$
by labeling ${v: F_\textit{exist} \to \{0,...,k\}}$,
where for every existentially quantified formula $\E\varphi \in F_\textit{exist}$ and a tree node $n$:
\li
\- $v_{\E\varphi,n}=0$ encodes that no claim that $\E\varphi$ holds in $n$ is made
   (similarly to the proposition $p_{\E\varphi}$ being false in the annotated tree), whereas

\- a value $v_{\E\varphi,n} \in \{1,...,k\}$
   requires that the word of a tree path with ID $v_{\E\varphi,n}$
   starting in $n$ and that follows $d(v_{\E\varphi,n})$ satisfies $\varphi$,
   i.e.,
   the word corresponding to
   $n,
   n\!\cdot\! d_n(v_{\E\varphi,n}),
   n\!\cdot\! d_n(v_{\E\varphi,n}) \!\cdot\! d_{n\!\cdot\! d_n(v_{\E\varphi,n})}(v_{\E\varphi,n}),...$
   satisfies $\varphi$
   (where $d_n$ denotes $d$ in node $n$).
\il
% ak: "similarly" means that the verifier consults the direction-successor for the corresponding q0
% ak: "similarly" uses the fact that in in the annotated trees
% having a node marked with a subformula means that
% the node witnesses the subformula with the marked out path
% (the path is marked out by the strategy and we start in atm state q0)
% For A\phi subformulas, all paths must satisfy \phi---this is ensured
% by the verifier
% (that ensures that all such paths are accepted by the automaton).

\begin{example}
The tree in Figure~\ref{fig:relabeled-tree} becomes a numbered computation tree
if we replace the propositional labels $p_{\EX}$ and $p_{\EG}$ with ID numbers as follows.
The root $\epsilon$ has $v_{\EX} = 1$ and $v_{\EG} = 4$,
the left child $\mathrm r$ has $v_{\EX} = 1$,
the node $\mathrm{rr}$ has $v_{\EX} = 2$,
the node $\mathrm{rrr}$ has $v_{\EX} = 3$.
Note that $id(q_0) = v_{\EX}$ and $id(q_0') = v_{\EG}$ whenever those $v$s are non-zero.
The nodes outside of the dashed path have $v_{\EX} = v_{\EG} = 0$,
meaning that no claims about satisfaction of the corresponding path formulas is made.
\end{example}

Initially, we use \emph{ID} labeling $v$ in addition with ($out, id, d, succ, p^\textit{univ}$),
where $p^\textit{univ}$ is a restriction of $p$ on $F_\textit{univ}$,
and then there is no relevant change in the way the (deterministic) verifier works.
I.e., a numbered computation tree can be turned into annotated computation tree, and vice versa,
such that the numbered tree is accepted iff the annotated tree is accepted.
% The direction rich $\leftarrow$ numbered is easy,
% the direction rich $\rightarrow$ numbered could be better formalized (?): it follows from the fact the number of equivalence classes passing through a node is less or equal to $k$.
%(Indeed: it suffices that an initial state of the respective existential word automaton is in the pre-image of ID for $id$ whenever ID is not $0$.)

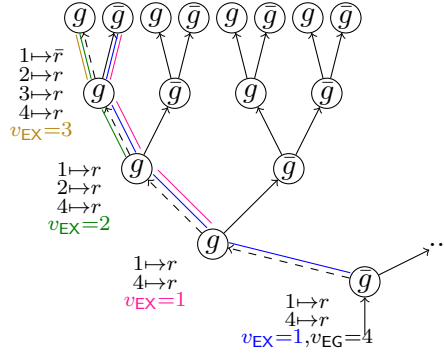
\begin{figure}[tb]
\center
\begin{tikzpicture}
	\begin{pgfonlayer}{nodelayer}
		\node [style=invisible] (0) at (1, 1) {...};
		\node [style=wn, initial below] (1) at (0, 0.5) {$\bar g$};
		\node [style=wn] (2) at (-2, 1) {$g$};
		\node [style=wn] (3) at (-3, 2) {$g$};
		\node [style=wn] (4) at (-1, 2) {$\bar g$};
		\node [style=wn] (5) at (-0.5, 3) {$\bar g$};
		\node [style=wn] (6) at (-1.5, 3) {$g$};
		\node [style=wn] (7) at (-2.5, 3) {$\bar g$};
		\node [style=wn] (8) at (-3.5, 3) {$g$};
		\node [style=wn] (9) at (-3.75, 4) {$g$};
		\node [style=wn] (10) at (-3.25, 4) {$\bar g$};
		\node [style=wn] (11) at (-2.75, 4) {$g$};
		\node [style=wn] (12) at (-2.25, 4) {$\bar g$};
		\node [style=wn] (13) at (-1.75, 4) {$g$};
		\node [style=wn] (14) at (-1.25, 4) {$\bar g$};
		\node [style=wn] (15) at (-0.75, 4) {$g$};
		\node [style=wn] (16) at (-0.25, 4) {$\bar g$};
		\node [style=textual] (17) at (-0.75, 0.25) {$\color{black}{_{1 \mapsto r}}$};
		\node [style=textual] (18) at (-0.75, -0) {$\color{black}{_{4 \mapsto r}}$};
		\node [style=textual] (19) at (-0.75, -0.25) {$_{{\color{blue}{v_{\sf EX}=1}}, v_{\sf EG}=4}$};
		\node [style=textual] (20) at (-2.75, 0.5) {$\color{black}{_{4 \mapsto r}}$};
		\node [style=textual] (21) at (-2.75, 0.75) {$\color{black}{_{1 \mapsto r}}$};
		\node [style=textual] (22) at (-2.75, 0.25) {$_{{\color{DeepPink}{v_{\sf EX}=1}}}$};
		\node [style=textual] (23) at (-3.75, 1.5) {$\color{black}{_{4 \mapsto r}}$};
		\node [style=textual] (24) at (-3.75, 1.25) {$_{{\color{Green}{v_{\sf EX}=2}}}$};
		\node [style=textual] (25) at (-3.75, 2) {$\color{black}{_{1 \mapsto r}}$};
		\node [style=textual] (26) at (-3.75, 1.75) {$\color{black}{_{2 \mapsto r}}$};
		\node [style=textual] (27) at (-4.25, 2.75) {$\color{black}{_{4 \mapsto r}}$};
		\node [style=textual] (28) at (-4.25, 2.5) {$_{{\color{DarkGoldenrod}{v_{\sf EX}=3}}}$};
		\node [style=textual] (29) at (-4.25, 3.5) {$\color{black}{_{1 \mapsto \bar r}}$};
		\node [style=textual] (30) at (-4.25, 3.25) {$\color{black}{_{2 \mapsto r}}$};
		\node [style=textual] (31) at (-4.25, 3) {$\color{black}{_{3 \mapsto r}}$};
	\end{pgfonlayer}
	\begin{pgfonlayer}{edgelayer}
		\draw [style=arrow] (1) to node{} (0);
		\draw [style=arrow] (4) to (5);
		\draw [style=arrow] (4) to (6);
		\draw [style=arrow] (3) to (7);
		\draw [style=dashed arrow] (3) to (8);
		\draw [transform canvas={yshift=0.3mm,xshift=0.8mm}, style=blue] (3) to (8);
		\draw [transform canvas={yshift=0.8mm,xshift=1.4mm}, style=pink] (3) to (8);
		\draw [transform canvas={yshift=-0.3mm,xshift=-0.5mm}, style=green] (3) to (8);
		\draw [style=dashed arrow, in=-45, out=135, looseness=0.75] (2) to node{} (3);
		\draw [transform canvas={yshift=0.7mm,xshift=0.7mm}, style=blue] (2) to (3);
		\draw [transform canvas={yshift=1.3mm,xshift=1.3mm}, style=pink] (2) to (3);
		\draw [style=arrow] (2) to (4);
		\draw [style=dashed arrow] (1) to node[right]{} (2);
		\draw [transform canvas={yshift=0.7mm,xshift=0.1mm}, style=blue] (1) to (2);
		\draw [style=dashed arrow] (8) to (9);
		\draw [transform canvas={xshift=-0.5mm,yshift=-0.1mm}, style=green] (8) to (9);
		\draw [transform canvas={xshift=-1mm,yshift=-0.2mm}, style=yellow] (8) to (9);
		\draw [style=arrow] (8) to (10);
		\draw [transform canvas={xshift=0.5mm,yshift=-0.1mm}, style=blue] (8) to (10);
		\draw [transform canvas={xshift=1mm,yshift=-0.2mm}, style=pink] (8) to (10);
		\draw [style=arrow] (7) to (11);
		\draw [style=arrow] (7) to (12);
		\draw [style=arrow] (6) to (13);
		\draw [style=arrow] (6) to (14);
		\draw [style=arrow] (5) to (15);
		\draw [style=arrow] (5) to (16);
	\end{pgfonlayer}
\end{tikzpicture}
\caption{%
  Numbered computation tree with redundant annotations removed.}
\label{fig:lean-numbered-tree}
\end{figure}

\ak{define 'tree path satisfies LTL', remove 'word' in many places}

Now we observe that the labelings $id$ and $succ$
are used only to witness that each word mapped out by $d$
is accepted by respective existential word automata.
I.e., $id$ and $succ$ make the verifier deterministic.
Let us remove $id$ and $succ$ from the labeling.
We call such trees \emph{lean-numbered computation trees};
they have labeling
$
(
out:O \to \bbB, 
v:F_\textit{exist} \to \{0,...,k\},
d:\{1,...,k\} \to \I,
p^\textit{univ}: F_\textit{univ} \to \bbB
)
$.
This makes the verifier nondeterministic.
We still have the property that
every accepting annotated computation tree can be turned into
an accepting lean-numbered computation tree,
and vice versa.
This step is shown in Figure~\ref{fig:stepD};
an example of a lean-numbered computation tree is in Figure~\ref{fig:lean-numbered-tree}.

\parbf{Step E (the final step)}
We show how labeling $(out,v,d,p^\textit{univ})$ 
allows for using LTL formulas instead of directly using automata for the acceptance check.
The encoding into LTL is as follows.
\li
\- For each existentially quantified formula $\E\varphi$,
   we introduce the following LTL formula
   (recall that $v_{\E\varphi} = 0$ encodes that we do \emph{not} claim that $\E\varphi$ holds in the current     tree node,
    and $v_{\E\varphi} \neq 0$ means that $\E\varphi$ does hold
    and $\varphi$ holds if we follow $v_{\E\varphi}$-numbered directions):
   \begin{equation}\label{eq:ltl-existential}
   \bigwedge_{j \in \{1,...,k\}} \G\Big[ v_{\E\varphi} = j ~\impl~ \big( \G d_j \impl \varphi'\big) \Big],
   \end{equation}
   where $\varphi'$ is obtained from $\varphi$
   by replacing the subformulas of the form $\E\psi$ by $v_{\E\psi} \neq 0$
   and the subformulas of the form $\A\psi$ by $p_{\A\psi}$.

\- For each subformula of the form $\A\varphi$, we simply take
   \begin{equation}\label{eq:ltl-universal}
   \G\Big[ p_{\A\varphi} ~\impl~ \varphi' \Big],
   \end{equation}
   where $\varphi'$ is obtained from $\varphi$ as before.

\- Finally, the overall LTL formula is the conjunction
   \begin{equation}\label{eq:ltl-full}
   \boxed{
   \Phi' \land \bigwedge_{\E\varphi \in F_\textit{exist}} \text{Eq.}\ref{eq:ltl-existential} ~\land \bigwedge_{\A\varphi \in F_\textit{univ}} \text{Eq.}\ref{eq:ltl-universal}
   }
   \end{equation}
   where the Boolean formula $\Phi'$ is obtained by replacing in the original \CTLstar formula
   every $\E\varphi$ by $v_{\E\varphi} \neq 0$ and every $\A\varphi$ by $p_{\A\varphi}$.
\il

\begin{example}\label{ex:ctlstar}
Let $I = \{r\}$, $O=\{g\}$.
Consider the \CTL formula
$$
\EG \neg g \land \AG\EF \neg g \land \EF g.
$$
The sum of states of individual NBWs is $5$
(assuming the natural translations),
so we introduce integer propositions $v_{\EF\!\bar{g}}$, $v_{\EG\!\bar{g}}$, $v_{\EF\!g}$, ranging over $\{0,...,5\}$,
%\footnote{%
%  We can slightly optimize the LTL formula by using smaller domains for $v_i$s:
%  each $v_i$ can vary over \hl{xx:non-intersecting}
%  where $Q_i$ is the number of states in the corresponding NBW.
%  }
and five Boolean propositions $d_1$, ..., $d_5$;
we also introduce the Boolean proposition $p_{\AG(v_{\EF\!\bar{g}}\neq 0)}$.
The LTL formula is:
\begin{align*}
&~~~v_{\EG\!\bar{g}} \neq 0 \land p_{\AG(v_{\EF\!\bar{g}}\neq 0)} \land v_{\EF\!g} \neq 0 ~\land \\
&\bigwedge_{j \in \{1...5\}}\G \left[
\begin{aligned}
& v_{\EF\!\bar{g}}=j ~\impl~ (\G d_j \impl \F \neg g) \\
& v_{\EG\!\bar{g}}=j ~\impl~ (\G d_j \impl \G \neg g) \\
& v_{\EF\!g}=j ~\impl~ (\G d_j \impl \F g)
\end{aligned}
\right] ~\land\\
&~~~\G\left[p_{\AG(v_{\EF\!\bar{g}}\neq 0)} \impl \G(v_{\EF\!\bar{g}}\neq 0)\right].
\end{align*}
Figure~\ref{fig:ctlstar:system} shows a system satisfying the LTL specification.

\begin{figure}[bt]
\center
\begin{tikzpicture}[->,>=stealth',shorten >=1pt,auto,node distance=3.2cm]
  \tikzset{every state/.style={minimum size=6mm,inner sep=0.0mm}, initial text={}}
  \tikzstyle{every edge} = [align=center,draw=black]

  \node[state,initial below] (0)
  [label={left:\specialcellC{$p_{\AG(v_{\EF\bar g}\neq 0)}$\\$v_{\EF\!\bar{g}}=v_{\EG\!\bar{g}}=2,d_2=\neg r$\\$v_{\EF\!g}=3, d_3=r$}},
   label={below right:$t_0$}] {$\neg g$};
  \node[state] (1)
  [right of=0,
  label={right:\specialcellL{$v_{\EF\!\bar{g}}=2,d_2=r$\\$v_{\EG\!\bar{g}}=v_{\EF\!g}=0$}},
  label={below left:$t_1$}] {$g$};

  \path
  (0) edge [bend left=10] node {$r$} (1)
  (1) edge [bend left=10] node {$r$} (0)
  (1) edge [loop above] node {$\neg r$} (1)
  (0) edge [loop above] node {$\neg r$} (0);
\end{tikzpicture}
\caption{A Moore machine for Example~\ref{ex:ctlstar}.
The witness for $\EG\neg g$ is:
$v_{\EG\!\bar{g}}(t_0)=2$, we move along $d_2=\neg r$ looping in $t_0$, thus the witness is $(t_0)^\omega$.
The witness for $\EF g$:
since $v_{\EF\!g}(t_0)=3$, we move along $d_3=r$ from $t_0$ to $t_1$,
where $d_3$ is not restricted, so let $d_3=\neg r$ (not drawn) and then the witness is $t_0 (t_1)^\omega$.
The satisfaction of $\AGEF \neg g$ means that every state has $v_{\EF\!\bar{g}} \neq 0$, which is true.
In $t_0$ we have $\neg g$, so $\EF \neg g$ is satisfied;
for $t_1$ we have $v_{\EF\!\bar{g}}(t_1)=2$ hence we move $t_1 \trans{r} t_0$ and $\EF \neg g$ is also satisfied.
}
\label{fig:ctlstar:system}
\end{figure}
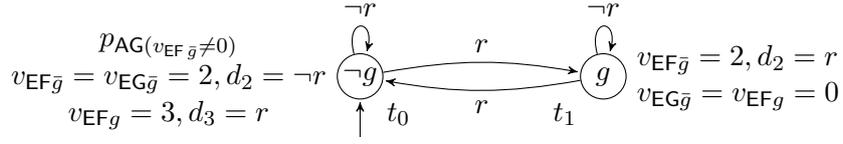
\end{example}

\begin{remark}[We \emph{need} propositions for universal subformulas]
It is intuitively clear that we need new propositions for existential subformulas.
But it is tempting to believe that we can skip introducing new propositions for universal subformulas
and directly use the subformulas instead of the propositions.
This is wrong.
Consider the \CTLstar formula $\AFAG g$.
Our reduction produces the LTL formula
$$
p_{\AF} \land \G [p_{\AF} \impl \F p_{\AG}] \land \G [p_{\AG} \impl \G g].
$$
If we substitute the new propositions with what they express ($p_{\AF}$ by $\F p_{\AG}$ and $p_{\AG}$ by $\G g$),
then we get $\FG g$.
But $\AFAG g$ is different from $\AFG g$.
\end{remark}

%So far we have talked about computation trees.
%Finally, note that, whenever there is a computation tree satisfying an LTL formula,
%there is a finite system satisfying it~\cite{LTL-finite-model-property}.
\ak{restore this note}

\smallskip

The whole discussion leads us to the theorem.
\begin{theorem}
  Let $I$ be the set of inputs and $O$ be the set of outputs,
  and $\Phi_\LTL$ be derived from a given $\Phi_\CTLstar$ as described above.
  Then:
  $$
  \textit{
  $\Phi_\CTLstar$ is realisable
  $~\Iff~$
  $\Phi_\LTL$ is realisable.
  }
  $$
\end{theorem}

\subsection{Complexity}
The translated LTL formula $\Phi_\text{LTL}$, due to Eq.~\ref{eq:ltl-existential},
in the worst case, can be exponentially larger than $\Phi_\CTLstar$,
$|\Phi_\LTL| = 2^{\Theta(|\Phi_\CTLstar|)}$.
Yet, the upper bound on the size of $UCW_{\Phi_\LTL}$ is $2^{\Theta(|\Phi_\CTLstar|)}$
rather than $2^{\Theta(|\Phi_\LTL|)}=2^{2^{\Theta(|\Phi_\CTLstar|)}}$,
because:
\li
\- the size of the UCW is additive in the size of the UCWs of the individual conjuncts, and
\- each conjunct UCW has almost the same size as a UCW of the corresponding subformula,
   since, for every LTL formula $\varphi$, $|UCW_{\G[p \impl (\G\!d \impl \varphi)]}| = |UCW_\varphi|+1$.%
   \footnote{To see this, recall that we can get $UCW_\psi$ by
     treating $NBW_{\neg\psi}$ as a UCW,
     and notice that $|NBW_{\F[p\land\G\!d\land\neg\varphi]}| = |NBW_{\neg\varphi}|+1$.}
\il
Determinising $UCW_{\Phi_\LTL}$ gives a parity game with up to $2^{2^{\Theta(|\Phi_\CTLstar|)}}$ states and
$2^{\Theta(|\Phi_\CTLstar|)}$ priorities~\cite{Schewe/09/determinise,Piterman07,Safra}.
The recent quasipolynomial algorithm~\cite{DBLP:conf/stoc/CaludeJKL017} for solving parity games
has a particular case for $n$ states and $log(n)$ many priorities,
where the time cost is polynomial in the number of game states.
This gives us $O(2^{2^{|\Phi_\CTLstar|}})$-time solution to the derived LTL synthesis problem.
The lower bound comes from the 2EXPTIME-completeness of the \CTLstar synthesis problem~\cite{RosnerThesis}.

\begin{theorem}
  Our solution to the \CTLstar synthesis problem
  via the reduction to \LTL synthesis
  is 2EXPTIME-complete.
\end{theorem}

\subsection*{Minimality}
Although the reduction to LTL synthesis preserves the complexity class,
it does not preserve the minimality of the systems.
Consider an existentially quantified formula $\E\varphi$.
A system path satisfying the formula may pass through the same system state more than once
and exit it in different directions.%
  \footnote{%
  E.g., in Figure~\ref{fig:annotated-model} the system path $t_0 t_1 t_1 (t_0)^\omega$,
  satisfying $\EX(g\land\X(g\land\F\neg g))$,
  double-visits state $t_1$ and exits it first in direction $r$ and then in $\neg r$,
  where $t_0$ is the system state on the left and $t_1$ is on the right.}
Our encoding forbids that.%
  \footnote{%
  Recall that with $\E\varphi$ we associate a number $v_{\E\varphi}$,
  such that whenever in a system state $v_{\E\varphi}$ is non-zero,
  then the path mapped out by $v_{\E\varphi}$-numbered directions satisfies the path formula $\varphi$.
  Therefore whenever $v_{\E\varphi}$-numbered path visits a system state,
  it exits it in the \emph{same} direction $d_{v_{\E\varphi}}$.
  %I.e., a $v_{\E\varphi}$-numbered path can visit a \ul{system state}
  %zero times, once, or infinitely many times.
  %I.e., any $v_{\E\varphi}$-numbered path is a simple lasso.
  }
I.e., in any system satisfying the derived LTL formula,
a system path mapped out by an ID has a unique outgoing direction from every visited state.
As a consequence, such systems are less concise.\ak{what is the worst case blowup?}
This is illustrated in the following example.

\begin{example}[Non-minimality]\label{ex:ctlstar:nonminimal}
Let $I = \{r\}$, $O=\{g\}$,
and consider the \CTLstar formula
$$
\EX(g \land \X(g \land \F\neg g))
$$
The NBW automaton for the path formula has 5 states (Figure~\ref{fig:nbw-x}),
so we introduce the integer proposition $v$ ranging over $\{0,...,5\}$
and Boolean propositions $d_1$, $d_2$, $d_3$, $d_4$, $d_5$.
The LTL formula is
$$
~~~v \neq 0 ~\land
\bigwedge_{j \in \{1...5\}}\!\!\!\!
\G\big[v=j ~\impl~ (\G d_j \impl \X(g \land \X(g \land \F\neg g)))\big]
$$
A smallest system for this LTL formula is in Figure~\ref{fig:ctlstar:nonminimal:system:ltl}.
It is of size is $3$,
while a smallest system for the original \CTLstar formula is of size $2$
(Figure~\ref{fig:annotated-model}).
\begin{figure}[tb]\center
\begin{tikzpicture}[->,>=stealth',shorten >=1pt,auto,node distance=2.2cm]
  \tikzset{every state/.style={minimum size=6mm,inner sep=0.0mm}, initial text={}}
  \tikzstyle{every edge} = [align=center,draw=black]

  \node[state,initial] (0)
  [label={above right:\specialcellC{$v=1$\\$d_1=\neg r$}},
   label={below:$t_0$}] {$\neg g$};

  \node[state] (1)
  [right of=0,
  label={right:{$d_1=r$}},
  label={below left:$t_1$}] {$g$};

  \node[state] (2)
  [below of=1,
  label=left:{$t_2$},
  label=right:{$d_1=r$}
  ] {$g$};

  \path
  (0) edge node[below] {$\neg r$} (1)
  (1) edge [loop above] node {$\neg r$} (1)
  (1) edge node {$r$} (2)
  (2) edge node {$1$} (0)
  (0) edge [loop above] node {$r$} (0);
\end{tikzpicture}
\caption{%
  A smallest Moore machine satisfying the LTL formula from Example~\ref{ex:ctlstar:nonminimal}.}
\label{fig:ctlstar:nonminimal:system:ltl}
\end{figure}
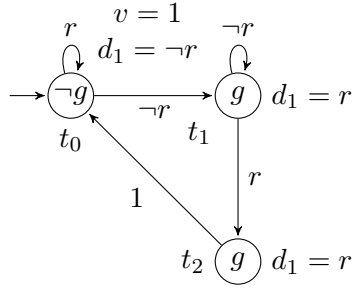
\end{example}

\subsection{Bounded Reduction}

While we have realisability equivalence for sufficiently large $k$, $k$ is a parameter,
where much smaller $k$ might suffice.
In the spirit of Bounded Synthesis,
it is possible to use smaller parameters in the hope of finding a system.
These systems might be of interest in that they guarantee a limited entanglement of different mapped out paths,
as they cap the number of such paths that can go through the same node of a tree.
Such systems are therefore simple wrt.\ this metric,
and this metric is independent of the automaton representation.
(As opposed to a lower bound for $k$
 that depends on the existential automata.)

\ak{CTL?}

\section{Checking Unrealisability of \CTLstar} \label{sec:ctlstar-unreal}

What does a witness of unrealisability for \CTLstar look like?
I.e., when a formula is unrealisable,
is there an ``environment model'', like in the LTL case,
which disproves any system model?

The LTL formula and the annotation shed light on this:
the system for the dualised case is a strategy how to choose original inputs
(depending on the history of $v$, $d$, $p$, and original outputs),
such that any path in the resulting tree violates the original LTL formula.
I.e., the spoiler strategy is a tree, whose nodes are labeled with original inputs,
and whose directions are defined by $v$, $d$, $p$, and original outputs.

% The spoiler would prefer to annotate tree nodes in conformance with some direction $d$
% and to satisfy $\neg\varphi'$
% in Eq.~\ref{eq:ltl-existential} or Eq.~\ref{eq:ltl-universal},
% whenever the implications are triggered.
% I.e. the spoiler has a current direction to follow,
% as well as a set of ``promising'' directions,

\begin{example}
Consider an unrealisable \CTLstar specification:
$\AG g \land \EFX\neg g$ with inputs $\{r\}$ and outputs $\{g\}$.
After reduction to LTL we get the specification:
inputs $\{r\}$, outputs $\{g,p_{\AG\!g},v_{\EFX\!\bar g},d_1,d_2\}$,
and the LTL formula
$$
p_{\AG\!g} \land v_{\EFX\!\bar g} \neq 0 \land
\G\big[ p_{\AG\!g} \impl \G g \big] \land\!\!\!
\bigwedge_{j \in \{1,2\}}\!\!\!\!\!\G\big[ (v_{\EFX\!\bar g} = j \land \G d_j)  \impl  \F\!\X\neg g \big].
$$
The dual specification is:
the system type is Mealy,
new inputs $\{g,p_{\AG\!g},v_{\EFX\!\bar g},d_1,d_2\}$,
new outputs $\{r\}$,
and the LTL formula is the negated original LTL:
$$
p_{\AG\!g} \land v_{\EFX\!\bar g} \neq 0 \land
\G\big[ p_{\AG\!g} \impl \G g \big] ~\impl\!
\bigvee_{j \in \{1,2\}}\!\!\!\!\!\F\big[ (v_{\EFX\!\bar g} = j \land \G d_j)  \land  \G\!\X g \big].
$$
This dual specification is realisable, and it exhibits, e.g., the following witness of unrealisability:
the output $r$ follows $d_1$ or $d_2$ depending on input $v_{\EFX\!\bar g}$.
(The new system needs two states.
 State $1$ describes ``I've seen $v_{\EFX\!\bar g}\in\{0,1\}$ and I output $r$ equal to  $d_1$'';
 from state $1$ we irrevocably go into state $2$ once $v_{\EFX\!\bar g}=2$ and make $r$ equal to $d_2$).
\end{example}

Although our encoding allows for checking unrealisability of \CTLstar
(via dualising the converted LTL specification),
this approach suffers from a very high complexity.
Recall that the LTL formula can become exponential in the size of a \CTLstar formula,
which could only be handled
because it became a big conjunction with sufficiently small conjuncts.
After negating it becomes a large disjunction,
which makes the corresponding UCW doubly exponential in the size of the initial \CTLstar specification
(vs.\ single exponential for the non-negated case).
%Not only can $k$ be exponential in the size of the $\CTLstar$ formula making the LTL formula exponential,
%all is multiplicative here. % we cannot seem to bound the game size by 2EXP (3EXP --- yes)
This seems---there may be a more clever analysis of the formula structure---%
to make the unrealisability check via reduction to LTL cost three exponents in the worst case
(vs.\ 2EXP by the standard approach).

What one could try is to let the new system player in the dualised game
choose a number of disjunctive formulas to follow,
and allow it to revoke the choice finitely many times.\ak{clarify}
This is conservative:
if following $m$ different disjuncts in the dualised formula is enough to win,
then the new system wins.
% There does not seem to be good complexity guarantees that go with this,
% but with a bit of luck that might work.

Alternatively, one could try to synthesise environment model for parts of the disjunction
increasing them until all disjunctions are used.
This is precise.

\ak{
\hl{%
  Is there a bounded procedure which is like $2^{2^{|\Phi_\CTLstar|}\cdot 2^{\textit{smallest size}}}$
  where the 2exp is in the SMT solver? (just like in the bounded synthesis?)}
}

\section{Experiments}\label{sec:experiments}

We implemented the \CTLstar to LTL converter {\small\tt ctl\_to\_ltl.py} inside PARTY~\cite{party}.
PARTY also has two implementations of Bounded Synthesis~\cite{BS},
one encodes the problem into SMT and another reduces the problem to safety games.
Also, PARTY has a \CTLstar synthesiser based on Bounded Synthesis idea
that encodes the problem into SMT (presented in Chapter~\ref{chap:bosy:ctlstar}).
In this section we compare those three solvers,
where the first two solvers take LTL formulas produced by our converter.
All logs and the code are available in repository {\small\url{https://github.com/5nizza/party-elli}},
the branch ``cav17''.

%The results are in Table~\ref{tab:optimizations}, let us analyse them.

\parbf{Specifications}
We created realisable arbiter-like \CTLstar specifications.
The number after the specification name indicates the number of clients.
All specifications have LTL properties in the spirit of ``every request must eventually be granted'' and the mutual exclusion of the grants,
plus some \CTLstar properties.
Below we provide details.
\li
\- ``res\_arbiter'' has the properties:\\
   $\bigwedge_{i \neq j} \AG \neg(g_i \land g_j) ~\land~ \bigwedge_{i} \AG (r_i \impl \F g_i) ~\land$\\
   $\bigwedge_{i} \AG\EFG (\neg g_i)$.

\- ``loop\_arbiter'' has the properties:\\
    $\bigwedge_{i \neq j} \AG \neg(g_i \land g_j) ~\land~ \bigwedge_{i} \AG (r_i \impl \F g_i) ~\land$\\
    $\bigwedge_{i} \AG\EFG (\neg g_i) \land$\\
    $\bigwedge_{i} \EFG g_i$.
   
\- ``postp\_arbiter'' has the properties:\\
   $\bigwedge_{i \neq j} \AG \neg(g_i \land g_j) ~\land~ \bigwedge_{i} \AG (r_i \impl \F g_i) ~\land$ \\
   $\bigwedge_{i} \neg g_i ~\land$\\
   $\bigwedge_{i} \AGEF(\neg g_i \land r_i \land \X(\neg g_i \land r_i \land X \neg g_i ))$.

\- ``prio\_arbiter'' has the properties:\\
   $\A\big[\GF\neg rm \impl $\\
   $~~~~~~\bigwedge_{i \neq j} \G \neg(g_i \land g_j) \land \G\neg(g_i \land gm) \land \bigwedge_{i} \G (r_i \impl \F g_i) \land \G(rm \impl \F gm) \land $\\
   $~~~~~\G(rm \impl \X(\bigwedge_i \neg g_i \U gm)\big]\land$\\
   $\bigwedge_{i} \AG\EFG (\neg g_i) \land \AG\EFG (\neg gm)$.\\
   (It additionally has the prioritised request input $rm$ and grant output $gm$.)

\- ``user\_arbiter'' contains only existential properties
   that specify different sequences of requests and grants.
\il

\parbf{LTL formula and automata sizes}
Our experiments confirm that the LTL formulas increase $\approx |Q|$ times when $k$ increases from $1$ to $|Q|$,
just as described by Eq.~\ref{eq:ltl-existential}.
But the increase does not incur the exponential blow up of the UCWs:
they also increase only $\approx |Q|$ times (just like the theory predicts).

\parbf{Synthesis time}
The table below compares different synthesis approaches for the (realisable) \CTLstar specifications described above.
The column $|\CTLstar|$ is the size of the non-reduced AST of the \CTLstar formula,
the column $|\LTL|$ has two numbers:
the size of the non-reduced AST of the LTL formula for $k=1$ ($k$ is the number of witness IDs)
and the size for $k$ being the upper bound (the sum of the number of states in all existential automata).
The column $|AHT|$ is the sum of the number of states in existential and universal automata.
The column $|UCW|$ is the number of states in the UCW of the translated LTL formula:
we show two numbers, for $k=1$ and when it is the upper bound.
Timings are in seconds, the timeout is 3 hours\footnote{Except for the last specification ``user\_arbiter1'' for which the timeout was 1 hour.} (denoted ``$to$'').
``Time \CTLstar'' is the synthesis time and [system size] required for
\CTLstar synthesizer {\small \tt star.py},
``time LTL(SMT)'' --- for synthesizer {\small\tt elli.py} which implements the original Bounded Synthesis for LTL via SMT~\cite{BS},
``time LTL(game)'' --- for synthesizer {\small\tt kid.py} which implements the original Bounded Synthesis for LTL via reduction to safety games~\cite{BS}.
Both ``time LTL'' columns have two numbers:
when $k$ is set to the minimal value for which the LTL is realisable,
and when $k$ is set to the upper bound.
The subscript near the number indicates the value of $k$:
e.g., $to_8$ means the timeout on all values of $k$ from 1 to $|Q|=8$;
$to_{12(3)}$ means there was the timeout for $k=|Q|=12$
and the last non-timeout was for $k=3$;
$20_1$ means 20 seconds and the minimal $k$ is 1.
The running commands were:
  ``{\small \tt elli.py --incr spec}'',
  ``{\small \tt star.py --incr spec}'',
  ``{\small \tt kid.py spec}''.

\begin{table}[h]
\scriptsize
\centering
\setlength{\tabcolsep}{3pt}
\begin{tabular}{ lcc|cc|r|cc }
\toprule
  & $|\CTLstar|$
  & \specialcellC{$|\LTL|$\\($k_{1}$:$k_{|Q|}$)}
  & $|\text{AHT}|$
  & \specialcellC{$|\text{UCW}|$\\($k_{1}$:$k_{|Q|}$)}
  & \specialcellC{time\\ \CTLstar} & \specialcellC{time\\LTL(SMT)\\($k_{min}$:$k_{|Q|}$)} & \specialcellC{time\\LTL(game)\\($k_{min}$:$k_{|Q|}$)} \\
\midrule
res\_arbiter3     &  65   &  ~78 : 127  &  9   &   7 : 9    & 25 [5]         &  $40_1:260_2$              &  $\mathbf{~7_1:20_2}$   \\
res\_arbiter4     &  97   &  109 : 168  &  10  &  ~8 : 10   & 7380 [7]       &    $to_1$                  &  $\mathbf{30_1:60_2}$   \\
loop\_arbiter2    &  49   &  105 : 682  &  12  &  11 : 41   & {\bf 2} [4]    &  $20_3:131_6$              & ~~~$18_3:to_{6(5)}$   \\
loop\_arbiter3    &  80   & ~183 : 1607 &  15  &  14 : 70   & {\bf 6360} [7] &  $to_{8}$                  & $to_{8}$   \\
postp\_arbiter3   &  113  & ~177 : 2097 &  19  & ~~15 : 114 & 3 [4]          &  ~~~~$\mathbf{2_1}:1735_{12}$  & ~~~~$20_1:to_{12(3)}$   \\
postp\_arbiter4   &  162  & ~276 : 4484 &  24  &  19 : $to$ & 2920 [5]       &  ~~~$\mathbf{68_1}:to_{16(5)}$  & ~~~~$70_1:to_{16(2)}$  \\
prio\_arbiter2    &  82   &  ~92 : 141  &  13  &  14 : 16   & 60 [5]         &    $14_1:19_2$~            & ~$\mathbf{9_1:17_2}$   \\
prio\_arbiter3    &  117  &  125 : 184  &  15  &  16 : 18   & $to$~~~~       &  $4318_1:to_2$~~~~         & ~$\mathbf{26_1:56_2}$~   \\
user\_arbiter1    &  99   & ~203 : 4323 &  23  &  23 : $to$ &  {\bf 3} [5]   &  $to_{16}$                 & $to_{16}$  \\
%AK: for the last row: all data was re-run on my laptop (and some data is "i am quite sure it should be like this...")
%\midrule
\bottomrule
\end{tabular}
\label{tab:optimizations} 
\end{table}

When the minimal $k$ is 1,
the game-based synthesiser is the fastest in most of the cases.
However, it struggles to find a system when we set $k$ to a ``large'' number
(see the timeouts in rows 3--6).
The LTL part of specifications ``res\_arbiter'' and ``prio\_arbiter''
is known to be easier for the game-based synthesiser than for the SMT-based ones%
---adding the simple resettability property does not change this.
For \CTLstar specifications whose minimal $k$ is ``large'' (``loop\_arbiter'' and ``user\_arbiter'' that requires $k>4$),
the specialised \CTLstar synthesiser outperforms
both the game-based and SMT-based synthesisers for the translated LTL specifications.
Our preliminary conclusion is that for \CTLstar specifications that do not require large $k$,
the reduction to LTL synthesis is beneficial.
(Currently we do not know how to predict if a large $k$ is required.)

\parbf{System sizes}
The reduction did not increase the system size in most of the cases
(for the cases ``loop\_arbiter3'', ``res\_arbiter4'', and ``user\_arbiter1''
 we do not know the minimal system size when synthesising from the LTL specification).

\section{Conclusion}\label{sec:conclusion}
We presented the reduction of \CTLstar synthesis problem to LTL synthesis problem.
The reduction preserves the worst-case complexity of the synthesis problem,
although possibly at the cost of larger systems.
The reduction allows the designer to write \CTLstar specifications
even when she has only an LTL synthesiser at hand.
We experimentally showed---on the \emph{small} set of specifications---%
that the reduction is practical when the number of existentially quantified formulas is small.

We briefly discussed how to handle unrealisable \CTLstar specifications.
Whether our suggestions are practical on typical specifications---%
this is still an open question.
A possible future direction is to develop a similar reduction for logics like ATL*~\cite{Alur97},
and to look into the problem of satisfiability of \CTLstar~\cite{ES84}.

\part[Excursion into Parameterized Systems]{Excursion Into\\Parameterized Systems\\ \ \\ \LARGE Guarded and Token-ring Systems}

\section*{Overview of Part II}

Concurrent systems are hard to implement and even harder to debug.
On the other side, they are relatively easy to specify.
Consider, for example, the arbiter serving many clients.
A possible specification is
$
  \forall i \neq j. \G \neg ( g_i \land g_j ) \land \G (r_i \impl \F g_i),
$
%\[
%\begin{array}{ll}
%  \forall i \neq j.~ & \G \neg ( g_i \land g_j ) \land \\
%  \forall i.~ & \G (r_i \impl \F g_i),
%  \end{array}
%\]
which says that, for every client, every request should be eventually granted,
and the grants are mutually exclusive.
If a human implements such an arbiter,
he would try to come up with a basic block that handles a single client,
and connect such a block into a system, that handles as many clients as needed.
On the other side,
the computer tries to synthesize a system as one monolithic block.
This hides the insight that
a system for $n+1$ clients is very similar to a system for $n$ clients.
This leads to the scalability problem, once we require a large number of clients.

The parameterized synthesis approach~\cite{JB14} addresses the issue.
The idea is---just like the human would do---%
to automatically synthesize a basic block that
can be arranged into a system of any desired size.
There are several ways to arrange such blocks into a system,
depending on how they communicate with each other.
In this thesis part we will look into two system architectures.

The first architecture is inspired by cache coherence protocols found in modern processors.
Such a protocol is described by states,
where transitions between states happen depending on whether or not there is a processor in a particular state.
I.e., the transitions are guarded.
Chapter~\ref{chap:guarded-systems} studies guarded systems.

The second kind of systems is token-ring systems.
In such a system, the single token circulates in the system.
A process possessing the token knows that no other process has the token.
Based on this information, the process can, for example, raise the grant.
If all processes raise the grant only when they posses the token,
then the grants will be mutually exclusive.
Chapter~\ref{chap:token-systems} studies token-ring systems.

For both architectures we study their parameterized synthesis problems.
The parameterized synthesis problems asks, given a parameterized specification,
to find a process implementation,
such that a system of any size composed of such processes,
satisfies the specification.
The solution to the seemingly difficult problem---we now ask for correctness of a system of \emph{any} size---is based on the cutoff reduction:
to synthesize a process that works for all system sizes,
it is enough to synthesize a process that works in a system of a cutoff size.
For example, for the specification of the arbiter mentioned above,
the cutoff for token-ring systems is 4.
This means that it is enough to find a process implementation that
works in a system with 4 such processes.
Once we find it, a system of size 5, 6, 7,... is also correct.

In Chapter~\ref{chap:guarded-systems} we prove cutoff results for guarded systems.
Our results extend the results of Emerson and Kahlon~\cite{EmersonK03}.
Our contribution concerns both parameterized synthesis and parameterized verification.
We prove new cutoff results that are applicable to a previously unconsidered setting
of open systems with liveness properties under fairness assumptions.
We also prove new cutoff results for deadlock detection.
The work is theoretical; it is yet to find its application.

In Chapter~\ref{chap:token-systems} on token-ring systems,
we extend the cutoffs of Emerson and Namjoshi~\cite{Emerso03}
to a new setting of fully asynchronous systems and richer specifications.
Then we apply them to an industrial arbiter protocol called AMBA.
Thus, we synthesize for the first time the AMBA protocol in the parameterized sense.

The chapters can be read in any order.

\chapter{Parameterized Guarded Systems} \label{chap:guarded-systems}
\hfill {\footnotesize\textit{This chapter is based on joint work with S.Au{\ss}erlechner and S.Jacobs~\cite{AJK16,SimonThesis}~~~~~~~~}}

\begin{quotation}
\noindent\textbf{Abstract.}
Guarded protocols were introduced in a seminal
paper by Emerson and Kahlon (2000), and describe systems of processes whose transitions
are enabled or disabled depending on the existence of other processes in
certain local states.
In this chapter we study parameterized model checking and synthesis of
guarded protocols, both aiming at formal correctness arguments for systems with any number of processes. 
Cutoff results reduce reasoning about systems with an arbitrary 
number of processes to systems of a determined, fixed size. Our work stems 
from the observation that existing cutoff results for guarded protocols
(i) are restricted to closed systems, and
(ii) are of limited use for liveness properties because reductions do not preserve fairness.
We close these gaps and obtain new cutoff results for open systems with liveness properties under
fairness assumptions.
Furthermore, we obtain cutoffs for the detection of global and local deadlocks,
which are of paramount importance in synthesis.
Finally, we prove tightness or asymptotic tightness for the new cutoffs.
\end{quotation}

\newcommand{\mP}{\mathcal{P}}
\newcommand{\mB}{\mathcal{B}}
\newcommand{\mF}{\mathcal{F}}
\newcommand{\mD}{\mathcal{D}}
\newcommand{\mC}{\mathcal{C}}

\newcommand{\pclass}{{\sf cl}}
\newcommand{\val}{{\sf val}}
\newcommand{\guard}{{\sf guard}}
\newcommand{\eval}{{\sf eval}}
\newcommand{\init}{{\sf init}\xspace}
\newcommand{\templateI}{A}
\newcommand{\templateII}{B}
\newcommand{\formula}{F}
\newcommand{\formulas}{\mathcal{\formula}}
\newcommand{\inputs}{\Sigma}
\newcommand{\localin}{\sigma}
\newcommand{\globIn}{E}

\newcommand{\visited}{\mathsf{Visited}}
\newcommand{\visInf}[2]{\mathsf{Visited}^\inf_{#1}\!(#2)}
\newcommand{\visFin}[2]{\mathsf{Visited}^\fin_{#1}\!(#2)}
\newcommand{\witness}{\mathsf{witness}}
\newcommand{\witfirst}{\mathsf{wit\_first}}
\newcommand{\witsecond}{\mathsf{wit\_second}}
\newcommand{\witlast}{\mathsf{wit\_last}}
\newcommand{\first}{\mathsf{first}}
\newcommand{\enter}{\mathsf{enter}}
\newcommand{\exit}{\mathsf{exit}}

\newcommand{\deadlock}{\mathsf{lock}}
\newcommand{\cutoff}{\mathsf{cutoff}}
\newcommand{\destutter}{\ensuremath{\mathsf{destutter}}\xspace}
\newcommand{\bound}{b}

\newcommand{\sched}{{{\sf move}}}
\newcommand{\enabled}{{{\sf en}}}
\newcommand{\moves}{{{\sf moves}}}

\definecolor{darkgreen}{rgb}{0,0.5,0}
\definecolor{darkblue}{rgb}{0,0,.5}
\definecolor{mygray}{gray}{.3}

%% macros for LTSs (states, state sets, etc.)
%% to change whole batch, just change \state command
\newcommand{\state}{q}
\newcommand{\initstate}{{\sf init}\xspace}
\newcommand{\stateset}{\expandafter\MakeUppercase\expandafter{\state}}
\newcommand{\State}{s}
\newcommand{\Stateset}{\expandafter\MakeUppercase\expandafter{\State}}
\newcommand{\LTS}{\mathcal \stateset}
%% transition function, state labels, ...
%\newcommand{\trans}{\ensuremath{\delta}}
%\newcommand{\Trans}{\ensuremath{\Delta}}
\newcommand{\labeling}{L}
\newcommand{\labelings}{\mathcal{\labeling}}
\newcommand{\outputs}{O}
\renewcommand{\time}{m}
\newcommand{\dead}{{\sf dead}\xspace}

\newcommand{\myloop}[1]{\ensuremath{#1 \rightarrow \initstate \rightarrow #1}\xspace}

%% set cardinality
\newcommand{\card}[1]{\left| {#1} \right|}
\newcommand{\Nat}{\ensuremath{\mathbb{N}}}
%% 'such that'
\newcommand{\smi}{\!\setminus\!}

%% specification, implication
\newcommand{\spec}{\Phi}
\newcommand{\pspec}{\Phi}

%\declaretheorem[name=Theorem]{thm}
%\declaretheorem[name=Lemma]{lem}
%\declaretheorem[name=Corollary]{cor}
%\declaretheorem[name=Observation]{obs}
%\declaretheorem[name=Tightness]{tightness}
\newtheorem{tightness}{Tightness}

\newcommand{\disj}[1]{\ensuremath{\exists\{#1\}}}
\newcommand{\conj}[1]{\ensuremath{\forall\{#1\}}}

% Disj proof
\newcommand{\sys}[1]{(A,B)^{(1,#1)}}
\newcommand{\cutoffsys}{\ensuremath{(A,B)^{(1,c)}}\xspace}
\newcommand{\largesys}{\ensuremath{(A,B)^{(1,n)}}\xspace}
\renewcommand{\interleave}{\ensuremath{\textit{interleave}}\xspace}
\newcommand{\last}{\mathsf{last}}
\renewcommand{\l}{\mathsf{l}}
\newcommand{\f}{\mathsf{f}}
\newcommand{\rpt}{\mathsf{rpt}}
\newcommand{\ppi}[1]{\mathsf{p}_{\circlearrowleft #1}}
\newcommand{\ei}[1]{\mathsf{e}_{\circlearrowleft #1}}
\newcommand{\qstar}{{\ensuremath{q^\star}}}
\newcommand{\estar}{{\ensuremath{\mathsf{e}^\star}}}
\newcommand{\fstar}{{\ensuremath{\mathsf{f}^\star}}}
\newcommand{\fin}{\textit{fin}}
\newcommand{\pref}{\mathsf{pref}}

\renewcommand{\reach}{\mathsf{reach}}
\newcommand{\tpref}{\textit{pref}}

\renewcommand{\inf}{\textit{inf}}
\newcommand{\lo}{\circlearrowleft}

\newcommand{\others}{\textit{others}}
\newcommand{\looop}{\textit{loop}}
\newcommand{\xloop}{\mathsf{xloop}}
\newcommand{\yloop}{\mathsf{yloop}}
\newcommand{\yround}{\mathsf{round}}
\newcommand{\inject}{\mathsf{inject}}

\iffinal
\newcommand{\gray}[1]{}
\newcommand{\blue}[1]{{\color{blue!80} #1}}
\else
\newcommand{\gray}[1]{{\color{black!50} #1}}
\newcommand{\blue}[1]{{\color{blue!80} #1}}
\fi

\newcommand{\transition}[3]{#1 \stackrel{#3}{\rightarrow} #2}
\newcommand{\slice}[2]{\ensuremath{[{#1}\!:\!{#2}]}}
\newcommand{\MinComp}[1]{\ensuremath{x(B^{\witfirst_{#1}})\slice{1}{\first_{#1}\!\!-\!\!1}}}
\newcommand{\flood}{\textit{flood}}
\newcommand{\yflood}{y^\flood}
\newcommand{\sflood}{s^\flood}
\newcommand{\eflood}{e^\flood}
\newcommand{\Pflood}{P^\flood}
\renewcommand{\fair}{\textit{fair}}
\newcommand{\yfair}{y^\fair}
\newcommand{\sfair}{s^\fair}
\newcommand{\efair}{e^\fair}
\newcommand{\Pfair}{P^\fair}
\newcommand{\evac}{\textit{evac}}
\newcommand{\yevac}{y^\evac}
\newcommand{\sevac}{s^\evac}
\newcommand{\eevac}{e^\evac}
\newcommand{\Pevac}{P^\evac}
\newcommand{\qloop}{q}
\newcommand{\eloop}{e}

\newcommand{\myparagraphraw}[1]{\smallskip\noindent{\bf#1}}
\newcommand{\myparagraph}[1]{\myparagraphraw{#1.}}

\NewEnviron{tikzLTS}{%
  \begin{tikzpicture}[node distance=2cm,inner sep=1pt,minimum size=0.5mm,bend angle=20]
	 	\tikzstyle{proc} = [rectangle,draw=black,fill=green!20,thick,inner sep=10pt]
	  	\tikzstyle{state} = [circle,draw=black,thick,inner sep=3pt]
	  	\tikzstyle{noproc} = [circle]
	  	\tikzstyle{lbl} = [rectangle,node distance=2cm]
  		\tikzstyle{pre} = [ <-,shorten <=2pt,shorten >=2pt, >=stealth', semithick]
	  	\tikzstyle{post} = [ ->,shorten <=2pt,shorten >=2pt, >=stealth', semithick]
    \BODY
  \end{tikzpicture}
}

\NewEnviron{tikzPath}{%
  \begin{tikzpicture}	
  	\tikzstyle{pstate} = [circle, fill=black,thick, inner sep=2pt, minimum size=2mm]
    \tikzstyle{hiddenstate} = [circle]
  	\tikzstyle{trans-notsched} = [ ->,shorten <=2pt,shorten >=2pt, >=stealth', semithick]
  	\tikzstyle{trans-sched} = [ ->,shorten <=2pt,shorten >=2pt, >=stealth', very thick] 
    \BODY
    \end{tikzpicture}
}

\section{Introduction}
\label{gua:sec:intro}

% Motivation: (synthesis is especially valueable in concurrent settings)
Concurrent hardware and software systems are notoriously hard to get right.
Formal methods like model 
checking or synthesis can be used to guarantee correctness, but the state explosion 
problem prevents us from using such methods for systems with a large number 
of components. Furthermore, correctness properties are often expected to hold 
for an \emph{arbitrary} number of components. Both problems can be solved by 
\emph{parameterized} model checking and synthesis approaches, which give correctness 
guarantees for systems with any number of components without considering every 
possible system instance explicitly.

While parameterized model checking (PMC) is undecidable in general~\cite{Suzuki88}, 
there exists a number of methods that decide the problem for specific classes 
of systems~\cite{German92,Emerson00,Emerso03}, as well as semi-decision 
procedures that are successful in many interesting 
cases~\cite{Kurshan95,Clarke08,KaiserKW10}.
In this chapter, we consider the \emph{cutoff} method that can guarantee properties of 
systems of arbitrary size by considering only systems of up to a certain 
fixed size, thus providing a decision procedure for PMC if components are finite-state.

We consider systems that are composed of an arbitrary number of processes,
each an instance of a process template from a given, finite set.  
Process templates can be viewed as synchronization 
skeletons~\cite{EmersonC82}, i.e., program abstractions that suppress information not 
necessary for synchronization. In our system model, processes communicate by guarded updates, 
where guards are statements about other processes that are interpreted either 
conjunctively (``every other process satisfies the guard'') or disjunctively 
(``there exists a process that satisfies the guard''). Conjunctive guards can 
model atomic sections or locks, disjunctive guards can model token-passing or to some extent pairwise rendezvous (cf.~\cite{EmersonK03}). 

This class of systems has been studied by Emerson and 
Kahlon~\cite{Emerson00}, and cutoffs that depend on the 
size of process templates are known for specifications of the form 
$\forall{\bar{p}}.\ \spec(\bar{p})$, 
where $\spec(\bar{p})$ is an $\LTLmX$ property over the local states of one or more processes 
$\bar{p}$. Note that this does not allow us to specify fairness 
assumptions, for two reasons: (i) to specify fairness, additional atomic propositions for enabledness and scheduling of processes are needed, and (ii) specifications with global fairness assumptions are of the form 
$(\forall{\bar{p}}.\ \fair(\bar{p})) \impl (\forall{\bar{p}}.\ \spec(\bar{p}))$.
Because neither is supported by \cite{Emerson00}, the existing cutoffs are of 
limited use for 
reasoning about liveness properties. 

Emerson and Kahlon~\cite{Emerson00} mentioned this limitation and illustrated it
using the process template on the figure on the right.
Transitions from the initial state $N$
\begin{wrapfigure}{r}{0.28\linewidth}
\vspace{-9pt}
\scalebox{0.84}{
\begin{tikzpicture}[->,>=stealth',shorten >=1pt,auto,node distance=2cm,
                    semithick]
  \tikzstyle{every state}=[align=center, anchor=center, inner sep=2.2pt,minimum size=0.5mm]
  \tikzstyle{every edge} = [align=center,draw=black,]
  
  \node[state] (N) [label={below:$$},] {$N$};
  \node[state] (T) [right of=N, label={below:$$},] {$T$};
  \node[state] (C) [right of=T, label={below:$$},] {$C$};

  \path 
  (N) edge  node {$true$} (T)
  (T) edge node {$\forall \{T,N$\}} (C)
  (C) edge [bend left=24] node {$true$} (N);
\end{tikzpicture}
}
\label{fig:mutex}
\vspace{-20pt}
\end{wrapfigure}
to the ``trying'' state $T$, and from the critical state $C$ to $N$ are always 
possible, while the transition from $T$ to $C$ is only possible if no other 
process is in $C$. The existing cutoff results can be 
used to prove safety properties like mutual exclusion for systems composed of 
arbitrarily many copies of this template. However, they cannot be used to prove 
starvation-freedom properties like 
$\forall{p}. \pforall \always (T_p \rightarrow \eventually C_p)$,
stating that every process $p$ that enters its local state $T_p$ will 
eventually enter state $C_p$, because without fairness 
of scheduling the property does not hold. 

Also, Emerson and Kahlon~\cite{Emerson00} consider only closed systems. 
Therefore, in this example, processes always try to enter $C$. 
In contrast, in open systems the transition to $T$ might be a reaction 
to a corresponding input from the environment that makes entering $C$ necessary. While it is possible to convert an open system to a closed system that is equivalent under \LTL\ properties, this comes at the cost of a blow-up.

\smallskip
\noindent {\bf Motivation.} Our work is inspired by applications in
parameterized synthesis~\cite{JB14}, where the goal is to automatically
construct process templates such that a given specification is satisfied in
systems with an arbitrary number of components.
In this setting,
one generally considers \emph{open systems} that interact with an uncontrollable environment (user).
Also, most specifications contain liveness properties that cannot be guaranteed without fairness assumptions.
Note that in the parameterized setting liveness properties cannot be reduced to safety properties,
because the size of a system is not bounded a priori.
Finally, we are interested in synthesizing deadlock-free systems.
\emph{Cutoffs} are essential for parameterized synthesis,
because they enable a semi-decision procedure to parameterized synthesis.

\smallskip
\noindent
{\bf Contributions.}
\begin{itemize}
\item We show that existing cutoffs for model checking of $\LTLmX$ properties are in general not sufficient for systems with \emph{fairness assumptions}, and provide new cutoffs for this case.
\item We improve some of the existing cutoff results, and give separate cutoffs for the problem of \emph{deadlock detection}, which is closely related to fairness. 
\item We prove \emph{tightness} or asymptotical tightness for all of our cutoffs, showing that smaller cutoffs cannot exist with respect to the parameters we consider.
\end{itemize}
Moreover, all of our cutoffs directly support \emph{open systems}, where each process may communicate with an adversarial environment. This makes the blow-up incurred by translation to an equivalent closed system unnecessary. 
Finally,
we will show in Sect.~\ref{gua:sec:paramsynt}
how to integrate our size-dependent cutoffs into the parameterized synthesis approach.
%The results presented here are based on a more detailed preliminary version of this paper~\cite{AJK15}.

\section{Related Work} \label{gua:sec:related}
\ak{ cite ``Parameterized Model-Checking of Timed Systems with Conjunctive Guards'' and another Sasha's paper that contain results for disj guards (including undec result for some CTL*-like logic). }

In this work we extend the results of Emerson and Kahlon~\cite{Emerson00} who study PMC of guarded protocols, but do not support fairness assumptions, nor provide cutoffs for 
deadlock detection.
In \cite{EmersonK03} they extended their work to systems with limited forms of guards and broadcasts, and also proved undecidability 
of PMC of conjunctive guarded protocols wrt.\ $\LTL$ (including $\nextt$), 
and undecidability wrt.\ $\LTLmX$ for systems with both 
conjunctive and disjunctive guards.

Bouajjani et al.~\cite{Bouajjani08} study parameterized model checking of
resource allocation systems (RASs). Such systems have a bounded number of resources, 
each owned by at most one process at any time. Processes are pushdown automata, 
and can request resources with high or normal priority.
RASs are similar to conjunctive guarded protocols in that certain
transitions are disabled unless a processes has a certain resource. 
RASs without priorities and where all the processes are finite state Moore machines can 
be converted to conjunctive guarded protocols (at the price of blow up), 
but not vice versa.
The authors study parameterized model checking wrt.\ $\LTLmX$ properties under certain fairness assumptions,
and deadlock detection.
%The authors study parameterized model checking wrt.\ $\LTLmX$ properties on arbitrary or on strong-fair runs,
%and (local or global) deadlock detection.
Their proofs are based on ideas of~\cite{Emerson00} (our proofs are also based on ideas of~\cite{Emerson00}).
 
German and Sistla~\cite{German92} considered global deadlocks and strong fairness properties for systems with pairwise rendezvous communication in a clique.
In such systems, processes communicate pairwise using messages:
one process sends a message and blocks until another process reads the message.
Emerson and Kahlon~\protect\cite{EmersonK03} have shown that disjunctive guard systems can be reduced to such pairwise rendezvous systems.
However, German and Sistla \cite{German92} do not provide cutoffs, 
nor do they consider deadlocks for individual processes,
and their specifications can talk about one process only.
Aminof et al.~\cite{AminofKRSV14} have 
recently extended these results to more general topologies, and have 
shown that for some decidable parameterized model checking problems there are no cutoffs,
even in cliques.
\ak{update the paragraphs, they also have some results on disj?}

Many of the decidability results above have been surveyed in our book~\cite{BloemETAL15}.

\section{Preliminaries} \label{gua:sec:prelim}\label{gua:sec:definitions}
Many definitions intersect with those defined in previous chapters,
but to keep the chapter self-contained we define them here.

Notation:
$\bbB = \{\true,\false\}$ is the set of Boolean values,
$\bbN$ is the set of natural numbers (excluding $0$),
$\bbN_0 = \bbN\cup\{0\}$,
$[k]$ is the set $\{i \in \bbN \| i \leq k\}$
and $[0..k]$ is the set $[k] \cup \{0\}$ for $k \in \bbN$.
For a sequence $x=x_1x_2\ldots$ denote the $i$-$j$-subsequence as $x\slice{i}{j}$,
i.e., ${x\slice{i}{j}}=x_i \ldots x_j$.

\subsection{System Model} \label{gua:sec:model}

We consider systems $A {\parallel} B^n$, usually written $\largesys$, 
consisting of
one copy of a process template $A$ and $n$ copies of a process template $B$,
in an interleaving parallel composition.%
%% AK: moved to a separate note, to be able to explain why we _cannot_ generalize for 1-conj
%\footnote{As shown in \cite{Emerson00}, cutoffs for this case generalize to cutoffs 
%          for systems of the form $A^m {\parallel} B^n$, and further 
%          to systems with an arbitrary number of process templates 
%          $U_1^{n_1} {\parallel} \ldots {\parallel} U_m^{n_m}$.} 
We distinguish objects that belong to different templates by indexing them with
the template. E.g., for process template $U \in \{A,B\}$, $Q_U$ is the set of
states of $U$. For this section, fix two disjoint finite sets $Q_A$, $Q_B$ as
sets of states of process templates $A$ and $B$, and a positive integer $n$.

\parbf{Processes} A \emph{process template} 
 is a transition system
  $U=(\stateset, \init, \inputs, \delta)$ with 
	\begin{itemize}
	\item $\stateset$ is a finite set of states including the
  initial state $\init$,
	\item $\inputs$ is a finite input alphabet,
	\item $\delta: \stateset \times \inputs \times \mP(Q_A \cupdot Q_B) \times \stateset$ is a guarded transition relation.
	\end{itemize}
A process template is \emph{closed} if $\inputs = \emptyset$, and otherwise \emph{open}.

By $\transition{q_i}{q_j}{e:g}$
we denote a process transition from $q_i$ to $q_j$
for input $e \in \Sigma$ and guarded by guard $g \in \mP(Q_A \cupdot Q_B)$.
We skip the input $e$ and guard $g$
if they are not important or can be inferred from the context.

We define the size $\card{U}$ of a process template $U \in \{A,B\}$ as $\card{\stateset_U}$. A copy of a template $U$ will be called a \emph{$U$-process}.
Different $B$-processes are distinguished by subscript, i.e., for $i \in [1..n]$, $B_i$ is the $i$th copy of $B$, and $\state_{B_i}$ is a state of $B_i$. A state of the $A$-process is denoted by $q_A$. 

For the rest of this subsection, fix templates $A$ and $B$. We assume that $\inputs_A \cap \inputs_B = \emptyset$. We will also write $p$ for a process in $\{ A, B_1, \ldots, B_n\}$, unless $p$ is specified explicitly.
We often denote the set $\{B_1,...,B_n\}$ as $\mB$.

\parbf{Disjunctive and conjunctive systems}
In a system $\largesys$,
consider the global state $s = (\state_A,\state_{B_1},\ldots,\state_{B_n})$ and
global input $e=(\localin_A,\localin_{B_1},\ldots,\localin_{B_n})$.
We write $s(p)$ for $q_p$, and $e(p)$ for $\sigma_p$.
A local transition $(\state_p,\localin_p,g,\state_p') \in \delta_U$ of a process $p$ is \emph{enabled for $s$ and $e$}
if the \emph{guard} $g$ is satisfied by the state $s$ wrt.\ the process $p$, written $(s,p) \models g$ (defined below).
The semantics of $(s,p) \models g$ differs for disjunctive and conjunctive systems:
\begin{align*}
\text{In disjunctive systems: } & (s,p) \models g \text{~~~iff~~~} 
\exists p' \in \{A,B_1,\ldots,B_n\} \setminus \{p\}:\ \ \state_{p'} \in g. \\
\text{In conjunctive systems: } & (s,p) \models g \text{~~~iff~~~} 
\forall p' \in \{A,B_1,\ldots,B_n\} \setminus \{p\}:\ \ \state_{p'} \in g.
\end{align*}

Note that we check containment in the guard (disjunctively or conjunctively) 
only for local states of processes \emph{different from} $p$. A process is \emph{enabled} for $s$ and $e$ if at least one of its transitions is enabled for $s$ and $e$, otherwise it is \emph{disabled}.

Like Emerson and Kahlon~\cite{Emerson00}, 
we assume that in conjunctive systems $\init_A$ and $\init_B$ are contained in all guards,
i.e., they act as neutral states.
Furthermore, we call a conjunctive system \emph{$1$-conjunctive} if every guard is of the form $(Q_A \cupdot Q_B) \setminus \{q\}$ for some $q \in Q_A\cupdot Q_B$.

Then, \largesys is defined as the transition 
system $(S,\init_S,\globIn,\delta)$ with 
\begin{itemize}
\item set of global states $S = \stateset_A \times \stateset_B^{n}$, 
\item global initial state $\init_S = (\initstate_A,\initstate_B,\ldots,\initstate_B)$, 
\item set of global inputs $\globIn = (\inputs_A) \times (\inputs_B)^{n}$,
\item and global transition relation $\delta \subseteq S \times \globIn \times S$ with $(s,e,s') \in \delta$ iff 
\begin{enumerate}[label=\roman*)] 
  \item $s=(\state_A,\state_{B_1},\ldots,\state_{B_n})$, 
  \item $e=(\localin_A, \localin_{B_1},\ldots,\localin_{B_n})$, and 
  \item $s'$ is obtained from $s$ by replacing one local state $\state_p$ with a new local state $\state_p'$, where $p$ is a $U$-process with local transition $(\state_{p},\localin_{p},g,\state_p') \in \delta_U$ and $(s,p) \models g$. 
        Thus, we consider so-called interleaved systems,
        where in each step exactly one process transits.
\end{enumerate}
\end{itemize}
We say that a system $\largesys$ is \emph{of type} $(A,B)$. It is called a
\emph{conjunctive system} if guards are interpreted conjunctively, and a
\emph{disjunctive system} if guards are interpreted disjunctively. 
A system is \emph{closed} if all of its templates are closed.

\parbf{Runs} 
A \emph{configuration} of a system is a triple $(s,e,p)$, where $s \in S$, $e 
\in \globIn$, and $p$ is either a system process, or the special symbol $\bot$.
 A \emph{path} of a system is a configuration sequence 
$x = (s_1,e_1,p_1),(s_2,e_2,p_2),\ldots$ such that, for all $\time < |x|$, there is a 
transition $(s_\time,e_\time,s_{\time+1}) \in \delta$ based on a local 
transition of process $p_\time$. We say that process 
$p_\time$ \emph{moves} at \emph{moment} $\time$. 
Configuration $(s,e,\bot)$ appears
 iff all processes are disabled for $s$ and $e$.
Also, for every $p$ and $\time < |x|$: 
either $e_{\time+1}(p) = e_\time(p)$ or process $p$ moves at moment $\time$. 
That is, the environment keeps the input to each process unchanged until 
the process can read it.\footnote{By only considering inputs that are actually processed, we 
approximate an 
action-based semantics. Paths that do not fulfill this requirement are not 
very interesting, since the environment can violate any interesting 
specification that involves input signals by manipulating them when the 
corresponding process is not allowed to move.} 

A system \emph{run} is a maximal path starting in the initial state. Runs are either infinite, or they end in a configuration $(s,e,\bot)$. We say that a run is \emph{initializing} if every 
%$B$-process 
process
that moves infinitely often also visits 
%$\initstate_B$ 
its $\initstate$ 
infinitely often.

Given a system path $x = (s_1,e_1,p_1),(s_2,e_2,p_2),\ldots$ and a process $p$, the \emph{local path} of $p$ in $x$ is the projection $x(p) = (s_1(p),e_1(p)),(s_2(p),e_2(p)),\ldots$ of $x$ onto local states and inputs of $p$.
Similarly, we define the projection on two processes $p_1,p_2$ denoted by $x(p_1,p_2)$.

%The \emph{destuttering} $\destutter(x)$\ak{make it work with inf runs} of a (local) path \sj{local path not defined} $x=x_0,x_1,\ldots$ is obtained by removing stuttering steps from the sequence, i.e., $\destutter(x)$ is the maximal subsequence $x'$ of $x$ such that for every $\time$ we have $x'_\time \neq x'_{\time+1}$. Two (local) paths $x$ and $y$ are \emph{stutter-equivalent}, written $x \simeq y$, if $\destutter(x)=\destutter(y)$. Define an extension of $\destutter$ to sets of paths in the obvious way. Then two systems $S_1, S_2$ are \emph{stutter-equivalent}, written $S_1 \simeq S_2$, if $\destutter(X_1) = \destutter(X_2)$, where $X_i$ is the set of all infinite runs of system $S_i$.

\parbf{Deadlocks and fairness}
A run is \emph{globally deadlocked} if it is finite.
An infinite run is \emph{locally deadlocked} for process $p$ if there exists $\time$ such that $p$ is disabled for all $s_{\time'},e_{\time'}$ with $\time'\ge \time$. A run is \emph{deadlocked} if it is locally or globally deadlocked.
A system \emph{has a (local/global) deadlock} if it has a (locally/globally) deadlocked run. Note that the absence of local deadlocks for all $p$ implies the absence of global deadlocks, but not the other way around.

A run $(s_1,e_1,p_1), (s_2,e_2,p_2),...$ is \emph{unconditionally-fair} if every process moves infinitely often. 
A run is \emph{strong-fair} if it is infinite and, for every process $p$, if $p$ is enabled infinitely often, then $p$ moves infinitely often.
%\sj{weak fairness needed?} Finally, $x$ is \emph{weak-fair} if it is infinite and for every process $p$, if there exists $t$ such that $p$ is enabled for every $s_{\time'}, e_{\time'}$ with $\time' \ge \time$, then $p$ moves infinitely often.
We will discuss the role of deadlocks and fairness in synthesis in Section~\ref{gua:sec:paramsynt}.

\begin{remark}[$A^m {\parallel} B^n$]
One usually starts with studying parameterized systems of the form $A^n$ (having one process template),
then proceeds to systems of the form $A^m {\parallel} B^n$ (having two templates)
and $U_1^{n_1} {\parallel} \ldots {\parallel} U_m^{n_m}$ (having an arbitrary fixed number of templates).
Our work studies systems $A {\parallel} B^n$,
which have one $A$-process and a parameterized number of $B$-processes,
because the results for such systems can be generalized to systems $U_1^{n_1} {\parallel} \ldots {\parallel} U_m^{n_m}$
(see~\cite{Emerson00} for details).
This generalization works for our results as well,
except for the cutoffs for deadlock detection that are restricted to 1-conjunctive systems of the form $A\parallel B^n$
(Section~\ref{gua:sec:cutoffs}).
\end{remark}

\subsection{Specifications}
\label{gua:sec:semantics}
Fix templates $(A,B)$.
We consider formulas in $\LTLmX$---$\LTL$ without the next-time operator $\nextt$---%
that are prefixed by path quantifiers $\E$ or $\A$
(for LTL and path quantifiers see Section~\ref{defs:ctlstar}).
Let $h(A,B_{i_1},\ldots,B_{i_k})$ be an $\LTLmX$ formula over atomic propositions from $Q_A \cup \Sigma_A$ and indexed propositions from $(Q_B \cup \Sigma_B) \times \{i_1,\ldots,i_k\}$.
For a system $\largesys$ with $n \geq k$ and every $i_j \in [1..n]$,
satisfaction of $\A h(A,B_{i_1},\ldots,B_{i_k})$ and $\E h(A,B_{i_1},\ldots,B_{i_k})$ is defined in the usual way.

\parbf{Parameterized specifications} \label{gua:sec:parameterized}
A \emph{parameterized specification} is a temporal logic formula
with indexed atomic propositions and quantification over indices. 
We consider formulas of the forms
$\forall{i_1,\ldots,i_k.} \A h(A,B_{i_1},\ldots,B_{i_k})$ and\\ 
$\forall{i_1,\ldots,i_k.} \E h(A,B_{i_1},\ldots,B_{i_k})$. 
For a given $n \geq k$, 
$$
\largesys \models \forall{i_1,{\ldots},i_k.} \A h(A,B_{i_1},{\ldots},B_{i_k})
$$
~iff~
$$
\largesys \models \!\!\!\!\!\!\!\!\bigwedge_{j_1 \neq {\ldots} \neq j_k \in [1..n]}\!\!\!\!\!\!\!\!\A h(A,B_{j_1},{\ldots},B_{j_k}).
$$ 
By symmetry of guarded systems (see~\cite{Emerson00}),
the second formula is equivalent to
$\largesys \models \A h(A,B_1,\ldots,B_k)$. 
The formula $\A h(A,B_1,\ldots,B_k)$ is denoted by $\A h(A,B^{(k)})$, 
and we often use it instead of the original $\forall{i_1,\ldots,i_k.} \A h(A,B_{i_1},...,B_{i_k})$.
For formulas with the path quantifier $\E$,
satisfaction is defined analogously
and is equivalent to satisfaction of $\E h(A,B^{(k)})$.
\begin{example}
Consider the formula
$$
\forall{i_1,i_2}.\A \big(\G (r_{i_1} \impl \F g_{i_1}) \land \G \neg (g_{i_1} \land g_{i_2})\big).
$$
By our definition, its satisfaction by a system $(A,B)^{(1,3)}$ means
\begin{align*}
(A,B)^{(1,3)} \models
\A \left(
\begin{aligned}
&\G(r_1 \impl \F g_1) \land \G(r_2 \impl \F g_2) \land \G(r_3 \impl \F g_3) \land \\
&\G \neg (g_1 \land g_2) \land \G \neg (g_1 \land g_3) \land \G \neg (g_2 \land g_3)
\end{aligned}
\right),
\end{align*}
where $g_1$ and $r_1$ refer to the propositions $g$ and $r$ of the process $B_1$,
$g_2$ and $r_2$ belong to $B_2$, and so on.
By symmetry, the latter satisfaction is equivalent to
\begin{align*}
(A,B)^{(1,3)} \models
\A \left(
\begin{aligned}
&\G(r_1 \impl \F g_1) \land \G(r_2 \impl \F g_2) \land \\
&\G \neg (g_1 \land g_2)
\end{aligned}
\right).
\end{align*}
Note that this formula talks about processes $B_1$ and $B_2$, but does not mention $B_3$.
\end{example}

\parbf{Specification of fairness and local deadlocks}
It is often convenient to express fairness assumptions and local deadlocks 
as parameterized specifications.
To this end,
define auxiliary atomic propositions $\sched_p$ and $\enabled_p$ for every process $p$ of system $(A,B)^{(1,n)}$. At moment $\time$ of a given run $(s_1,e_1,p_1),(s_2,e_2,p_2), \ldots$, let $\sched_p$ be true whenever $p_\time = p$, and let $\enabled_p$ be true if $p$ is enabled for $s_\time, e_\time$. Note that we only allow the use of these propositions to define fairness, but not in general specifications.
Then, an infinite run is 
\begin{itemize}
\item \emph{local-deadlock-free} if it satisfies $\forall{p}. \GF \enabled_p$, abbreviated as $\spec_{\neg dead}$,
\item \emph{strong-fair} if it satisfies $\forall{p}. \GF \enabled_p \impl \GF \sched_p$, abbreviated as $\spec_{strong}$, and 
\item \emph{unconditionally-fair} if it satisfies $\forall{p}. \GF \sched_p$, abbreviated as $\spec_{uncond}$.
%\item \sj{needed?:}\emph{weak-fair} if it satisfies $\forall{p}. \A \spec_{weak}$, where $\spec_{weak} = \FG \enabled_p \impl \GF \sched_p$.
\end{itemize}

If $f \in \{strong, uncond\}$ is a fairness notion and 
$\A h(A,B^{(k)})$
a specification, then we write
$\A_{f} h(A,B^{(k)})$ for $\A (\spec_{f} 
\rightarrow h(A,B^{(k)}))$.
Similarly, we write $\E_{f} h(A,B^{(k)})$ for $\E (\spec_{f} \land h(A,B^{(k)}))$.

\subsection{Model Checking and Synthesis Problems}
\label{gua:sec:nonparameterized_synthesis}
Given a system $\largesys$ and a specification $\A h(A,B^{(k)})$, where $n \ge k$. Then:
\begin{itemize}
\item the \emph{model checking problem} is to decide whether $\largesys \models \A h(A,B^{(k)})$,
\item the \emph{deadlock detection problem} is to decide whether $\largesys$
      does not have global nor local deadlocks,
%\item the \emph{deadlock detection problem} is to decide whether all runs of $\largesys$
%are infinite and $\largesys \models \A \spec_{\neg dead}$, 
%i.e., there are no local deadlocks,
\item the \emph{parameterized model checking problem} (PMCP) is to decide whether $\forall m \ge n:\ (A,B)^{(1,m)} \models \A h(A,B^{(k)})$, and 
%\item the \emph{parameterized deadlock detection problem} is to decide whether for all $m \ge n$, all runs of $(A,B)^{(1,m)}$ are infinite and $(A,B)^{(1,m)} \models \A \spec_{\neg dead}$.
\item the \emph{parameterized deadlock detection problem} is to decide whether, 
      for all $m \ge n$, $(A,B)^{(1,m)}$ does not have global nor local deadlocks.
\end{itemize}
For a given number $n \in \bbN$ and specification $\A h(A,B^{(k)})$ with $n \ge k$,
\begin{itemize}
\item the \emph{template synthesis problem} is to find process templates $A,B$ such that
$\largesys \models \A h(A,B^{(k)})$ and $\largesys$ does not have global deadlocks\footnote{\label{footnote:local-deadlocks}Here we do not explicitly mention local deadlocks because they can be specified as a part of $\A h(A,B^{(k)})$.}
\item 
the \emph{bounded template synthesis problem} for a pair of bounds $(\bound_A,\bound_B) \in \bbN \times \bbN$ 
is to solve the template synthesis problem with
$\card{A} \leq \bound_A$ and $\card{B} \leq \bound_B$.
\item the \emph{parameterized template synthesis problem} is to find process templates $A,B$ such that $\forall m \ge n:\ (A,B)^{(1,m)} \models \A h(A,B^{(k)})$ and $(A,B)^{(1,m)}$ does not have global deadlocks\footnoteref{footnote:local-deadlocks}.
\end{itemize}
Similarly, we define problems for specifications having $\E$ instead of $\A$.
The definitions can be flavored with different notions of fairness.

\section{Reduction Method and Challenges} \label{gua:sec:paramsynt}
 
We show how to use existing cutoff results of Emerson and Kahlon~\cite{Emerson00} to reduce the PMCP to a standard model checking problem,
and parameterized template synthesis to template synthesis.
We note the limitations of the existing results that are crucial in the context of synthesis.

\subsection*{Reduction by Cutoffs} \label{page:gua:def:cutoff}

\parbf{Cutoffs}
A \emph{cutoff} for a system type $(\templateI,\templateII)$ and a specification $\spec$ is a number $c \in \bbN$ such that:
\[ 
\forall n \ge c: \left( \cutoffsys \models \spec ~~\Iff~~ \largesys \models \spec \right).
\]
Similarly,
a \emph{cutoff for deadlock detection} for a system type $(\templateI,\templateII)$ is a number $c \in \bbN$ such that:
\[ 
\forall n \ge c: \left( \cutoffsys \textit{ has a deadlock} ~~\Iff~~ \largesys \textit{ has a deadlock}\right).
\]
Here, ``has a deadlock'' means ``there is a locally or globally deadlocked run''.

For the systems and specifications presented in this work, cutoffs can be computed from 
the size of the process template $B$ and the number $k$ of copies of $B$ 
mentioned in the specification, and are given as expressions like 
$\card{B}+k+1$.

\begin{remark}\label{re:EK_cutoffs}
Our definition of a cutoff is different from that of Emerson and Kahlon~\cite{Emerson00}, and instead similar to, e.g., Emerson and Namjoshi~\cite{Emerso03}. The reason is that we want the following property to hold for any $(A,B)$ and $\Phi$: 
\begin{quote}
if $n_0$ is the smallest number such that ~$\forall n \geq n_0:\ \largesys \models \Phi$,
then any $c<n_0$ is not a cutoff, any $c\geq n_0$ is a cutoff.
\end{quote}
We call $n_0$ the \emph{tight} cutoff.
The definition of Emerson and Kahlon~\cite[page 2]{Emerson00} requires that
$\forall{n\leq c}. \largesys \models \Phi$
if and only if
$\forall{n \geq 1}: \largesys \models \Phi$, and thus allows stating $c<n_0$ as a cutoff if $\Phi$ does not hold for all $n$.
\end{remark}

\parbf{Parameterized synthesis}
We encourage the reader to revisit Chapter~\ref{defs:bounded_synthesis} on page~\pageref{page:defs:bounded_synthesis}
to recall how bounded synthesis works in the case of non-distributed systems.
Now we adapt the procedure to guarded parameterized systems.
In parameterized model checking
a cutoff allows us to check whether any ``big'' system satisfies the specification
by checking it in the cutoff system.
A similar reduction applies to the parameterized synthesis problem~\cite{JB14}.
For guarded protocols, we obtain the following 
\emph{semi-decision procedure for parameterized synthesis}\ak{is it decidable?}:
\begin{enumerate}
  \item[0.] set initial bound $(\bound_A,\bound_B)$ on the size of the process templates;
  \item[1.] determine the cutoff for $(\bound_A,\bound_B)$ and $\spec$;
  \item[2.] solve the bounded template synthesis problem for cutoff, size bound, and $\spec$;
  \item[3.] if successful, return $(A,B)$, else increase $(\bound_A,\bound_B)$ and goto (1).
\end{enumerate}
This procedure was implemented inside our parameterized synthesis tool PARTY~\cite{party}
by Simon Au{\ss}erlechner as a part of his Master Thesis~\cite{SimonThesis}.

\subsection*{Existing Cutoff Results}
Emerson and Kahlon~\cite{Emerson00} have shown:

\begin{theorem}[Disjunctive Cutoff Theorem] \label{thm:disj-cutoff-pairs}
    For closed disjunctive systems $A{\parallel}B^n$,
    $\card{B}+2$ is a cutoff {$^{(\dagger)}$} for formulas of the
    form $\A h(A,B^{(1)})$ and $\E h(A,B^{(1)})$, and for global
    deadlock detection.
\end{theorem}
 
\begin{theorem}[Conjunctive Cutoff Theorem] \label{thm:conj-cutoff}
    For closed conjunctive systems ${A{\parallel}B^n}$,
    $2\card{B}$ is a cutoff {$^{(\dagger)}$} for formulas of the
    form $\A h(A)$ and $\E h(A)$, and for global deadlock detection.
    For formulas of the form $\A h(B^{(1)})$ and $\E h(B^{(1)})$,
    $2\card{B}+1$ is a cutoff.
\end{theorem}
\noindent
In the above theorems,
$h(A)$ (resp.\ $h(B^{(1)})$) means that the formula talks about the $A$-process only (resp.\ $B_1$).

\begin{remark} ${(\dagger)}$
Note that Emerson and Kahlon \cite{Emerson00} proved these results for
a different definition of a cutoff (see Remark \ref{re:EK_cutoffs}).  
Their results also hold for our definition, except possibly for
global deadlocks.  For the latter case to hold with the new cutoff definition, one 
also needs to prove the direction ``global deadlock in the cutoff system implies global
deadlock in a large system'' (later called Monotonicity Lemma).
In Sections~\ref{gua:sec:proofs-disj-deadlock-unfair} and \ref{gua:sec:proofs-disj-deadlock-fair},
Sections~\ref{gua:sec:proofs-conj-deadlock-unfair} and \ref{gua:sec:proofs-conj-deadlock-fair},
we prove these lemmas for the case of general deadlock (global \emph{or} local).
\end{remark}

\subsection*{Challenge: Open Systems}
For any open system $S$ there exists a closed system $S'$ such that $
S$ and $S'$ cannot be distinguished by $\LTL$ specifications 
(e.g., see Manna and Pnueli~\cite{Manna92}). Thus, one approach to PMC for open 
systems is to use a translation between open and closed systems, and then use the 
existing cutoff results for closed systems.

%While such an approach works in theory, it is not feasible when cutoffs 
%depend on the size of process templates: in this case the conversion not only results in a 
%blowup of the local state space of each process, but also in the number of 
%processes that we need to consider. 
While such an approach works in theory, it might not be feasible in practice:
since cutoffs depend on the size of the process templates,
and the translation blows up the process template,
it also blows up the cutoffs.
Thus, cutoffs that directly support open systems are important.

\subsection*{Challenge: Liveness and Deadlocks under Fairness}
We are interested in cutoff results that support liveness properties.
Consider a specification $\Phi=h(A,B^{(k)})$.
In general, we would like to consider only runs where all processes move infinitely often,
i.e.,
use the unconditional fairness assumption $\forall{p}. \GF \sched_p$ and thus have $\A_{uncond}\Phi$.
However, this would mean that we accept all systems that always go into a local deadlock,
since then the assumption is violated (i.e., there will be no unconditionally-fair runs).
This is especially undesirable in synthesis, because the synthesizer often tries to violate the assumptions to satisfy the specification.
To avoid this,
we require the absence of local deadlocks.
But local deadlocks may appear due to unfair scheduling.
Therefore we require the absence of local deadlocks under the strong fairness assumption,
i.e., we require satisfaction of the formula
$\A_{strong} \spec_{\neg dead}=\big(\forall{p}. (\GF \enabled_p \impl \GF \sched_p)\big) \impl \forall{p}. \GF\enabled_p$.
This formula can be roughly read as ``the absence of local deadlocks under fair scheduling''.
Since absence of global deadlocks and absence of local deadlocks under strong fairness imply unconditional fairness,
we can safely use $\A_{uncond}\Phi$.

%In these systems, processes may be disabled depending on their input and the global state. \change{Thus, strong fairness $\forall{p}. (\GF \enabled_p \impl \GF \sched_p)$ is an insufficient assumption, since the environment can easily violate liveness properties by choosing inputs and scheduling such that some process is only enabled finitely often.}
%{Thus, strong fairness $\forall{p}. (\GF \enabled_p \impl \GF \sched_p)$ is an
%insufficient assumption, since the environment can easily violate a process's liveness
%property by choosing inputs and scheduling such that the process never moves 
%after some moment.}
%%
%Moreover, using the unconditional fairness assumption $\forall{p}. \GF \sched_p$
%\remove{for the complete specification} is also undesirable, since then we would accept all systems that always go into a local deadlock. This is especially undesirable in synthesis.
%To exclude this case, we require absence of local deadlocks under the strong fairness assumption. 

In summary, for a parameterized specification $\spec$, we consider satisfaction of
\[
\begin{array}{lllll}
\textit{``all runs are infinite''} &~~\land~~& \A_{strong} \spec_{\neg dead} & ~~\land~~ & \A_{uncond} \spec.
\end{array}
\]
This is equivalent to $\textit{``all runs are infinite''} \land \A_{strong} (\spec_{\neg dead} \,\land\, \spec)$, but by considering the form above we can separate the tasks of deadlock detection and of model checking $\LTLmX$-properties, and obtain modular cutoffs. 
(The phrase ``all runs are infinite'' is another way of saying ``all runs have no global deadlocks''.)

%%%In the following, we present cutoffs for problems of the forms 
%%%(i) $\A_{uncond} \spec$ and
%%%(ii) $\E_{strong} \spec_{dead} \lor \textit{``some run is finite''}$
%%%(and the variants of (i) with $\E$ path quantifier).
%%%%, as well as for the detection of global deadlocks.
%%%% AK: the previous version reads like we provide cutoffs for three problems -- we have only two

\section{New Cutoff Results} \label{gua:sec:cutoffs}

We present new cutoff results that extend Theorems~\ref{thm:disj-cutoff-pairs} and \ref{thm:conj-cutoff}.
The new and previous results are summarized in the table below.

\begin{table}[h]
\centering
%\vspace{-10pt}
\label{table:cutoffs}
\centering
\setlength{\tabcolsep}{2pt}
{%\resizebox{\linewidth}{!}{
\begin{tabular}{ r|c|c|c|c }
   & \specialcellC{$h(A,B^{(k)})$ \\ no fairness} & \specialcellC{deadlock detection \\ no fairness} & \specialcellC{$h(A,B^{(k)})$ \\ uncond. fairness} & \specialcellC{deadlock detection \\ strong fairness}  \\[9pt]
\hline
Disjunctive~ & $|B| + k + 1$ &
        $2|B| - 1$ & 
        $2|B| + k - 1$ &
        $2|B| - 1$ 
          \\[4pt]
\hline
Conjunctive~ & 
        $k+1$ &  
        $2|B|-2~(*)$ & 
        $k+1~(*)$ &  
        $2|B|-2~(*)$
\end{tabular}
}
%\vspace{-10pt}
\end{table}
The table distinguishes between disjunctive and conjunctive systems (in rows).
In the columns,
we consider satisfaction of properties $h(A,B^{(k)})$ and the existence of deadlocks,
with and without fairness assumptions.
All results hold for open systems, and for both path quantifiers $\pforall$ and $\pexists$.
Cutoffs depend on the size of process template $B$ and the number $k \geq 1$ of $B$-processes a property talks about.

Results marked with a $(*)$ are for a restricted class of systems:
for conjunctive systems with fairness, we require infinite runs to be 
initializing, i.e., all non-deadlocked 
%$B$-processes 
processes
return to 
$\init$
%$\initstate_B$ 
infinitely often.\footnote{This assumption is in the same 
flavor as the restriction that $\initstate_A$ and $\initstate_B$ appear in 
all conjunctive guards. Intuitively, the additional restriction makes sense 
since conjunctive systems model shared resources, and everybody who takes a 
resource should eventually release it.}
Additionally, the cutoffs for deadlock detection in conjunctive systems only support $1$-conjunctive systems.
\ak{explain why we need it}
%The reason for this restriction will be explained in Remark~\ref{rem:general-conj-tough} and Appendix~\ref{gua:sec:app-conj}.

All cutoffs in the table are tight---no smaller cutoff can exist for 
this class of systems and properties---except for the case of deadlock 
detection in disjunctive systems without fairness. There, the cutoff is 
asymptotically tight, i.e., it must increase linearly with the size 
of the process template.

Note that the table does not describe all possible combinations:
for example, we do not consider satisfaction of $h(A,B^{(k)})$ on strong-fair runs.
But the results in the table are the most interesting, from our view,
for parameterized synthesis.

In the following sections we prove the results.

\section{Proof Structure}\label{gua:sec:proof-structure}

The proofs for the cutoff results, new and original,
are based on two lemmas, Monotonicity and Bounding~\cite{Emerson00}.
When combined together, the lemmas give a cutoff.
We state the lemmas, and discuss them in the context of deadlock detection and fairness.
The detailed proofs are in Sections~\ref{gua:sec:proofs-disj} and \ref{gua:sec:proofs-conj}.
Note that we only consider properties of the form $h(A,B^{(1)})$---the
proof ideas extend to general properties $h(A,B^{(k)})$ without difficulty.
Similarly, in most cases the proof ideas extend to open systems
without major difficulties---mainly because when we construct a simulating
run, we have the freedom to choose the input that is needed.
% AK: in the main text of the paper no difficulties are seen
Only for the case of deadlock detection we have to handle open systems explicitly.

\smallskip
\noindent
{\bf 1) Monotonicity lemma:} if a behavior is possible in a (conjunctive or disjunctive) system with $n \in \Nat$ copies of $B$,
then it is also possible in a (conjunctive or disjunctive resp.) system with one additional process:
\[
\largesys \models \pexists h(A,B^{(1)}) 
~\implies~
(A,B)^{(1, n+1)} \models \pexists h(A,B^{(1)}), 
\]
and if a deadlock is possible in $(A,B)^{(1, n)}$, then it is possible in $(A,B)^{(1, n+1)}$.

\parit{Discussion}
The lemma is easy to prove for properties 
$\pexists h(A,B^{(1)})$ in both disjunctive and conjunctive systems, by letting the 
additional process stay in its initial state $\init_B$ forever (see~\cite{Emerson00}).
This cannot disable transitions with disjunctive guards, as
these check for \emph{existence} of a local state in another process (and we 
do not remove any processes), and it cannot disable conjunctive guards since 
they contain $\init_B$ by assumption. 
However, this construction violates fairness, since the new process
never moves. This can be resolved in the disjunctive case by letting the
additional process mimic all transitions of an existing process. But in
general this does not work in conjunctive systems (due to the non-reflexive
interpretation of guards).
For this case and for deadlock detection, the proof is not
trivial and may only work for $n \geq c$, for some lower bound $c \in \Nat$.
The following sections provide the details.

\smallskip
\noindent
{\bf 2) Bounding lemma:} there exists a number $c \in \Nat$ such that
a behavior is possible in a system with $c$ copies of $B$ if it is possible in a system with
$n \geq c$ copies of process $B$:
\[
(A,B)^{(1, c)} \models \pexists h(A,B^{(1)})
~\impliedby~
(A,B)^{(1, n)} \models \pexists h(A,B^{(1)})
,
\]
and a deadlock is possible in \cutoffsys if it is possible in \largesys.

\parit{Discussion}
For disjunctive systems,
the main difficulty is that removing processes
might falsify guards of the local transitions of $A$ or $B_1$ in a given run.
To address this,
Emerson and Kahlon~\cite{Emerson00} came up with so-called flooding construction (described later).
For conjunctive systems,
removing processes from a run is easy for the case of infinite runs,
since a transition that was enabled before cannot become disabled.
Here, the difficulty is in preserving deadlocks,
because removing processes may enable processes that were deadlocked before.
The next sections explain how to address this.

\parbf{Tightness}
Recall from Section~\ref{gua:sec:paramsynt} that
$c$ is a tight cutoff iff $c$ is a cutoff and there are templates $(A,B)$ and a property $\Phi$,
such that
$$\sys{c-1} \not\models \Phi \textit{ and } \sys{c} \models \Phi.$$
For deadlock detection this is equivalent to:
$\sys{c-1}$ does not have a deadlock but $\sys{c}$ does.
%For properties of the form $\E h(A,B^{(1)})$ and $\A h(A,B^{(1)})$,
%due to the equivalence ``$\sys(n) \models \E h$'' $\equiv$ ``$\sys{n} \not\models \A\neg h$'',
%this is equivalent to having:
%$\sys{c-1} \not\models \E h(A,B^{(1)})$ and $\sys{c} \models \E h(A,B^{(1)})$.
To prove tightness, we provide a template $(A,B)$ and a property.

The next sections contains all the proofs of the results in the table.
For each row and column, we prove monotonicity and bounding lemmas,
as well as tightness.
Note that for simplicity the proofs are for the case of $h(A,B_1)$,
while the generalization to the case $h(A,B^{(k)})$ follows.

\section{Proof Techniques for Disjunctive Systems} \label{gua:sec:proofs-disj}

\subsection{\LTLmX\ Properties without Fairness: Existing Constructions}
\label{gua:sec:ideas-disj-nofair}

We revisit the main techniques
of the original proof of Theorem~\ref{thm:disj-cutoff-pairs}~\cite{Emerson00}. 

\begin{lemma}[Monotonicity: Disj, \LTLmX, Unfair] \label{disj:le:NonFairDisjunctiveMono}
    For disjunctive systems:
    \begin{align*}
    &\forall n \geq 1:\\
    &(A,B)^{(1,n)} \models \pexists h(A,B_1)
    \ \Impl \
    (A,B)^{(1,n+1)} \models \pexists h(A,B_1).
    \end{align*}
\end{lemma}
\begin{proof}
Given a run $x$ of $(A,B)^{(1,n)}$,
we construct a run $y$ of $(A,B)^{(1,n+1)}$: 
copy $x$ into $y$ and keep the additional process in the initial state.
\end{proof}

As for the bounding lemma, we construct an infinite run $y$ of $\cutoffsys$ 
with $y \models h(A,B^{(1)})$, 
based on an infinite run $x$ of $\largesys$ with $n>c$ and $x \models h(A,B^{(1)})$. 
The idea is to copy local runs $x(A)$ and $x(B_1)$ into $y$, 
and construct runs of other processes in a way 
that enables all transitions along $x(A)$ and $x(B_1)$. 
The latter is achieved with the flooding construction.

\myparagraph{Flooding construction \cite{Emerson00}}
Given a run $x = (s_1,e_1,p_1), (s_2,e_2,p_2) \ldots$ of $\largesys$, let
$\visited_\mB(x)$ be the set of all local states visited by $B$-processes in $x$,
i.e., $\visited_\mB(x) = \{ q \in Q_B \| \exists m \exists i.\ s_m(B_i) = q \}$. 

For every $q \in \visited_\mB(x)$ there is a local run of \largesys, say $x(B_i)$,
that visits $q$ first, say at moment $m_q$. Then, saying that process 
$B_{i_q}$ of \cutoffsys \emph{floods $q$} means:
$$y(B_{i_q}) = x(B_i)\slice{1}{m_q}(q)^\omega.$$ 
In words: the run $y(B_{i_q})$ is the same as $x(B_i)$ until moment $\time_q$,
and after that the process never moves.

The construction achieves the following. 
If we copy local runs of $A$ and $B_1$ from $x$ to $y$, 
and in $y$ for every $q \in \visited_\mB(x)$ introduce one process that floods $q$, 
then: 
if in $x$ at some moment $\time$ there is a process in state $q'$, 
then in $y$ at moment $\time$ there will also be a process (different from 
$A$ and $B_1$) in state $q'$. Thus, every transition of $A$
 and $B_1$, which is enabled at moment $\time$ in $x$, will also be enabled in $y$. 

\begin{lemma}[Bounding: Disj, \LTLmX, Unfair] \label{disj:le:NonFairDisjunctiveBounding}
    For disjunctive systems:
    \begin{align*}
    \forall n \geq |B|+2:\ 
    (A,B)^{(1,|B|+2)} \models \pexists h(A,B_1)
    \ \ \Implied \ \ 
    \largesys \models \pexists h(A,B_1).
    \end{align*}
\end{lemma}

The proof of the lemma is from \cite[Lemma 4.1.2]{Emerson00}. 
We recapitulate it to introduce the notions of ``a process floods a state'', 
\destutter, \interleave, and ``process mimics another process'',
which are used in our proofs later.

\begin{proof}[Proof idea]
The lemma is proved by copying local runs $x(A)$ and $x(B_1)$,
and flooding all states in $\visited_\mB(x)$.
To ensure that at least one process moves infinitely often
in $y$, we copy one additional (infinite) local run from $x$. Finally, it may
happen that the resulting collection of local runs violates the interleaving 
semantics requirement. To resolve this, we add stuttering steps into local 
runs whenever two or more processes move at the same time, and we 
remove global stuttering steps in $y$. Since the only difference between 
$x(A,B_1)$ and $y(A,B_1)$ are stuttering steps, $y$ and $x$ satisfy the same $
\LTLmX$-properties $h(A,B^{(1)})$. 
Since $\card{\visited_\mB(x)} \leq 
\card{B}$, we need at most $1+\card{B}+1$ copies of $B$ in \cutoffsys.
\end{proof}

\begin{proof}
Let $c = |B|+2$ and $n \geq c$. Let $x=(s_1,e_1,p_1), (s_2,e_2,p_2) \ldots$ be a run of $\largesys$ that satisfies $\pexists h(A,B_1)$. We construct a run $y$ of the cutoff system $\cutoffsys$ with $y(A, B_1) \simeq x(A, B_1)$.

Let $\visited(x)$ be the set of all visited states by B-processes in run $x$: $\visited(x) = \{ q \| \exists m \exists i: s_m(B_i) = q \}$. 

Construct the run $y$ of \cutoffsys as follows.
\li
  \-[a.] We copy runs of $A$ and $B_1$ from $x$ to $y$:
         $y(A)=x(A)$, $y(B_1)=x(B_1)$;

  \-[b.] Since $x$ is infinite, it has at least one infinitely moving process, denoted $B_\infty$. Devote one unique process $B_\infty$ in \cutoffsys that copies the behaviour of $B_\infty$ of \largesys: $y(B_\infty)=x(B_\infty)$.

  \-[c.] For every $q \in \visited$, there is a process of \largesys, denoted $B_i$, that visits $q$ first, at moment denoted $m_q$. Then devote one unique process in \cutoffsys, denoted $B_{i_q}$, that \emph{floods $q$}: set $y(B_{i_q}) = x(B_i)\slice{1}{m_q}(q)^\omega$. In words: the run $y(B_{i_q})$ repeats exactly that of $x(B_i)$ till moment $m_q$, after which the process is never scheduled.

  \-[d.] Let any other process $B_i$ of \cutoffsys not used in the previous steps (if any) \emph{mimic} the behavior of $B_1$ of \cutoffsys: $y(B_i) = y(B_1)$.
\il
The figure illustrates the construction.\ak{\init should be flooded}
On the left is $\largesys$ and on the right is $(A,B)^{(1,|B|+2)}$
(i.e., $(A,B)^{(1,5)}$, since $|B|=3$ in the figure).

\begin{figure}[hp]
\centering
\scalebox{0.7}{
%\usetikzlibrary{calc,arrows,shapes,fit}

\begin{tikzpicture}
\tikzstyle{every label} = [font=\normalsize];
\tikzstyle{state} = [circle, inner sep=2pt, minimum size=2mm, draw=black, node distance=0.5cm]
\tikzstyle{hiddenstate} = [state, circle,inner sep=0pt,outer sep=0pt, minimum size=0pt, draw=gray]
\tikzstyle{topstate} = [hiddenstate, node distance=0.75cm]
\tikzstyle{state21} = [state, rectangle]
\tikzstyle{state22} = [state, diamond]
\tikzstyle{state23} = [state, regular polygon,circle]

%----------------------------------------------------------------------------------------------------------------------------

%T1_1
\node (t11_top) [topstate, label={[name=t11_top-label] $A_1$}] {};
% \draw ($(t11_top)-(.2,0)$) -- ($(t11_top)+(.2,0)$);
\node (t11_1) [hiddenstate, below of=t11_top] {};
\node (t11_2) [hiddenstate, below of=t11_1] {};
\node (t11_3) [hiddenstate, below of=t11_2] {};
\node (t11_4) [hiddenstate, below of=t11_3] {};
\node (t11_5) [hiddenstate, below of=t11_4] {};
\node (t11_6) [hiddenstate, below of=t11_5] {};
\draw[|-] (t11_top) -- (t11_6);

%T2_1
\node (t21_top) [topstate, right of=t11_top, label={[name=t21_top-label] $B_1$}] {};
% \draw ($(t21_top)-(.2,0)$) -- ($(t21_top)+(.2,0)$);
\node (t21_1) [hiddenstate, below of=t21_top] {};
\node (t21_2) [hiddenstate, below of=t21_1] {};
\node (t21_3) [hiddenstate, below of=t21_2] {};
\node (t21_4) [hiddenstate, below of=t21_3] {};
\node (t21_5) [hiddenstate, below of=t21_4] {};
\node (t21_6) [hiddenstate, below of=t21_5] {};
\draw[|-] (t21_top) -- (t21_6);

%T2_2
\node (t22_top) [topstate, right of=t21_top, label={[name=t22_top-label] $B_2$}] {};
\node (t22_1) [state21, below of=t22_top] {};
\node (t22_2) [hiddenstate, below of=t22_1] {};
\node (t22_3) [state22, below of=t22_2] {};
\node (t22_4) [hiddenstate, below of=t22_3] {};
\node (t22_5) [state23, below of=t22_4] {};
\node (t22_6) [hiddenstate, below of=t22_5] {};
\draw[|-] (t22_top) -- (t22_1);
\draw[] (t22_1) -- (t22_3);
\draw[] (t22_3) -- (t22_5);
\draw[] (t22_5) -- (t22_6);

% other processes placeholder
\node (t2i_top) [topstate, right of=t22_top] {};

% T_2_m
\node (t2m_top) [topstate, right of=t2i_top, label={[name=t2m_top-label] $B_{n-1}$}] {};
\node (t2m_1) [hiddenstate, below of=t2m_top] {};
\node (t2m_2) [state23, below of=t2m_1] {};
\node (t2m_3) [hiddenstate, below of=t2m_2] {};
\node (t2m_4) [state22, below of=t2m_3] {};
\node (t2m_5) [hiddenstate, below of=t2m_4] {};
\node (t2m_6) [hiddenstate, below of=t2m_5] {};
\draw[|-] (t2m_top) -- (t2m_2);
\draw[] (t2m_2) -- (t2m_4);
\draw[] (t2m_4) -- (t2m_6);

% T_2_n
\node (t2n_top) [topstate, right of=t2m_top, label={[name=t2n_top-label] $B_n$}] {};
\node (t2n_1) [hiddenstate, below of=t2n_top] {};
\node (t2n_2) [hiddenstate, below of=t2n_1] {};
\node (t2n_3) [hiddenstate, below of=t2n_2] {};
\node (t2n_4) [hiddenstate, below of=t2n_3] {};
\node (t2n_5) [hiddenstate, below of=t2n_4] {};
\node (t2n_6) [hiddenstate, below of=t2n_5, label={[name=t2n_6-label]below:$\infty$}] {};
\draw[|-] (t2n_top) -- (t2n_6);

%\draw[dashed] (t2n_2) -- (t2n_4);
%\draw[dashed] (t2n_4) -- (t2n_6);

\coordinate (t11_old-top) at (t11_top);
\coordinate (t11_old-bottom) at (t11_6);
\coordinate (t21_old-top) at (t21_top);
\coordinate (t21_old-bottom) at (t21_6);
\coordinate(t22_old-bottom) at (t22_6);
\coordinate (t2n_old-top) at (t2n_top);
%\coordinate(t2n_old-bottom) at (t2n_6);
\coordinate(t2m_old-bottom) at (t2m_6);
\coordinate(t2n_old-bottom) at ($(t2n_6-label)-(0,0.1)$);

% other processes line
\coordinate (t22_center) at ($(t22_3.center)$);
\coordinate (t2m_center) at ($(t2m_3.center)$);
\node (othrerssss) at ($(t22_center)+(0.8,0)$) {\textbf{\ldots}};
% \draw[dotted,semithick] ($(t22_center)+(0.4,0)$) to ($(t2m_center)-(0.3,0)$); % a bit cheating

% box
\node (box_outer) [draw=black, fit=(t11_top-label.center) (t2n_6), inner sep=0.5cm] {} ;
%\draw[solid] ($(t11_top-label)!0.5!(t21_top-label)+(0,0.5)$) to ($(t11_6)!0.5!(t21_6)-(0,0.5)$) ;
% \draw[solid] ($(t21_top-label)!0.5!(t22_top-label)+(0,0.5)$) to ($(t21_6)!0.5!(t22_6)-(0,0.5)$) ;

%----------------------------------------------------------------------------------------------------------------------------

%T1_1
\node (t11_top) [topstate, label={[name=t11_top-label] $A_1$}, right of=t2n_old-top, node distance=1.5cm] {};
% \draw ($(t11_top)-(.2,0)$) -- ($(t11_top)+(.2,0)$);
\node (t11_1) [hiddenstate, below of=t11_top] {};
\node (t11_2) [hiddenstate, below of=t11_1] {};
\node (t11_3) [hiddenstate, below of=t11_2] {};
\node (t11_4) [hiddenstate, below of=t11_3] {};
\node (t11_5) [hiddenstate, below of=t11_4] {};
\node (t11_6) [hiddenstate, below of=t11_5] {};
\draw[|-] (t11_top) -- (t11_6);

%T2_1
\node (t21_top) [topstate, right of=t11_top, label={[name=t21_top-label] $B_1$}] {};
% \draw ($(t21_top)-(.2,0)$) -- ($(t21_top)+(.2,0)$);
\node (t21_1) [hiddenstate, below of=t21_top] {};
\node (t21_2) [hiddenstate, below of=t21_1] {};
\node (t21_3) [hiddenstate, below of=t21_2] {};
\node (t21_4) [hiddenstate, below of=t21_3] {};
\node (t21_5) [hiddenstate, below of=t21_4] {};
\node (t21_6) [hiddenstate, below of=t21_5] {};
\draw[|-] (t21_top) -- (t21_6);

%T2_2
\node (t22_top) [topstate, right of=t21_top, label={[name=t22_top-label] $B_2$}] {};
\node (t22_1) [state21, below of=t22_top] {};
\node (t22_2) [state21, below of=t22_1] {};
\node (t22_3) [state21, below of=t22_2] {};
\node (t22_4) [state21, below of=t22_3] {};
\node (t22_5) [state21, below of=t22_4] {};
\node (t22_6) [hiddenstate, below of=t22_5] {};
\draw[|-] (t22_top) -- (t22_1);
\draw[dashed] (t22_1) -- (t22_2);
\draw[dashed] (t22_2) -- (t22_3);
\draw[dashed] (t22_3) -- (t22_4);
\draw[dashed] (t22_4) -- (t22_5);
\draw[dashed] (t22_5) -- (t22_6);

% other processes placeholder
\node (t2i_top) [topstate, right of=t22_top] {};

% T_2_3
\node (t23_top) [topstate, right of=t22_top, label={[name=t23_top-label] $B_3$}] {};
\node (t23_1) [state21, below of=t23_top] {};
\node (t23_2) [hiddenstate, below of=t23_1] {};
\node (t23_3) [state22, below of=t23_2] {};
\node (t23_4) [state22, below of=t23_3] {};
\node (t23_5) [state22, below of=t23_4] {};
\node (t23_6) [hiddenstate, below of=t23_5] {};
\draw[|-] (t23_top) -- (t23_1);
\draw[] (t23_1) -- (t23_3);
\draw[dashed] (t23_3) -- (t23_4);
\draw[dashed] (t23_4) -- (t23_5);
\draw[dashed] (t23_5) -- (t23_6);

% T_2_4
\node (t24_top) [topstate, right of=t23_top, label={[name=t24_top-label] $B_4$}] {};
\node (t24_1) [hiddenstate, below of=t24_top] {};
\node (t24_2) [state23, below of=t24_1] {};
\node (t24_3) [state23, below of=t24_2] {};
\node (t24_4) [state23, below of=t24_3] {};
\node (t24_5) [state23, below of=t24_4] {};
\node (t24_6) [hiddenstate, below of=t24_5] {};
\draw[|-] (t24_top) -- (t24_2);
\draw[dashed] (t24_2) -- (t24_3);
\draw[dashed] (t24_3) -- (t24_4);
\draw[dashed] (t24_4) -- (t24_5);
\draw[dashed] (t24_5) -- (t24_6);

% T_2_5
\node (t25_top) [topstate, right of=t24_top, label={[name=t25_top-label] $B_5$}] {};
\node (t25_1) [hiddenstate, below of=t25_top] {};
\node (t25_2) [hiddenstate, below of=t25_1] {};
\node (t25_3) [hiddenstate, below of=t25_2] {};
\node (t25_4) [hiddenstate, below of=t25_3] {};
\node (t25_5) [hiddenstate, below of=t25_4] {};
\node (t25_6) [hiddenstate, below of=t25_5, label={[name=t25_6-label]below:$\infty$}] {};
\draw[|-] (t25_top) -- (t25_6);

\coordinate (t11_new-top) at (t11_top);
\coordinate (t11_new-bottom) at (t11_6);
\coordinate (t21_new-top) at (t21_top);
\coordinate (t21_new-bottom) at (t21_6);
\coordinate (t22_new-top) at (t22_top);
\coordinate (t22_new-bottom) at (t22_6);
\coordinate(t23_new-bottom) at (t23_6);
\coordinate(t24_new-bottom) at (t24_6);
%\coordinate(t25_new-bottom) at (t25_6);
\coordinate(t25_new-bottom) at ($(t25_6-label)-(0,0.1)$);

% box
\node (box_outer) [draw=black, fit=(t11_top-label.center) (t25_6), inner sep=0.5cm] {} ;
%\draw[solid] ($(t11_top-label)!0.5!(t21_top-label)+(0,0.5)$) to ($(t11_6)!0.5!(t21_6)-(0,0.5)$) ;
% \draw[solid] ($(t21_top-label)!0.5!(t22_top-label)+(0,0.5)$) to ($(t21_6)!0.5!(t22_6)-(0,0.5)$) ;

\draw[->, bend right, dotted,semithick] (t22_old-bottom) to (t22_new-bottom);
\draw[->, bend right, dotted,semithick] (t22_old-bottom) to (t23_new-bottom);
\draw[->, bend right, dotted,semithick] (t2m_old-bottom) to (t24_new-bottom);
%\draw[->, bend right, dotted,semithick] (t2n_old-bottom) to (t25_new-bottom);
%\draw[->, bend left, dotted,semithick] (t2m_top-label.north) to (t24_top-label.north);
\draw[->, bend left=5, dotted,semithick] (t2n_top-label.north) to (t25_top-label.north);

\end{tikzpicture}
}
\end{figure}

The correctness follows from the observation that any transition of any process at any moment $m$ of $y$ was done by some process in $x$ at moment $m$, and hence is enabled at $m$. Also note that, if $\geq 2$ processes transit simultaneously in $y$, then the guards of their transitions will be enabled even if both of them are removed from the state space\ak{vague}. Note that it is possible that in $y$:
\li
  \- more than one process transits at the same moment. Then, \emph{\interleave} the transitions of such processes, namely arbitrarily sequentialize them. \ak{why are enabled}
  \- at some moment no processes move. Then remove elements of the run $y$ -- the resulting run is denoted $\destutter(y)$.
\il
This construction uses $|\visited| + 2 \leq |B|+2$ copies of B (ignoring case (d)).
\end{proof} 

\begin{tightness}[Disj, \LTLmX, Unfair] \label{obs:disj:tight_prop}
    The cutoff in Lemma~\ref{disj:le:NonFairDisjunctiveBounding} is tight.
    I.e., for any $k$ there exist process templates $(A,B)$ with $|B| = k$ 
    and $\LTLmX$ formula $h(A,B_1)$ such that:
    $$
    (A,B)^{(1,|B|+2)} \models \pexists h(A,B_1) ~~and~~ 
    (A,B)^{(1,|B|+1)} \not\models \pexists h(A,B_1).
    $$
\end{tightness}
\begin{proof}
The idea of the proof relies on the subtleties of the definition of a run: it is infinite (thus not globally deadlocked), and in each step of a run exactly one process moves. 

Consider the templates from Figure~\ref{fig:obs:disj:tight_prop} and let $\pexists h(A,B_1) = \pexists (\eventually 3_{B_1} \land \eventually\always (2_{B_1} \land {end}_A))$. In words: there exists a run in a system where process $B_1$ visits $3_B$ and process $B_1$ with $A$ eventually always stay in $2_B$ and ${end}_A$.
\begin{figure}[tb]
\centering
\begin{subfigure}[b]{0.35\textwidth}\center
\scalebox{0.66}{%!TEX root = table.tex
\begin{tikzpicture}[node distance=2.1cm,inner sep=1pt,minimum size=0.5mm,->,>=latex]

\node[initial below, state] (a_1) {$1_A$};
\node (dots) [right of=a_1] {$\ldots$};
\node[state] (a_all) [right of=dots] {${all}_A$};
\node[state] (a_end) [right of=a_all] {${end}_A$};

\path (a_1) edge [above] node {$\disj{1_B}$} (dots);
\path (dots) edge [above] node {$\disj{{\card{B}}_B}$} (a_all);
\path (a_all) edge [above] node {$\disj{3_B}$} (a_end);

\end{tikzpicture}}
\caption*{Template A}
\end{subfigure}
\begin{subfigure}[b]{0.62\textwidth}\center
\scalebox{0.66}{%!TEX root = table.tex
\begin{tikzpicture}[node distance=2.1cm,inner sep=1pt,minimum size=0.5mm,->,>=latex]

\node[initial below, state] (b_1) {$1_B$};
\node[state] (b_2) [left of=b_1] {$2_B$};
\node[state] (b_3) [right of=b_1]{$3_B$};
\node (dots) [right of=b_3] {$\ldots$};
\node[state] (b_k) [right=2.3cm of dots] {${\card{B}}_B$}; 

\path (b_1) edge [above] node {$\disj{1_B}$} (b_2);
\path (b_1) edge [above] node {$\disj{1_B}$} (b_3);
\path (b_3) edge [above] node {$\disj{3_B}$} (dots);
\path (dots) edge [above] node {$\disj{{\card{B}{-}1}_B}$} (b_k);
\path (b_k) [loop above] edge [right] node {$\disj{{\card{B}}_B}$} (b_k);
\path (b_3) [bend right=60] edge [above] node {$\disj{{all}_A}$} (b_1);

\end{tikzpicture}}
\caption*{Template B}
\end{subfigure}
\caption{Templates for proving Tightness~\ref{obs:disj:tight_prop}}
\label{fig:obs:disj:tight_prop}
\end{figure}
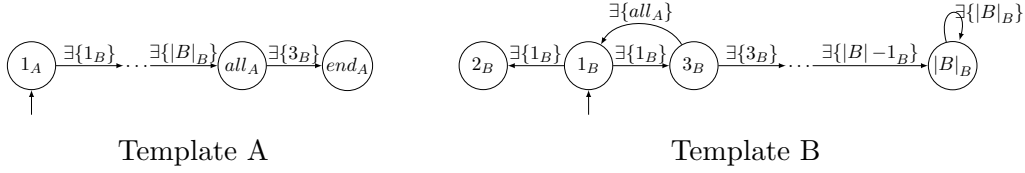

We need one process in every state of $B$ to enable the transitions of $A$ to ${all}_A$. Only when $A$ in ${all}_A$, $B_1$ can move $3_B \to 1_B$, and then at some point to $2_B$. After $B_1$ moves $3_B \to 1_B$, $A$ moves ${all}_A \to {end}_A$, which requires process $B_{i \neq 1}$ in $3_B$. Finally, to make the run infinite, there should be at least two processes in the state $k_B$.
Hence, every infinite run satisfying the formula needs at least $|B|+2$ $B$-processes.
\sj{other cases are covered in general lemma below}
\end{proof}

\subsection{\LTLmX\ Properties with Fairness: New Constructions} \label{gua:sec:ideas-disj-fair}

As for the case without fairness, proving the monotonicity lemma is simple.

\begin{lemma}[Monotonicity: Disj, \LTLmX, Fair] \label{disj:le:FairDisjunctiveMonotonicity}
For disjunctive systems:
\begin{align*}
& \forall n \geq 1: \\
&(A,B)^{(1, n)} \models \pexists_{uncond} h(A,B_1) 
\implies
(A,B)^{(1,n+1)} \models \pexists_{uncond} h(A,B_1),\\
\end{align*}
\end{lemma}
\begin{proof}
In run $x$ of $(A,B)^{(1,n)}$ with $n \geq 1$ all processes move infinitely often. 
Hence let the run $y$ of $(A,B)^{(1,n+1)}$ copy $x$, 
and let the new process mimic an infinitely moving B process of $(A,B)^{(1,n)}$.
\end{proof}

To prove the bounding lemma, we introduce two new constructions.
We need new constructions, because the flooding construction does not preserve fairness,
and also cannot be used to construct deadlocked runs,
since it does not preserve disabledness of transitions of processes $A$ or $B_1$.

Consider the proof task of the bounding lemma for disjunctive systems with fairness:
given an unconditionally fair run $x$ of 
\largesys with 
$x \models h(A,B^{(1)})$, we want to construct an unconditionally fair run $y$ 
of \cutoffsys with $y \models h(A,B^{(1)})$. In contrast to unfair systems, we 
need to ensure that all processes move infinitely often in $y$. 
The insight is 
that after a finite time all processes will start looping 
around some set $\visited^\inf$ of states. We construct a run $y$ that
mimics this. To this end, we introduce two constructions. \emph{Flooding with
evacuation} is similar to flooding, but instead of keeping
processes in their flooding states forever it evacuates the processes into 
$\visited^\inf$. \emph{Fair extension} lets all processes move infinitely 
often without leaving $\visited^\inf$.

\myparagraph{Flooding with evacuation}
Given a subset $\mF \subseteq \mB$
%\footnote{In this section we will use only $\mP_1 = \{B_1\}$, 
%          but for the case of deadlocks we will need a set.  
%          $\mP_1$ is a set of processes whose local runs we will copy 
%          from $x$ to $y$ later in the proof.} 
and an infinite run $x=(s_1,e_1,p_1)\ldots$ of \largesys, 
define
\begin{align}
& \visInf{\mF}{x} = \{ q \|\! \exists \text{ infinitely many } ~~~~ m\!:  
s_m(B_i) = q 
\text{ for some } B_i \in \mF \} \label{disj:def_vinf_wrt} \\
& \visFin{\mF}{x} = \{ q \|\! \exists \text{ only finitely many } m\!:  
s_m(B_i) = q
\text{ for some } B_i\in \mF \} \label{disj:def_vfin_wrt}
\end{align}
Let $q \in \visFin{\mF}{x}$.
In run $x$ there is a moment $f_q$ when $q$
is reached for the first time by some process from $\mF$, denoted $B_{\first_q}$. 
Also, in run $x$ there is a moment $l_q$ such that:
$s_{l_q}(B_{\last_q})=q$ for some process $B_{\last_q} \in \mF$, 
and $s_t(B_i)\neq q$ for all $B_i \in \mF$,
$t > l_q$%
---i.e., when some process from $\mF$ is in state $q$ for the last time in $x$. 
Then, saying that process $B_{i_q}$ of \cutoffsys\ 
{\em floods $q \in \visFin{\mF}{x}$ and then evacuates into $\visInf{\mF}{x}$} 
means: 
$$
y(B_{i_q}) = x(B_{\first_q})\slice{1}{f_q} \ \cdot\ (q)^{(l_q - f_q + 1)} \cdot \ 
x(B_{\last_q})\slice{l_q}{m} \ \cdot \ (q')^\omega,
$$
where $q'$ is the state in $\visInf{\mF}{x}$ that $x(B_{\last_q})$ reaches first, 
at some moment $\time \geq l_q$.
In words, process $B_{i_q}$ mimics process $B_{\first_q}$ until it reaches $q$, 
then does nothing until process $B_{\last_q}$ starts leaving $q$, 
then it mimics $B_{\last_q}$ until it reaches $\visInf{\mF}{x}$.

The construction ensures: 
if we copy local runs of all processes not in $\mF$ from $x$ to $y$, 
then all transitions of $y$ are enabled. 
This is because, for any process $p$ of $\cutoffsys$ that takes a transition in $y$ at any moment, 
the set of states visible to process $p$ is a superset of the set of states 
visible to the original process in \largesys whose transitions process $p$ copies.

\myparagraph{Fair extension} 
\ak{explain intuition about those three sets}
\ak{adapt to dead}
Here, we consider a path $x$ that is the postfix of an unconditionally fair run $x'$ of $\largesys$, 
starting from the moment where no local states from $\visFin{\mB}{x'}$ are visited anymore. 
We construct a corresponding unconditionally-fair path $y$ of $\cutoffsys$, 
where no local states from $\visFin{\mB}{x'}$ are visited.

Formally, let $n \geq 2|B|$, and $x$ an unconditionally-fair path of $\largesys$ such that
$\visFin{\mB}{x}=\emptyset$.
Let $c \geq 2|B|$, and $s_1'$ a state of \cutoffsys
with
\li
\- $s_1'(A_1)=s_1(A_1)$, $s_1'(B_1)=s_1(B_1)$;

\- for every $q \in \visInf{B_2..B_n}{x} \smi \visInf{B_1}{x}$,
   there are two processes $B_{i_q}, B_{i_q'}$ of \cutoffsys
   that start in $q$, i.e., $s_1'(B_{i_q})=s_1'(B_{i_q'})=q$;

\- for every $q \in \visInf{B_2..B_n}{x} \cap \visInf{B_1}{x}$,
   there is one process $B_{i_q}$ of \cutoffsys
   that starts in $q$;

\- for some $\qstar \in \visInf{B_2..B_n}{x} \cap \visInf{B_1}{x}$,
   there is one additional process of \cutoffsys, 
   different from any in the above, 
   called $B_{i_\qstar'}$,
   that starts in $\qstar$;
   and

\- any other process $B_i$ of \cutoffsys 
   starts in some state of $\visInf{B_2..B_n}{x}$.
\il
Note that, if $\visInf{B_2..B_n}{x}\cap \visInf{B_1}{x} = \emptyset$, 
then the third and fourth pre-requisite are trivially satisfied.

The fair extension extends state $s_1'$ of \cutoffsys 
to an unconditionally-fair path $y=(s'_1,e'_1,p'_1)\ldots$ 
with $y(A_1,B_1) = x(A_1,B_1)$ as follows.
\li
\-[(a)] $y(A_1)=x(A_1)$, $y(B_1)=x(B_1)$.

\-[(b)] For every $q \in \visInf{B_2..B_n}{x} \smi \visInf{B_1}{x}$: 
       in run $x$ there is $B_i \in \{B_2..B_n\}$ 
       that starts in $q$ and visits it infinitely often. 
       Let $B_{i_q}$ and $B_{i'_q}$ of \cutoffsys mimic $B_i$ in turns: 
       first $B_{i_q}$ mimics $B_i$ until it reaches $q$, 
       then $B_{i'_q}$ mimics $B_i$ until it reaches $q$, and so on.

\-[(c)] Arrange the states of $\visInf{B_2..B_n}{x}\cap \visInf{B_1}{x}$ 
       in some order $(\qstar, q_1, \ldots, q_l)$.  
       The processes $B_{i_\qstar'}, B_{i_\qstar}, B_{i_{q_1}}, \ldots, B_{i_{q_l}}$ 
       behave as follows.
       Start with $B_{i_\qstar'}$: 
       when $B_1$ enters $\qstar$ in $y$, it carries%
       \footnote{``Process $B_1$ starting at moment $m$ carries process $B_i$ 
                 from $q$ to $q'$'' means: process $B_i$ mimics 
                 the transitions of $B_1$ starting at moment $m$ at $q$ 
                 until $B_1$ first reaches $q'$.}
       $B_{i_\qstar'}$             from $\qstar$ to $q_1$, 
       then carries $B_{i_{q_1}}$ from $q_1$ to $q_2$, \ldots, 
       then carries $B_{i_{q_l}}$ from $q_l$ to $\qstar$, 
       then carries $B_{i_\qstar}$ from $\qstar$ to $q_1$, 
       then carries $B_{i_\qstar'}$ from $q_1$ to $q_2$, 
       then carries $B_{i_{q_1}}$ from $q_2$ to $q_3$,
       and so on.

%\-[c2.] otherwise, $\visited_{\inf\cap B_1}{x} = \{ q^\star \}$. Then $B_1$ only ever makes transitions $q^\star \to q^\star$, thus let process $B_{i_{q^\star}}$ mimic this.

\-[(d)] Any other $B_i$ of \cutoffsys,
       starting in $q \in \visInf{B_2..B_n}{x}$,
       mimics $B_{i_q}$.
\il
Note that parts (b) and (c) of the construction ensure that there is always at
       least one process in every state from $\visInf{B_2..B_n}{x}$. This
       ensures that the guards of all transitions of the construction are satisfied.
Excluding processes in (d), the fair extension uses up to $2|B|$ copies of $B$.%
\footnote{A careful reader may notice that,
          if
          $|\visInf{B_1}{x}|=1$ and $|\visInf{B_2..B_n}{x}|=|B|$,
          then the construction uses $2|B|+1$ copies of $B$.
          But one can slightly modify the construction for this special case,
          and remove process $B_{i_\qstar'}$ from the pre-requisites.}

Now we are ready to prove the bounding lemma.

\begin{lemma}[Bounding: Disj, \LTLmX, Fair] \label{disj:le:FairDisjunctiveBounding}
For disjunctive systems:
\begin{align*}
&\forall n>2|B|: \\
&(A,B)^{(1,2|B|)} \models \pexists_{uncond} h(A,B_1) &
&\impliedby& &
(A,B)^{(1,n)} \models \pexists_{uncond} h(A,B_1),\\
\end{align*}
\end{lemma}
\begin{proof}
\sj{for weak or strong fairness, the same construction can be used; evacuation is not necessary, but also doesn't increase the cutoff if we use it; difficulty: show that cutoff is still tight
}
Let $c=2\card{B}$. 
Given an unconditionally-fair run $x$ of $\largesys$,
we construct an unconditionally-fair run $y$ of the cutoff system $\cutoffsys$ 
such that $y(A,B_1)$ is stuttering equivalent to $x(A,B_1)$.

Note that in $x$ there is a moment $m$ such that all local states that are visited after $m$ are in $\visInf{\mB}{x}$.

The construction has two phases. In the first phase, we apply flooding for states in $\visInf{\mB}{x}$, and flooding with evacuation for states in $\visFin{\mB}{x}$:
\li
\-[(a)] $y(A)=x(A)$, $y(B_1)=x(B_1)$;

\-[(b)] for every $q \in \visInf{B_2..B_n}{x} \smi \visInf{B_1}{x}$, 
       devote two processes of $\cutoffsys$ that flood $q$;

\-[(c)] for some $\qstar \in \visInf{B_2..B_n}{x} \cap \visInf{B_1}{x}$,
       devote one process of \cutoffsys that floods $\qstar$;

\-[(d)] for every $q \in \visFin{B_2..B_n}{x}$, 
       devote one process of $\cutoffsys$ that 
       floods $q$ and evacuates into $\visInf{B_2..B_n}{x}$;
       and

\-[(e)] let other processes (if any) mimic process $B_1$.
\il
The phase ensures that at moment $m$ in $y$, 
there are no processes in $\visFin{\mB}{x}$, 
and all the pre-requisites of the fair extension are satisfied.

The second phase applies the fair extension, 
and then establishes the interleaving semantics 
as in the bounding lemma in the non-fair case.
The overall construction uses up to $2|B|$ copies of $B$.
% Indeed: 
% fin&smi=0, fin&cap=0,
% smi = InfB2..Bn - cap
% InfB2..Bn <= B - fin
% Then:
% 1+2smi+1cap+1+fin <= 1+2(InfB2..Bn-cap)+cap+1+fin = 
%                      2+2InfB2..Bn-cap-fin <=
%                      2+2B-cap-fin
% now split case:
% everywhere fin>0 (otherwise 2|B| follows from the analysis of the fair extension)
% note: if cap=0, then we actually have (recall fair pre)
%       1+2smi+fin =< 1+2(B-InfB1-fin)+fin = 1+2B-2InfB1-2fin =< 2B-3
% thus cap>0,fin>0
% then 2+2B-cap-fin =< 2B
\end{proof}

\begin{tightness}[Disj, \LTLmX, Fair] \label{obs:disj:fair_tight_prop}
The cutoff in Lemma~\ref{disj:le:FairDisjunctiveBounding} is tight.
I.e., 
for any $k$ there exist process templates $(A,B)$ with $|B| = k$ 
and $\LTLmX$ formula $h(A,B_1)$ such that:
$$
(A,B)^{(1,2|B|)} \models \pexists h(A,B_1) ~~and~~ 
(A,B)^{(1,2|B|-1)} \not\models \pexists h(A,B_1).
$$
\ak{what happens if we bound $T_A$?}
\end{tightness}
\begin{proof}
Consider process templates $A,B$ from Figure~\ref{fig:obs:disj:fair_tight_prop}
and the property $\pexists \true$.
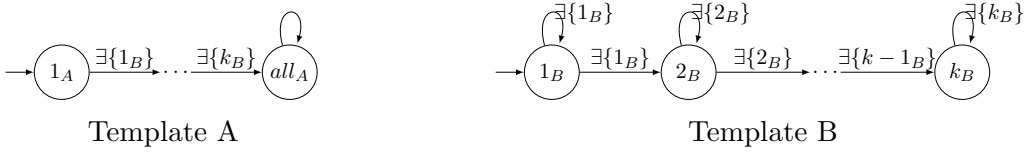
\begin{figure}[tb]
\centering
\begin{subfigure}[b]{0.35\textwidth}\center
\scalebox{0.75}{%!TEX root = table.tex
\begin{tikzpicture}[node distance=2cm,inner sep=1pt,minimum size=0.5mm,->,>=latex]

\node[state,initial left] (a_1) {$1_A$};
\node (dots) [right of=a_1] {$\ldots$};
\node[state] (a_all) [right of=dots] {${all}_A$};

\path (a_1) edge [above] node {$\disj{1_B}$} (dots);
\path (dots) edge [above] node {$\disj{k_B}$} (a_all);
\path (a_all) [loop above] edge node {} (a_1);

\end{tikzpicture}}
\caption*{Template A}
\end{subfigure}
\hspace{1cm}
\begin{subfigure}[b]{0.55\textwidth}\center
\scalebox{0.75}{%!TEX root = table.tex
\begin{tikzpicture}[node distance=2.4cm,inner sep=1pt,minimum size=0.5mm,->,>=latex]

\node[state,initial left] (b_1) {$1_B$};
\node[state] (b_2) [right of=b_1] {$2_B$};
\node (dots) [right of=b_2] {$\ldots$};
\node[state] (b_k) [right of=dots] {$k_B$}; 

\path (b_1) edge [above] node {$\disj{1_B}$} (b_2);
\path (b_2) edge [above] node {$\disj{2_B}$} (dots);
\path (dots) edge [above] node {$\disj{{k-1}_B}$} (b_k);

\path (b_1) [loop above] edge [right] node {$\disj{1_B}$} (b_1);
\path (b_2) [loop above] edge [right] node {$\disj{2_B}$} (b_2);
\path (b_k) [loop above] edge [right] node {$\disj{k_B}$} (b_k);

\end{tikzpicture}}
\caption*{Template B}
\end{subfigure}
\caption{Templates for proving Tightness~\ref{obs:disj:fair_tight_prop}}
\label{fig:obs:disj:fair_tight_prop}
\end{figure}
\end{proof}

\subsection{Deadlocks without Fairness: Updated Constructions} \label{gua:sec:proofs-disj-deadlock-unfair}

%The lemma for deadlock detection, for fair and unfair cases,
%is proven for $n \geq |B|+1$.
%In the case of local deadlocks, 
%process $B_{n+1}$ mimics a process that moves infinitely often in $x$.
%In the case of global deadlocks, 
%by pigeon hole principle, 
%in the global deadlock state there is a state $q$ with at least two processes in it---let process $B_{n+1}$ mimic a process that deadlocks in $q$.

\begin{lemma}[Monotonicity: Disj, Deadlocks, Unfair]
\label{mono_lem_disj_deadlocks_unfair}
    For disjunctive systems:
    $$\forall n\geq |B|+1: (A,B)^{(1,n)} \textit{ has a deadlock} \ 
    \Impl\ 
    (A,B)^{(1,n+1)} \textit{ has a deadlock.}$$
\end{lemma}
\begin{proof}
Given a deadlocked run $x$ of $(A,B)^{(1,n)}$,
we build a deadlocked run of $(A,B)^{(1,n+1)}$. 
If the run $x$ is locally deadlocked,
then it has at least one infinitely moving process, 
thus let the additional process mimic that process. 
If the run $x$ is globally deadlocked run, 
then due to $n>|B|$ in some state there are at least two processes deadlocked. 
Thus, let the new process mimic a process deadlocked in that state---%
the run constructed will also be globally deadlocked.
\end{proof}

\begin{lemma}[Bounding: Disj, Deadlocks, Unfair] \label{lem_disj_deadlocks_unfair}
For disjunctive systems:
\li
  \- with $c=|B|+2$ and any $n>c$:
  $$(A,B)^{(1,c)} \textit{ has a local deadlock} \ \Implied\ (A,B)^{(1,n)} \textit{ has a local deadlock;}$$
  
  \- with $c=2|B| - 1$ and any $n>c$
  $$(A,B)^{(1,c)} \textit{ has a global deadlock} \ \Implied\ (A,B)^{(1,n)} \textit{ has a global deadlock;} $$
  
  \- with $c=2|B|-1$ and any $n>c$:
  $$(A,B)^{(1,c)} \textit{ has a deadlock} \ \Implied\ (A,B)^{(1,n)} \textit{ has a deadlock.}$$
\il
\ak{seems not tight}
\end{lemma}
\begin{proof}[Proof idea]
First, consider the case of global deadlocks.
The insight is to divide deadlocked local states into two disjoint sets, 
$\dead_1$ and $\dead_2$, as follows.
Given a globally deadlocked run $x$ of \largesys, 
for every $q \in \dead_1$,
there is a process of \largesys deadlocked in $q$ with input $i$,
that has an outgoing transition guarded ``$\exists q$''%
---hence, adding one more process into $q$ would unlock the process.
%\sj{do we always consider inputs correctly? what if $q \in \dead_1$ for some $e$, but $q \in \dead_2$ for $e'$?}\ak{thanks, modified, now it is impossible}
In contrast, $q \in \dead_2$ if any process deadlocked in $q$
stays deadlocked after adding more processes into $q$.
Let us denote the set of $B$-processes deadlocked in $\dead_1$ by $\mD_1$.
Finally, abuse the definition in Eq.~\ref{disj:def_vfin_wrt}
and denote by $\visFin{\mB\smi\mD_1}{x}$ the set of states
that are visited by $B$-processes not in $\mD_1$ before reaching a deadlocked state.

Given a globally deadlocked run $x$ of \largesys with $n\geq 2|B|-1$, 
we construct a globally deadlocked run $y$ of \cutoffsys with $c = 2|B|-1$ as follows.
\li
\- We copy from $x$ into $y$ the local runs of processes in $\mD_1 \cup \{A\}$;
\- flood every state of $\dead_2$;
   and
\- for every $q \in \visFin{\mB\smi\mD_1}{x}$,
   flood $q$ and evacuate into $\dead_2$.
\il
The construction ensures: 
(1) for any moment and any process in $y$,
    the set of local states that are visible to the process includes all the states that were visible 
    to the corresponding process in \largesys whose transitions we copy;
(2) in $y$, there is a moment when all processes deadlock in $\dead_1 \cup \dead_2$.

For the case of local deadlocks, 
the construction is slightly more involved,
since we also need to copy the behaviour of an infinitely moving process.
\end{proof}

\begin{proof}
Given a (globally or locally) deadlocked run of $\largesys$,
we construct (globally or locally) deadlocked run of $\cutoffsys$, 
where $c$ depends on the nature of the given run. 
We do this using the construction template. 

Let $\mB=\{B_1,...,B_n\}$.
The template depends on the set $\mC \subseteq \{B_1,...,B_c\}$
and is as follows.
\li
  \-[a.] Set $y(A)=x(A)$;
  \-[b.] for every $B_i \in \mC$, set $y(B_i)=x(B_i)$;
  \-[c.] for every $q \in \visInf{\mB\smi\mC}{x}$, 
         devote one process of \cutoffsys that floods $q$;
  \-[d.] for every $q \in \visFin{\mB\smi\mC}{x}$, 
         devote one process of \cutoffsys that floods $q$ 
         and then evacuates into $\visInf{\mB\smi\mC}{x}$;
         and
  \-[e.] let other processes (if any) mimic some process from (c).
\il

\myparagraph{1) Local deadlock}
We distinguish three cases: 
\li
  \-[1a)] $A$ deadlocks, $B_1$ moves infinitely often;
  \-[1b)] $A$ moves infinitely often, $B_1$ deadlocks; and
  \-[1c)] $A$ neither deadlocks nor moves infinitely often, 
          $B_1$ deadlocks, $B_2$ moves infinitely often.
\il

\myparagraphraw{1a:} ``$A$ deadlocks, $B_1$ moves infinitely often''. 

Let $c=|B|+1$, and $\mC=\{B_1\}$.
Note that $\visInf{B_2..B_n}{x} \neq \emptyset$. 
The resulting construction uses 
$|\visFin{B_2..B_n}{x}| + |\visInf{B_2..B_n}{x}| + 1 
 \leq 
 |B| + 1$ 
copies of B.
\ak{seems tight}\ak{correctness}

\myparagraphraw{1b:} ``$A$ moves infinitely often, $B_1$ deadlocks''. 

Let $c=|B|+1$, and $\mC=\{B_1\}$.
Let $q_\bot$ be the state in which $B_1$ deadlocks.
Instantiate the construction template.

Process $B_1$ of \cutoffsys is deadlocked in $y$ starting from some moment $d$,
because any state it sees (in $\visInf{A,B_2..B_n}{x}$)
was also seen by $B_1$ in \largesys in $x$ at some moment $d' \geq d$
(note that $d'$ may be not the same moment as $d$).
%\footnote{Note about open systems: here we use the fact from the definitions 
%          that inputs to $B_1$ do not change.
%          This ensures that the set of states that $B_1$ should not see in order
%          to stay deadlocked does not change over time.}

\myparagraphraw{1c:} ``$A$ neither deadlocks nor moves infinitely often, 
                       $B_1$ deadlocks, $B_2$ moves infinitely often''. 

Instantiate the construction template with $c=|B|+2$ and $\mC = \{B_1,B_2\}$.
\ak{seems not tight}\ak{correctness}

\smallskip
Finally, $|B|+2$ is a (possibly not tight) cutoff for local deadlock detection problem.

\myparagraph{2) Global deadlock}
Let $x=(s_1,e_1,p_1)...(s_d,e_d,\bot)$ be a globally deadlocked run of $\largesys$ 
with $n\geq c$.

Let us abuse the definition of $\visInf{\mF}{x}$ and $\visFin{\mF}{x}$,
in Eq.~\ref{disj:def_vinf_wrt} and \ref{disj:def_vfin_wrt} resp., 
and adapt it to the case of finite runs.
To this end, given a finite run $x=(s_1,e_1,p_1)...(s_d,e_d,\bot)$, 
extend it to the infinite sequence $(s_1,e_1,p_1)...(s_d,e_d,\bot)^\omega$, 
and apply the definition of $\visInf{\mF}{x}$ and $\visFin{\mF}{x}$ to the sequence.

Let $\mD_1$ be the set of processes deadlocked in unique states:
$\forall p\in \mD_1 \nexists p' \neq p: s_d(p')=s_d(p)$.
Instantiate the construction template with $\mC = \mD_1$ and $c=2|B|-1$.
\footnote{$2|B|-1$ copies is enough, because: 
          $\visFin{\mB\smi\mC}{x} \cap \visInf{\mB\smi\mC}{x} = \emptyset$,
          $\visInf{\mB\smi\mC}{x} \cap \visInf{\mC}{x} = \emptyset$,
          and if $\visFin{\mB\smi\mC}{x} \neq \emptyset$, 
          then $\visInf{\mB\smi\mC}{x} \neq \emptyset$.}
\ak{seems not tight}

%The construction uses $|dead1| + |dead2| + |\visited_{\fin-P_\bot^1}(x)| \leq 2|B|-1$ copies of B.\ak{seems not tight}\ak{CHECK}\ak{correctness}

\myparagraph{3) Deadlocks}
As the cutoff for the deadlock detection problem we take the largest cutoff in (1)--(2), namely, $2|B|-1$,
but it may be not tight%
---finding the tight cutoffs for local deadlock and for deadlock detection problems is an open problem.

\ak{tried to refine but could not -- the trial is commented out}
\end{proof}

\begin{tightness}[Disj, Deadlocks, Unfair] \label{obs:disj:tight_deadlock}
The cutoff $c=2|B|-1$ for deadlock detection in disjunctive systems is \emph{asymptotically optimal but possibly not tight}.
I.e., for any $k$ there are templates $(A,B)$ with $|B|=k$ such that:
$$
(A,B)^{(1,|B|-1)} \textit{ does not have a deadlock, but } (A,B)^{(1,|B|)} \textit { does}.
$$
\end{tightness}
\begin{proof}
Figure~\ref{fig:obs:disj:tight_deadlock} illustrates templates $(A,B)$ to prove the asymptotic optimality of cutoff $2|B|-1$ for deadlock detection problem. Template $A$ is any that never deadlocks. The system has a local deadlock only when there are at least $|B|$ copies of $B$, which is a constant factor of $2|B|-1$.
\begin{figure}[tb] \centering
\makebox[0.4\textwidth][c]{
\scalebox{0.75}{\begin{tikzpicture}[node distance=2.4cm,inner sep=1pt,minimum size=0.5mm,->,>=latex]

\node[initial left, state] (b_1) {$1_B$};
\node[state] (b_2) [right of=b_1] {$2_B$};
\node (dots) [right of=b_2] {$\ldots$};
\node[state] (b_k) [right of=dots] {$k_B$}; 

\path (b_1) edge [above] node {$\disj{1_B}$} (b_2);
\path (b_2) edge [above] node {$\disj{2_B}$} (dots);
\path (dots) edge [above] node {$\disj{{k-1}_B}$} (b_k);

\path (b_1) [loop above] edge [right] node {} (b_1);
\path (b_2) [loop above] edge [right] node {} (b_2);
% \path (b_k) [loop above] edge [right] node {$\disj{b_k}$} (b_k);

\path (b_1) edge [below, bend right=22] node {} (b_k);
\path (b_2) edge [below, bend right=15] node [near start] {$\disj{{k}_B}$} (b_k);

\end{tikzpicture}

% \begin{tikzLTS}
% \tikzstyle{state} = [circle,draw=black,thick,inner sep=1.5pt, text width={width("$\text{init}_2$")}, align=center]
% \tikzstyle{boxLabel} = [yshift=-0.4cm]
% \tikzstyle{every node} = [node distance=1.75cm]
% 
% \node[state] (t1) {$b_1$};
% \node[state] (t2) [below of=t1] {$b_2$};
% \node[] (t3) [below of=t2] {$$};
% \node[] (t4) [below of=t3, node distance=1cm] {$$};
% \node[state] (t5) [below of=t4] {$b_k$};
% 
% \path (t1) [post,loop right] edge node [right] (l1) {} (t1);
% \path (t1) [post] edge node[right] {$\disj{b_1}$} (t2);
% \path (t2) [post,loop right] edge node [right] (l2) {} (t2);
% \path (t2) [post,dashed] edge node[right] {$\disj{b_2}$} (t3);
% \path (t3) [dotted] edge (t4);
% \path (t4) [post,dashed] edge node[right] {$\disj{b_{k-1}}$} (t5);
% \path (t5) [post,loop right] edge node[right] (l5) {$\disj{b_k}$} (t5);
% 
% \node (U2) [draw=black, fit=(t1) (t2) (t3) (t4) (t5) (l1) (l2) (l5), inner sep=0.75cm] {} ;
% \node [boxLabel] at (U2.north) {Template $B$};
% 
% \end{tikzLTS}
}
}
\caption{Templates for proving Tightness~\ref{obs:disj:tight_deadlock}}
\label{fig:obs:disj:tight_deadlock}
\end{figure}
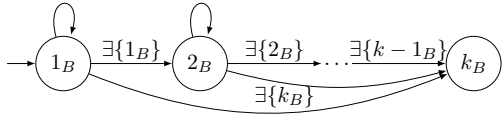
\end{proof}

%--- global deadlocks: fair and unfair:
%C: processes that dead1
%F: processes that dead2
%We copy local runs of dead1.
%We flood deadlocked states of dead2, and flood and evacuate non-deadlocked states of dead2.
%
%--- local deadlocks: unfair:
%I: processes that move infinitely often
%D: processes that dead
%copy local run of one process from I, 
%copy one local run of process from C,
%flood and evacuate finitely visited states by processes except copied
%
%--- local deadlocks: fair:
%I: processes that move infinitely often
%C: processes that dead1
%F: processes that dead2
%Copy local runs of C, 
%flood and evacuate finitely visited states of F\\C, 
%flood dead or infinitely often visited states of F\\C.

\subsection{Deadlocks with Fairness: New Constructions} \label{gua:sec:proofs-disj-deadlock-fair}

\begin{lemma}[Monotonicity: Disj, Deadlocks, Fair] \label{mono_lem_disj_deadlocks_fair}
For disjunctive systems, on strong-fair or finite runs:
$$
\forall n\geq |B|+1: (A,B)^{(1,n)} \textit{ has a deadlock} 
\ \Impl\ 
(A,B)^{(1,n+1)} \textit{ has a deadlock.}
$$
\end{lemma}
\begin{proof}
See proof of Lemma~\ref{mono_lem_disj_deadlocks_unfair}.
\end{proof}

\begin{lemma}[Bounding: Disj, Deadlocks, Fair] \label{le:disj:fair_tight_deadlock}
For disjunctive systems, on strong-fair or finite runs:
\li
  \- with $c=2|B|-1$ and any $n>c$:
  $$(A,B)^{(1,c)} \textit{ has a local deadlock} \ \Implied\ (A,B)^{(1,n)} \textit{ has a local deadlock;}$$
  
  \- with $c=2|B| - 1$ and any $n>c$
  $$(A,B)^{(1,c)} \textit{ has a global deadlock} \ \Implied\ (A,B)^{(1,n)} \textit{ has a global deadlock;} $$
  
  \- with $c=2|B|-1$ and any $n>c$:
  $$(A,B)^{(1,c)} \textit{ has a deadlock} \ \Implied\ (A,B)^{(1,n)} \textit{ has a deadlock.}$$
\il
\end{lemma}
The proofs are similar to that of Lemma~\ref{lem_disj_deadlocks_unfair} (the case without fairness):
the case of global deadlocks is exactly the same,
the case of local deadlocks differ---we additionally use the fair extension to ensure the resulting run is fair.
\begin{proof}

\providecommand{\deadOne}[1]{\dead_{<2}(#1)}
\providecommand{\deadTwo}[1]{\dead_2(#1)}

If $\largesys$ has a global deadlock, 
then the fairness does not influence the cutoff, 
and the proof from Lemma~\ref{lem_disj_deadlocks_unfair}, 
case ``Global Deadlocks'', applies and gives the cutoff $2|B|-1$. 
Hence below consider only the case of local deadlocks. 

Given a strong-fair deadlocked run $x$ of $\largesys$, 
we first construct a strong-fair deadlocked run $y$ of $\cutoffsys$ 
with $c=2|B|$ and then argue that $c$ can be reduced to $2|B|-1$. 
The construction is similar to that in Lemma~\ref{lem_disj_deadlocks_unfair} 
-- the differences originate from the need to infinitely move 
non deadlocked processes.

Let $\deadOne{x}$ be the set of deadlocked states in the run $x$ 
that are only deadlocked if there is no other process in the same state, 
and let $\mD_1$ be the set of processes deadlocked 
in the run $x$ in $\deadOne{x}$. 
Let $\deadTwo{x}$ be the set of states that are deadlocked 
in the run $x$ even if there is another process in the same state. 

We note the following:
\li
\- $|\mD_1| = |\deadOne{x}| \leq |B|$;

\- $\deadOne{x} \cap \deadTwo{x} = \emptyset$;

\- $\visFin{\mB\smi \mD_1}{x} \cap \deadOne{x} \neq \emptyset$
   is possible, because a state from $\visFin{\mB\smi \mD_1}{x}$ 
   can first be visited by a process in $\mB \smi \mD_1$, 
   and later be deadlocked because of the process in $\mD_1$;

\- $\deadTwo{x} \subseteq \visInf{\mB\smi \mD^1}{x}$,
   and hence $\visFin{\mB \smi \mD_1}{x} \cap \deadTwo{x} = \emptyset$.
\il

The construction has two phases. 
The first phase is as follows.
\li
\-[a.] For every $p \in \{A\} \cup \mD_1$, set $y(p)=x(p)$;

\-[b.] for every $q \in \deadTwo{x}$, 
       devote one process of $\cutoffsys$ that floods it;

\-[c.] for every $q \in \visInf{\mB \smi \mD_1}{x} \smi \deadTwo{x}$,
       devote two processes of $\cutoffsys$ that flood it;

\-[d.] for every $q \in \visFin{\mB \smi \mD_1}{x}$, 
       devote one process of $\cutoffsys$ that floods it
       and then evacuates into $\visInf{\mB \smi \mD_1}{x}$;
       and

\-[e.] let other processes (if any) mimic some process from (c).
\il
After this phase all $B$ processes will be in 
$\visInf{\mB \smi \mD_1}{x} \cup \deadOne{x}$. 

The second phase applies to processes in 
$\visInf{\mB \smi \mD_1}{x} \smi \deadTwo{x}$ the fair extension%
\footnote{The fair extension requires the run $x$ to be unconditionally-fair, 
          but here we have a run in which all processes that are not deadlocked
          move infinitely often.
          To adapt the construction to this case:
          copy local runs of processes $\{A\} \cup \mD_1$,
          and do not extend local runs of processes that are in a
          state in $\dead_2$.}.

% The resulting configuration sequence is a run of $\cutoffsys$ by correctness of the flooding, evacuation, fair extension, interleaving and destuttering constructions\ak{prove}. Furthermore, for every $q \in \deadOne$ we have exactly one process in $\cutoffsys$ that eventually stays in $q$, and for every $q \in \deadTwo$ -- at least one such process; they are eventually deadlocked in $q$ because the states that appear infinitely often in the run $x$ of $\cutoffsys$ are the same as in the resulting run of $\largesys$.

How many processes does the construction use? 
Note that the sets 
$\deadOne{x} \cup \visFin{\mB \smi \mD_1}{x}$, 
$\deadTwo{x}$, 
$\visInf{\mB \smi \mD_1}{x} \smi \deadTwo{x}$ 
are disjoint, thus:
\begin{align}
& | \visFin{\mB \smi \mD_1}{x} | + |\deadOne{x}| + |\deadTwo{x}| + 2| \visInf{\mB \smi \mD_1}{x} \smi \deadTwo{x} | \leq \label{disj:eq:1} \\
& 2|\visFin{\mB \smi \mD_1}{x} \cup \deadOne{x}| + |\deadTwo{x}| + 2| \visInf{\mB \smi \mD_1}{x} \smi \deadTwo{x} | \leq \label{disj:eq:2} \\
& |B| + |\visFin{\mB \smi \mD_1}{x} \cup \deadOne{x}| + | \visInf{\mB \smi \mD_1}{x} \smi \deadTwo{x} | \leq 2|B|  \nonumber
\end{align}
Let us reduce the estimate to $\leq 2|B|-1$:
\li
  \- assume that $\deadTwo{x} = \emptyset$ 
     (otherwise, Eq.\ref{disj:eq:1} and the sets disjointness give $2|B|-1$);
     and

  \- assume that $ \visFin{\mB \smi \mD_1}{x} \neq \emptyset$ 
     (the other case together with eq.\ref{disj:eq:2}, 
      the sets disjointness, and the first item gives $2|B|-1$);

  \- hence, the construction in step (d) evacuates the process in 
     $q \in \visFin{\mB \smi \mD_1}{x}$ 
     into 
     $ \visInf{\mB \smi \mD_1}{x} \smi \deadTwo{x}$. 
     Hence modify step (c) of the construction 
     and for $q$ devote a single process of $\cutoffsys$ that floods it. 
     This will give $\leq 2|B|-1$.
\il
This concludes the proof.

\end{proof}

\begin{tightness}[Disj, Deadlocks, Fair]
\label{obs:disj:fair_tight_deadlock}
The cutoff $c=2|B|-1$ for deadlock detection in disjunctive systems on strong-fair or finite runs is tight.
I.e., for any $k$ there are templates $(A,B)$ with $|B|=k$ such that:
$$
(A,B)^{(1,2|B|-2)} \textit{ does not have a deadlock, but } (A,B)^{(1,2|B|-1)} \textit { does}.
$$
\end{tightness}
\begin{proof}
Figure~\ref{gua:fig:tight_disj_dead_fair}
shows process templates $(A,B)$ such that any system $\largesys$ with $n\leq 2|B|-2$ does not deadlock on strong-fair runs, but larger systems do.
\begin{figure}[htpb]
\centering
\begin{subfigure}[b]{0.45\textwidth}\center
\scalebox{0.75}{% !TEX root = table.tex
\begin{tikzpicture}[node distance=2.5cm,inner sep=1pt,minimum size=0.5mm,->,>=latex]
% \tikzset{every state/.style={circle,minimum size=.5cm,inner sep=.03cm}}

\node[initial left, state] (a_1) {};
% \node[state] (a_2) [right of=a_1] {$a_2$};
\node (dots) [right of=a_1] {$\ldots$};
\node[state] (a_k) [right of=dots] {}; 
\node[state] (r) [below=0.8cm of dots] {$r_A$};

\path (a_1) edge [above] node {$\disj{1_B}$} (dots);
\path (dots) edge [above] node {$\disj{{k-1}_B}$} (a_k);
\path (a_1) edge [bend right=15] node {} (r);
\path (a_k) edge [bend left=15] node {} (r);
\path (dots) edge [above, dotted] node {} (r);
\path (a_k) edge [bend right=35,above] node {$\disj{k_B}$} (a_1);

% \path (a_1) [loop above] edge [right] node {$\disj{a_1}$} (a_1);
% \path (a_2) [loop above] edge [right] node {$\disj{a_2}$} (a_2);
% \path (a_k) [loop above] edge [right] node {$\disj{a_k}$} (a_k);
\path (r) [loop below] edge [right] node {} (r);

\end{tikzpicture}}
\label{fig:disj:tight_fair_deadlock_tmplA}
\caption*{Template A}
\end{subfigure}
\begin{subfigure}[b]{0.45\textwidth}\center
\scalebox{0.75}{% !TEX root = table.tex
\begin{tikzpicture}[node distance=2.5cm,inner sep=1pt,minimum size=0.5mm,->,>=latex]
% \tikzset{every state/.style={circle,minimum size=.5cm,inner sep=.03cm}}

\node[state,initial left] (b_1) {$1_B$};
\node (dots) [right of=b_1] {$\ldots$};
\node[state] (b_k) [right of=dots] {$k_B$}; 
\node[state] (r) [below=0.8cm of dots] {};

\path (b_1) edge [above] node {$\disj{1_B}$} (dots);
\path (dots) edge [above] node {$\disj{{k-1}_B}$} (b_k);
\path (b_1) edge [bend right=15] node {} (r);
\path (b_k) edge [bend left=15] node {} (r);
\path (dots) edge [above, dotted] node {} (r);

\path (b_1) [loop above] edge [right] node {$\disj{1_B}$} (b_1);
\path (b_k) [loop above] edge [right] node {$\disj{k_B}$} (b_k);
\path (r) [loop below] edge [right] node {$\disj{r_A}$} (r);
\path (dots) [loop above,dotted,in=60,out=120,distance=12mm] edge [right] node {} (dots);

\end{tikzpicture}}
\caption*{Template B}
\end{subfigure}
\caption{Templates $(A,B)$ used in Tightness~\ref{obs:disj:fair_tight_deadlock}.}
\label{gua:fig:tight_disj_dead_fair}
\end{figure}
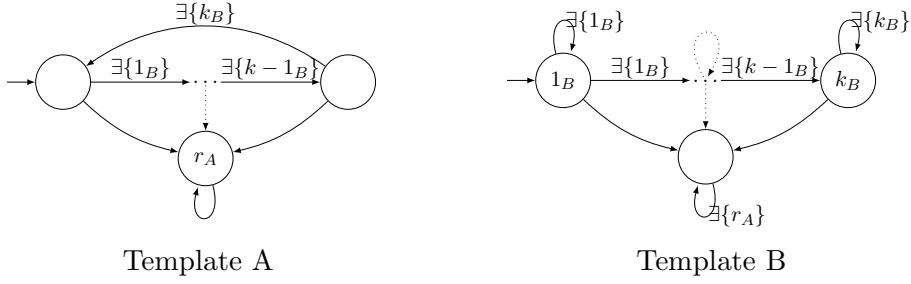
\end{proof}

\section{Proof Techniques for Conjunctive Systems} \label{gua:sec:proofs-conj}

\subsection{\LTLmX\ Properties Without Fairness: Existing Constructions} \label{gua:sec:ideas-conj-nofair}

The Monotonicity Lemma is proven~\cite{Emerson00} by keeping the additional process in the initial state.

\begin{lemma}[Monotonicity: Conj, \LTLmX, Unfair] \label{le:ConjMonotonicityLemma}
For conjunctive systems,
\begin{align*}
\forall n \geq 1:\ 
(A,B)^{(1,n)} \models \pexists h(A,B_1)
\ \ \Impl \ \ 
(A,B)^{(1,n+1)}\models \pexists h(A,B_1).
\end{align*}
\end{lemma}
\begin{proof}
Let the new process stutter in $\init$ state.
\end{proof}

To prove the Bounding Lemma,
Emerson and Kahlon \cite{Emerson00} suggest
to simply copy the local runs $x(A)$ and $x(B_1)$ into $y$. 
In addition,
we may need one more process that moves infinitely often
to ensure that an infinite run of \largesys will result in an infinite run of \cutoffsys.
All transitions of copied processes will be enabled
because removing processes from a conjunctive system
cannot disable a transition that was enabled before.

\begin{lemma}[Bounding: Conj, \LTLmX, Unfair] \label{le:ConjunctiveBoundingLemma}
For conjunctive systems,
\begin{align*}
\forall n \geq 2:\ 
(A,B)^{(1,2)} \models \pexists h(A,B_1)
\ \ \Implied \ \ 
\largesys \models \pexists h(A,B_1).
\end{align*}
\end{lemma}
The proof is inspired by the first part of the proof of \cite[Lemma 5.2]{Emerson00}.
\begin{proof}
Let $x=(s_1,e_1,p_1) (s_2,e_2,p_2) \ldots$ be a run of $\largesys$. 
Note that, by the semantics of conjunctive guards, 
the transitions along any local run of $x$ will also be enabled 
in any system $\cutoffsys$ with $c \leq n$, 
where the processes exhibit a subset of the local runs of $x$. 
Thus, we obtain a run of $\cutoffsys$ by copying a subset of the local runs of $x$, 
and removing elements of the new global run where all processes stutter.
\sj{should we put this as a general lemma somewhere?}

Then, based on an infinite run $x$ of the original system, 
we construct an infinite run $y$ of the cutoff system. 
Let $y(A)=x(A)$ and $y(B_1)=x(B_1)$. 
The second copy of template $B$ in $(A,B)^{(1,2)}$ is needed to ensure that 
the run $y$ is infinite, i.e., at least one process moves infinitely often. 
If both $x(A)$ and $x(B_1)$ eventually deadlock, 
then there exists a process $B_i$ of $\largesys$ that makes infinitely many moves, 
and we set $y(B_2) = x(B_i)$. 
Otherwise, we set $y(B_2) = x(B_2)$.
\end{proof}

\begin{tightness}[Conj, \LTLmX, Unfair] \label{obs:conj:tight_prop}
The cutoff $c=2$ is tight for parameterized model checking of properties $\pexists h(A,B_1)$ in the 1-conjunctive systems, i.e., there is a system type $(A,B)$ and property $Eh(A,B_1)$ which is not satisfied by $(A,B)^{(1,1)}$ but is by $(A,B)^{(1,2)}$.
\end{tightness}
\begin{proof}
Figure~\ref{fig:obs:conj:tight_prop} shows templates $(A,B)$, $\pexists h(A,B_1) = \pexists \eventually b$. An infinite run that satisfies the formula needs one copy of $B$ that stays in the initial state, and one that moves into $b$.
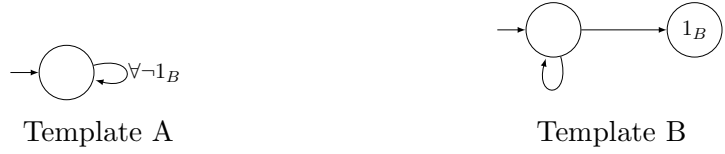
\begin{figure}[tb]
\centering
\begin{subfigure}[b]{0.45\textwidth}\center
\scalebox{0.75}{% !TEX root = table.tex
\begin{tikzpicture}[node distance=1.5cm,inner sep=1pt,minimum size=0.5mm,->,>=latex]

\node[state,initial left] (init) {};
\path (init) [loop right] edge [right] node {$\forall \neg 1_B$} (init);

\end{tikzpicture}
  }
\caption*{Template A}
\end{subfigure}
\begin{subfigure}[b]{0.45\textwidth}\center
\centering
\scalebox{0.75}{% !TEX root = table.tex
\begin{tikzpicture}[node distance=1.5cm,inner sep=1pt,minimum size=0.5mm,->,>=latex]

\node[state,initial left] (init) {};
\node[state] (b_1) [right= of init] {$1_B$};

\path (init) edge [above] node {} (b_1);
\path (init) [loop below] edge [right] node {} (init);

\end{tikzpicture}
  }
\caption*{Template B}
\end{subfigure}
\caption{Templates used to prove Tightness~\ref{obs:conj:tight_prop}}
\label{fig:obs:conj:tight_prop}
\end{figure}
\end{proof}

\subsection{\LTLmX\ Properties with Fairness: New Constructions} \label{gua:sec:ideas-conj-fair}

In this section, subscript $i$ in path quantifiers, $\pexists_i$ and $\pforall_i$, 
denotes the quantification over initializing runs.

The proof of the Bounding Lemma is the same as in the non-fair case,
noting that, if the original run is unconditional-fair,
then so will be the resulting run.

\begin{lemma}[Bounding: Conj, \LTLmX, Fair] \label{le:FairConjunctiveBounding Lemma}
For unconditionally-fair initializing runs of conjunctive systems:
\begin{align*}
&\forall n \geq 1:\\
& (A,B)^{(1,1)} \models \pexists_{uncond} h(A,B_1)
\ \Implied \
(A,B)^{(1,n)} \models \pexists_{uncond} h(A,B_1).
\end{align*}
\end{lemma}
\begin{proof}
Given an unconditionally-fair [initializing] run $x$ of $\largesys$ with $n>c$ construct an unconditionally-fair [initializing] run $y$ in the cutoff system $(A,B)^{(1,1)}$: copy the local runs of processes $A$, $B_1$.
%  and copy the behaviour of a process of $\largesys$ that moves infinitely often in the run $x$. Since $y$ is the result of removal of a number of local runs from $x$, it is an unconditionally-fair initializing run of $\cutoffsys$.
\end{proof}

Proving the Monotonicity Lemma is more difficult,
since the fair extension construction from disjunctive 
systems does not work for conjunctive systems%
---if an additional process mimics the transitions of an existing process
   then it disables transitions of the form 
   $\transition{q}{q'}{\textit{``\,}\forall\neg q\textit{\!''}}$ or
   $\transition{q}{q'}{\textit{``\,}\forall\neg q'\textit{\!''}}$.
Hence, we add the restriction of initializing runs, 
which allows us to construct a fair run as follows.
The additional process $B_{n+1}$ ``shares'' a local run $x(B_i)$ 
with an existing process $B_i$ of $(A,B)^{(1,n+1)}$: 
one process stutters in $\init_B$ while the other makes transitions from $x(B_i)$, 
and whenever $x(B_i)$ enters $\init_B$
(this happens infinitely often),
the roles are reversed. 
Since this changes the behavior of $B_i$, 
$B_i$ should not be mentioned in the formula, 
i.e., we need $n\geq 2$ for a formula $h(A,B^{(1)})$. 

\begin{lemma}[Monotonicity: Conj, \LTLmX, Fair] \label{le:ConjMonFair}
For unconditionally-fair initializing runs of conjunctive systems:\sj{generalization is obvious; $n \ge k+1$ in general case}
\begin{align*}
& \forall n \ge 2:\\
& (A,B)^{(1,n)} \models \pexists_{uncond,i} h(A,B_1)
\ \Impl \
(A,B)^{(1,n+1)} \models \pexists_{uncond,i} h(A,B_1).
\end{align*}
\end{lemma}
\begin{proof}
Given a unconditionally-fair initializing run $x$ of $\largesys$, we construct a unconditionally-fair initializing run $y$ in $(A,B)^{(1,n+1)}$, with one additional process $p$. 
First, copy all local runs of all processes of $(A,B)^{(1,n)}$ from the run $x$ into $y$.
Then, let process $p'$ stutter in $\init$ until some other process $p \neq B_1$ enters $\initstate$. 
Then, exchange the roles of processes $p'$ and $p$: let $p$ stutter in $\initstate$, while $p'$ takes the transitions of $p$ from the original run, until it enters $\initstate$. And so on.
In this way, we continue to interleave the run between $p'$ and $p$, and obtain a unconditionally-fair initializing run for all processes, with $y(A,B_1)=x(A,B_1)$. 
Thus, if $\largesys \models \pexists h(A,B_1)$, then $(A,B)^{(1,n+1)} \models \pexists h(A,B_1)$.
\end{proof}

\begin{tightness}[1-Conj, \LTLmX, Fair] \label{obs:conj:tight_prop_fair}
The cutoff $c=2$ is tight for parameterized model checking of $\pexists h(A,B_1)$ 
on unconditionally-fair initializing runs in 1-conjunctive systems, 
i.e., 
there is a system type $(A,B)$ and property $\pexists h(A,B_1)$ 
which is satisfied by $(A,B)^{(1,1)}$ but not by $(A,B)^{(1,2)}$.
\end{tightness}
\begin{proof}
Figure~\ref{fig:obs:conj:tight_prop_fair} shows templates $(A,B)$;
$\pexists h(A,B_1) = \pexists \FG (b_{init} \impl a_1)$.
\begin{figure}[tb]
\centering
\begin{subfigure}[b]{0.45\textwidth}\center
\scalebox{0.75}{% !TEX root = table.tex
\begin{tikzpicture}[node distance=1.5cm,inner sep=1pt,minimum size=0.5mm,->,>=latex]

\node[state,initial left] (init) {${init}_A$};
\node[state] (a_1) [right= of init] {$1_A$};

\path (init) edge [above] node {} (a_1);
\path (a_1)  [bend left=20] edge [below] node {} (init);

\end{tikzpicture}
  }
\caption*{Template A}
\end{subfigure}
\begin{subfigure}[b]{0.45\textwidth}\center
\scalebox{0.75}{% !TEX root = table.tex
\begin{tikzpicture}[node distance=1.5cm,inner sep=1pt,minimum size=0.5mm,->,>=latex]

\node[state,initial left] (init) {${init}_B$};
\node[state] (b_1) [right= of init] {$1_B$};
\node[state] (b_2) [right= of b_1] {$2_B$}; 

\path (init) edge [above] node {$\forall\neg 1_B$} (b_1);
\path (b_1)  edge [above] node {$\forall\neg 1_A$} (b_2);
\path (b_2)  [bend left=20] edge [below] node {$\forall\neg 2_B$} (init);

\end{tikzpicture}
  }
\caption*{Template B}
\end{subfigure}
\caption{Templates used to prove Tightness~\ref{obs:conj:tight_prop_fair}}
\label{fig:obs:conj:tight_prop_fair}
\end{figure}
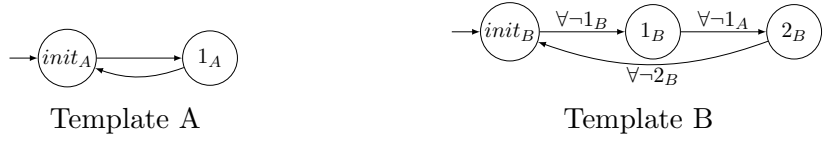
\end{proof}

\subsection{Deadlocks Without Fairness: Updated Constructions}\label{gua:sec:proofs-conj-deadlock-unfair}

\begin{lemma}[Monotonicity: Conj, Deadlocks, Unfair] \label{le:ConjunctiveMonotonicityLemmaDeadlocks}
For conjunctive systems:
$$
\forall n\geq 1: (A,B)^{(1,n)} \textit{ has a deadlock} 
\ \Impl\ 
(A,B)^{(1,n+1)} \textit{ has a deadlock.}
$$
\end{lemma}
\begin{proof}
Given a deadlocked run $x$ of $(A,B)^{(1,n)}$, we construct a deadlocked run of $(A,B)^{(1,n+1)}$.
Let $y$ copy run $x$, and keep the new process in $\init$.
If $x$ is globally deadlocked and $d$ is the moment when the deadlock happens in $x$,
then schedule the new process arbitrarily after moment $d$.
Thus, it is possible that the newly constructed system run is only locally deadlocked,
while the original run is globally deadlocked.
\end{proof}

As for the Bounding Lemma,
in the case of global deadlock detection, 
Emerson and Kahlon~\cite{Emerson00} suggest to copy a subset of the original local runs.
For every local state $q$ that is present in the final state of the run, 
we need at most two local runs that end in this state. 
In the case of local deadlocks, 
our construction uses the fact that systems are 1-conjunctive.
In 1-conjunctive systems, if a process is deadlocked, 
then there is a set of states $DeadGuards$ that all need to be populated by other processes
in order to disable all transitions of the deadlocked process. 
Thus, the construction copies: 
(i) the local run of a deadlocked process, 
(ii) for each $q \in DeadGuards$, the local run of a process 
     that is in $q$ at the moment of the deadlock, and
(iii) the local run of an infinitely moving process.

\begin{lemma}[Bounding: 1-Conj, Deadlocks, Unfair] \label{le:ConjunctiveBoundingLemmaDeadlocks}
For 1-conjunctive systems:
\li
  \- with $c=2|Q_B\smi \{ \init \}|$ and any $n>c$ \footnote{This statement also applies to systems without restriction to $1$-conjunctive guards.}
  $$(A,B)^{(1,c)} \textit{ has a global deadlock} \ \Implied\ (A,B)^{(1,n)} \textit{ has a global deadlock;} $$
  
  \- with $c=|Q_B\smi \{ \init \}|+2$ and any $n>c$:
  $$(A,B)^{(1,c)} \textit{ has a local deadlock} \ \Implied\ (A,B)^{(1,n)} \textit{ has a local deadlock;}$$
  
  \- with $c=2|Q_B\smi \{ \init \}|$ and any $n>c$:
  $$(A,B)^{(1,c)} \textit{ has a deadlock} \ \Implied\ (A,B)^{(1,n)} \textit{ has a deadlock.}$$
\il
\end{lemma}
\begin{proof}
The proof is inspired by the second part of the proof of \cite[Lemma 5.2]{Emerson00}, 
but in addition to global we consider local deadlocks. 

\myparagraph{Global deadlocks} 
Let $c=2|Q_B\smi\{ \init \}|)$. 
Let run $x = (s_1,e_1,p_1)\ldots(s_d,e_d,\bot)$ of \largesys 
with $n>c$ be globally deadlocked. 
We construct a globally deadlocked run $y$ in $\cutoffsys$ as follows.
\li
  \-[a.] For every $q \in s_d \setminus \{\init\}$:
  \li
    \- if $s_d$ has two processes in state $q$, 
       then devote two processes of \cutoffsys that mimic the behaviour 
       of the two of \largesys correspondingly;

    \- otherwise, $s_d$ has only one process in state $q$, 
       then devote one process of \cutoffsys that mimics the process of \largesys;
  \il
  \-[b.] for every process of \cutoffsys not used in the construction (if any): 
         let it mimic an arbitrary $B$-process of \largesys
         that was not yet used in the construction in item (a) nor (b).
\il
The construction uses $\leq 2|Q_B\smi \{ \init \}|$ processes $B$.
Note that the proof does not assume that the system is 1-conjunctive.

\myparagraph{Local deadlocks} 
Let $c = |Q_B\smi \{ \init \}|+2$. 
Let run $x = (s_1,e_1,p_1)\ldots$ of \largesys with $n>c$ be locally deadlocked. 
We will construct a run $y$ of \cutoffsys 
where at least one process deadlocks and exactly one process moves infinitely often.

Wlog. we distinguish three cases:
\li
\-[1.] $A$ moves infinitely often in $x$, and $B_1$ deadlocks;
\-[2.] $A$ deadlocks, and $B_1$ moves infinitely often; and
\-[3.] $A$ neither deadlocks nor moves infinitely often, $B_1$ deadlocks, 
       $B_2$ moves infinitely often.
\il

\myparagraph{1} ``$A$ moves infinitely often in $x$, and $B_1$ deadlocks''.

Let $q_\bot, e_\bot$ be the deadlocked state and input of $B_1$ in $x$, 
and let $d$ be the moment from which $B_1$ is deadlocked.

Let $DeadGuards=\{q_1,\ldots,q_k\}$ be the set of states
such that for every $q_i \in DeadGuards$ there is an outgoing transitions 
from $q_\bot$ with $e_\bot$ guarded ``$\forall \neg q_i$'',
and assume $DeadGuards \neq \emptyset$
(if it is empty, then we keep every process in $\init$ 
 until someone reaches $q_\bot$ and then schedule the rest arbitrarily). 
(Recall that $q_i \in Q_B \cupdot Q_A$.)

The construction is as follows.
\li
  \-[a.] $y(A)=x(A)$, $y(B_1)=x(B_1)$.
  \-[b.] For each $q \in DeadGuards$, at moment $d$ in $x$
         there is a process $p_q$ in state $q$. 
         If $p_q \in \{B_1,...,B_n\}$, 
         then let one process of \cutoffsys mimic it till moment $d$, 
         and then stutter in $q$.
  \-[c.] Let other processes of \cutoffsys (if any) stay in $\init$.
\il
The construction uses (if ignore (c)) $\leq |Q_B\smi \{ \init \}|+1$ processes $B$.

Note: 
the assumption of 1-conjunctive systems implies that,
in order to deadlock $B_1$,
we need a process in each state in $BlockGuards$.
This implies that having a process in each state of $BlockGuards$ does not disable 
any $A$'s transition after moment $d$.

\myparagraph{2} ``$A$ deadlocks, and $B_1$ moves infinitely often'': 
use the construction from (1).

\myparagraph{3} 
``$A$ neither deadlocks nor moves infinitely often, 
  $B_1$ deadlocks, $B_2$ moves infinitely often''. 
Use the construction from (1), and additionally: $y(B_2)=x(B_2)$. 
Thus, the construction uses (if ignore (c)) $\leq |Q_B \smi \{ \init \}|+2$ 
processes $B$.

\myparagraph{Deadlocks}
Take the higher value among the cases considered above $c=2|Q_B\smi \{ \init \}|$: 
if $x$ is locally deadlocked then the Monotonicity Lemma ensures 
that there is a deadlocked run in \cutoffsys.
\end{proof}

\begin{tightness}[1-Conj, Deadlocks, Unfair] \label{obs:conj:tight_deadlock}
The cutoff $c=2|B|-2$ is tight for parameterized deadlock detection in the 1-conjunctive systems, i.e., for any $k$ there is a system type $(A,B)$ with $|B|=k$ such that there is a deadlock in $(A,B)^{(1,2|B|-2)}$, but not in $(A,B)^{(1,2|B|-3)}$. 
\end{tightness}
\begin{proof} 
Figure~\ref{fig:obs:conj:tight_deadlock} provides templates $(A,B)$ that proves the observation. In the figure the edge with $\forall{\neg b_1},\ldots,\forall{\neg b_k}$ denotes edges with guards $\forall{\neg b_1},\ldots,\forall{\neg b_k}$. To get the global deadlock we need at least two processes in each $b_i \in \{b_1,\ldots,b_k\}$. Note that the system does not have local deadlocks.\ak{show that cutoffs for local deadlocks are also tight}
\begin{figure}[h]
\vspace{-10pt}
\centering
\begin{subfigure}[b]{0.45\textwidth}\center
\scalebox{0.75}{% !TEX root = table.tex
\begin{tikzpicture}[node distance=2.3cm,inner sep=1pt,minimum size=0.5mm,->,>=latex]

\node[initial left, state] (init) {};

\path (init) [loop right] edge [right] node {$\forall\neg 1_B$} (init);
\path (init) [loop right,dotted,distance=26mm] edge [right] node {...} (init);
\path (init) [loop right,distance=38mm] edge [right] node {$\forall\neg k_B$} (init);

\end{tikzpicture}
  }
\caption*{Template A}
\end{subfigure}
\hspace{1cm}
\begin{subfigure}[b]{0.45\textwidth}\center
\scalebox{0.75}{% !TEX root = table.tex
\begin{tikzpicture}[node distance=2.4cm,inner sep=1pt,minimum size=0.5mm,->,>=latex]

\node[initial above, state] (init) {$init$};
\node[state] (b_1) [left=2.7cm of init] {$1_B$};
\node (dots) [below=1cm of init] {$\ldots$};
\node[state] (b_k) [right=2.7cm of init] {$k_B$}; 

% \path (b_1.20) edge [above] node {\specialcell{$\forall{\neg b_1}$\\...\\$\forall{\neg b_k}$}} (init.160);
\path (b_1.20) edge [above] node {$\forall{\neg 1_B},...,\forall{\neg k_B}$} (init.160);
\path (init.200) edge [above] node {} (b_1.340);

\path (b_k.160) edge [above] node {$\forall{\neg 1_B},...,\forall{\neg k_B}$} (init.20); 
\path (init.340) edge [above] node {} (b_k.200); 

\path (init.282) edge [left, dotted] node {} ($(dots)+(0.1,0.1)$);
\path ($(dots)-(0.1,-0.1)$) edge [left, dotted] node {$\forall{\neg 1_B},...,\forall{\neg k_B}$} (init.257);

% \path (b_1) [loop left] edge [left] node {$\forall\neg b_1$} (b_1);
% \path (b_1) [loop left,dotted,distance=26mm] edge [left] node {...} (b_1);
% \path (b_1) [loop left,distance=38mm] edge [left] node {$\forall\neg b_k$} (b_1);

% \path (b_k) [loop right] edge [right] node {$\forall\neg b_1$} (b_k);
% \path (b_k) [loop right,dotted,distance=26mm] edge [right] node {...} (b_k);
% \path (b_k) [loop right,distance=38mm] edge [right] node {$\forall\neg b_k$} (b_k);

\end{tikzpicture}
  }
\caption*{Template B}
\end{subfigure}
\caption{Templates used to prove Tightness~\ref{obs:conj:tight_deadlock}}
\label{fig:obs:conj:tight_deadlock}
\end{figure}
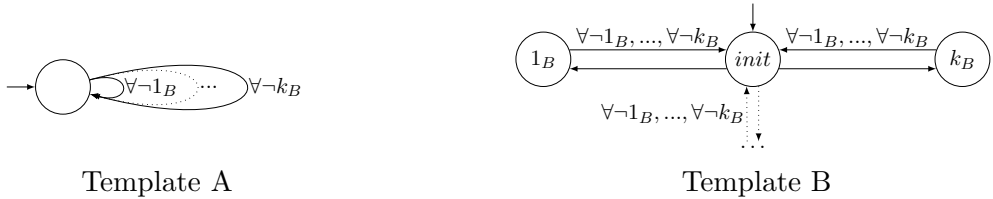
\end{proof}

\subsection{Deadlocks with Fairness: New Constructions} \label{gua:sec:proofs-conj-deadlock-fair}

The Monotonicity Lemma is proven by keeping process $B_{n+1}$ in the initial state, 
and copying the runs of deadlocked processes.
If the run of \largesys is globally deadlocked,
then process $B_{n+1}$ may keep moving in the constructed run,
i.e., the run may be only locally deadlocked. 
In the case of a local deadlock in \largesys, we distinguish two cases: 
there is an infinitely moving $B$-process, or all $B$-processes are deadlocked 
(and thus $A$ moves infinitely often).
In the latter case, we use the same construction as in the global deadlock case
(the correctness argument uses the fact that systems are 1-conjunctive, 
 runs are initializing, and there is only one process of type $A$).
In the former case, we copy the original run, and let $B_{n+1}$ share
a local run with an infinitely moving $B$-process.

\begin{lemma}[Monotonicity: Conj, Deadlocks, Fair] \label{le:FairConjunctiveMonotonicityLemmaDeadlocks}
For 1-conjunctive systems on strong fair initializing or finite runs:
$$
\forall n\geq 1: (A,B)^{(1,n)} \textit{ has a deadlock}
\ \Impl\ 
(A,B)^{(1,n+1)} \textit{ has a deadlock.}
$$
\end{lemma}
\begin{proof}\ak{check the minimal value of $n$ (1 or 2?)}
Let $x$ be a globally deadlocked or locally deadlocked strong-fair initializing run of $(A,B)^{(1,n)}$.
We will build a globally deadlocked or locally deadlocked strong-fair initializing run 
of $(A,B)^{(1,n+1)}$.

If $x$ is finite, then $y$ is the copy of $x$, and the new process stays in $\init_B$
until every process becomes deadlocked, and then is scheduled arbitrarily.
Note that $y$ constructed this way may be locally deadlocked 
rather than globally deadlocked as $x$ is.

Now consider the case when $x$ is locally deadlocked strong-fair initializing.

Let $\mD$ be the set of deadlocked $B$-processes in $x$, and $d$ be the moment 
when the processes become deadlocked.

Consider the case $\visInf{\mB\smi\mD}{x} \neq \emptyset$:
copy $x$ into $y$, and let the new process $B_{n+1}$ wait in $\init_B$ 
and interleave the roles with a process $B$ that moves infinitely often in $x$, 
as described in the proof of Lemma~\ref{le:ConjMonFair}.

Consider the case $\visInf{\mB\smi\mD}{x} = \emptyset$:
every $B$ process of $(A,B)^{(1,n)}$ is deadlocked and thus $\mD = \mB$.
Define 
$$
DeadGuards\! =\! \big\{q \| \exists B_i \in \mD
                      \textit{ with a transition guarded ``\,}
                      {\forall \neg q} 
                      \textit{\!'' in } (s_d(B_i),e_d(B_i))\big\}.
$$
Note that $Q_A \cap DeadGuards = \emptyset$, because $A$ visits infinitely often $\init_A$
and we consider 1-conjunctive systems.
Hence, copy $x$ into $y$, and let the new process $B_{n+1}$ wait in $\init_B$ 
until every process $B_1,...,B_n$ become deadlocked, and then schedule $B_{n+1}$ arbitrarily.
%
% See the proof of Lemma~\ref{le:ConjunctiveMonotonicityLemmaDeadlocks}.
%%% AK: this won't work because the process B_{n+1} in init should move inf often or deadlock,
%%% the arbitrary scheduling can lead to unlocking of every one
\end{proof}

As for the Bounding Lemma,
we use a construction that is similar to that of properties under fairness for disjunctive systems (Sect.~\ref{gua:sec:ideas-disj-fair}):
  in the setup phase, 
  we populate some ``safe'' set of states with processes,
  and then we extend the runs of non-deadlocked processes 
  to satisfy strong fairness, 
  while ensuring that deadlocked processes never get enabled.

\begin{lemma}[Bounding: 1-Conj, Deadlocks, Fair] \label{le:FairConjunctiveBoundingLemmaDeadlocks}
For 1-conjunctive systems on strong-fair initializing or finite runs:
\ak{no real need for initializing -- but easier to explain}
\li
  \- with $c=2|Q_B\smi \{ \init \}|$ and any $n>c$:
  $$
  \cutoffsys \textit{ has a global deadlock} 
  \ \Implied\ 
  \largesys \textit{ has a global deadlock;}
  $$

  \- with $c=2|Q_B\smi \{ \init \}|+1$ and any $n>c$ (when $|Q_B|>2$):
  $$
  \cutoffsys \textit{ has a local deadlock} 
  \ \Implied\ 
  \largesys \textit{ has a local deadlock;}
  $$

  \- with $c=2|Q_B\smi \{ \init \}|$ and any $n>c$:
  $$
  \cutoffsys \textit{ has a deadlock} 
  \ \Implied\ 
  \largesys \textit{ has a deadlock.}
  $$
\il
\end{lemma}
\begin{proof}
\providecommand{\deadOne}{\dead_1}
\providecommand{\deadTwo}{\dead_2}

\myparagraph{Global deadlocks}
$c=2|Q_B \smi \{\init_B\}|$, 
see Lemma~\ref{le:ConjunctiveBoundingLemmaDeadlocks}, 
the fairness does not matter on finite runs.

\myparagraph{Local deadlocks}
Let $c=2|Q_B\smi \{ \init_B \}|$. 
Let $x= (s_1,e_1,p_1)\ldots$ be a locally deadlocked strong-fair intitializing run 
of $\largesys$ with $n>c$. 
We construct a locally deadlocked strong-fair initializing run $y$ of $\cutoffsys$.

Let $\mD$ be the set of deadlocked processes in $x$. 
Let $d$ be the moment in $x$ starting from which every process in $\mD$ is deadlocked.

Let $\dead(x)$ be the set of states in which processes $\mD$ of \largesys
are deadlocked.

Let $\deadTwo(x) \subseteq \dead(x)$ be the set of deadlocked states such that: 
for every $q \in \deadTwo(x)$, 
there is a process $P \in \mD$ with $s_d(P) = q$ 
and that for input $e_{\geq d}(P)$ has a transition guarded with ``$\forall \neg q$''.
Thus, a process in $q$ is deadlocked with $e_d(P)$
only if there is another process in $q$ in every moment $\geq d$.

Let $\deadOne(x) = \dead(x)\smi\deadTwo(x)$.
I.e., 
for any $q \in \deadOne(x)$, there is a process $P$ of \largesys 
which is deadlocked in $s_d(P) = q$ with input $e_d(P)$,
and no transitions from $q$ with input $e_d(P)$ are guarded with ``$\forall \neg q$''.

Define
$$
DeadGuards\!=\!
\big\{ q \| \exists B_i \in \mD
         \textit{ with a transition guarded ``\,}
         {\forall \neg q} 
         \textit{\!'' in } (s_d(B_i),e_d(B_i)) \big\}.
$$
Figure~\ref{fig:conj-deadlocks-venn} illustrates properties of sets 
$DeadGuards$, $\deadOne$, $\deadTwo$, $\visInf{\mB\smi\mD}{x}$.
\ak{check how $A$'s states affect all those sets, currently i assumed that they are all subsets of $Q_B$}

\begin{figure}[hptb]
\begin{mdframed}
\centering
\includegraphics[width=0.7\textwidth]{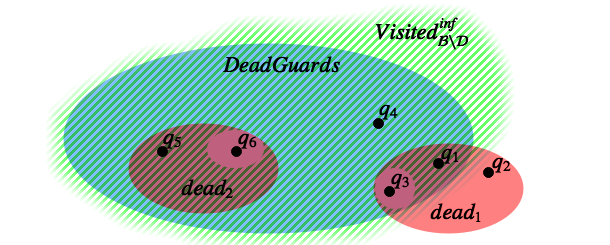}
\captionsetup{singlelinecheck=off}
\caption[fig:conj-deadlocks-venn]{%
Venn diagram for sets $DeadGuards$, $\deadOne$, $\deadTwo$, $\visInf{\mB\smi\mD}{x}$:
\begin{itemize}
\item[($q_1$)] $\deadOne \cap DeadGuards \cap \visInf{\mB\smi\mD}{x} \neq \emptyset$ is possible:
               in $x$, 
               there is a process deadlocked in state $q_1$,
               there is a non-deadlocked process that visits $q_1$ infinitely often,
               and there is a process deadlocked in a state $q \neq q_1$ 
               with a transition guarded ``$\forall \neg q_1$'' 

\item[($q_3$)] $\deadOne \cap DeadGuards \smi \visInf{\mB\smi\mD}{x} \neq \emptyset$ is possible:
               similarly to $q_1$, 
               except that no non-deadlocked processes visit $q_3$ infinitely often

\item[($q_2$)] $\deadOne \smi (\visInf{\mB\smi\mD}{x} \cup DeadGuards) \neq \emptyset$ is possible:
               in $x$, 
               there is a process deadlocked in state $q_2$,
               no other processes visit $q_2$ infinitely often,
               and no processes are deadlocked with a transition guarded ``$\forall \neg q_2$''

\item[($q_4$)] $DeadGuards \smi \dead \neq \emptyset$ is possible:
               there is a process deadlocked in a state $q \neq q_4$ 
               with a transition guarded  ``$\forall \neg q_4$''

\item[($q_5$)] $\deadTwo \cap \visInf{\mB\smi\mD}{x} \cap DeadGuards \neq \emptyset$ is possible:
               there is at least one process deadlocked in $q_5$ with a transition guarded ``$\forall \neg q_5$'',
               and some non-deadlocked process visits $q_5$ infinitely often
               (this process does not deadlock in $q_5$, 
                because in $q_5$ it receives an input different from that of the deadlocked processes)

\item[($q_6$)] $\deadTwo \cap DeadGuards \smi \visInf{\mB\smi\mD}{x} \neq \emptyset$ is possible:
               similarly to $q_5$, except no non-deadlocked processes visit $q_6$ infinitely often
\end{itemize}
}
\label{fig:conj-deadlocks-venn}
\end{mdframed}
\end{figure}

Let us assume $DeadGuards \neq \emptyset$---the other case is straightforward.\ak{check}

The construction has two phases, the setup and the looping phase.

In the {\bf setup phase}, we copy from $x$ into $y$:
\li
\-[a.] $y(A) = x(A)$;

\-[b.] for every $q \in \deadOne$: 
   devote one process of \cutoffsys that copies 
   a process of \largesys deadlocked in $q$;

\-[c.] for every $q \in \deadTwo \setminus \visInf{\mB\smi\mD}{x}$: 
   devote two processes of \cutoffsys that copy 
   the behaviour of two processes of \largesys that deadlock in $q$;

\-[d.] for every $q \in \deadTwo \cap \visInf{\mB\smi\mD}{x}$:
   in $x$, 
   there is a process, $B_q^\inf \in \mB\smi\mD$, that visits $q$ infinitely often,
   and there is a process, $B_q^\bot \in \deadTwo$, deadlocked in $q$.
   Then:
\li
   \-[1.] devote one process of \cutoffsys that copies the behaviour of $B_q^\bot$, and
   \-[2.] devote one process of \cutoffsys that copies the behaviour of $B_q^\inf$ 
          until it reaches $q$ at a moment after $d$,
          and then provide the same input as $B_q^\bot$ receives at moment $d$.
          This will deadlock the process;
\il

\-[e.] for every $q \in DeadGuards \setminus \dead$:
       note that $q \in \visInf{\mB\smi\mD}{x}$ and, thus, there is a process, 
       $B_q^\inf \in \mB\smi\mD$, 
       that visits $q$ infinitely often.
       Devote one process of \cutoffsys that copies the behaviour of $B_q^\inf$ 
       until it reaches $q$ at a moment after $d$;

\-[f.] if $DeadGuards \setminus \dead \neq \emptyset$ 
       or $A \in \mD$,
       then devote one process that stays in $\init_B$.
       The process will be used in the looping phase to ensure that the run $y$ is infinite,
       and that every process of \cutoffsys used in (e) 
       moves infinitely often (and thus $y$ is strong-fair);
       and
%       Note that if $A$ moves infinitely often in $x$ and $DeadGuards \smi \dead = \emptyset$,
%       then there is no need for such additional infinitely moving process.

\-[g.] let any other process of \cutoffsys (if any) 
       copy behaviour of a process of \largesys 
       that was not used in the construction so far (including this step).
\il
\ak{go through every item, and prove it is necessary (by giving an example)}
The setup phase ensures: 
in every state $q \in \dead$,
there is at least one process deadlocked in $q$ at moment $d$ in $y$. 
Now we need to ensure that the non-deadlocked processes described 
in steps (e) and (f) move infinitely often,
which is done using the looping extension described bellow.

The looping phase is applied to processes in (e) and (f) only\footnote%
{%
  If there are no such processes, then the setup phase produces the sought run $y$.
}.

Order arbitrarily 
$DeadGuards \smi \dead = (q_1,\ldots,q_k) \subseteq \visInf{\mB\smi\mD}{x}$.
Note that $\init_B \not\in (q_1,...,q_k)$.
Let $\mP$ be the set of processes of \cutoffsys used in steps (e) or (f).
Note that $|\mP| = |(q_1,...,q_k)| + 1$.

The \textbf{looping phase} is:
set $i=1$, and repeat infinitely the following.
\li
  \- Let $P_\init \in \mP$ be the process that is currently in $\init_B$, 
     and $P_{q_i} \in \mP$ -- in $q_i$.
     
  \- Let $B_{q_i} \in \visInf{\mB\smi\mD}{x}$ be a process of \largesys 
     that visits $q_i$ and $\init_B$ infinitely often.
     Let $P_\init$ of \cutoffsys copy transitions of $B_{q_i}$
     on some path $\init_B \to \ldots \to g_i$,
     then let $P_{g_i}$ copy transitions of $B_{q_i}$ on some path 
     $g_i \to \ldots \to \init_B$. 
     For copying we consider only the paths of $B_{q_i}$ that happen after moment $d$.

  \- $i=i \oplus 1$.
\il

The number of copies of $B$ that the construction uses in the worst case is 
(i.e., the item (g) is not used, and we assume $Q_B>2$, $DeadGuards \smi \dead = \emptyset$, and $A \in \mD$):
$$
1_{(f)} + 2|\deadTwo|_{(c),(d)} + |\deadOne|_{(b)} 
 \leq 
2|Q_B \smi \{\init_B\}| + 1.
$$

\myparagraph{Deadlocks}
The largest value of $c$ among those for ``Local Deadlocks'' 
and for ``Global Deadlocks'' can be used as the sought value of $c$ 
for the case of general deadlocks.
But it will not be the smallest one.
In the proof of the case ``Local Deadlocks'', in the setup phase, 
item (e) can be modified for the case when $A \in \mD$:
since we do not need to ensure that $y$ is infinite, 
we avoid allocating a process in state $\init_B$.
For a given locally deadlocked strong-fair run, the setup phase may produce
the globally deadlocked run, but that is allright for the case of general deadlocks.
With this note, for the general case $c = 2|Q_B \smi \{\init_B\}|$.
\end{proof}

\begin{tightness}[1-Conj, Deadlocks, Fair] \label{obs:conj:tight_deadlock_fair}
The cutoff $c=2|B|-2$ is tight for deadlock detection on strong-fair initializing
or finite runs in the 1-conjunctive systems, 
i.e., 
for any $k>2$ there is a system type $(A,B)$ with $|B|=k$ such that 
there is a strong-fair initializing deadlocked run in $(A,B)^{(1,2|B|-2)}$, 
but not in $(A,B)^{(1,2|B|-3)}$.
\end{tightness} 
\begin{proof} 
Consider the same templates as in Tightness~\ref{obs:conj:tight_deadlock}.
%
% AK: we can claim the below, but we need to note that 
% the monotonicity lemma for local deadlocks also holds (straightforward)
% let's comment this out for now.
%  Note that the cutoff $c=2|B|-1$ stated in the previous Lemma 
% for the case of local deadlocks is also tight.
% To prove this, take the templates from Observation~\ref{obs:conj:tight_deadlock}
% and modify slightly the template B:
% add the unguarded self-loop to $\init$.
%%%
% \begin{figure}[Htb]
% \centering
% \subfloat[Template A]{
% \centering
% \makebox[0.4\textwidth][c]{
% \scalebox{0.75}{\input{guarded-systems/img/conj_tight_deadlock_tmplA}}
% \label{fig:conj:tight_deadlock_tmplA}
% }}
% \subfloat[Template B]{
% \centering
% \makebox[0.6\textwidth][c]{
% \scalebox{0.75}{\input{guarded-systems/img/conj_tight_deadlock_tmplB}}
% \label{fig:conj:tight_deadlock_tmplB}
% }}
% \caption{Templates $(A,B)$ used to prove the tightness of the cutoff $c=2|B|-2$ for the deadlock detection in 1-guard conjunctive systems.
% In the figure the edge with $\forall{\neg b_1},\ldots,\forall{\neg b_k}$ denotes edges with guards $\forall{\neg b_1},\ldots,\forall{\neg b_k}$ (Observation~\ref{obs:conj:tight_deadlock}).\ak{check me}}
% \label{fig:conj:tight_dead_tmpl}
% \end{figure}
\end{proof}

\section{Conclusion} \label{gua:sec:concl}

We have extended the cutoffs for guarded protocols of Emerson and 
Kahlon~\cite{Emerson00} to support local deadlock detection, fairness 
assumptions, and open systems.
In particular, our results imply the decidability of the parameterized model checking problem for this class of systems and specifications,
which to the best of our knowledge was unknown before. 
%Our results allow us to model check
%guarded protocols that satisfy not 
%only safety, but also liveness conditions, for an arbitrary number of 
%components. 
Furthermore, the cutoff results can easily be integrated into 
the parameterized synthesis approach~\cite{JB14}.
%~\cite{Jacobs14,Khalimov13,Khalimov13a}.

%An approach for using cutoff results in 
%synthesis has been introduced by 
%Jacobs and Bloem~\cite{Jacobs14}. It has been described in detail for the 
%case of 
%token-passing systems. Follow-up papers have shown how to make the approach 
%more efficient~\cite{KhalimovJB13b}, and how to use it for the synthesis of a 
%large 
%case study, the AMBA bus arbiter~\cite{BloemJK14}.

Since conjunctive guards can model atomic sections and read-write locks, 
and disjunctive guards can model pairwise rendezvous 
(for some classes of specifications, see~\cite{EmersonK03}), 
our results apply to a wide spectrum of systems models.
But the expressive power of the model %and flexibility of the results 
comes at a high cost: cutoffs are linear in the size of a process, and 
are shown to be tight (with respect to this parameter).
For conjunctive systems, our new results are restricted to systems with
1-conjunctive guards, effectively only allowing to model a single shared
resource. 
We conjecture that our proof methods can be extended to systems with
more general conjunctive guards, at the price of bigger cutoffs.
We leave this extension and the question of finding cutoffs that are independent of the size of processes for future research.

We did preliminary experiments~\cite{SimonThesis} by implementing the synthesizer inside our parameterized synthesizer PARTY~\cite{party}.
It is a possible future work to find and apply it to real-world applications.
%
%
%We are working on a prototype implementation (\url{https://bitbucket.org/parsy/guarded_synthesis/}), which however is currently limited to very small systems.\sj{we should re-formulate or remove the comment on implementation} 
%In future work, we will try to lift the restrictions of our results for conjunctive systems, and investigate cutoffs that are independent of the size of the components' state spaces.
%\ak{remove this promise?}
\ak{note that EK have better complexities for 'for all paths' properties. 
As a future work, one can look if our cutoffs can be improved.}
%This is due 
%to the growth of the cutoff (linearly) and the set of possible transition 
%guards (doubly exponential) in the size of process templates. 
%In the future, 
%we will look into cutoffs that are independent of the size of process 
%templates.

\chapter{Parameterized Token Rings} \label{chap:token-systems}

\newcommand\VarNames{\textsf{Vars}}
\newcommand{\pring}{\ensuremath{\mathcal {R}}}
\renewcommand{\trans}[3]{#1 \stackrel{{#3}}{\rightarrow} #2}
\newcommand{\trcv}{\mathsf{rcv}}
\newcommand{\tsnd}{\mathsf{snd}}
\newcommand{\sch}{\mathsf{sch}}
\newcommand{\token}{\mathsf{tok}}
\newcommand{\powerset}[1]{2^{#1}}
\newcommand{\Locals}{Q}
\newcommand{\LocalsI}{\Locals_0}
\newcommand{\ActionsProc}{\Sigma_{\mathsf{pr}}}

\newcommand{\IndSet}{\bbN}

\newcommand{\Opr}{{\mathrm{O}_{\mathrm{pr}}}}
\newcommand{\Osys}{{\mathrm{O}_{\mathrm{sys}}}}

\newcommand{\IprAll}{{\mathrm{I}_{\mathrm{pr}}}}
\newcommand{\Iloc}{\mathrm{I}_{loc}}
\newcommand{\Iglob}{\mathrm{I}_{glob}}
\newcommand{\Ipr}{\IprAll}
\newcommand{\Isys}{{\mathrm{I}_{\mathrm{sys}}}}

% AMBA signals
\newcommand{\ambasignal}[1]{\textsc{#1}}
\newcommand{\ambasignali}[2][i]{\textsc{#2}[\textrm{#1}]}
\newcommand{\hgrant}{\ambasignal{\textcolor{blue}{hgrant}}}
\newcommand{\hgranti}[1][i]{\ambasignali[#1]{\textcolor{blue}{hgrant}}}
\newcommand{\hbusreq}{\ambasignal{\textcolor{red}{hbusreq}}}
\newcommand{\hbusreqi}[1][i]{\ambasignali[#1]{\textcolor{red}{hbusreq}}}
\newcommand{\hready}{\ambasignal{\textcolor{red}{hready}}}
\newcommand{\hreadyi}[1][i]{\ambasignali[#1]{\textcolor{red}{ready}}}
\newcommand{\hlock}{\ambasignal{\textcolor{red}{hlock}}}
\newcommand{\hlocki}{\ambasignali[i]{\textcolor{red}{hlock}}}
\newcommand{\hmastlock}{\ambasignal{\textcolor{blue}{hmastlock}}}
\newcommand{\hmastlocki}[1][i]{\ambasignali[#1]{\textcolor{blue}{hmastlock}}}
\newcommand{\hburst}{\ambasignal{\textcolor{red}{hburst}}}
\newcommand{\hbursti}{\ambasignali[i]{\textcolor{red}{hburst}}}
\newcommand{\hmaster}{\ambasignal{\textcolor{blue}{hmaster}}}
\newcommand{\hmasteri}[1][i]{\ambasignali[#1]{\textcolor{blue}{hmaster}}}
\newcommand{\hstart}{\ambasignal{\textcolor{blue}{start}}}
\newcommand{\hstarti}[1][i]{\ambasignali[#1]{\textcolor{blue}{start}}}
\newcommand{\hlocked}{\ambasignal{\textcolor{blue}{locked}}}
\newcommand{\hlockedi}[1][i]{\ambasignali[#1]{\textcolor{blue}{locked}}}
\newcommand{\hdecide}{\ambasignal{\textcolor{blue}{decide}}}
\newcommand{\hdecidei}{\ambasignali[i]{\textcolor{blue}{decide}}}
\newcommand{\tok}{\ambasignal{\textcolor{blue}{tok}}}
\newcommand{\toki}{\ambasignali[i]{\textcolor{blue}{tok}}}
\newcommand{\hincr}{\ambasignal{incr}}
\newcommand{\hburstfour}{\ambasignal{burst4}}
\newcommand{\hburstthree}{\ambasignal{burst3}}
\newcommand{\norequestsi}{\ambasignali{\textcolor{red}{no\_req}}}
\newcommand{\norequests}{\ambasignal{\textcolor{red}{no\_req}}}
\newcommand{\send}{\ambasignal{send}}
\newcommand{\sendi}{\ambasignali[i]{send}}

\hfill {\footnotesize\textit{This chapter is based on joint work with R.Bloem and S.Jacobs~\cite{Khalimov13,party,BJK14}}~~~~~~~~}

\begin{quotation}
\noindent\textbf{Abstract.}
Parameterized synthesis was recently proposed as a way
to circumvent the poor scalability of current synthesis tools.
The method uses cutoff results in token rings to
reduce the problem to bounded distributed synthesis,
and ultimately to a sequence of SMT problems. 
But experiments show that the size of the specification is a major issue. 
In this chapter we
 (1) propose several optimizations of the approach, and
 (2) perform a parameterized synthesis case study on the industrial arbiter protocol AMBA.

In the first part of this chapter,
we optimize the reduction of the parameterized to distributed synthesis.
To this end,
we refine the cutoff reduction using modularity and abstraction.
The evaluation, using our specially developed parameterized synthesizer PARTY,
shows that the optimizations lead to several orders of magnitude speed-ups.

In the second part, we perform parameterized synthesis case study
on the industrial arbiter protocol AMBA.
The AMBA protocol has been used as a benchmark for many reactive synthesis tools,
because it is hard to synthesize an implementation that can serve a large number of clients.
We show how to use parameterized synthesis to obtain a component that serves a single master,
and can be arranged in a ring of arbitrarily many components.
We describe new tricks---a cutoff extension tailored for AMBA and decompositional synthesis---%
that together with the previously described optimizations allowed us
to synthesize a component with 14 states in about 1 hour.
\end{quotation}

\section{Introduction}

By automatically generating correct implementations from a temporal logic 
specification, reactive synthesis tools can relieve system designers from 
tedious and error-prone tasks like low-level manual implementation and 
debugging. This great benefit comes at the cost of high computational complexity 
of synthesis, which makes synthesis of large systems an ambitious goal. 
For instance, Bloem et al.~\cite{Bloem12}
synthesize an arbiter for the ARM AMBA Advanced High
Performance Bus (AHB)~\cite{AMBAspec}. The results, obtained using 
RATSY~\cite{Bloem10c},
show that both the size of the implementation and the time for synthesis
increase steeply with the number of masters that the arbiter can
handle. This is unexpected, since an arbiter for $n+1$ masters is very 
similar to an arbiter for $n$ masters, and manual implementations grow only 
slightly with the number of masters. While recent results show that 
synthesis time and implementation size can be improved in standard LTL 
synthesis tools~\cite{BS,GodhalCH13}, the fundamental problem of increasing 
complexity with the number of masters can only be solved by adapting the 
synthesis approach itself.

To this end, Jacobs and Bloem~\cite{JB14} introduced the 
\emph{parameterized synthesis} approach.
A simple example of a parameterized specification is the following LTL specification of a simple arbiter:
\[ \begin{array}{ll}
  \forall i \neq j.~ & \G \neg ( g_i \land g_j ) \land \\
  \forall i.~ & \G (r_i \impl \F g_i).
  \end{array}
\]
In parameterized synthesis, we synthesize a building block that can be cloned to form a system that satisfies such a specification, for any number of components.

Jacobs and Bloem~\cite{JB14} showed that parameterized synthesis is undecidable in general, but semi-decision procedures can be found for classes of systems with cutoffs, i.e., where parameterized verification can be reduced to verification of a system with a bounded number of components. They presented a semi-decision procedure for token-ring networks, building on results by Emerson and Namjoshi~\cite{Emerso03}, which show that for the verification of parameterized token rings, a cutoff of $5$ is sufficient for a certain class of specifications. Following these results, parameterized synthesis reduces to distributed synthesis in token rings of (up to) $5$ identical processes. To solve the resulting problem, a modification of the SMT encoding of the distributed bounded synthesis problem by Finkbeiner and Schewe~\cite{BS} was used.

Experiments with the parameterized synthesis method~\cite{JB14} revealed that only very small specifications could be handled with this encoding. For example, the simple arbiter presented before can be synthesized in a few seconds for a ring of size $4$, which is the sufficient cutoff for this specification. However, synthesis does not terminate within $2$ hours for a specification that also excludes spurious grants, in a ring of the same size. Furthermore, the previously proposed method uses cutoff results of Emerson and Namjoshi~\cite{Emerso03} and therefore inherits a restricted language support and cannot handle specifications in assume-guarantee style~\cite{Bloem12}.
This precludes the approach from being applied to the AMBA protocol.

In this chapter we address both issues.

In the first part of the chapter (Section~\ref{tok_rings:sec:bs-and-optimizations}),
we optimize the reduction of the parameterized to distributed synthesis.
We use the fact that
(a) token-ring systems consist of isomorphic processes,
(b) different properties may require different cutoffs, and
(c) when model checking the behaviours of some fixed processes,
    the behaviours of the others can be abstracted.
The evaluation,
using our specially developed parameterized synthesizer PARTY,
show that the optimizations lead to several orders of magnitude speed-ups.
 
In the second part of the chapter (Section~\ref{amba:sec}),
we perform parameterized synthesis case study on the industrial arbiter protocol AMBA.
The AMBA protocol has been used as a benchmark for many reactive synthesis tools,
because it is hard to synthesize an implementation that can serve
a large number of clients. We show how to use parameterized synthesis
to obtain a component that serves a single master, and can be arranged
in a ring of arbitrarily many components.
We describe new tricks%
---a cutoff extension tailored for AMBA and decompositional synthesis---%
that together with the previously described optimizations allowed us
to synthesize a component with 14 states in about 1 hour.

The chapter starts with definitions in Section~\ref{tok_rings:defs},
where we introduce token-ring systems, parameterized specifications and problems.
Then we state known cutoff results and a slight generalization.
Section~\ref{tok_rings:sec:bs-and-optimizations} describes
the SMT encoding of the bounded synthesis for token-ring systems,
followed by optimizations and experiments.
Then we proceed to the AMBA case study (Section~\ref{amba:sec}).
We describe the protocol and its parameterized specification.
Section~\ref{amba:sec:handling-amba} contains the main contribution:
(1) we rewrite the specification into the form feasible to parameterized synthesis and
(2) we extend the known cutoffs to handle the resulting AMBA specification.
In Section~\ref{amba:sec:experiments} on experiments,
we describe the crucial optimization ``decompositional synthesis''
and report synthesis timings.

%%%%%%%%%%%%%%%%%%%%%%%%%%%%%%%%%%%%%%%%%%%%%%%%%%%%%%%%%%%%%%%%%%%%%%%%%%%%%%%%%%%%
\section{Definitions} \label{tok_rings:defs}

\subsection{Token-ring Systems} \label{tok_rings:defs:system}
In this section we define token ring systems---%
the LTS that consists of replicated copies of a process connected in a uni-directional ring.
Transitions in a token ring system are either internal or synchronized
(in which one process sends the token to the next process along the ring).
The token starts in a non-deterministically chosen process.

We start by recalling a (non-deterministic) labeled transition system.
A \emph{labeled transition system (LTS)}
is a tuple
$(I,O,Q,Q_0,\delta,out)$
where
 $I$ is the set of {\em inputs},
 $O$ is the set of {\em outputs} disjoint from $I$,
 $Q$ is the set of {\em states},
 $Q_0 \subseteq Q$ is the set of {\em initial states},
 $\delta \subseteq Q \times 2^I \times Q$ is the {\em transition relation},
 and $out:Q \to 2^O$ is the \emph{output function} (also called {\em state-labeling} function).

Fix two disjoint sets:
a set $\Opr$ of process template \emph{output variables} that contains two distinguished output variable,
$\tsnd$ and $\token$,
and a set $\Ipr$ of process template \emph{input variables} that contains a distinguished input variable $\trcv$.
We always assume that $\Ipr$ and $\Opr$ are disjoint.

\parbf{Process template} \label{page:tok_rings:defs:process_template}
A {\em process template} $P$ is an LTS
$(\Ipr, \Opr, Q, Q_0, \delta, out, A_{loc})$:
\begin{enumerate}[label*=\roman*)]
\item 
The state set $Q$ is finite and can be partitioned into two non-empty disjoint sets: $Q = T \cupdot NT$. 
    States in $T$ are said to {\em have the token}.

\item
The initial state set is $Q_0  = \{\iota_t, \iota_n\}$ for some $\iota_t \in T, \iota_n \in NT$.

\item
The output function is $out: Q \rightarrow 2^{\Opr}$ and it satisfies:
\li
\- for every $t \in NT$: $\token \not\in out(t)$ and for every $t \in T$: $\token \in out(t)$,
\- for every $t \in Q$: $\tsnd \in out(t) \impl t \in T$.
\il

\item
Let $\ActionsProc = 2^\Ipr$.
Let $\ActionsProc^\trcv = \{ i \in \ActionsProc \| \trcv \in i \}$, 
$\ActionsProc^{\neg \trcv} = \ActionsProc \setminus \Sigma^\trcv$,
$T^\tsnd = \{ q \in \Locals \| \tsnd \in out(q) \}$,
$T^{\neg \tsnd} = T \setminus T^\tsnd$.
Then the transition function: 
\begin{equation}\label{tok_rings:eq:process-trans}
\delta \subseteq 
T^\tsnd \times \ActionsProc^{\neg \trcv}\times NT  ~~\cup~~
NT\times\ActionsProc^\trcv\times T   ~~\cup~~
NT\times\ActionsProc^{\neg \trcv}\times NT   ~~\cup~~
T^{\neg \tsnd}\times\ActionsProc^{\neg \trcv}\times T.
\end{equation}

Also, $\delta$ is non-terminating: 
for every $q \in NT$ and every $i \in \ActionsProc$
  there exists $\trans{q}{q'}{i}$; 
and for every $q \in T$ and every $i \in \ActionsProc^{\neg \trcv}$
  there exists $\trans{q}{q'}{i}$.

\item [$\dagger$)]
$A_{loc}$ is a \emph{fairness} condition over $\Ipr \cup \Opr$.
We require that on every infinite path from an initial state and satisfying $A_{loc}$,
from any state with the token, $q \in T$, the process reaches a state $q'$ where it sends the token.
(In LTL this can be written as $A_{loc} \impl \G(\token \impl \F \tsnd)$.)
We call this requirement ($\dagger$).
We omit $A_{loc}$ in the LTS tuple when it is not important.
\end{enumerate}

\parbf{Ring topology $R$}
A {\em ring} is a directed graph $R = (V,E)$,
where the set of vertices is $V = \{ 1,\ldots,k\}$ for some $k \in \IndSet$,
and the set of edges is $E = \{(i,i_{mod |V|}+1) \mid i\in V\}$.
We will skip ``$mod |V|$'' and write $i+1$.
Vertices are called {\em process indices}.

\parbf{Token-ring system $P^R$}
Fix a ring topology $R=(V,E)$.

Let $\Isys = (\Iloc \times V) \cupdot \Iglob$
be the \emph{system input variables},
where local inputs $\Iloc$ and global inputs $\Iglob$ are
such that $\Ipr = \Iloc \cupdot \Iglob$.
%Define $\ActionsSys = \powerset{\Isys}$.
For system input ${\sf in} \in 2^\Isys$,
let ${\sf in}(v) = \{ i \in {\sf in} \| i\in \Iloc\times\{v\} \cup \Iglob \}$
denote the input to process $v$ (including global inputs).

Let $\Osys = \Opr \times V$ be the \emph{system output variables}. 
For $(p,i)$ in $\Osys$ or in $\Isys\setminus \Iglob$ we write $p_i$. 

Given a process template $P =  (\Ipr, \Opr, Q, Q_0, \delta, out)$
and a token ring topology $R = (V,E)$,
the {\em token-ring system} $P^R$ is the LTS $(\Isys,\Osys,S,S_0,\Delta,Out)$:

\li
\- The set $S$ of \emph{global states} is $Q^V$,
   i.e., all functions from $V$ to $Q$.
   If $s \in Q^V$ is a global state
   then $s(i)$ denotes the local state of the process with index $i$.

\- The set of \emph{global initial states} $S_0$ contains all $s_0 \in Q_0^V$
   in which exactly one of the processes has the token.

\- The labeling $Out(s): S \to 2^\Osys$ is:
   for every $s \in S$: $p_i \in Out(s)$ iff $p \in out(s(i))$, for $p \in \Opr$ and $i \in V$.
\il

Finally, we define the {\em global transition relation} $\Delta$. 
In a {\em fully asynchronous token ring},
a subset of the processes can make a transition in each step of the system.
Thus, $\Delta$ consists of the following set of transitions:

 \begin{itemize}
  \item 
  An {\em internal transition} is an element $(s,{\sf in},s')$ of $S \times 2^\Isys \times S$, for which there are process indices $M \subseteq V$ such that 
  \begin{enumerate}[label*=\roman*)]
    \item
    for all $v \in M$: $\tsnd \not\in out(s(v))$ and $\trcv \not \in{\sf in}(v)$,

    \item 
    for all $v\in M$: $\trans{s(v)}{s'(v)}{in(v)}$ is a transition of $P$, and

    \item 
    for all $u \in V \setminus M$: $s(u) = s'(u)$.
  \end{enumerate}

  \item
  A {\em token-passing transition}  is an element $(s,{\sf in},s')$ of $S \times 2^\Isys \times S$
  for which there are two process indices $v$ and $w=v+1$ and process indices $M \subset V$
  with $\{v,w\} \subseteq M$ such that

  \begin{enumerate}[label*=\roman*)]
    \item
    $\tsnd \in out(s(v))$, 
    and $\forall{u \in M\setminus \{v\}} : \tsnd \not\in out(s(u))$%
    ---i.e., only process $v$ sends the token,

    \item
    $\trcv \in {\sf in}(w)$ and
    for all $u \in M \setminus \{w\}$: $\trcv \not \in{\sf in}(u)$%
    ---i.e., only process $w$ receives the token,

    \item
    for every $u\in M$: $\trans{s(u)}{s'(u)}{in(u)}$ is a transition of $P$, and

    \item
    for every $u \in V \setminus M$: $s'(u) = s(u)$.
  \end{enumerate}

 \end{itemize}

Special cases of the fully asynchronous token ring are the \emph{synchronous token ring}
and the \emph{interleaving token ring}.
In a synchronous token ring, $M=V$ for internal and token-passing transitions,
i.e., at each step all the processes simultaneously make a transition.
In an interleaving token ring,
$M = \{v\}$ for some $v\in V$ for internal transitions,
and $M=\{v,w\}$ for $(v,w)\in E$ for token-passing transitions,
i.e., at each moment either \emph{exactly one} process makes an internal transition,
or one process sends a token to the next process.

An example of processes arranged in a token ring is in Figure~\ref{fig:ring-architecture}.

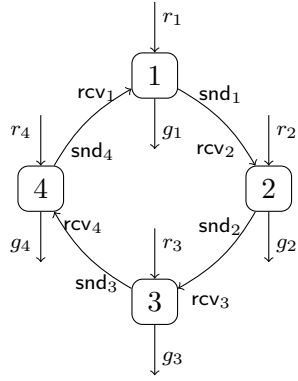
\begin{figure}[tb]\center
\begin{tikzpicture}
	\begin{pgfonlayer}{nodelayer}
		\node [style=text ellipse] (0) at (0, 1.5) {1};
		\node [style=text ellipse] (1) at (1.5, -0) {2};
		\node [style=text ellipse] (2) at (-1.5, -0) {4};
		\node [style=text ellipse] (3) at (0, -1.5) {3};
		\node [style=invisible] (4) at (0, 2.5) {};
		\node [style=invisible] (5) at (-1.5, 1) {};
		\node [style=invisible] (6) at (-1.5, -1) {};
		\node [style=invisible] (7) at (0, -0.5) {};
		\node [style=invisible] (8) at (0, -2.5) {};
		\node [style=invisible] (9) at (1.5, 1) {};
		\node [style=invisible] (10) at (1.5, -1) {};
		\node [style=invisible] (11) at (0, 0.5) {};
		\node [style=textual] (12) at (0.75, 1.25) {$~~_{\tsnd_1}$};
		\node [style=textual] (13) at (0.25, 2.25) {$_{r_1}$};
		\node [style=textual] (14) at (-0.75, 1.25) {$_{\trcv_1}$};
		\node [style=textual] (15) at (1.75, 0.75) {$_{r_2}$};
		\node [style=textual] (16) at (1.75, -0.75) {$_{g_2}$};
		\node [style=textual] (17) at (1, -0.5) {$_{\tsnd_2}~~$};
		\node [style=textual] (18) at (0.75, -1.5) {$_{\trcv_3}$};
		\node [style=textual] (19) at (0.75, 0.5) {$~_{\trcv_2}$};
		\node [style=textual] (20) at (-1, 0.5) {$~~~_{\tsnd_4}$};
		\node [style=textual] (21) at (-1, -0.5) {$~_{\trcv_4}$};
		\node [style=textual] (22) at (-0.75, -1.25) {$_{\tsnd_3}$};
		\node [style=textual] (23) at (-1.75, 0.75) {$_{r_4}$};
		\node [style=textual] (24) at (-1.75, -0.75) {$_{g_4}$};
		\node [style=textual] (25) at (0.25, -2.25) {$_{g_3}$};
		\node [style=textual] (26) at (0.25, -0.75) {$_{r_3}$};
		\node [style=textual] (27) at (0.25, 0.75) {$_{g_1}$};
	\end{pgfonlayer}
	\begin{pgfonlayer}{edgelayer}
		\draw [style=arrow, bend left=15, looseness=1.00] (0) to (1);
		\draw [style=arrow, bend left=15, looseness=1.00] (1) to (3);
		\draw [style=arrow, bend left=15, looseness=1.00] (3) to (2);
		\draw [style=arrow, bend left=15, looseness=1.00] (2) to (0);
		\draw [style=arrow, label=df] (4) to (0);
		\draw [style=arrow] (5) to (2);
		\draw [style=arrow] (2) to (6);
		\draw [style=arrow] (9) to (1);
		\draw [style=arrow] (1) to (10);
		\draw [style=arrow] (7) to (3);
		\draw [style=arrow] (3) to (8);
		\draw [style=arrow] (0) to (11);
	\end{pgfonlayer}
\end{tikzpicture}
\caption{Token ring system with 4 processes.
  Every process has input $r$ and output $g$.
  Additionally, every process has input $\trcv$ and output $\tsnd$ that are used for passing the token.
  Thus, $\Ipr=\{r,\trcv\}$ and $\Opr = \{g,\tsnd\}$.
  In this example, $\Iglob$ is empty.}
\label{fig:ring-architecture}
\end{figure}

\parbf{System runs}
Fix a ring topology $R = (V,E)$ and a process template $P$.
A \emph{run of a token ring system}
$P^{R}=(\Isys,\Osys,S,S_0,\Delta,Out)$ is a maximal-finite or infinite sequence
$x=(s_1,{\sf in}_1,M_1)(s_2,{\sf in}_2,M_2)\ldots$,
where:
\li
\- $s_1 \in S_0$, $s_k \in S$ and ${\sf in}_k \in 2^\Isys$ for any $k \le |x|$,
\- for all $k < |x|: (s_k,{\sf in}_k,s_{k+1}) \in \Delta$,
\- for all $k < |x|$: $M_k$ is the set of processes transiting in $(s_k,{\sf in}_k,s_{k+1})$
   (see $M$ in the definition of $\Delta$).
\il

\subsection{Parameterized Systems}

The \emph{parameterized ring} is the function $\pring: n \mapsto \pring(n)$,
where $n \in \bbN$ and $\pring(n)$ is the ring with $n$ vertices.
A \emph{parameterized token-ring system} is a function $P^\pring: n \mapsto P^{\pring(n)}$,
where $n\in\bbN$ and $P$ is a given process template.
To disambiguate, we explicitly write
``parameterized [fully asynchronous][interleaving][synchronous] token-ring system''.

\subsection{Parameterized Specifications}\label{tok_rings:defs:indexed-ltl}

\emph{Parameterized specification} is a tuple $\tpl{\Ipr,\Iglob,\Opr,\Phi}$,
where $\Ipr$ is a set of process template inputs (global and local),
$\Iglob$ is a set of global inputs,
$\Opr$ is a set of process template outputs,
and $\Phi$ is an indexed LTL formula over $\Ipr$ and $\Opr$.
Intuitively, an indexed LTL formula is an LTL formula with indexed variables and quantification over indices.
Below we define indexed LTL and its sublogic, prenex-indexed LTL.

\subsection*{Indexed LTL}

\parbf{Syntax}
Let $\VarNames$ denote the set of variable names (that will be used as process indices).
Let $cond$ be a Boolean formula over atoms of the form $x = y$ or $x = y+1$,
for arbitrary $x,y$ from $\VarNames$.
Then an indexed LTL formula $\Phi$ over $\Ipr$, $\Iglob$, and $\Opr$ has the grammar:
\begin{align*}
\Phi ~=~ & \forall v. (cond \impl \Phi) \| \exists v. (cond \land \Phi) \| \\
         & \Phi\land\Phi \| \neg \Phi \| \\
         & e \| i_v \| o_v \| \Phi \U \Phi \| \X_v \Phi
\end{align*}
where $v \in \VarNames$, $i \in \Iloc$, $o \in \Opr$, $e\in \Iglob$.
We will write $\forall x \neq y:\Phi$ instead of $\forall x\forall y: (x \neq y) \impl \Phi$,
and $\exists x \neq y:\Phi$ instead of $\exists x\exists y: x \neq y \land \Phi$.

\parbf{Semantics}
We define the semantics for sentence formulas only:
a formula $\Phi$ is a \emph{sentence} iff every variable $v$ mentioned in the formula
is in the scope of a quantifier over that variable.
E.g., $r_x$ is not a sentence, while $\forall x: r_x$ is.

Let $\Phi$ be a sentence.
Let $P^R$ be a token-ring system with $R=(V,E)$ and $\pi$ be an infinite run of the system.
Define $\pi\models \Phi$ iff $\pi \models \Phi_V$ (this satisfaction is defined later),
where $\Phi_V$ is constructed from $\Phi$ as follows.
\li
\-[1.]
Replace every single-quantified subformula $\forall v.\phi$ of $\Phi$
with $\bigwedge_{i \in V} \phi[v \mapsto i]$;
replace every single-quantified subformula $\exists v.\phi$
with $\bigvee_{i \in V} \phi[v \mapsto i]$.
Here $\phi[v \mapsto i]$ denotes the formula $\phi$
in which $v$ is substituted by $i$.
E.g., $r_x[x \mapsto 5]$ is $r_5$.

\-[2.]
Repeat step (1) until all quantifiers disappear.
The resulting formula is $\Phi_V$.
Note that conditions $cond$ like $x\neq y$ get simplified into $\true$ or $\false$.
\il
E.g., $\exists x\exists y.x\neq y \land g_x \land g_y$ becomes
$\bigvee_{(x,y) \in V\times V}.x\neq y \land g_x\land g_y$.

\parbf{Definition of ``system satisfies $\Phi$''}
Fix a $P=(\Ipr,\Opr,Q,Q_0,\delta,out)$,
global inputs $\Iglob$, and a token ring $R=(V,E)$.
Let $\Phi$ be an indexed LTL over $\Ipr$, $\Opr$, and $\Iglob$.
Then $P^R \models \Phi$ iff for every infinite system run $\pi$: $\pi\models\Phi$.
An infinite system run $\pi=(s_1,in_1,M_1) (s_2, in_2, M_2) ... \in (S\times 2^{\Isys}\times 2^V)^\omega$
satisfies $\Phi$
iff
$(Out(s_1),in_1) (Out(s_2), in_2)... \models \Phi_V$.
The latter satisfaction is standard except for the operator $\X$.
Given a $v \in V$ and the original run,
$(Out(s_1),in_1) (Out(s_2), in_2)... \models \X_v \varphi$
iff
$(Out(s_i),in_i) (Out(s_{i+1}), in_{i+1})... \models \varphi$
where $i$ is the second\footnote{Why ``second'', not the first one?
  This is the consequence of the fact that we group the input \emph{to be read} with the current output.
  E.g., $\X_v r_v$ should refer to $r_v$ read when transiting \emph{from} the next state
  rather than referring to $r_v$ read when transiting \emph{into} the next state.}
  smallest $i$ such that $v \in M_i$.
Intuitively, $\X_v \varphi$ requires $\varphi$ to hold on the suffix run
that skips one transition of the process $v$ and that starts with $v$ transiting.
In formulas of the form $\forall i.(...\X_i...)$,
we usually skip the subscript in $\X_i$ and write $\X$.
(The next operator $\X_i$ presented here is inspired by the action-based semantics from~\cite{Emerso03}.)

\subsection*{Prenex-indexed LTL}

Let us abbreviate by $\forall x_{cond}.\phi$ the formula $\forall x.cond \impl \phi$,
and by $\exists x_{cond}.\phi$ the formula $\exists x. cond \land \phi$.
When the quantifier is not important, we write $Q x_{cond}.\phi$.

An indexed LTL formula $\Phi$ is \emph{prenex-indexed} iff it is of the form
$$
Q {v^1}_{cond_{v^1}}...Q {v^k}_{cond_{v^k}}: \phi.
$$
We call $\Phi$ \emph{$k$-indexed}, because it has $k$ quantifiers.
Let $\LTLmX$ refer to LTL formulas that do not use $\X$.

Note that prenex-indexed \LTL is not as expressive as (non-prenex) indexed LTL.
For example, formula $\F\forall x. p_x$ does not have an equivalent prenex-indexed form.

Most of existing and our cutoff results are restricted to prenex-indexed LTL formulas
with the empty set of global inputs.

\begin{remark}[$\forall i.A_i \impl \forall j.G_j$ is not prenex-indexed]\label{tok_rings:rem:gr1-not-prenex}
In the previous section we defined ``a system satisfies an indexed LTL formula''.
If we use the path quantifier $\A$ explicitly,
then, as usually, a system satisfies an LTL formula $\varphi$,
$sys\models \varphi$, is equivalent to $sys \models \A\varphi$,
where $\varphi$ is treated as a path formula of \CTLstar.
Now consider $\forall i.A_i \impl \forall j.G_j$.
If rewritten with the path quantifier $\A$, it is $\A(\forall i.A_i \impl \forall j.G_j)$.
There is no way to turn it into the form $Q {v^1}... Q {v^k} \A \phi$
and this formula is not prenex-indexed.
\end{remark}

\subsection{Parameterized Synthesis Problem}
The \emph{parameterized synthesis problem (for token rings)} is:

\smallskip\noindent
\emph{Given}: parameterized specification $\tpl{\Ipr,\Iglob,\Opr,\Phi}$ \\
\emph{Return}: process template $P=(\Ipr,\Opr,Q,q_0,\delta,out)$ such that
          for every $n$: $P^{\pring(n)} \models \Phi$,
          or ``unrealizable'' if no such template exists.
\smallskip

We can similarly define the parameterized \emph{model checking} problem,
in which the process template is given as input.

Furthermore, we will use the variants of these problems,
which ask whether all systems \emph{larger than a given $n_0$}
satisfy the formula.
We call such problems {\em parameterized$_{>n_0}$}.

The parameterized synthesis for token rings is undecidable~\cite{JB14},
even for prenex 2-indexed specifications without global inputs:
\begin{theorem}[\cite{JB14}, Theorem 3.5]\label{tok_rings:thm:param-synth-is-undec}
The parameterized synthesis problem of interleaving token rings,
without global inputs,
formulas $\forall i \neq j. \varphi(i,j)$,
is undecidable,
where $\varphi(i,j)$ is an $\LTLmX$ formula over processes $i,j$.
\end{theorem}
\noindent This result follows from the undecidability of synthesis of
distributed systems with two processes~\cite{DBLP:conf/focs/PnueliR90}.
The problem \emph{is} decidable for prenex 1-indexed specifications.

\section{Reduction by Cutoffs}\label{tok_rings:sec:reduction}

The definition of a cutoff is the same as in Section~\ref{gua:sec:paramsynt} on page~\pageref{gua:sec:paramsynt},
we repeat it here for completeness.
A \emph{cutoff} for parameterized specification $\tpl{\Ipr,\Iglob,\Opr,\Phi}$ 
and process template $P=(\Ipr,\Opr,Q,Q_0,\delta,out)$ is a number $c \in \bbN$ such that
$$
\forall n \geq c. \big(P^{\pring(c)} \models \Phi ~\Iff~ P^{\pring(n)} \models \Phi\big).
$$

Cutoffs reduce the parameterized synthesis and model checking problems to
their non-parameterized variants.
E.g., if the cutoff is $2$
then the answer to the parameterized$_{> 2}$ model checking problem
``$\forall n>2: P^{\pring(n)} \models \Phi$''
is the same as the answer to the non-parameterized model checking problem
``$P^{\pring(2)} \models \Phi$''.

\subsection*{Known Cutoffs}

  In a seminal paper~\cite{Emerso95b,Emerso03} Emerson and Namjoshi proved the following cutoff results.
  \begin{theorem}[\cite{Emerso03}] \label{tok_rings:thm:cutoffs}
    Let $P=(\Ipr,\Opr,Q,Q_0,\delta,out,A_{loc})$ be a process template,
    $\Iglob=\emptyset$ (no global inputs),
    $\tpl{\Ipr,\Opr,\Phi}$ a parameterized specification.
    Assume that the scheduler is interleaving.
    Then $c$ is a cutoff depending on $\Phi$:
    \li
    \- $c=2$ for $\forall i.~ \phi(i)$,
    \- $c=3$ for $\forall i.\forall j_{j=i+1}.~ \phi(i,i+1)$,
    \- $c=4$ for $\forall i.\forall j_{i\neq j}.~ \phi(i,j)$,
    \- $c=5$ for $\forall i.\forall j_{i\neq j}.\forall k_{k=i+1}.~ \phi(i,i+1,j)$.
    \il
  \end{theorem}

  The above cutoff results are restricted to token-ring architectures
  and do not allow for specifications of the more general $k$-indexed form.
  Later in~\cite{AJKR14} we extended the results to more general networks (directed graphs),
  where the processes can control the directions in which to send and receive the token,
  and systems can pass more than one token.
  The paper also studied $k$-indexed \CTLstar properties,
  also with a bounded alternation depth of path quantifiers.

\section{Bounded Synthesis of Parameterized Token Rings}\label{tok_rings:sec:bs-and-optimizations}

\subsection{SMT Encoding}

We encourage the reader to revisit Chapter~\ref{defs:bounded_synthesis} on page~\pageref{page:defs:bounded_synthesis}
to recall how bounded synthesis works in the case of non-distributed systems.
We adapt the encoding to the case of (distributed) token ring systems as follows.

Let us start with SMT constraint about a process template.

\parbf{SMT constraints for a process template}
Let us encode the definition of process template from
Section~\ref{tok_rings:defs:system} on page~\pageref{tok_rings:eq:process-trans}:
%
%\footnote{We could also encode this into LTL,
%  but using SMT constraints is more efficient.%
%}:
\li
\- Introduce a special output $\token: T \to \bbB$
   such that $\token(t)$ holds iff $t \in T$
   (recall that we divide the states $Q = T\cupdot NT$).
   Let us encode Eq.~\ref{tok_rings:eq:process-trans} (on page~\pageref{tok_rings:eq:process-trans}),
   which specifies:
   a process template can send the token only if it has the token;
   sending the token means a process template loses the token;
   if a process template receives the token and currently does not have it,
   then it has the token after the transition.
   \begin{equation} \label{tok_rings:eq:snd_rcv_tok}
   \forall_i. \G\left[
   \begin{array}{l}
     \tsnd_i \impl \token_i \\
     \token_i \impl (\tsnd_i \iff \X\neg\token_i)\\
     \neg\token_i \impl (\trcv_i \iff \X\token_i)
   \end{array}
   \right]
   \end{equation}
\- What is left is the condition $(\dagger)$ from the process template definition.
   We introduce the following LTL formula:
   \begin{equation}\label{tok_rings:eq:tok_release}
   \forall i.~(A_{loc})_i \impl \G(\token_i \impl \F \tsnd_i),
   \end{equation}
   i.e., a process does not lock the token if the fairness condition $A_{loc}$ is satisfied.

\il

\parbf{SMT constraints for a system}
Now let us encode particularities of (distributed) token-ring systems.

\li

\- We compose the system transition function out of process transition functions.
   Note that all processes share the \emph{same} transition function;
   the input arguments to the function reflect for what process it is used.
   To account for scheduling,
   we introduce additional system inputs $\sch_1, ..., \sch_k$
   (where $k$ is the number of processes in a ring),
   and require that a process $i \in \{1,...,k\}$ can transit only when $\sch_i$ is true
   (and hence a process does not see its inputs when it is not scheduled).

\- The scheduler model (asynchronous/synchronous/interleaving) defines the constraints
   on the scheduling variables $\sch_1,...,\sch_k$.
   For synchronous token rings, all scheduling variables are set to true.
   For interleaving scheduling, exactly one of the scheduling variables is set to true,
   except for the token-passing transitions where the two processes transit simultaneously.
   For asynchronous scheduling, any number (including zero) of the scheduling variables can be true.
   To specify fair scheduling (for the interleaving or asynchronous cases),
   we use the constraint $\bigwedge_i \GF sch_i$.
   This constraint is added as the assumption to the original formula,
   when we translate the formula into an automaton.
   \ak{why it does not break prenex-indexed formulas?}

\- To ensure that the topology is the token ring
   (where every process sends the token to its single neighbor),
   we manipulate process input $\tsnd$ and output $\trcv$ in the natural way.
   For example,
   if process $i$ is ready to send the token, i.e., it is in a state $t$ and $\tsnd(t)$ holds,
   then once it is scheduled we set $\trcv_{i+1}$ to true.
   I.e., $\tsnd_{i}$ is connected to $\trcv_{i+1}$:
   \begin{equation*}
   \G [\tsnd_i \iff \trcv_{i+1}]
   \end{equation*}
\il

Thus, given an LTL formula $\varphi$,
we want to synthesize a token-ring system that satisfies:
\begin{equation}\label{tok_rings:eq:full_formula}
\boxed{
\forall i. (\GF\sch_i \land \G(\tsnd_i \iff \trcv_{i+1}))
~\impl~
  \varphi \land \forall_i\left(
   \begin{array}{l}
     \G [\tsnd_i \impl \token_i] \\
     \G [\token_i \impl (\tsnd_i \iff \X\neg\token_i)] \\
     \G [\neg\token_i \impl (\trcv_i \iff \X\token_i)] \\
     (A_{loc})_i \impl \G(\token_i \impl \F \tsnd_i)
   \end{array}
   \right)
}
\end{equation}

\begin{example}\label{tok_rings:ex:simple_arb}
Consider a specification of a simple arbiter.
A process template has inputs $I=\{r,\trcv\}$, outputs $O=\{g,\tsnd\}$,
the original parameterized LTL formula specifying the arbiter is:
$$
\begin{array}{rl}
  \forall i \neq j. & \G \neg ( g_i \land g_j ) \\
         \forall i. & \G (r_i \impl \F g_i) \land \neg g_i.
  \end{array}
$$
By Theorem~\ref{tok_rings:thm:cutoffs}, the cutoff is 4.
We set $A_{loc}=\true$, instantiate the above formula, and synthesise a token-ring system.
The process synthesised using our tool PARTY~\cite{party} is in Figure~\ref{tok_rings:fig:simple_arb}.

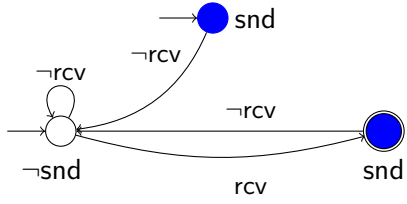
\begin{figure}[tb]\center
\begin{tikzpicture}
	\begin{pgfonlayer}{nodelayer}
		\node [style=wn, fill=blue, double, minimum size=0.5cm] (0) at (3.5, -0) {};
		\node [initial, style=wn] (1) at (-0.75, -0) {};
		\node [initial, style=wn, color=blue] (2) at (1.25, 1.5) {};
		\node [style=textual] (3) at (1.75, 1.5) {$~\tsnd$};
		\node [style=textual] (4) at (3.5, -0.5) {$\tsnd$};
		\node [style=textual] (5) at (-1, -0.5) {$~~\neg\tsnd$};
		\node [style=textual] (6) at (-0.75, 0.75) {\footnotesize$\neg\trcv$};
		\node [style=textual] (7) at (1.75, 0.25) {\footnotesize$\neg\trcv$};
		\node [style=textual] (8) at (1.75, -0.75) {\footnotesize$\trcv$};
		\node [style=textual] (9) at (0.5, 1) {\footnotesize$\neg\trcv$};
	\end{pgfonlayer}
	\begin{pgfonlayer}{edgelayer}
		\draw [style=arrow, in=120, out=60, loop] (1) to ();
		\draw [style=arrow, bend right=15, looseness=1.00] (1) to (0);
		\draw [style=arrow] (0) to (1);
		\draw [style=arrow, bend left, looseness=1.00] (2) to (1);
	\end{pgfonlayer}
\end{tikzpicture}
\caption{Process template synthesized from the specification of a simple arbiter
 (Example~\ref{tok_rings:ex:simple_arb}).
 There are two initial states, with and without the token.
 The blue-filled states have the token,
 the double state has $g$, the others have $\neg g$.
 The process template grants whenever it has the token (except for the initial state),
 ignoring the request.
 The exclusivity of the token ensures the mutual exclusion of the grants.}
\label{tok_rings:fig:simple_arb}
\end{figure}
\end{example}
\ak{if you call it SMT encoding, then you need to provide an example of the constraints}

\subsection{Optimizations} \label{tok_rings:sec:optimizations}

In this section we describe high-level optimizations that are not specific to the SMT encoding.
The first two optimizations, \emph{incremental solving} and \emph{modular generation of constraints},
are sound and complete.
The third, \emph{specification strengthening},
is based on automatic rewriting of the specification and introduces incompleteness.
The last optimization, \emph{hub-abstraction} is sound and complete.

\subsubsection{Incremental Solving}

Theorem~\ref{tok_rings:thm:cutoffs} states that it is sufficient to synthesize a token ring of cutoff size $c$.
However, a solution for a smaller number of processes can still be correct in bigger rings.
We propose to proceed incrementally, synthesizing first a ring of size 1, then 2, ..., up to $c$.
After synthesizing a process that works in a ring of size $n$,
we check whether it satisfies the specification also in a ring of size $n+1$.
Only if the result is negative,
we start to synthesize a ring of size $n+1$.

\subsubsection{Modular Constraints for Conjunctive Properties}

A useful property of the SMT encoding for parameterized synthesis is that
we can separate conjunctive specifications into their parts,
generate constraints for the parts separately,
and then search for a solution that satisfies the conjunction of all constraints.
In the following,
for a parameterized specification $\varphi$ and a number of processes $k$,
let $C(\varphi,k)$ be the set of SMT constraints generated by the bounded synthesis procedure.
Note that $C(\varphi,k)$ is of the form $\exists P (...)$.
When a process template $P$ is given,
let ``$P \models C(\varphi,k)$'' mean that the constraints $C(\varphi,k)$ are satisfied
when instantiated with the process $P$.
\begin{theorem}
Let $\varphi_1$ and $\varphi_2$ be prenex-indexed formulas
such that $n_1$ is a cutoff for $\varphi_1$ and $n_2$ is a cutoff for $\varphi_2$.
Then:
$$
P \models C(\varphi_1,n_1) \land C(\varphi_2,n_2) ~~\Impl~~
P^{\pring(k)} \models \varphi_1 \land \varphi_2 \textit{~~for every~} k \geq max(n_1,n_2).
$$
\end{theorem}

The theorem allows us to use different cutoffs for sub-parts of a formula.
By conjoining the resulting constraints of all parts,
we obtain an SMT problem such that every solution satisfies the complete formula.
For example, for a formula
$$
\begin{array}{ll}%
\forall i \neq j.~ & \G \neg ( g_i \land g_j ) \\
\forall i.~ & \G (r_i \impl \F g_i),
\end{array}
$$
we generate constraints for a ring of size $4$ for the first conjunct,
and we generated constraints for a ring of size $2$ for the second conjunct.
This is useful for formulas where the local (1-indexed) part is more complex than the global part,
like our more complex arbiter examples.

\subsubsection{Specification Strengthening and Handling Assumptions} \label{tok_rings:sec:localising}

To handle specifications in assume-guarantee style,
we strengthen them in two rewriting steps,
which are sound but incomplete.
This turns them into the prenex-indexed form
(the only form, for which we know how to do parameterized synthesis).

Consider a formula in assume-guarantee style $A_L \land A_S \impl G_L \land G_S$,
where each of the conjuncts is in the prenex-indexed form,
and $L$ and $S$ denote respectively liveness and safety.
Notice that this formula, as a whole,
  is \emph{not} in prenex-indexed form,
  since it contains process quantifiers inside the path quantifier $\A$
  (if written explicitly, it says $\pforall (\forall i... \impl \forall j...)$,
  which is 2-indexed but not \emph{prenex}-indexed).

\parbf{Safety-liveness assumptions}
Our first strengthening is based on the intuition that
often $A_L$ is not needed to obtain $G_S$,
so we strengthen the formula to $(A_S \impl G_S) \land (A_L \land A_S \impl G_L)$.
This step is incomplete for specifications where the system can
falsify liveness assumptions $A_L$ and therefore ignore guarantees,
or if the assumptions $A_S \land A_L$ are unrealizable but $A_S$ is realizable.
Both of the cases often hint at the problems with the specification%
\footnote{The well known class of GR1 specifications~\cite{Bloem12},
  which can be used to describe industrial systems,
  does not use liveness assumptions for safety guarantees.
  Furthermore, for GR1 specifications Klein and Pnueli~\cite{Klein10}
  describe a similar separation of safety guarantees from liveness assumptions.
  They introduce ``well separated'' assumptions,
  which are such that the system cannot falsify them at any state,
  and show that ``well separation'' of assumptions is sufficient for the rewriting to be sound.
  Incomplete cases represent specifications where the system can falsify assumptions and ignore guarantees.%
  }.

\parbf{Localizing assumptions}
Consider a $2$-indexed formula in assume-guarantee style,
$\forall_i A_i \impl \forall_j G_j$,
where $A_i$ and $G_j$ refer to process $i$ and $j$ respectively.
Originally,
we want to plug this formula into Eq.~\ref{tok_rings:eq:full_formula}
and synthesize for the resulting formula.
Instead,
we localize it---turn $(\forall i...) \impl (\forall j...)$ into $\forall i (... \impl ...)$---and get:
\begin{equation}\label{tok_rings:eq:localised}
\forall i:~ \big(\!\GF\sch_i \land \G(\tsnd_i \iff \trcv_{i+1})\big)
\impl
  \left(
   \begin{array}{l}
     \G [\tsnd_i \impl \token_i] \\
     \G [\token_i \impl (\tsnd_i \iff \X\neg\token_i)] \\
     \G [\neg\token_i \impl (\trcv_i \iff \X\token_i)] \\
     A_i \impl \G(\token_i \impl \F \tsnd_i) \\
     A_i \land \GF \token_i \impl G_i
   \end{array}
   \right)
\end{equation}
A few notes:
\li

\- This formula implies the original formula Eq.\ref{tok_rings:eq:full_formula}
   where we set $\varphi = \forall_i A_i \impl \forall_i G_i$ and $A_{loc}=A_i$.
   Setting $A_{loc} = A_i$---%
   requiring $\G(\token_i \impl \F \tsnd_i)$ to hold under the assumption $A_i$---%
   is reasonable:
   it says that if the environment violates the assumption $A_i$,
   then we are not required to release the token.
   Note that if token releasing is required despite $A_i$,
   then the rewriting is unsound
   (it may result in incorrect solutions wrt.\ Eq.\ref{tok_rings:eq:full_formula}).

\- In this formula,
   the non-technical part (where the technical part encodes the token-ring properties)
   is in the \emph{prenex}-indexed fragment.
   Indeed, the non-technical part corresponds to $\forall i.~(A_i \land \GF\token_i \impl G_i)$,
   which is prenex 1-indexed LTL formula.
   In contrast, the original formula $\forall i.A_i \impl \forall j.G_j$ is \emph{not} prenex-indexed,
   because it corresponds to $\A(\forall i.A_i \impl \forall j.G_j)$,
   if we explicitly write the path quantifier $\A$.
   Hence for the new formula we can use the cutoff results of Theorem~\ref{tok_rings:thm:cutoffs},
   but we could not for the original one.

\- Adding $\forall_i \GF\token_i$ to the first constraint is crucial.
   Otherwise,
   the final formula becomes too restrictive and we may miss solutions.
   The reason why $\G\F\token_i$ may prevent this is that
   $\G\F\token_i$ may work as a local trigger of a violation of an assumption.
   This is confirmed in the ``Pnueli'' arbiter experiment,
   where a violation of one of the assumptions $A_i$ prevents fair token passing in the ring,
   falsifying $\G\F\token_j$ for all $j\neq i$.

\- Filiot et al.~\cite{Filiot11} describe a similar rewriting heuristic,
   in the context of monolithic synthesis.
   Our version differs in that we add $\GF\token_i$ assumptions
   before localization to prevent missing the solutions.

\il

\subsubsection{Hub-abstraction}

  Inspired by the work~\cite{Clarke04c}, we introduce the hub abstraction optimization.
  Recall that for 1-indexed properties $\forall i.\varphi(i)$, a cutoff is 2,
  meaning that it is enough to consider a token ring system with two processes.
  Furthermore, by symmetry of the processes,
  it holds that
  $P^{\pring(2)}\models \forall i. \varphi(i) \Iff P^{\pring(2)} \models \varphi(1)$
  (for details, see~\cite{Emerso03}).
  The hub abstraction suggests to replace the process $P_2$ of a system with the hub process
  whose whole purpose is to pass the token.
  This may reduce the state space,
  because we replace the original process $P_2$ by the small hub process.
  We emulate the hub process using the environment assumptions,
  thus considering only one real process.
  The assumptions are:
  \li
  \-[(1)]
     if the process does not have the token,
     then the environment eventually sends the token (raises the input $\trcv$):
     $\G ( \neg \token \impl \F \trcv)$,

  \-[(2)]
     if the process has the token, then the environment does not send the token:
     $\G ( \token \impl \neg\trcv )$.
  \il
  The final formula to synthesize is:
  \begin{equation}\label{tok_rings:eq:hub-abstraction}
  \big(\!\GF\sch \land \G(\neg\token \impl \F\trcv) \land \G(\token \impl \neg\trcv)\big)
  ~\impl~
    \varphi \land \left(
     \begin{array}{l}
       \G [\tsnd \impl \token] \\
       \G [\token \impl (\tsnd \iff \X\neg\token)] \\
       \G [\neg\token \impl (\trcv \iff \X\token)] \\
       A_{loc} \impl \G(\token \impl \F \tsnd)
     \end{array}
     \right)
  \end{equation}

  Note that the assumption (1) states that the token cannot get stuck
  in the hub process (and thus in the original process that the hub abstracts).
  This does not always hold,
  because we only require to pass the token if $A_{loc}$ holds.
  This means that the hub-abstraction is not sound wrt.\ Eq.\ref{tok_rings:eq:full_formula},
  i.e.,
  there is a process $P$ (and $A_{loc}$) and $\forall i.\varphi(i)$ such that
  $P \models \text{Eq.\ref{tok_rings:eq:hub-abstraction}}$
  but $P^{\pring(2)} \not\models \text{Eq.\ref{tok_rings:eq:full_formula}}$.

  However, we will use the hub abstraction in the context of assume-guarantee 1-indexed specifications
  in the form of Eq.\ref{tok_rings:eq:localised} (``$\forall i.A_i \land \GF\token_i \impl G_i$'').
  Since Eq.\ref{tok_rings:eq:localised} only requires the guarantee to hold on paths
  where the token is passed infinitely often,
  the following result holds.

  \begin{theorem}
    For every process template $P$
    and LTL formula $\varphi$ of the form $A \impl G$:
    $$
    P \models \text{Eq.\ref{tok_rings:eq:hub-abstraction} (where $A_{loc}=A$)}
    ~~\Impl~~
    P^{\pring(2)} \models \text{Eq.\ref{tok_rings:eq:localised}}.
    $$
  \end{theorem}

Furthermore, we can replace $\GF\sch$ with \emph{true}, which can introduce unsoundness wrt.\ Eq.\ref{tok_rings:eq:localised}.
But this step is sound for formulas where the environment cannot
violate guarantees by not scheduling a process.
This is true for all examples we consider in the next section.

\subsection{Evaluating Optimizations}\label{tok_rings:optimisations:experiments}

For the evaluation of optimizations we developed an automatic
parameterized synthesis tool PARTY~\cite{party}.
The tool and the benchmarks are available at {\small\url{https://github.com/5nizza/Party/}}.
PARTY
\li
\-[(1)] identifies the cutoff of a given LTL specification,
\-[(2)] adds token-ring specific guarantees and assumptions to the specification,
\-[(3)] translates the modified specification into a UCT using LTL3BA~\cite{LTL3BA},
\-[(4)] for a given cutoff and system size bound, builds the SMT constraints,
\-[(5)] solves the constraints using SMT solver Z3 v.4.1~\cite{Moura08}.
        If the solver reports unsatisfiability,
        then no model for the current bound exists,
        and the tool goes to step 4 and increases the bound until
        the user interrupts execution or a model is found.
        A model synthesized represents a Moore machine
        that can be copied to form a token-ring system of any size.
\il

We run the experiments on a single core of a Linux machine
with two 4-core 2.66 GHz Intel Xeon processors and 64~GB RAM.
Reported times in tables include all the steps of the tool.
For long running examples, SMT solving contributes most of the time. 
Timings reported in tables are timings of one particular run,
although we observed that the behaviour of optimizations timings does not change much on different runs.

For the evaluation of optimizations we run the tool,
with different sets of optimizations enabled,
on three examples:
a simple arbiter, a full arbiter, and a ``Pnueli'' arbiter.
All benchmarks contain the mutual exclusion property $\forall i \neq j. \G\neg g_i \land g_j$,
for which a cutoff is $4$ according to Theorem~\ref{tok_rings:thm:cutoffs}.
We show solving times in Table~\ref{tok_rings:tab:opt}.
The horizontal axis of the table has columns for token rings of different sizes.
Each successive optimization below includes previous optimizations.

\parbf{Incremental solving}
% \todo{redo experiment: check solution for ring 6}
% LTL3BA fails to produce automaton for rings of size 6 in a reasonable time
Solving times can be sped up considerably by synthesizing a ring of size $2$,
then checking whether the solution is correct for a ring of size $4$.
For instance, for the full arbiter, the general solution was found in $\approx\!24$ seconds
when synthesizing a ring of size $2$ (time from the ``original'' row in Table~\ref{tok_rings:tab:opt}).
Checking if the solution is correct for a ring of size $4$ takes additional $\approx\!30$ seconds,
thus reducing the synthesis time from more than 2 hours (column ``full4'' in the same row)
to $\approx\!54$ seconds.
Times for incremental solving are not given in the table,
because its contribution is small when optimizations ``strengthening'', ``modular'', and ``async hub'' are applied.

\parbf{Strengthening}
This version refers to two optimizations described in Section~\ref{tok_rings:sec:optimizations}:
localizing of assume-guarantee properties and
rewriting liveness assumptions from properties with safety guarantees.
Formula rewriting significantly reduces the size of the automaton:
for example, the automaton corresponding to the ``Pnueli'' arbiter in a ring of size 4
reduces its size from 1700 to 31 states
(from 41 to 16 for the full arbiter).

\parbf{Modular}
In this version,
constraints for formulas of the form $\phi_i \land \phi_{i,j}$ are
generated separately for local properties $\phi_i$ and for global properties $\phi_{i,j}$,
using the same symbols for transition and output functions.
Constraints for $\phi_i$ are generated for a ring of size 2,
and constraints for $\phi_{i,j}$ for a ring of size 4.
These sets of constraints are then conjoined in one query and given to the SMT solver.
Such separate generation of constraints leads to smaller automata and queries,
resulting in approximately $10$x speed up.

\parbf{Hub abstractions}
By replacing one of the processes in a ring of size 2 with assumptions on its behavior,
we reduce the synthesis of a ring of size two to the synthesis of a single process.
In row ``async hub'' the process is synthesized in an asynchronous setting,
while in row ``sync hub'' the process is assumed to be always scheduled.
On these examples, the speed up is insignificant.

\begin{table}[tb]
\caption{Effect of optimizations on synthesis time (in seconds, t/o=2h)}
\label{tok_rings:tab:opt}
\small
\centering
\setlength{\tabcolsep}{3pt}
\begin{tabular}{ lrrrrrrrrr }
\toprule
 & simple4 & full2 & full3 & full4 & pnueli2 & pnueli3 & pnueli4 & pnueli5 & pnueli6 \\
\midrule
% some older results
original          & 3 & 24   & 934 & t/o & 23 & 6737 & t/o & t/o & t/o\\
% results_Oct19_11_28/
strengthening    & 1 & 6 & 81 & 638 & 2 & 13 & 90 & 620 & 6375 \\
% results_Oct19_11_28/
modular           & 1 & 4 & 8 & 13 & 2 & 4 & 11 & 49 & 262 \\
% results_Oct19_11_28/
async hub               & 1 & 2 & 2 & 5 & 2 & 3 & 9 & 37 & 236 \\
% results_Oct19_11_28/
sync hub                & 1 & 1 & 2 & 4 & 2 & 3 & 8 & 42 & 191 \\
\midrule
total speedup      & $3$ & $20$ & $10^2$ & $\ge\!\!10^3$ & $10$ & $10^3$ & $\ge\!\!10^3$ & $\ge\!\!10^2$ & $\ge\!\!40$ \\
\bottomrule
\end{tabular}
\end{table}

\subsection{Discussion}

We showed how optimizations of the SMT encoding, along with modular application of cutoff results, strengthening and abstraction techniques, leads to a significant speed-up of parameterized synthesis. Experimental results show speed-ups of more than three orders of magnitude for some examples.

%The current bottleneck of SMT-based bounded (and thus, parameterized) synthesis is the construction of the UCT automaton. In our experiments, LTL3BA could not generate the UCT for an AMBA arbiter with only 1 client within two hours. Therefore, we think that it will be important to develop techniques that help us to avoid construction of the whole automaton (for example by separate tracking of assumptions and guarantees violations, as in~\cite{Ehlers12}).

In the next section,
we use these optimizations to tackle AMBA specification.
This will not work out of the box and we will introduce more tricks specifically tailored to the AMBA.

\section{AMBA Protocol Case Study}\label{amba:sec}

We demonstrate how to synthesize a parameterized implementation of the AMBA AHB,
with guaranteed correctness for any number of masters.
To this end, we translate the LTL specification of the 
AMBA AHB (as found in~\cite{BarbaraThesis}) into a version that is suitable 
for parameterized synthesis in token rings, and address several challenges 
with respect to theoretical applicability and practical feasibility:
\li
\- We show how to localize global input and output signals (those that cannot be assigned to one particular master).
   This is necessary since our approach is based on the replication of components that act only on local information.
\- We extend the cutoff results to fully asynchronous timing model and systems with two process templates.
\- We describe further optimizations that make synthesis feasible,
   in particular based on the insight that the AMBA protocol features three 
   different types of accesses, and the control structures for these accesses 
   can be synthesized step-by-step.
\il

\subsection{Description of the AMBA Protocol} \label{sec:amba}

ARM’s \emph{Advanced Microcontroller Bus Architecture} (AMBA)~\cite{AMBAspec} 
is a communication bus for a number of masters and
clients on a microchip. One of the crucial parts of AMBA is
the \emph{Advanced High-performance Bus} (AHB),
a system bus for the efficient connection of processors, memory, and devices.

For convenience, the input signals are depicted in \textcolor{red}{red} color
and the outputs are \textcolor{blue}{blue}.

The bus arbiter ensures that only one master accesses the bus at any time.
Masters send \hbusreq\ to the arbiter if they want access, and receive
\hgrant\ if they are allowed to access it.
Masters can also ask for different
kinds of \emph{locked transfers} that cannot be interrupted.

The exact arbitration 
protocol for AMBA is not specified. Our goal is to synthesize a protocol 
that guarantees safety and liveness properties. According to the 
specification, any device that is connected to the bus will react to an input 
with a delay of one time step. I.e., we are considering Moore machines. In 
the following, we introduce briefly which signals are used to
realize the arbiter of this bus for masters.

\parbf{Requests and grants}
The identifier of the master which is 
currently active is stored in the $n+1$-bit signal \hmaster[$n$:0], 
with $n$ chosen such that the number of masters fit into $n+1$ bits.
To request the 
bus, master $i$ raises signal \hbusreqi. The arbiter decides who will be 
granted the bus next by raising signal \hgranti. When
the client raises \hready, the bus access starts at the next tick, and there 
is an update \hmaster[$n$:0] := i, where \hgranti\ is currently active.

\parbf{Locks and bursts}
A master can request a locked access by raising both \hbusreqi\ and \hlocki. In this case, 
the master additionally sets \hburst[1:0] to either \ambasignal{single} 
(single cycle access), \hburstfour\ (four cycle burst) or \hincr\ (unspecified 
length burst). For a \hburstfour\ access, the bus is
locked until the client has accepted 4 inputs from the master (each signaled by
raising \hready). In case of a \hincr\ access, the bus is
locked until \hbusreqi\ is lowered. The arbiter raises \hmastlock\ if the bus 
is currently locked.

\parbf{LTL specification}\ak{where $\G \neg g_1 \land g_2$?}
The original natural-language specification~\cite{AMBAspec} has been translated into a formal specification in the GR(1) fragment of LTL before in~\cite{BarbaraThesis,Bloem12,GodhalCH13}.
Figure~\ref{amba:fig:AMBAspec} shows the environment assumptions and system guarantees from~\cite{BarbaraThesis} that serve as the basis for our parameterized specification.
The full specification is 
$ (A1 \land \ldots \land A4) \rightarrow (G1 \land \ldots \land G11)$.
\begin{figure}[t]
  \fbox{%
  \begin{minipage}{\textwidth}
  \input{token-systems/amba-spec}
  \end{minipage}}
\caption{Specification of the AMBA AHB~\cite{BarbaraThesis}, in the GR(1) fragment of LTL.
  The inputs are:
  \hburst, \hbusreq[i], \hready, and \hlock[i].
  The outputs are:
  \hmastlock, \hmaster, \hstart, \hdecide, \hlocked, and \hgrant[i].}
\label{amba:fig:AMBAspec}
\end{figure}

\parbf{Challenges}
The AMBA specification has global inputs and outputs (those are without ``$[i]$''),
distinguishes 0 from non-0 processes (G10.1 and G11),
has the assume-guarantee form $\A(\forall i. A1\land...\land A4 \impl \forall i.G1\land...\land G11)$
(thus not in the prenex-indexed form),
has the process quantification inside a temporal operator in G10.2 $\G(\forall i... \impl ...)$,
and requires a synchronous mode of execution (all processes transit simultaneously).
Thus we cannot apply the cutoff results (Theorem~\ref{tok_rings:thm:cutoffs} on page~\pageref{tok_rings:thm:cutoffs})
for parameterized synthesis.
The next section shows how to handle this.

% We will show how we can use our approach for parameterized synthesis in token 
% rings for synthesizing component implementations for the AMBA AHB 
% specification, such that every component serves the requests of one master, 
% and the composition of an arbitrary number of components is correct by 
% construction. To illustrate our solution, we first review the parameterized 
% synthesis approach, and then mention a number of problems that need to be 
% solved for the AMBA case study.

%\subsection{Additional Definitions}
%

\subsection{Handling the AMBA Specification}\label{amba:sec:handling-amba}

  This section shows how to rewrite the AMBA specification into a form
  admissible to the parameterized synthesis.
  We not only rewrite the specification,
  but also extend the cutoff results~\cite{Emerso03} to the resulting class of specifications.
  Note that the resulting specification is not the same as the original AMBA (but closely resembles it),
  due to constraints of the token-ring architecture.
  (For example, token rings cannot ensure immediate granting of a client,
   because the token has to travel to the corresponding process first.)
  The resulting specification describes a round-robin arbiter
  with different granting schemes and one special process.

%%%%The AMBA specification (Figure~\ref{amba:fig:AMBAspec}) has the form
%%%%$
%%%%(A1 \land \ldots \land A4) \impl (G1 \land \ldots \land G11).
%%%%$
%%%%The specification poses
%%%%the following challenges to existing parameterized synthesis approach~\cite{JB12}
%%%%and cutoff results~\cite{Emerso03,AJKR14}:

\subsubsection{Special 0-process: two process templates $(A,B)$}

  The specification distinguishes between master number 0 and all other masters.
  We support this by synthesizing two different process implementations,
  $A$ for the 0-process and $B$ for non-0 processes:
  the $A$-process serves master 0 and the $B$-processes serve the other masters.
  We denote a token-ring system composed of one $A$ process and $n$ copies of $B$
  using the notation $(A,B)^{(1,n)}$.
  The modified parameterized synthesis problem is to find $(A,B)$ such that
  $\forall n: (A,B)^{(1,n)} \models \Phi$.
  Later we will separate the specification into two parts:
  one will talk about process $A$, another will talk about $B$-processes.

\subsubsection{Localizing global outputs}

  The AMBA specification has global outputs \hmastlock, \hmaster, \hstart, \hdecide, and \hlocked.
  They depend on the global state of the system,
  which is not handled by the work on parameterized model checking of token ring systems~\cite{Emerso03,AJKR14}.
  To overcome this,
  we introduce local versions of the global outputs and build global outputs from them:
  \li
  \- $\hmaster = i$ whenever $\hmasteri$ is high, and 
  \- for every global output $glob$ from $\{\hmastlock, \hstart, \hdecide, \hlocked\}$,
     $glob = \exists{i}.\ \toki \land glob_i$.
  \il
    
  We replace each global output with its local version, e.g., $\hstart$ is replaced by $\hstarti$.
  Note that the limited communication interface (via token passing)
  does not make AMBA specification unrealizable,
  although processes cannot access the value of global outputs when they do not possess the token.
  Intuitively, this is because the token is the shared resource that guarantees mutual exclusion of grants, and therefore the values of these global signals should always be controlled by the process that has the token. In particular, outputs \hdecide\ and \hstart\ are used to decide when to raise a grant and when to start and end a bus access\footnote{The original AMBA specification~\cite{AMBAspec} does not have these signals---they were introduced to simplify the formalization of the specification~\cite{BarbaraThesis}.}, which should only be done when the token is present.
  Similarly, signals \hmastlock\ and \hmaster\ should be controlled by the process that currently controls the bus (and hence has the token). 

  Finally, we mentioned many times that the token should be used to ensure the mutual exclusion of grants.
  Let us explicitly add this requirement into the specification,
  namely we add G12: $\forall i.\,\hgranti \rightarrow \toki$.
  (The original formula contains only an \emph{implicit} mutual exclusion property:
   G4 defines how \hmaster\ is updated by the \hgranti\ signals,
   which can only be satisfied if \hgranti\ are mutually exclusive.)

\subsubsection{Splitting the specification into two \& other small rewritings}

  Once we localized global outputs,
  we can talk about splitting the AMBA specification into two parts.
  At first, each part will be in the assume-guarantee form,
  where the assumptions talk about all the processes (the only $A$ and all $B$),
  but the guarantees will be separated into
  (i) guarantees for the $B$-processes and
  (ii) guarantees for the process $A$.

  After the localization, guarantees G10.1 and G10.2 become:
  \[\begin{array}{rrlr}
  \forall \mathrm i\neq 0\!:\!\!\!\!\! & \always & \neg \hgranti \impl (\neg \hgranti \weakuntil \hbusreqi) & (G10.1)\\[3pt]
  & \always & (\hdecide[0] \land \forall \mathrm i.\neg \hbusreqi) \impl \nextt \hgrant[0] & (G10.2) \\[3pt]
  \end{array}\]
  Thus, G10.1 is used for $B$-processes,
  while G10.2 is used for the process $A$.
  Let us talk more about G10.2, because it has two issues.

  The first issue with G10.2 is that it requires an immediate reaction to a situation when no process receives a bus request. 
  This is unrealizable in token rings,
  because mutual exclusion of the grants requires possession of the token,
  and the token transmission takes time.
  We modify G10.2 to allow the process $A$ to wait for the token and then immediately react:
  $$
  \G
  \big(
  \neg \tok[0] \land \X \tok[0] \land \forall \mathrm i. \neg \hbusreqi
  \impl
  \nextt \hgrant[0]
  \big).
  \footnote{Maybe you expected to see
    $\always \ (\hdecide[0] \land (\forall \mathrm i: \neg \hbusreqi) \land \nextt \tok[0])
    \impl \nextt \hgrant[0])$, 
    but this makes the specification unrealizable due to the token ring requirement 
    $\always (\tok \impl \eventually \send)$: 
    the environment falsifies it by making $\forall \mathrm i: \neg \hbusreqi$ true whenever the process raises $\hdecide$ (and hence the process should continue granting and cannot release the token).
    To regain realizability one could add additional assumptions (something like $\G\F (\hdecide\land\neg\norequests)$). 
    Instead, we decided to change slightly the specification.%
  }
  $$

  The second issue with G10.2 is the quantifier $\forall{i}$ \emph{inside} the temporal operator $\G$
  (such specifications were not studied in parameterized model checking of token rings).
  It requires the process $A$ to know about inputs of all $B$-processes,
  as it needs to react to a situation where \hbusreqi\ is low for every process.
  To get rid of the nesting $\G(\forall i...)$,
  we introduce a new global input $\norequests$,
  and add the assumption $\forall{i}.\G (\hbusreqi \impl \neg \norequests)$.
  Then G10.2 becomes:
  $$
  \G \ 
  (\neg \tok[0] \land \X \tok[0] \land \norequests)
  \impl \X \hgrant[0]).
  $$
  This strengthens the specification,
  because the environment can set $\neg\norequests$
  even when there are no requests.
  This concludes the discussion of G10.2.

  The last asymmetric property is the guarantee G11. We split it into two parts:
  \li
  \- G11.1: $\neg \hgranti \land \neg \hmastlocki$ for $B$-processes and
  \- G11.2: $\tok[0] \impl \hgrant[0] \land \hmaster[0] \land \neg \hmastlock[0]$ (for $A$).
  \il
  % In the general case, if allow index quantifiers inside temporal operators, then parameterized model checking becomes undecidable even for systems with no synchronization at all~\cite[Appendix A]{Igor}.
  % However, for the special case of G10.2 their undecidability proof breaks.\ak{so what?}
  % since G10.2 contains only a single index quantifier inside the topmost $\always$ operator.

\subsubsection{Localizing global inputs}

  The AMBA specification in Figure~\ref{amba:fig:AMBAspec} uses global inputs \hburst,\hready,
  and \norequests\ that we introduced in the previous section.

  First, we introduce local versions $\hburst[i]$, $\hready[i]$, and $\norequests[i]$,
  and add the assumption $\forall i\neq j.\ loc_i = loc_j$ for $loc \in \{\hburst, \hready,\norequests\}$.
  This rewriting does not change the specification.
  The specification becomes
  \begin{align*}
  &\A\Big(\forall{i \neq j}_{i,j \in \{A,B_1,...,B_n\}}. \Phi({i,j}) ~\impl~ \forall{k \in \{B_1,...,B_n\}}. \Psi(k)\Big) ~\land\\
  &\A\Big(\forall{i \neq j}_{i,j \in \{A,B_1,...,B_n\}}. \Phi'({i,j}) ~\impl~ \Psi'(A)\Big),
  \end{align*}
  where each $\Phi(i,j)$ and $\Phi'(i,j)$ talk about propositions of processes $i$ and $j$,
  and $\Psi(k)$ and $\Psi'(A)$ talk about propositions of process $k$ and $A$ respectively.
  Note that $\Psi \neq \Psi'$ because we split the guarantees for $B$-processes and the process $A$
  (the assumptions also slightly differ, so we use $\Phi$ and $\Phi'$).
  Now the specification does neither have global inputs nor global outputs.

  Second, we drop the newly introduced assumptions.
  This means that the original global inputs $\hburst$, $\hready$, and $\norequests$ may have different values for different processes,
  i.e., they are not ``global'' anymore.
  This strengthens the specification,
  because dropping the assumptions enables more environment behaviors (and the formula is universal $\A(...)$).
  The resulting specification becomes
  \begin{equation}\label{eq:amba:assume-guarantee-spec}
  \begin{aligned}
  &\A\big(\forall{i \in \{A,B_1,...,B_n\}}. \Phi(i) ~\impl~ \forall{j \in \{B_1,...,B_n\}}. \Psi(j)\big) ~\land\\
  &\A\big(\forall{i \in \{A,B_1,...,B_n\}}. \Phi'(i) ~\impl~ \Psi'(A)\big).
  \end{aligned}
  \end{equation}
  Figures~\ref{fig:AMBASpecNewI}~and~\ref{fig:AMBASpecNew0} define $\Phi$, $\Psi$, $\Phi'$, and $\Psi'$.
  We stress that $\Phi(i)$ and $\Phi'(i)$
  has to be conjoined over all $B$-processes \emph{and} the process $A$ to form the assumptions.

  \begin{figure}%[t]
    \fbox{%
    \begin{minipage}{\textwidth}
    \input{token-systems/amba-spec-new}
    \end{minipage}
    }
    \caption{Parameterized AMBA specification for $B$-processes:
      assumptions $\Phi(i)$ and guarantees $\Psi(i)$.
      G10.2 is omitted, since it is only needed for the process $A$.%
    }
    \label{fig:AMBASpecNewI}
  \end{figure}
  \begin{figure}%[t]
    \fbox{%
    \begin{minipage}{\textwidth}
    \input{token-systems/amba-spec-new0}
    \end{minipage}}
    \caption{Parameterized AMBA specification for the process $A$:
      modifications wrt.\ Figure~\ref{fig:AMBASpecNewI}.
      The index $0$ denotes the process $A$.%
    }
    \label{fig:AMBASpecNew0}
  \end{figure}

\subsubsection{Resulting parameterized specification and cutoffs}

  We still cannot apply the cutoff results to the specification in Eq.\ref{eq:amba:assume-guarantee-spec},
  because it is in the assume-guarantee form (and thus is not prenex-indexed)
  and has the synchronous timing model.

  To handle the synchronous timing model, we synthesize a more general case of fully asynchronous systems
  (those work under all ranges of schedulers from the synchronous to the interleaving one).
  This represents a more difficult synthesis task,
  but if the synthesizer finds such a system,
  then the system works in the synchronous setting too
  (because we have universal properties).

  To handle the assume-guarantee issue,
  we localize the assumptions as described in Eq.\ref{tok_rings:eq:localised} on page~\pageref{tok_rings:eq:localised},
  by strengthening
  $\A(\forall i. \phi(i) \impl \forall j. \psi(j))$ into
  $\A\big(\forall i. (\phi(i)\land\GF\token_i \impl \psi(i))\big)$
  where $A_i$ in Eq.\ref{tok_rings:eq:localised} is $\phi(i)$.
  The final specifications are (in LTL):
  \begin{equation}\label{eq:amba:final-spec}
  \boxed{
  \begin{aligned}
    \text{for $B$-processes:}
    &~~\forall{i \in \{B_1,...,B_n\}}. \Phi(i)\land\GF\token_i \impl \Psi(i)
    \text{ and } A_{loc} = \Phi,\\
    \text{for the process $A$:}
    &~~\Phi'(A) ~\impl~ \Psi'(A) \text{ and } A_{loc} = \Phi',
  \end{aligned}
  }
  \end{equation}
  where $\Phi$, $\Phi'$, $\Psi$, and $\Psi'$ are defined in Figures~\ref{fig:AMBASpecNewI}~and~\ref{fig:AMBASpecNew0}
  (Recall that $A_{loc}$ is a formula over process propositions
   such that $\G[\token \land A_{loc} \impl \F \tsnd)]$,
   see definitions on page~\pageref{page:tok_rings:defs:process_template}.)

  Let us prove cutoffs for specifications of the above form.
  The parameterized synthesis problem can be separated into two: find $(A,B)$ such that
  \begin{align*}
  \begin{aligned}
  &\forall n: \largesys ~\models~ \forall{i \in \{B_1,...,B_n\}}. \Phi(i) ~\impl~ \Psi(i) ~\land\\
  &\forall n: \largesys ~\models~ \Phi'(A) \impl \Psi'(A),
  \end{aligned}
  \end{align*}
  where $A_{loc}$ is either $\Phi$ or $\Phi'$.

  \begin{theorem}
    Given two process templates,
    $A=(\Ipr,\Opr,Q^A,Q_0^A,\delta^A,out^A,A^A_{loc})$ and
    $B=(\Ipr,\Opr,Q^B,Q_0^B,\delta^B,out^B,A^B_{loc})$,
    and let $\Iglob=\emptyset$ (no global inputs).
    Assume that initially the process $A$ has the token.
    Then a cutoff is $(1,1)$ (one $A$-process and one $B$-process) for the following PMCPs:
    \li
    \-[(1)] $\forall n: (A,B)^{(1,n)} \models A^A_{loc} \land \GF\token_A ~\impl~ \psi(A)$,
    \-[(2)] $\forall n: (A,B)^{(1,n)} \models \forall i: A^B_{loc,i} \land \GF\token_i ~\impl~ \psi(B_i)$.
    \il
    where
    $A^B_{loc,i}$ is $A^B_{loc}$ with all propositions subscripted with $i$,
    $\psi(p)$ is an LTL formula over propositions of a process $p \in \{A,B_1,...,B_k\}$.
  \end{theorem}

  \begin{proof}[Proof idea]
    The proof is inspired by the original proof~\cite{Emerso03}.

    \parbf{Item (1)}
    Fix an arbitrary $n>1$ and let $\varphi(A) = A^A_{loc}\land \GF \token_A \impl \psi(A)$.
    We prove that
    $$
    (A,B)^{(1,1)} \models \varphi(A) ~~\Iff~~ \largesys \models \varphi(A).
    $$

    Consider direction $\Rightarrow$. After contra-positioning:
    $$
    (A,B)^{(1,1)} \not\models \varphi(A) ~~\Leftarrow~~ \largesys \not\models \varphi(A).
    $$
    Given a system run of $(A,B)^{(1,1)}$ that satisfies $\neg\varphi(A)$,
    we build a system run of $(A,B)^{(1,n)}$ that satisfies $\neg\varphi(A)$.
    The construction is in Figure~\ref{amba:fig:simulation_proof_1}.
    \begin{figure}[t]\centering{
      \includegraphics[scale=0.7]{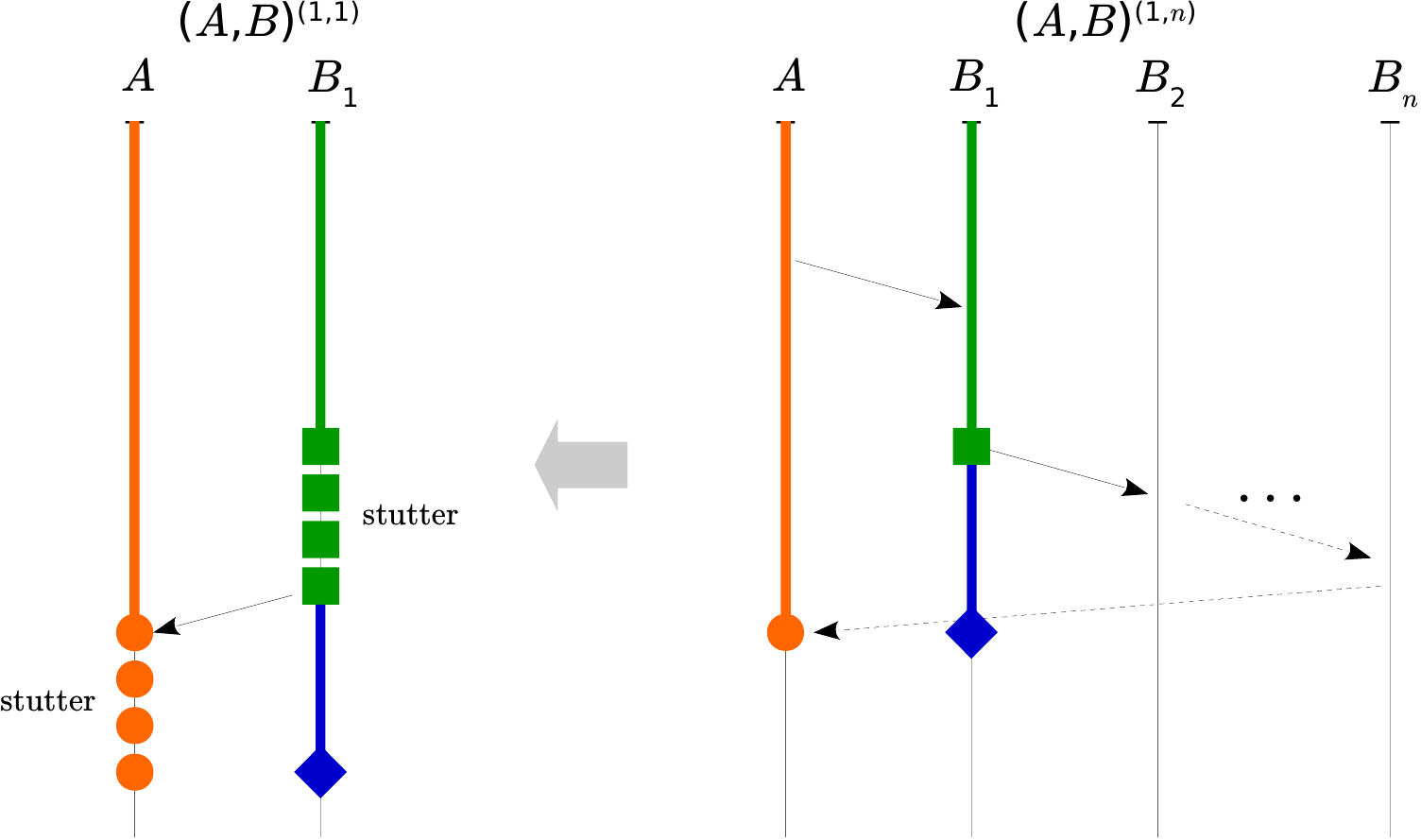}}
      \caption{Constructing a run of a cutoff system from a run of a large system.
        Vertical lines depict (local) paths of the processes,
        the horizontal lines mean the token transmission.
        The process $A$ starts with the token.%
      }
      \label{amba:fig:simulation_proof_1}
    \end{figure}
    We copy the behaviors of processes $A$ and $B_1$ until before $B_1$ sends the token.
    At this moment, we postpone sending the token by $B_1$ and stutter\footnote{
      To ``stutter a process $p$'' means ``not to schedule it''.
      As a result, a stuttered process neither reads inputs nor changes its state.
      In the figures it is shown by repeating a state.%
    }
    it,
    while the process $A$ continues execution until it gets into state \textcolor{orange}{$\CIRCLE$}
    ready to receive the token.
    Then $B_1$ transmits the token to process $A$.
    After that we move process $B_1$ into state \textcolor{blue}{$\blacklozenge$},
    while $A$ stutters in \textcolor{orange}{$\CIRCLE$}.
    Now we are in the original situation and repeat the construction.
    Since the property talks about process $A$ only, the resulting run satisfies it.
    Finally, we assumed that the processes of the large system pass the token infinitely often.
    If some process $B_x \in \{B_1,...,B_n\}$ holds the token forever,
    then we use its behavior for $B_1$ in the cutoff system
    (this may require to insert stuttering steps into behaviors of $B_1$ and $A$ of the cutoff system,
     to synchronize their (finitely many) token transmissions).

    Consider direction $\Leftarrow$. After contra-positioning:
    $$
    (A,B)^{(1,1)} \not\models \varphi(A) ~~\Rightarrow~~ \largesys \not\models \varphi(A).
    $$
    Given a system run of $(A,B)^{(1,n)}$ that satisfies $\neg\varphi(A)$,
    we build a system run of $(A,B)^{(1,1)}$ that satisfies $\neg\varphi(A)$.
    Figure~\ref{amba:fig:simulation_proof_2} shows how to construct a run of a system
    that has one more $B$ process than the cutoff system.
    By repeating the construction we can add the necessary $n-1$ $B$-processes.
    The construction works as follows.
    \begin{figure}[t]\centering{
      \includegraphics[scale=0.7]{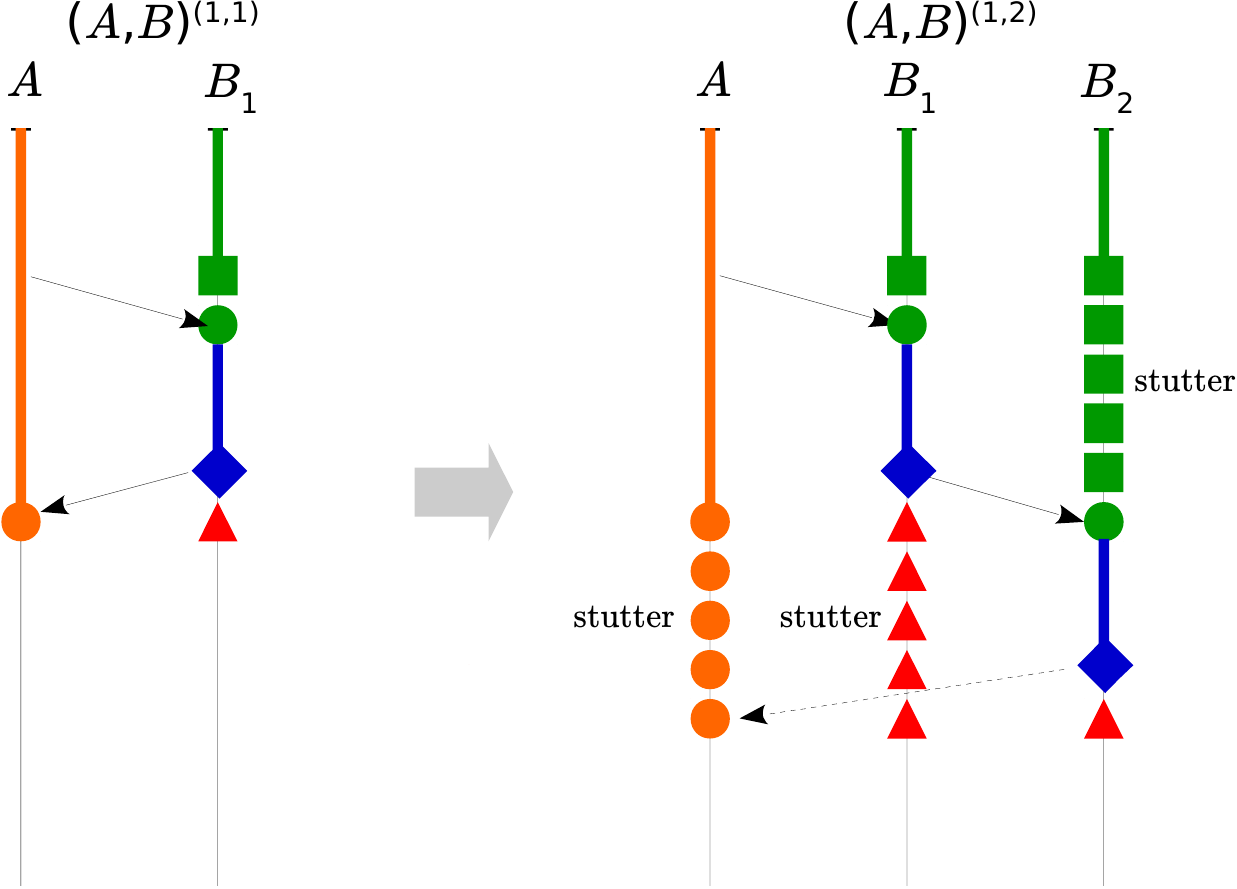}}
      \caption{Constructing a run of a system $(A,B)^{(1,2)}$ from a run of a cutoff system $(A,B)^{(1,1)}$.
        Vertical lines depict (local) paths of the processes,
        the horizontal lines mean the token transmission.
        The process $A$ starts with the token.%
      }
      \label{amba:fig:simulation_proof_2}
    \end{figure}
    The new process $B_2$ copies the behavior of $B_1$ until before $B_1$ receives the token
    (i.e., up to the state \textcolor{green}{$\rule{0.6em}{0.6em}$}).
    Then it stutters in \textcolor{green}{$\rule{0.6em}{0.6em}$} awaiting for the token from process $B_1$.
    After that it copies $B_1$ behavior from state \textcolor{green}{$\CIRCLE$} till \textcolor{blue}{$\blacklozenge$},
    while processes $A$ and $B_1$ stutter.
    Then $B_2$ sends the token to $A$ and we return to the original situation.
    Finally, the case of $B_1$ or $A$ holding the token forever is straightforward.

    \parbf{Item (2)}
    Consider the case $\forall i. \varphi(B_i)$.
    First, we use the symmetry argument: for every $n$,
    $$
    \largesys \models \forall i.\varphi(B_i) ~\Iff~ \largesys \models \varphi(B_1).
    $$
    It holds because, for every $B_i$ and system run that satisfies $\neg\varphi(B_i)$,
    we can construct a run that satisfies $\neg\varphi(B_1)$.
    The latter is possible because all $B$-processes start without the token and $\varphi$ is 1-indexed%
    \footnote{%
      In contrast,
      the symmetry argument will not work for properties of the form $\forall i. \varphi(A,B_i)$,
      because $B_1$, $B_n$, and $B_{x \in \{2,...,n-1\}}$ have different ``relation'' to $A$.
      For example, take the formula $\forall i. \G(\token_A \impl \token_A \W \token_i)$.
      The (wrongly applied) symmetry argument would produce $\G(\token_A \impl \token_A \W \token_1)$,
      which says that the token moves from $A$ to $B_1$ (trivially true in every system),
      but the original formula does not hold.%
    }.

    After applying the symmetry argument,
    we can use the very same constructions as in item (1),
    see Figures~\ref{amba:fig:simulation_proof_1}~and~\ref{amba:fig:simulation_proof_2}.
    Let us only note the case when the token is stuck in some process.
    As for the construction in Figure~\ref{amba:fig:simulation_proof_2},
    this is simple: the token will be stuck in $A$ or in $B_1$ in the large system too.
    Consider the case in Figure~\ref{amba:fig:simulation_proof_1},
    when the token gets stuck in some process $B_d$ for $d \neq 1$.
    This is the only place in the proof where we use the peculiar structure of the formula to verify:
    $A^B_{loc,1} \land \GF\token_1 \impl \psi(B_1)$.
    Recall that the contra-position negates it and gives 
    $A^B_{loc,1} \land \GF\token_1 \land \neg \psi(B_1)$.
    Thus, in the large system the process $B_1$ receives the token infinitely often,
    and we can simply ignore the case\footnote{%
      We \emph{did} consider the case in other proof branches,
      to avoid relying on the peculiarity of the formula.
      We conjecture that in the case of (more general) properties of the form $\forall i. \varphi(B_i)$ (without $\GF\token_i$),
      the cutoff increases to $(1,2)$.%
    }.
  \end{proof}

  Let us note that without the assumption ``$A$ starts with token'' the constructions break.
  We conjecture that in this case a cutoff increases to $(1,2)$.

\subsection{Experiments} \label{amba:sec:experiments}

  In this section,
  we describe optimizations that are crucial for the synthesis of the parameterized AMBA,
  and present synthesis timings and resulting implementations.
  Most of the optimizations were already described in Section~\ref{tok_rings:sec:optimizations}.
  One interesting and not previously described optimization is ``Decompositional synthesis'',
  where the specification is synthesized incrementally, starting from a subset of the properties.
  It is this optimization that allowed us to synthesize the AMBA.

\parbf{Prototype}
  Our prototype is based on our tool \textsc{Party}~\cite{party},
  a synthesizer of parameterized token rings.
  \textsc{Party} is written in Python,
  uses LTL3BA~\cite{LTL3BA} for automata translation and Z3~\cite{Moura08} for SMT solving.
  The prototype and specification files can be found at \small{\url{https://github.com/5nizza/Party/}} (branch `amba-gr1').
  The experiments were run on a x86\_64 machine with $2.6$GHz CPU, $12$GB RAM, Ubuntu OS.

\parbf{Synchronous hub abstraction (Section~\ref{tok_rings:sec:optimizations})}
  Synchronous hub abstraction can be applied to 1-indexed specifications.
  It lets the environment simulate all but one process, and always schedules this process.
  Thus, the synthesizer searches for a process template in the synchronous setting
  with additional assumptions on the environment, namely: 
  (i) the environment sends the token to the process infinitely often, and
  (ii) the environment never sends the token to the process if it already has it.
  The synchronous hub abstraction is \emph{sound and complete} for 1-indexed properties.
  After applying this optimization \emph{any} monolithic synthesis method
  can be applied to the resulting specification
  (in the form of Eq.\ref{tok_rings:eq:hub-abstraction} on page~\pageref{tok_rings:eq:hub-abstraction}).

\parbf{Hardcoding states with and without the token~\cite[Section4]{Khalimov13}}
  The number of states with and without the token in a process template
  defines the degree of the parallelism in a token ring.
  Parallelism increases with the number of states that do not have the token.
  In the AMBA case study, any grants related action depends on having the token.
  Thus we divide the states in the process template:
  (a) one state does not have the token, while
  (b) all other states have the token.
  We do this by hardcoding the $\tok$ output function.

\parbf{Decompositional synthesis of different grant schemes}
  The idea of the decompositional synthesis is:
  synthesize a subset of the properties,
  then synthesize a larger subset using the model from the previous step as the basis.
  %   As will be noted in the experiments section, all previous optimizations still do not allow us to synthesize AMBA.
  %   It is a successive synthesis approach of properties together with hardcoding token states (next paragraph) that have made it.
  Consider an example of the synthesis of the non-0-process of AMBA.
  The flow is: 
  \begin{enumerate}
  \- Assume that every request is a locked request of type \hburstfour, i.e., 
     add the assumption $\always(\hlocki\land\hburst=\hburstfour)$ to the specification.
     This implicitly removes guarantee G2 and assumption A1 from the specification.
     Synthesize the model. The resulting model has $10$ states (states $t0,..,t9$ and transitions between them in Figure~\ref{amba:fig:ith-model}).

  \- Use the model found in the previous step as the basis:
     assert the number of states, values of output functions in these states, transitions for inputs that satisfy the previous assumption.
     Transitions for inputs that violate the assumption from step 1 are not asserted, and thus are left to be synthesized.

     Now relax the assumptions: allow locked and non-locked \hburstfour\ requests, i.e., replace the previous assumption with $\always(\hburst=\hburstfour)$.
     Again, this implicitly removes G2 and A1.
     In contrast to the last step, now guarantee G3 is not necessarily `activated' if there is a request.

     Synthesize the model.
     This may require increasing the number of states (and it does, in the case of non-0 process)---add new states and keep assertions on all the previous states.

  \- Assert the transitions of the model found, as in the previous step.\\
     Remove all added assumptions and consider the original specification.
     Synthesize the final model.
\end{enumerate}

Although for AMBA this approach was successful, it is not clear how general it is.
For example, it does not work if we start with locked \hburstfour\ and \hready\ always high,
and then try to relax it.
Also, the separation into sets of properties to be synthesized was done manually.

\begin{figure}[t]\centering{
\includegraphics[width=\textwidth]{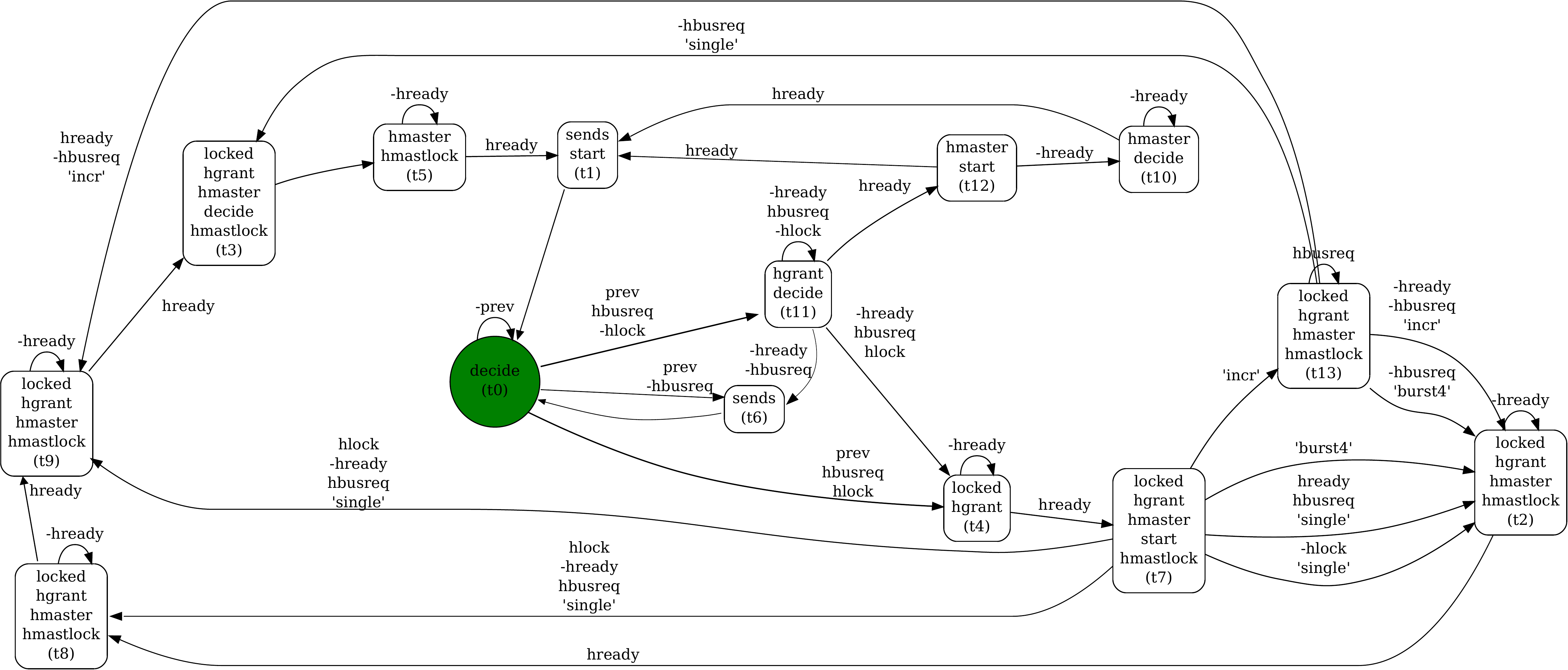}}
\caption{Synthesized model of non-$0$-processes (after manual simplification).
Circle green state ($t0$) is without the token, other states are with the token.
The initial state is $t0$.
States are labeled with their active outputs.
Edges are labeled with inputs, a missing input variable means ``don't care''.
`Burst4' means $\hburst=\hburstfour$, `incr' means $\hburst=\hincr$, `single' means neither of them.
In the first step of decompositional synthesis states $t0,..,t9$ were synthesized, in the second $t10,..,t12$ were added, in the final step state $t13$ was added.}
\label{amba:fig:ith-model}
\end{figure}
\begin{figure}[t]\centering{
\includegraphics[width=\textwidth]{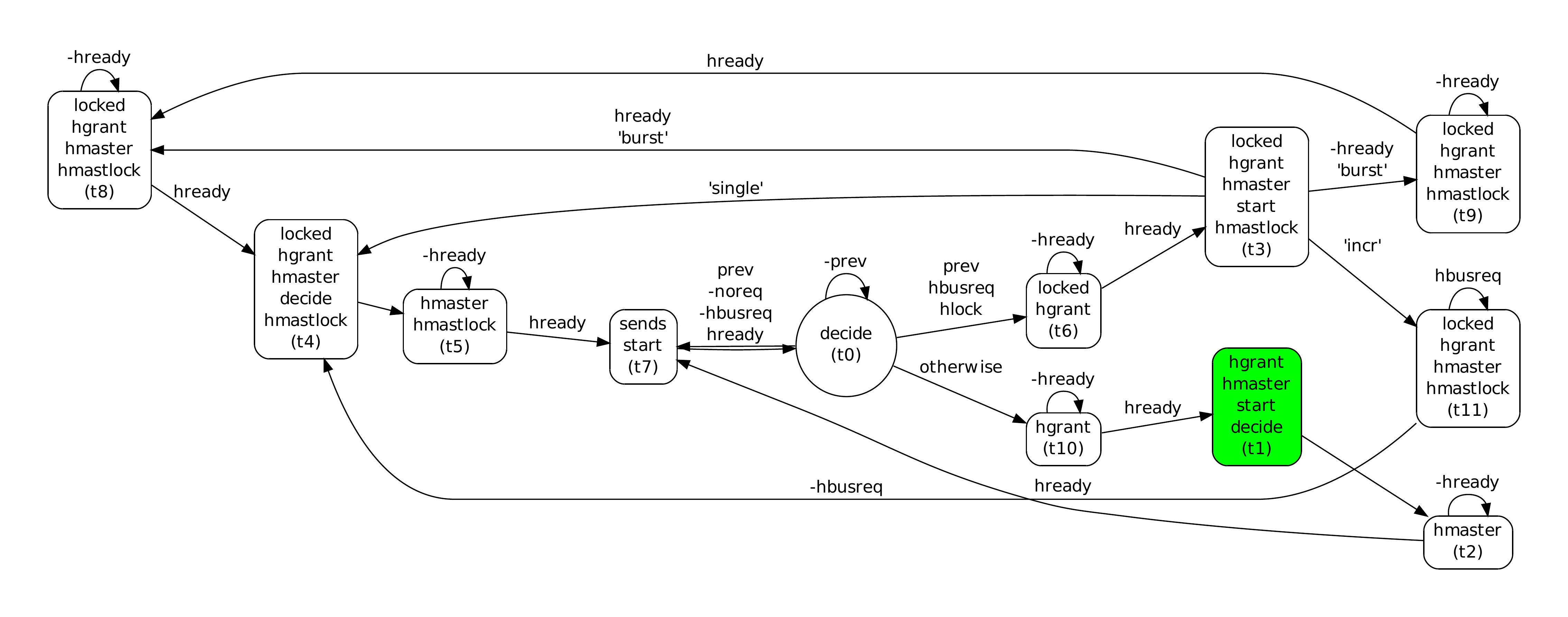}}
\caption{Synthesized model of $0$-processes (after manual simplification).
Circle state ($t0$) is without the token, other states are with the token.
The initial state is $t1$.
States are labeled with their active outputs.
Edges are labeled with inputs, a missing input variable means ``don't care''.
`Burst' means $\hburst=\hburstthree$ (recall we decreased the length of bursts for 0 process), `incr' means $\hburst=\hincr$, `single' means neither of them.
In the first step of decompositional synthesis states $t0,..,t10$ were synthesized, in the second only transitions were synthesized, but no new states added, in the final step $t11$ was added.}
\label{amba:fig:zero-model}
\end{figure}

\parbf{Results}
Synthesis times are in Tables~\ref{amba:tab:non-zero-process}~and~\ref{amba:tab:zero-process}, 
the model synthesized for non-0 process is in Figure~\ref{amba:fig:ith-model}.
The table has timings for the case when all optimizations described in this section are enabled --- it was not our goal to evaluate the optimizations separately, but to find a combination that works for the AMBA case study.

For the $0$-process we considered a simpler version with burst lengths reduced to 2/3 instead of the original 3/4 ticks.
With the original length, the synthesizer could not find a model within 2 hours (it hanged checking 11 state models, while the model has at least 12 states).

Without the decompositional approach,
the synthesizer could not find a model for for non-0 process of the AMBA specification within 5 hours.

\begin{table}[tb]
\footnotesize
\centering
\begin{minipage}[b]{0.45\textwidth}
\centering
\begin{tabular}{ l|cc }
Addit.\ assumptions                          & time & \#states \\
\hline
\rule{0pt}{3ex} \specialcellC{$\always \hlock$ \\ $\always \hburst=\hburstfour$} 
                       & 16min.  & 10  \\
\hline
\rule{0pt}{2ex} \specialcellC{$\always \hburst=\hburstfour$} & 13sec.  & 13  \\
\hline
\rule{0pt}{2ex} 
-- (Full Specification) & 1min.  & 14 \\
\end{tabular}
\caption{Results for non-$0$ process.\\~}
\label{amba:tab:non-zero-process}
\end{minipage}
\hspace{0.5cm}
\begin{minipage}[b]{0.45\textwidth}
\centering
\begin{tabular}{ l|cc }
Addit. assumptions & time & \#states \\
\hline
\rule{0pt}{3ex} \specialcellC{$\always \hlock$ \\ $\always \hburst=\hburstfour$} 
                       & 3h.  & 11  \\
\hline
\rule{0pt}{2ex} \specialcellC{$\always \hburst=\hburstfour$} & 1min.  & 11  \\
\hline
\rule{0pt}{2ex} 
-- (Full Specification) & 1m30s.  & 12 \\
\end{tabular}
\caption{Results for $0$-process \\ (bursts reduced: $3/4 \rightarrow 2/3$).}
\label{amba:tab:zero-process}
\end{minipage}
\end{table}

%\ak{`decide' is not faithfully simulated in token passing simulations?}

\subsection{Discussion}

We have shown that parameterized synthesis in token rings can be used to 
solve benchmark problems of significant size, in particular the well-known 
AMBA AHB specification that has been used as a synthesis benchmark for a long 
time. To achieve this goal, we slightly extended the cutoff results that 
parameterized synthesis is based on, and used a number of optimizations in 
the translation of the specification and the synthesis procedure itself to 
make the process feasible.

This is the first time that the AMBA case study, or any other realistic case,
has been solved by an automatic synthesis 
procedure for the parameterized case. However, some of the steps in the 
procedure are manual or use an ad-hoc solution for the specific problem at 
hand, like the limited extension of cutoff results for global inputs, the 
construction of suitable functions to convert local to global outputs, or the 
decompositional synthesis for different grant schemes. Generalizing and 
automating these approaches is a possible future work.

Our synthesized implementation is such 
that the size of the parallel composition grows only linearly with the number of 
components. Thus, for this case study our approach does not only solve the 
problem of increasing synthesis time for a growing number of components, but 
also the problem of implementations that need an exponential amount of memory 
in the number of components. We pay for this small amount of memory with a 
less-than-optimal reaction time, as processes have to wait for the token in 
order to grant a request. This restriction could be remedied by extending the 
parameterized synthesis approach to different system models, e.g., processes 
that coordinate by guarded transitions~\cite{EmersonK03},
or communicate via broadcast messages~\cite{EsparzaFM99}.

\section{Conclusion}\label{tok_rings:sec:conclusion}

In this chapter,
we studied the parameterized synthesis of token-ring systems from the applied perspective.
The starting point was the original approach of Bloem and Jacobs~\cite{JB14},
which could be applied only to toy specifications.
We suggested several optimizations that made
it applicable to larger ``made-up'' specifications.
Then we tackled the real-life specification, that of the AMBA bus protocol,
and suggested further optimizations.
This required us to extend the theory behind the approach.
In the end,
we synthesized a solution for the AMBA specification in the parameterized sense,
for the first time ever.

\backmatter

\bibliographystyle{plain}
\phantomsection
\addcontentsline{toc}{chapter}{Bibliography}
\bibliography{refs}

\begin{thebibliography}{10}

\bibitem{syntcomp}
{SYNTCOMP}.
\newblock \url{http://www.syntcomp.org/}, 2017.

\bibitem{Alur97}
Rajeev Alur, Thomas Henzinger, and Orna Kupferman.
\newblock Alternating-time temporal logic.
\newblock In {\em Journal of the ACM}, pages 100--109. IEEE Computer Society
  Press, 1997.

\bibitem{AJKR14}
B.~Aminof, S.~Jacobs, A.~Khalimov, and S.~Rubin.
\newblock Parameterized model checking of token-passing systems.
\newblock In {\em VMCAI}, volume 8318 of {\em LNCS}, pages 262--281. Springer,
  2014.

\bibitem{AminofKRSV14}
B.~Aminof, T.~Kotek, S.~Rubin, F.~Spegni, and H.~Veith.
\newblock Parameterized model checking of rendezvous systems.
\newblock In {\em {CONCUR}}, volume 8704 of {\em LNCS}, pages 109--124.
  Springer, 2014.

\bibitem{AMBAspec}
{ARM Ltd.}
\newblock {AMBA} specification (rev.2).
\newblock Available from www.arm.com, 1999.

\bibitem{SimonThesis}
S.~Au{\ss}erlechner.
\newblock {Parameterized Synthesis of Guarded Systems (Master Thesis)}.
\newblock {\em TU Graz Library}, May 2015.
\newblock Available at
  {\small\url{https://diglib.tugraz.at/download.php?id=576a77d1edae0&location=browse}}.

\bibitem{AJK16}
Simon Au{\ss}erlechner, Swen Jacobs, and Ayrat Khalimov.
\newblock Tight cutoffs for guarded protocols with fairness.
\newblock In Barbara Jobstmann and K.~Rustan~M. Leino, editors, {\em VMCAI},
  volume 9583 of {\em LNCS}, pages 476--494. Springer, 2016.

\bibitem{LTL3BA}
Tom{\'a}s Babiak, Mojm\'{\i}r Kret\'{\i}nsk{\'y}, Vojtech Reh{\'a}k, and Jan
  Strejcek.
\newblock {LTL} to {B\"u}chi automata translation: Fast and more deterministic.
\newblock In {\em TACAS}, volume 7214 of {\em LNCS}, pages 95--109. Springer,
  2012.

\bibitem{PrinciplesMC}
Christel Baier and Joost-Pieter Katoen.
\newblock {\em Principles of model checking}, volume 26202649.
\newblock MIT press Cambridge, 2008.

\bibitem{bounded-pctl}
Nathalie Bertrand, John Fearnley, and Sven Schewe.
\newblock {Bounded Satisfiability for PCTL}.
\newblock In Patrick C{\'e}gielski and Arnaud Durand, editors, {\em CSL},
  volume~16 of {\em LIPICS}, pages 92--106, Dagstuhl, Germany, 2012. Schloss
  Dagstuhl--Leibniz-Zentrum fuer Informatik.

\bibitem{RybHornSynth}
Tewodros Beyene, Swarat Chaudhuri, Corneliu Popeea, and Andrey Rybalchenko.
\newblock A constraint-based approach to solving games on infinite graphs.
\newblock {\em SIGPLAN Not.}, 49(1):221--233, January 2014.

\bibitem{BeDi}
Tewodros~Awgichew Beyene.
\newblock {\em Temporal Program Verification and Synthesis as Horn Constraints
  Solving}.
\newblock {PhD} dissertation, Technical University of Munich, 2015.

\bibitem{SMT}
A.~Biere, A.~Biere, M.~Heule, H.~van Maaren, and T.~Walsh.
\newblock {\em Handbook of Satisfiability: Volume 185 Frontiers in Artificial
  Intelligence and Applications}.
\newblock IOS Press, Amsterdam, The Netherlands, The Netherlands, 2009.

\bibitem{bjornerhorn}
Nikolaj Bj{\o}rner, Arie Gurfinkel, Ken McMillan, and Andrey Rybalchenko.
\newblock Horn clause solvers for program verification.

\bibitem{Bloem12}
R.~Bloem, B.~Jobstmann, N.~Piterman, A.~Pnueli, and Y.~Sa'ar.
\newblock Synthesis of reactive(1) designs.
\newblock {\em J. Comput. Syst. Sci.}, 78(3):911--938, 2012.

\bibitem{Nico}
Roderick Bloem, Nicolas Braud-Santoni, and Swen Jacobs.
\newblock Synthesis of self-stabilising and byzantine-resilient distributed
  systems.
\newblock In Swarat Chaudhuri and Azadeh Farzan, editors, {\em CAV}, volume
  9779 of {\em Lecture Notes in Computer Science}, pages 157--176. Springer,
  2016.

\bibitem{Bloem2015}
Roderick Bloem, Krishnendu Chatterjee, Swen Jacobs, and Robert K{\"o}nighofer.
\newblock Assume-guarantee synthesis for concurrent reactive programs with
  partial information.
\newblock In Christel Baier and Cesare Tinelli, editors, {\em TACAS}, pages
  517--532, Berlin, Heidelberg, 2015. Springer Berlin Heidelberg.

\bibitem{VacSy}
Roderick Bloem, Hana Chockler, Masoud Ebrahimi, and Ofer Strichman.
\newblock Synthesizing non-vacuous systems.
\newblock In Ahmed Bouajjani and David Monniaux, editors, {\em VMCAI}, pages
  55--72, Cham, 2017. Springer International Publishing.

\bibitem{Bloem10c}
Roderick Bloem, Alessandro Cimatti, Karin Greimel, Georg Hofferek, Robert
  K{\"o}nighofer, Marco Roveri, Viktor Schuppan, and Richard Seeber.
\newblock Ratsy - a new requirements analysis tool with synthesis.
\newblock In {\em CAV}, volume 6174 of {\em LNCS}, pages 425--429. Springer,
  2010.

\bibitem{BJK14}
Roderick Bloem, Swen Jacobs, and Ayrat Khalimov.
\newblock Parameterized synthesis case study: {AMBA} {AHB}.
\newblock In {\em SYNT}, volume 157 of {\em {EPTCS}}, pages 68--83, 2014.

\bibitem{BloemETAL15}
Roderick Bloem, Swen Jacobs, Ayrat Khalimov, Igor Konnov, Sasha Rubin, Helmut
  Veith, and Josef Widder.
\newblock {\em Decidability of Parameterized Verification}.
\newblock Synthesis Lectures on Distributed Computing Theory. Morgan {\&}
  Claypool Publishers, September 2015.
\newblock 170 pages.

\bibitem{DBLP:journals/sigact/BloemJKKRVW16}
Roderick Bloem, Swen Jacobs, Ayrat Khalimov, Igor Konnov, Sasha Rubin, Helmut
  Veith, and Josef Widder.
\newblock Decidability in parameterized verification.
\newblock {\em {SIGACT} News}, 47(2):53--64, 2016.

\bibitem{CTLsynt-via-LTLsynt}
Roderick Bloem, Sven Schewe, and Ayrat Khalimov.
\newblock {CTL*} synthesis via {LTL} synthesis.
\newblock In {\em SYNT Workshop}. EPTCS, 2017.

\bibitem{Bouajjani08}
A.~Bouajjani, P.~Habermehl, and T.~Vojnar.
\newblock Verification of parametric concurrent systems with prioritised {FIFO}
  resource management.
\newblock {\em Formal Methods in System Design}, 32(2):129--172, 2008.

\bibitem{BL69}
J.~R. B{\"u}chi and L.~H. Landweber.
\newblock Solving sequential conditions by finite-state strategies.
\newblock {\em Transactions of the American Mathematical Society},
  138:295--311, 1969.

\bibitem{DBLP:conf/stoc/CaludeJKL017}
Cristian~S. Calude, Sanjay Jain, Bakhadyr Khoussainov, Wei Li, and Frank
  Stephan.
\newblock Deciding parity games in quasipolynomial time.
\newblock In Hamed Hatami, Pierre McKenzie, and Valerie King, editors, {\em
  Proceedings of the 49th Annual {ACM} {SIGACT} Symposium on Theory of
  Computing, {STOC} 2017, Montreal, QC, Canada, June 19-23, 2017}, pages
  252--263. {ACM}, 2017.

\bibitem{Chandra:1981:ALT:322234.322243}
Ashok~K. Chandra, Dexter~C. Kozen, and Larry~J. Stockmeyer.
\newblock Alternation.
\newblock {\em J. ACM}, 28(1):114--133, January 1981.

\bibitem{Church63}
Alonzo Church.
\newblock Logic, arithmetic, and automata.
\newblock In {\em International Congress of Mathematicians (Stockholm, 1962)},
  pages 23--35. Institute Mittag-Leffler, Djursholm, 1963.

\bibitem{Clarke08}
E.~M. Clarke, M.~Talapur, and H.~Veith.
\newblock Proving ptolemy right: The environment abstraction framework for
  model checking concurrent systems.
\newblock In {\em TACAS}, volume 4963 of {\em LNCS}, pages 33--47. Springer,
  2008.

\bibitem{Clarke04c}
E.~M. Clarke, M.~Talupur, T.~Touili, and H.~Veith.
\newblock Verification by network decomposition.
\newblock In {\em CONCUR}, volume 3170 of {\em LNCS}, pages 276--291. Springer,
  2004.

\bibitem{ctl-origin}
Edmund~M Clarke and E~Allen Emerson.
\newblock Design and synthesis of synchronization skeletons using branching
  time temporal logic.
\newblock In {\em Workshop on Logic of Programs}, pages 52--71. Springer, 1981.

\bibitem{DBLP:journals/toplas/ClarkeES86}
Edmund~M. Clarke, E.~Allen Emerson, and A.~Prasad Sistla.
\newblock Automatic verification of finite-state concurrent systems using
  temporal logic specifications.
\newblock {\em {ACM} Trans. Program. Lang. Syst.}, 8(2):244--263, 1986.

\bibitem{de2012synthesizing}
Emanuele De~Angelis, Alberto Pettorossi, and Maurizio Proietti.
\newblock Synthesizing concurrent programs using answer set programming.
\newblock {\em Fundamenta Informaticae}, 120(3-4):205--229, 2012.

\bibitem{Moura08}
L.~De~Moura and N.~Bj{\o}rner.
\newblock Z3: An efficient {SMT} solver.
\newblock In {\em TACAS}, volume 4963 of {\em LNCS}, pages 337--340. Springer,
  2008.

\bibitem{spot}
Alexandre Duret-Lutz, Alexandre Lewkowicz, Amaury Fauchille, Thibaud Michaud,
  Etienne Renault, and Laurent Xu.
\newblock Spot 2.0 --- a framework for {LTL} and $\omega$-automata
  manipulation.
\newblock In {\em Proceedings of the 14th International Symposium on Automated
  Technology for Verification and Analysis (ATVA'16)}, volume 9938 of {\em
  Lecture Notes in Computer Science}, pages 122--129. Springer, October 2016.

\bibitem{EmersonC82}
E.~A. Emerson and E.~M. Clarke.
\newblock Using branching time temporal logic to synthesize synchronization
  skeletons.
\newblock {\em Sci. Comput. Program.}, 2(3):241--266, 1982.

\bibitem{Emerson00}
E.~A. Emerson and V.~Kahlon.
\newblock Reducing model checking of the many to the few.
\newblock In {\em CADE}, volume 1831 of {\em LNCS}, pages 236--254. Springer,
  2000.

\bibitem{EmersonK03}
E.~A. Emerson and V.~Kahlon.
\newblock Model checking guarded protocols.
\newblock In {\em LICS}, pages 361--370. IEEE Computer Society, 2003.

\bibitem{Emerso95b}
E.~A. Emerson and K.~S. Namjoshi.
\newblock Reasoning about rings.
\newblock In {\em Proc.\ Principles of Programming Languages}, pages 85--94,
  1995.

\bibitem{Emerso03}
E.~A. Emerson and K.~S. Namjoshi.
\newblock On reasoning about rings.
\newblock {\em Foundations of Computer Science}, 14:527--549, 2003.

\bibitem{ctlstar-origin}
E.~Allen Emerson and Joseph~Y. Halpern.
\newblock {`Sometimes' and `Not Never' Revisited: On Branching versus Linear
  Time Temporal Logic}.
\newblock {\em J. ACM}, 33(1):151--178, January 1986.

\bibitem{EJ99}
E.~Allen Emerson and Charanjit~S. Jutla.
\newblock The complexity of tree automata and logics of programs.
\newblock {\em SIAM J. Comput.}, 29(1):132--158, September 1999.

\bibitem{ES84}
E.~Allen Emerson and A.~Prasad Sistla.
\newblock Deciding full branching time logic.
\newblock {\em Information and Control}, 61(3):175 -- 201, 1984.

\bibitem{EsparzaFM99}
Javier Esparza, Alain Finkel, and Richard Mayr.
\newblock On the verification of broadcast protocols.
\newblock In {\em LICS}, pages 352--359. {IEEE} Computer Society, 1999.

\bibitem{Filiot11}
Emmanuel Filiot, Naiyong Jin, and Jean-Fran\c{c}ois Raskin.
\newblock Antichains and compositional algorithms for {LTL} synthesis.
\newblock {\em Form. Methods Syst. Des.}, 39(3):261--296, 2011.

\bibitem{BS}
Bernd Finkbeiner and Sven Schewe.
\newblock Bounded synthesis.
\newblock {\em {STTT}}, 15(5-6):519--539, 2013.

\bibitem{German92}
S.~M. German and A.~P. Sistla.
\newblock Reasoning about systems with many processes.
\newblock {\em J. ACM}, 39(3):675--735, 1992.

\bibitem{GodhalCH13}
Yashdeep Godhal, Krishnendu Chatterjee, and Thomas~A. Henzinger.
\newblock Synthesis of amba ahb from formal specification: a case study.
\newblock {\em STTT}, 15(5-6):585--601, 2013.

\bibitem{JB12}
S.~Jacobs and R.~Bloem.
\newblock Parameterized synthesis.
\newblock In {\em TACAS}, volume 7214 of {\em LNCS}, pages 362--376. Springer,
  2012.

\bibitem{JB14}
S.~Jacobs and R.~Bloem.
\newblock Parameterized synthesis.
\newblock {\em Logical Methods in Computer Science}, 10:1--29, 2014.

\bibitem{BarbaraThesis}
Barbara Jobstmann.
\newblock {\em Applications and Optimizations for LTL Synthesis}.
\newblock PhD thesis, Graz University of Technology, 2007.

\bibitem{jurdzinski2000small}
Marcin Jurdzi{\'n}ski.
\newblock Small progress measures for solving parity games.
\newblock In {\em Annual Symposium on Theoretical Aspects of Computer Science},
  pages 290--301. Springer, 2000.

\bibitem{KaiserKW10}
A.~Kaiser, D.~Kroening, and T.~Wahl.
\newblock Dynamic cutoff detection in parameterized concurrent programs.
\newblock In {\em CAV}, volume 6174 of {\em LNCS}, pages 645--659. Springer,
  2010.

\bibitem{Khalimov13}
A.~Khalimov, S.~Jacobs, and R.~Bloem.
\newblock Towards efficient parameterized synthesis.
\newblock In {\em VMCAI}, volume 7737 of {\em LNCS}, pages 108--127. Springer,
  2013.

\bibitem{DBLP:journals/corr/Khalimov16}
Ayrat Khalimov.
\newblock Specification format for reactive synthesis problems.
\newblock In {\em Proceedings Fourth Workshop on Synthesis, {SYNT} 2015, San
  Francisco, CA, USA, 18th July 2015.}, pages 112--119, 2015.

\bibitem{CTLstarCAV}
Ayrat Khalimov and Roderick Bloem.
\newblock Bounded synthesis for streett, rabin, and ctl{*}.
\newblock In Rupak Majumdar and Viktor Kuncak, editors, {\em Computer Aided
  Verification - 29th International Conference, {CAV} 2017, Heidelberg,
  Germany, July 24-28, 2017, Proceedings, Part {II}}, volume 10427 of {\em
  Lecture Notes in Computer Science}, pages 333--352. Springer, 2017.

\bibitem{party}
Ayrat Khalimov, Swen Jacobs, and Roderick Bloem.
\newblock Party parameterized synthesis of token rings.
\newblock In {\em Computer Aided Verification}, pages 928--933. Springer, 2013.

\bibitem{2017arXiv171204291K}
P.~{Klampfl}, R.~{Koenighofer}, R.~{Bloem}, A.~{Khalimov}, A.~{Abu-Yonis}, and
  S.~{Moran}.
\newblock {OpenSEA: Semi-Formal Methods for Soft Error Analysis}.
\newblock {\em ArXiv e-prints}, December 2017.

\bibitem{Klein10}
Uri Klein and Amir Pnueli.
\newblock Revisiting synthesis of gr(1) specifications.
\newblock In Sharon Barner, Ian~G. Harris, Daniel Kroening, and Orna Raz,
  editors, {\em Haifa Verification Conference}, volume 6504 of {\em Lecture
  Notes in Computer Science}, pages 161--181. Springer, 2010.

\bibitem{klenze2016fast}
Tobias Klenze, Sam Bayless, and Alan~J Hu.
\newblock Fast, flexible, and minimal {CTL} synthesis via {SMT}.
\newblock In {\em International Conference on Computer Aided Verification},
  pages 136--156. Springer, 2016.

\bibitem{KupfermanV05}
O.~Kupferman and M.~Y. Vardi.
\newblock Safraless decision procedures.
\newblock In {\em FOCS}, pages 531--542, 2005.

\bibitem{informatio}
Orna Kupferman and Moshe~Y. Vardi.
\newblock Church's problem revisited.
\newblock {\em Bulletin of Symbolic Logic}, 5(2):245--263, 1999.

\bibitem{ATA}
Orna Kupferman, Moshe~Y. Vardi, and Pierre Wolper.
\newblock An automata-theoretic approach to branching-time model checking.
\newblock {\em J. ACM}, 47(2):312--360, March 2000.

\bibitem{Kurshan95}
R.~P. Kurshan and K.~L. McMillan.
\newblock A structural induction theorem for processes.
\newblock {\em Inf. and Comp.}, 117(1):1--11, 1995.

\bibitem{Manna92}
Z.~Manna and A.~Pnueli.
\newblock Temporal specification and verification of reactive modules.
\newblock {\em Weizmann Institute of Science Technical Report}, 1992.

\bibitem{DBLP:conf/lop/MannaW81}
Zohar Manna and Pierre Wolper.
\newblock Synthesis of communicating processes from temporal logic
  specifications.
\newblock In Dexter Kozen, editor, {\em Logics of Programs, Workshop, Yorktown
  Heights, New York, May 1981}, volume 131 of {\em Lecture Notes in Computer
  Science}, pages 253--281. Springer, 1981.

\bibitem{MS95}
David~E Muller and Paul~E Schupp.
\newblock Simulating alternating tree automata by nondeterministic automata:
  New results and new proofs of the theorems of rabin, mcnaughton and safra.
\newblock {\em Theoretical Computer Science}, 141(1-2):69--107, 1995.

\bibitem{Piterman07}
Nir Piterman.
\newblock {From Nondeterministic B\"uchi and Streett Automata to Deterministic
  Parity Automata}.
\newblock {\em {Logical Methods in Computer Science}}, {Volume 3, Issue 3},
  August 2007.

\bibitem{Nir06}
Nir Piterman and Amir Pnueli.
\newblock Faster solutions of {Rabin} and {Streett} games.
\newblock In {\em 21th {IEEE} Symposium on Logic in Computer Science {(LICS}
  2006), 12-15 August 2006, Seattle, WA, USA, Proceedings}, pages 275--284,
  2006.

\bibitem{pnueli1977temporal}
Amir Pnueli.
\newblock The temporal logic of programs.
\newblock In {\em Foundations of Computer Science, 1977., 18th Annual Symposium
  on}, pages 46--57. IEEE, 1977.

\bibitem{DBLP:conf/popl/PnueliR89}
Amir Pnueli and Roni Rosner.
\newblock On the synthesis of a reactive module.
\newblock In {\em Conference Record of the Sixteenth Annual {ACM} Symposium on
  Principles of Programming Languages, Austin, Texas, USA, January 11-13,
  1989}, pages 179--190. {ACM} Press, 1989.

\bibitem{Pnueli1989}
Amir Pnueli and Roni Rosner.
\newblock {\em On the synthesis of an asynchronous reactive module}, pages
  652--671.
\newblock Springer Berlin Heidelberg, Berlin, Heidelberg, 1989.

\bibitem{DBLP:conf/focs/PnueliR90}
Amir Pnueli and Roni Rosner.
\newblock Distributed reactive systems are hard to synthesize.
\newblock In {\em 31st Annual Symposium on Foundations of Computer Science, St.
  Louis, Missouri, USA, October 22-24, 1990, Volume {II}}, pages 746--757.
  {IEEE} Computer Society, 1990.

\bibitem{ctlsat}
Nicola Prezza.
\newblock {CTL} (computation tree logic) {SAT} solver.

\bibitem{Rabin69}
M.O. Rabin.
\newblock {\em Automata on Infinite Objects and Church's Problem}.
\newblock Number~13 in Conference Series in Mathematics. American Mathematical
  Society, 1969.

\bibitem{RosnerThesis}
Roni Rosner.
\newblock {\em Modular synthesis of reactive systems}.
\newblock PhD thesis, PhD thesis, Weizmann Institute of Science, 1992.

\bibitem{Safra}
Shmuel Safra.
\newblock On the complexity of omega-automata.
\newblock In {\em 29th Annual Symposium on Foundations of Computer Science,
  White Plains, New York, USA, 24-26 October 1988}, pages 319--327. {IEEE}
  Computer Society, 1988.

\bibitem{Schewe/09/determinise}
Sven Schewe.
\newblock Tighter bounds for the determinisation of {B\"u}chi automata.
\newblock In {\em Proceedings of the Twelfth International Conference on
  Foundations of Software Science and Computation Structures (FoSSaCS 2009),
  22--29 March, York, England, UK}, volume 5504 of {\em Lecture Notes in
  Computer Science}, pages 167--181. Springer-Verlag, 2009.

\bibitem{Suzuki88}
I.~Suzuki.
\newblock Proving properties of a ring of finite state machines.
\newblock {\em Inf. Process. Lett.}, 28(4):213--214, 1988.

\bibitem{Vardi:1985:IUL:22145.22173}
M~Y Vardi and L~Stockmeyer.
\newblock Improved upper and lower bounds for modal logics of programs.
\newblock In {\em Proceedings of the Seventeenth Annual ACM Symposium on Theory
  of Computing}, STOC '85, pages 240--251, New York, NY, USA, 1985. ACM.

\bibitem{LTL-vs-CTL}
Moshe~Y Vardi.
\newblock Branching vs. linear time: Final showdown.
\newblock In {\em TACAS}, volume~1, pages 1--22. Springer, 2001.

\bibitem{DBLP:journals/iandc/VardiW94}
Moshe~Y. Vardi and Pierre Wolper.
\newblock Reasoning about infinite computations.
\newblock {\em Inf. Comput.}, 115(1):1--37, 1994.

\bibitem{DBLP:conf/focs/WolperVS83}
Pierre Wolper, Moshe~Y. Vardi, and A.~Prasad Sistla.
\newblock Reasoning about infinite computation paths (extended abstract).
\newblock In {\em 24th Annual Symposium on Foundations of Computer Science,
  Tucson, Arizona, USA, 7-9 November 1983}, pages 185--194. {IEEE} Computer
  Society, 1983.

\end{thebibliography}

\end{document}